\documentclass[12pt]{article}

\usepackage{fontspec}

\renewcommand\arraystretch{0.8}

\usepackage{graphicx}
\usepackage{amssymb}
\usepackage{amsmath}
\usepackage{amsthm}
\usepackage{geometry}
\numberwithin{equation}{section}
\usepackage{natbib}
\usepackage{fullpage}
\usepackage[colorlinks=true,linkcolor=blue,citecolor=blue,urlcolor=blue]{hyperref}
\usepackage{booktabs}
\usepackage{graphicx}     
\usepackage{subcaption}   
\usepackage{multirow}
\usepackage{siunitx}
\usepackage[table]{xcolor}
\usepackage{bm}
\usepackage{enumerate}
\usepackage{hyperref}
\hypersetup{
    pdfencoding=auto,  
    psdextra,          
}

\usepackage[T1]{fontenc}
\usepackage{fontspec} 

\newtheorem{theorem}{Theorem}[section]
\newtheorem{lemma}[theorem]{Lemma}

\theoremstyle{definition}

\newcommand{\Ind}{\mathcal{I}} 
\newcommand{\Jnd}{\mathcal{J}}
\newcommand{\Esp}{\mathbb{E}}
\newcommand{\bX}{\boldsymbol{X}}
\newcommand{\bZ}{\boldsymbol{Z}}
\newcommand{\bW}{\boldsymbol{W}}
\newcommand{\bV}{\boldsymbol{V}}
\newcommand{\bPsi}{\bm{\hat{\psi}}}
\newcommand{\trans}{T}

\theoremstyle{remark}

\DeclareMathOperator{\tr}{\mathop{tr}}
\DeclareMathOperator{\Var}{\mathop{Var}}

\DeclareMathOperator{\argmin}{\mathop{\arg\min}}

\allowdisplaybreaks[4]
\title{Adaptive Test for High Dimensional Quantile Regression}
\author{Ping Zhao$^1$, Zhenyu Liu$^2$ and Dan Zhuang$^3$\\Tiangong University$^1$, Nankai University$^2$, Fujian Normal University$^3$}
\date{\today}

\begin{document}
	
\maketitle

\begin{abstract}
Testing high-dimensional quantile regression coefficients is crucial, as tail quantiles often reveal more than the mean in many practical applications.
Nevertheless, the sparsity pattern of the alternative hypothesis is typically unknown in practice, posing a major challenge. To address this, we propose an adaptive test that remains powerful across both sparse and dense alternatives.
We first establish the asymptotic independence between the max-type test statistic proposed by \citet{tang2022conditional} and the sum-type test statistic introduced by \citet{chen2024hypothesis}. Building on this result, we propose a Cauchy combination test that effectively integrates the strengths of both statistics and achieves robust performance across a wide range of sparsity levels. Simulation studies and real data applications demonstrate that our proposed procedure outperforms existing methods in terms of both size control and power.
\end{abstract}

\section{Introduction}

Testing regression coefficients in linear models is one of the fundamental problems in statistics. 
The classical F-test provides an efficient and powerful solution under fixed dimension and sufficiently large sample size.
However, in modern high-dimensional settings where the number of covariates exceeds the sample size, the sample covariance matrix becomes singular, making the traditional F-test infeasible. To address this challenge, \citet{zhong2011tests} proposed a modification that replaces the sample covariance matrix in the F-statistic with the identity matrix, thereby enabling valid inference under high dimensionality. This idea was further extended by \citet{CuiGuoZhong2018} and \citet{ZhangYaoShao2018}, who developed improved procedures with enhanced power and theoretical guarantees. 

Despite their success, these methods rely on the least-squares framework, which is sensitive to heavy-tailed errors and outliers—common features in real-world data. Quantile regression \citep{KoenkerBassett1978} offers a robust alternative by modeling conditional quantiles of the response, thereby capturing heterogeneous effects across the distribution and maintaining robustness to outliers and non-Gaussian noise. Motivated by these advantages, high-dimensional inference based on quantile regression has attracted increasing attention in recent years, leading to active research on hypothesis testing, variable selection, and confidence interval construction under flexible error distributions. 

For the fixed-dimensional setting, numerous test procedures based on quantile regression have been developed (see, e.g., \cite{KoenkerMachado1999,Koenker2000,kocherginsky2005practical,WangHe2007,WangZhuZhou2009,feng2011wild,wang2018wild,wang2018testing}). These include Wald-type test, quantile likelihood ratio test, rank score test, and various resampling-based approaches. Such methods have been shown to be powerful and asymptotically valid under classical low-dimensional regimes. However, when the dimensionality of the covariates exceeds the sample size, the performance of these traditional procedures deteriorates substantially due to the challenges in parameter estimation and the breakdown of asymptotic approximations.

To address these challenges in the high-dimensional context, \cite{tang2022conditional} proposed a \textit{conditional marginal score-type test}, which adopts a max-type statistic to detect sparse alternatives effectively. This test is particularly powerful when only a small subset of the regression coefficients deviate from the null hypothesis. In contrast, \cite{chen2024hypothesis} developed a \textit{U-statistic--based score test}, which aggregates information across all coordinates and thus achieves higher power under dense alternatives, where many small effects are present simultaneously.

Nevertheless, in practical applications, the underlying sparsity structure of the alternative hypothesis is typically unknown. Relying solely on a sparse-oriented or dense-oriented procedure may therefore lead to substantial power loss. This motivates the development of {adaptive testing procedures} that can automatically adjust to the unknown sparsity level. Such adaptive methods aim to combine the advantages of both max-type and sum-type tests, maintaining robust performance across a wide spectrum of alternatives---from extremely sparse to fully dense scenarios. 

Recently, a growing body of literature has focused on combining sum-type and max-type test procedures to achieve adaptivity in various high-dimensional testing problems (see, e.g., \cite{feng2022maxsum,wang2023change,feng2022alpha,wang2024rank,chen2024hdreg,ma2024alpha,feng2024asyIndep,feng2024whitenoise,wang2024new,wang2024fisher,liu2025spatial}). The central idea underlying these studies is to first establish the asymptotic independence between the sum-type and max-type test statistics, and then to combine their respective $p$-values through a suitable aggregation method, such as the Cauchy combination test, Fisher’s method, or the minimum $p$-value approach. These hybrid strategies effectively leverage the strengths of both tests, leading to procedures that maintain good power against both sparse and dense alternatives.

Motivated by this line of research, we extend the idea to the high-dimensional quantile regression framework. Specifically, we show that the sum-type test statistic proposed by \cite{chen2024hdreg} is asymptotically independent of the max-type test statistic developed by \cite{tang2022conditional}. Building upon this asymptotic independence property, we construct a Cauchy combination test that adaptively integrates information from both types of statistics. Simulation studies and empirical analyses demonstrate that the proposed adaptive procedure performs robustly across a wide range of sparsity structures, exhibiting superior power regardless of whether the alternative hypothesis is sparse or dense.

The remainder of this paper is organized as follows. Section~\ref{sec:method} establishes the asymptotic independence between the sum-type and max-type test statistics and introduces the proposed adaptive testing procedure. Section~\ref{sec:simulation} presents extensive simulation studies, and Section~\ref{sec:realdata} illustrates the practical performance of the proposed method through a real data application. Section~\ref{sec:discussion} provides concluding remarks and discussions. All technical proofs are deferred to the Appendix.

\section{Method}\label{sec:method}
Let $\{(Y_i, \boldsymbol{Z}_i, \boldsymbol{X}_i)\}_{i=1}^n$ be an independent and identically distributed (i.i.d.) sample from the following linear model:
\begin{align}
    Y_i = \boldsymbol{Z}_i^{\top}\boldsymbol{\alpha}_{\tau} + \boldsymbol{X}_i^{\top}\boldsymbol{\beta}_{\tau} + \varepsilon_i,
\end{align}
where $\varepsilon_i$ is independent of $(\boldsymbol{Z}_i, \boldsymbol{X}_i)$ and has density $f_{\varepsilon}$ satisfying $P(\varepsilon_i > 0) = 1 - \tau$. Here, $\boldsymbol{Z}_i \in \mathbb{R}^q$ and $\boldsymbol{X}_i \in \mathbb{R}^p$ denote the covariate vectors of dimensions $q$ and $p$, respectively. Without loss of generaltiy, similar to \cite{tang2022conditional}, we assume that $E(\boldsymbol{Z}_i)=(1,\boldsymbol{0}_{q-1})^{\top}$ 
and $E(\boldsymbol{X}_i)=\boldsymbol{0}$. Let $\boldsymbol{\alpha}_{\tau} = \left(\alpha_{1,\tau}, \ldots, \alpha_{q,\tau}\right)^{\top}$ and $\boldsymbol{\beta}_{\tau} = \left(\beta_{1,\tau}, \ldots, \beta_{p,\tau}\right)^{\top}$ be the true quantile-specific coefficient vectors associated with $\boldsymbol{Z}_i$ and $\boldsymbol{X}_i$, respectively. We are interested in testing the existence of an association between $\boldsymbol{X}_i$ and the $\tau$th conditional quantile of $Y_i$, after adjusting for the effect of $\boldsymbol{Z}_i$. Formally, we consider the hypothesis test
\begin{align}
    H_0: \boldsymbol{\beta}_{\tau} = \boldsymbol{0}
    \quad \text{versus} \quad
    H_1: \boldsymbol{\beta}_{\tau} \neq \boldsymbol{0}.
\end{align}

First, we estimate the marginal effect of $\boldsymbol{Z}$ as
$$
\widehat{\boldsymbol{\alpha}}_{\tau}=\underset{\boldsymbol{\alpha} \in \mathbb{R}^q}{\operatorname{argmin}} \sum_{i=1}^n \rho_\tau\left(Y_i-\boldsymbol{Z}_i^{\top} \boldsymbol{\alpha}\right),
$$
where $\rho_\tau(t)=t\{\tau-I(t<0)\}$ is the quantile check loss function. To test the null hypothesis $H_0: \boldsymbol{\beta}_{\tau} = \boldsymbol{0}$, we first define
\begin{align}
    \hat{\psi}_i = I\!\left(Y_i - \boldsymbol{Z}_i^{\top} \hat{\boldsymbol{\alpha}}_{\tau} \le 0 \right) - \tau,
\end{align}
where $\hat{\boldsymbol{\alpha}}_{\tau}$ denotes a consistent estimator of $\boldsymbol{\alpha}_{\tau}$ under $H_0$.  
Intuitively, $\hat{\psi}_i$ represents the quantile score function evaluated at the estimated conditional $\tau$th quantile of $Y_i$ given $\boldsymbol{Z}_i$. Based on $\{\hat{\psi}_i\}_{i=1}^n$, \citet{chen2024hdreg} proposed a sum-type test statistic to detect dense alternatives:
\begin{align}
    T_{\mathrm{SUM}}
    = \frac{2}{n(n-1)} \sum_{i \ne j} 
      \boldsymbol{X}_i^{\top} \boldsymbol{X}_j \,
      \hat{\psi}_i \hat{\psi}_j.
\end{align}
Let $\boldsymbol{\Sigma}_x = \mathrm{Cov}(\boldsymbol{X}_i)$ denote the covariance matrix of $\boldsymbol{X}_i$. 
Under some mild regularity conditions, it can be shown that
\begin{align}
    \frac{n T_{\mathrm{SUM}}}
         {\tau(1-\tau)\sqrt{2\,\mathrm{tr}(\boldsymbol{\Sigma}_x^2)}}
    \xrightarrow{d} N(0,1),
\end{align}
under the null hypothesis $H_0$. 
Hence, $T_{\mathrm{SUM}}$ provides an asymptotically valid test for dense alternatives where many components of $\boldsymbol{\beta}_{\tau}$ deviate slightly from zero. In contrast, for sparse alternatives where only a few components of $\boldsymbol{\beta}_{\tau}$ are nonzero, 
\citet{tang2022conditional} proposed a max-type test statistic defined as
\begin{align}
    T_{\mathrm{MAX}} 
    = \max_{1 \le j \le p} S_{j,\tau}^2, 
    \quad 
    S_{j,\tau}
    = \frac{1}{\sqrt{n}} 
      \sum_{i=1}^n 
      \frac{W_{ij}\hat{\psi}_i}
           {\sqrt{\tau(1-\tau)\,\| \boldsymbol{W}_{\cdot j} \|^2 / n}},
\end{align}
where $\boldsymbol{W}_{\cdot j} = (W_{1j}, \ldots, W_{nj})^{\top}$ is the $j$th column of the projected design matrix
\begin{align}
    \boldsymbol{W}
    = \left(\mathbf{I}_n - \mathbf{Z} (\mathbf{Z}^{\top} \mathbf{Z})^{-1} \mathbf{Z}^{\top}\right) \mathbf{X},
\end{align}
with $\mathbf{Z} = (\boldsymbol{Z}_1^{\top}, \ldots, \boldsymbol{Z}_n^{\top})^{\top}$ 
and $\mathbf{X} = (\boldsymbol{X}_1^{\top}, \ldots, \boldsymbol{X}_n^{\top})^{\top}$.
Under suitable regularity conditions, they established that under $H_0$,
\begin{align}
    P\!\left[
        T_{\mathrm{MAX}} 
        - 2\log(p)
        + \log\{\log(p)\} 
        \le x
        \,\middle|\, H_0
    \right]
    \to 
    \exp\!\left\{-\pi^{-1/2} \exp\!\left(-\frac{x}{2}\right)\right\},
\end{align}
which implies that the null distribution of $T_{\mathrm{MAX}}$ asymptotically follows a type-I extreme value (Gumbel) distribution.

The two test statistics $T_{\mathrm{SUM}}$ and $T_{\mathrm{MAX}}$ are designed for different alternative regimes: the sum-type statistic is powerful against dense but weak signals, whereas the max-type statistic performs better when the signal is sparse and strong. In practice, the underlying alternative structure—whether dense or sparse—is typically unknown.  A test that is powerful against one type of alternative may lose power under the other. Therefore, it is desirable to develop an {adaptive test procedure} that can automatically adjust to both scenarios and maintain high power across a broad range of alternatives.

To this end, we aim to construct an adaptive test by combining the sum-type and max-type statistics introduced above. 
A key step toward this goal is to establish the {asymptotic independence} between $T_{\mathrm{SUM}}$ and $T_{\mathrm{MAX}}$ under the null hypothesis. 
Intuitively, $T_{\mathrm{SUM}}$ aggregates information from all coordinates of $\boldsymbol{X}_i$ through an average correlation structure, while $T_{\mathrm{MAX}}$ captures the most extreme signal among individual coordinates. 
Since these two statistics summarize the data through fundamentally different mechanisms, their asymptotic dependence tends to vanish as both $n$ and $p$ grow.

We need the following assumptions:
\begin{itemize}
    \item[(A1)]  The density $f_{\varepsilon}(t)$ of $\varepsilon$ is positive and away from zero. In addition, $f_{\varepsilon}(t)$ is differentiable and has bounded derivative.
    \item[(A2)] Let $\boldsymbol{Z}=(1,\tilde{\boldsymbol{Z}}^\top)^\top \in \mathbb{R}^q,\tilde{\boldsymbol{Z}}\in \mathbb{R}^{q-1}$. We assume $U=(\tilde{\boldsymbol{Z}}^\top,\boldsymbol{X}^\top)^\top \sim N(\boldsymbol{0},\boldsymbol{\Sigma})$ where 
    \begin{align*}
    \boldsymbol{\Sigma}=\left(
    \begin{array}{cc}
    \boldsymbol{\Sigma}_z&\boldsymbol{\Sigma}_{zx}\\
    \boldsymbol{\Sigma}_{xz} &\boldsymbol{\Sigma}_{x}
    \end{array}
    \right).
    \end{align*}
\item[(A3)] The eigenvalues of $\boldsymbol{\Sigma}$ are all bounded. Let $\boldsymbol{\Sigma}_{x|z}=\boldsymbol{\Sigma}_x-\boldsymbol{\Sigma}_{xz}\boldsymbol{\Sigma}_{z}^{-1}\boldsymbol{\Sigma}_{zx}$ and $\mathbf{R}_{x|z}=(r_{ij})_{1\le i,j\le p}$ is the corresponding correlation matrix. We assume that $\max_{1\le i,j\le p}|r_{ij}|\le r_0<1$.
\item[(A4)] The dimension $q$ is fixed and $\log p=o(n^{1/4}/\log (n)^{3/4})$.
\end{itemize}

Condition (A1) coincides with condition (C1) in \cite{chen2024hypothesis} and condition (A3) in \cite{tang2022conditional}, and is commonly assumed in the literature. To meet condition (A1) in \cite{tang2022conditional} and condition (C4) in \cite{chen2024hypothesis} regarding the distributional assumption of $\boldsymbol{U}$, we impose a multivariate normal model. Extending these assumptions to more general non-Gaussian settings is an interesting direction for future research. In addition, condition (A3) guarantees that condition (A2) in \cite{tang2022conditional} and condition (C2) in \cite{chen2024hypothesis} are satisfied.

\begin{theorem}\label{th0}
    Under conditons (A1)-(A4), under the null hypothesis, we have
\begin{align*}
P\left[\frac{n T_{\mathrm{SUM}}}
         {\tau(1-\tau)\sqrt{2\,\mathrm{tr}(\boldsymbol{\Sigma}_x^2)}} \le x,T_{\mathrm{MAX}} 
        - 2\log(p)
        + \log\{\log(p)\} 
        \le y\right]\to \Phi(x)G(y),
\end{align*}
where $\Phi(x)$ is the cumulative distribution function of $N(0,1)$ and $G(y)=\left\{-\pi^{-1/2} \exp\!\left(-\frac{y}{2}\right)\right\}$.
\end{theorem}

Based on Theorem~\ref{th0}, which establishes the asymptotic independence between the sum-type and max-type test statistics, we are able to construct a combined test that leverages the strengths of both procedures. In particular, we adopt a Cauchy combination approach \citep{liu2020cauchy} to integrate the evidences from the two statistics. Other combination strategies can also be employed, such as Fisher's method  or the minimum $p$-value approach, among others.

Recall that the sum-type test is powerful against dense alternatives, while the max-type test is particularly sensitive to sparse strong signals. Since neither statistic uniformly dominates the other across all scenarios, a natural strategy is to combine their corresponding $p$-values, denoted by $p_{\mathrm{sum}}$ and $p_{\mathrm{max}}$, respectively.

Motivated by the asymptotic independence result, we consider the following Cauchy combination statistic:
\[
T_{\mathrm{CC}} 
    = \tan\Big\{ \big(0.5 - p_{\mathrm{sum}}\big)\pi \Big\}
      + 
      \tan\Big\{ \big(0.5 - p_{\mathrm{max}}\big)\pi \Big\}.
\]
Under the null hypothesis, the asymptotic independence of the two $p$-values implies that $T_{\mathrm{CC}}$ follows a standard Cauchy distribution in the limit. Consequently, the combined $p$-value can be calculated in closed form as
\[
p_{\mathrm{CC}}
    = \frac{1}{2} - \frac{1}{\pi}\arctan\!\left( T_{\mathrm{CC}} \right).
\]

The resulting test rejects the null when $p_{\mathrm{CC}} \le \alpha$. This combination procedure inherits the robustness and adaptiveness of the Cauchy method: it retains high power under dense alternatives, due to the contribution from the sum-type statistic, while simultaneously maintaining sensitivity to sparse alternatives through the max-type statistic. Therefore, the proposed Cauchy combination test provides a unified and adaptive framework capable of handling a wide spectrum of high-dimensional alternatives.

Additionally, we consider the following alternative hypothesis:
\begin{itemize}
\item[(A5)] $\|\boldsymbol{\beta}_{\tau}\|_0=o(p^{1/2})$, $\boldsymbol{\beta}_{\tau}^\top \boldsymbol{\Sigma}_x  \boldsymbol{\beta}_{\tau}=o(1)$, $\boldsymbol{\beta}_{\tau}^\top \boldsymbol{\Sigma}_x^3  \boldsymbol{\beta}_{\tau}=o(n^{-1}\tr(\boldsymbol{\Sigma}_x^2))$, $E(\boldsymbol{X}_i^\top \varsigma(\boldsymbol{X}_i))\boldsymbol{\Sigma}_xE(\varsigma(\boldsymbol{X}_i)\boldsymbol{X}_i )=o(n^{-1}\tr(\boldsymbol{\Sigma}_x^2))$ and $E(\boldsymbol{X}_i^\top \tilde{R}_i(\boldsymbol{t}))\boldsymbol{\Sigma}_xE( \tilde{R}_i(\boldsymbol{t})\boldsymbol{X}_i)=o(n^{-2}H_z(\boldsymbol{t}_1)\tr(\boldsymbol{\Sigma}_x^2))$, where $\varsigma(\boldsymbol{X}_i)=E(I(Y_i<\boldsymbol{Z}_i^\top \boldsymbol{\alpha}_\tau + \boldsymbol{X}_i ^\top\boldsymbol{\beta}_{\tau} )| \boldsymbol{X}_i)-\tau$, $h_{\tau}(\boldsymbol{\alpha}_{\tau},\boldsymbol{\beta}_{\tau})=I(Y-\boldsymbol{Z}^\top \boldsymbol{\alpha}_{\tau}-\boldsymbol{X}^\top \boldsymbol{\beta}_{\tau}>0)-\tau$, $\tilde{R}_i(\boldsymbol{t})=h_\tau(\boldsymbol{\alpha}_{\tau}+n^{-1}\boldsymbol{t}_1,\boldsymbol{\beta}_{\tau}+n^{-1}\boldsymbol{t}_2)-h_{\tau}(\boldsymbol{\alpha}_{\tau},\boldsymbol{\beta}_{\tau})$, $H_z(\boldsymbol{t}_1)=\boldsymbol{t}_1^\top \boldsymbol{\Sigma}_z\boldsymbol{t}_1$, $\boldsymbol{t}=(\boldsymbol{t}_1^\top,\boldsymbol{t}_2^\top)\in \mathbb{R}^{p+q}$, $\boldsymbol{t}_1\in \mathbb{R}^q, \boldsymbol{t}_2\in \mathbb{R}^p$.
\end{itemize}
Under the above alternative hypothesis, the following theorem further establishes the asymptotic independence of the two test statistics.

\begin{theorem}\label{th1}
    Under conditons (A1)-(A5), we have
\begin{align*}
&P\left[\frac{n T_{\mathrm{SUM}}}
         {\tau(1-\tau)\sqrt{2\mathrm{tr}(\boldsymbol{\Sigma}_x^2)}} \le x,T_{\mathrm{MAX}} 
        - 2\log(p)
        + \log\{\log(p)\} 
        \le y\right]\to \\
    &P\left[\frac{n T_{\mathrm{SUM}}}
         {\tau(1-\tau)\sqrt{2\mathrm{tr}(\boldsymbol{\Sigma}_x^2)}} \le x\right]P\left[T_{\mathrm{MAX}} 
        - 2\log(p)
        + \log\{\log(p)\} 
        \le y\right].
\end{align*}
\end{theorem}

\cite{li2023} show that the Cauchy combination test can achieve substantially higher power than the minimal $p$-value method, that is, the test based on $\min\{p_{max},p_{{sum}}\}$ (referred to as the minimal $p$-value combination). Denote the corresponding power function by 
\[
\beta_{max\wedge {sum},\alpha}
    = P\left(\min\{{p}_{max},{p}_{sum}\}\leq 1-\sqrt{1-\alpha}\right).
\]
It is straightforward to see that
\begin{align}\label{power_H1}
\beta_{{max}\wedge {sum}, \alpha} 
&\ge P(\min\{{p}_{max},{p}_{sum}\}\leq \alpha/2)\nonumber\\
&= \beta_{{max},\alpha/2}+\beta_{{sum},\alpha/2}
   -P({p}_{max}\leq \alpha/2,\,{p}_{sum}\leq \alpha/2)\nonumber\\
&\ge \max\{\beta_{{max},\alpha/2},\,\beta_{{sum},\alpha/2}\}.
\end{align}
Moreover, under $H_1$ in condition (A5), the asymptotic independence established in Theorem~\ref{th1} yields
\begin{align}\label{power_H1np}
\beta_{{max}\wedge {sum}, \alpha}
    \ge \beta_{{max},\alpha/2}+\beta_{{sum},\alpha/2}
      -\beta_{{max},\alpha/2}\beta_{{sum},\alpha/2} + o(1).
\end{align}

For small $\alpha$, the values of $\beta_{M,\alpha/2}$ and $\beta_{S,\alpha/2}$ are close to $\beta_{M,\alpha}$ and $\beta_{S,\alpha}$, respectively. Consequently, combining \eqref{power_H1} and \eqref{power_H1np} implies that the minimal $p$-value combination test retains at least the power of the better-performing individual test, and often improves upon both when the two tests capture complementary features of the alternative.

\section{Simulation}\label{sec:simulation}

We conduct extensive simulation experiments to examine the finite-sample performance of the proposed Cauchy-type combination test $T_{\mathrm{CC}}$. 
The proposed test is compared with the sum-type test $T_{\mathrm{SUM}}$ \citep{chen2024hypothesis} and the max-type test $T_{\mathrm{MAX}}$ \citep{tang2022conditional}.

Covariates are generated from the independent component model
\begin{align*}
\boldsymbol{U} = \boldsymbol{\Sigma}^{1/2} \boldsymbol{u},
\end{align*}
where $\boldsymbol{u}=(u_1,\ldots,u_{p+q-1})^\top$ consists of i.i.d.\ entries drawn from one of the following centered distributions:  
(i) Normal, (ii) Laplace, (iii) Logistic, and (iv) the $t_2$ distribution.  
This setup allows us to evaluate robustness under light- to heavy-tailed distributions.

We consider three covariance structures for $\boldsymbol{\Sigma}$:
\begin{itemize}
    \item[(I)] $\boldsymbol{\Sigma}=\mathbf{I}_{p+q-1}$;
    \item[(II)] $\boldsymbol{\Sigma}_{ij}=0.5^{|i-j|}$ for $1\le i,j\le p+q-1$;
    \item[(III)] $\boldsymbol{\Sigma}=\mathbf{I}_{p+q-1}+ \boldsymbol{b}\boldsymbol{b}^\top - \mathbf{B}$,  
    where $\mathbf{B}=\operatorname{diag}(b_1^2,\dots,b_l^2)$ and the first $\lfloor p^{0.3}\rfloor$ components of $\boldsymbol{b}$ are independently sampled from $\operatorname{Uniform}(0.7,0.9)$ while the remaining components are zero.
\end{itemize}


We consider three quantile levels, $\tau \in \{0.25, 0.5, 0.75\}$, with sample sizes $n \in \{100,150\}$ and dimensions $p \in \{120,240\}$.
Tables~\ref{tab1}--\ref{tab3} summarize the empirical size performance across all experimental settings, including Cases~1--3, varying $(n,p)$ combinations, and multiple underlying distributions. We observe that most of the proposed tests are able to control the empirical size well. In particular, across different quantile levels and a wide range of underlying distributions, the empirical type I error rates remain close to the nominal level, indicating stable finite-sample performance.

For power evaluation, we generate data under $H_1$ by setting $\boldsymbol{\beta}=(\beta_1,\dots,\beta_q)^\top$ to be a sparse vector. The first $s$ entries of $\beta$ are nonzero and satisfy
\[
\|\boldsymbol{\beta}\|^2 = 0.5, 
\qquad \beta_i \sim \mathcal{N}(0,1), \quad i=1,\dots,s.
\]

\begin{table}[htbp] 
\centering 
\caption{Empirical size across different $(n,p)$ configurations and methods for Cases~1--3 ($\tau = 0.5$, $s = 9$; 2000 replications).} 
\label{tab1}
\scriptsize \setlength{\tabcolsep}{3.5pt} \renewcommand{\arraystretch}{1.2} 
\begin{tabular}{c c c | ccc ccc ccc ccc} 
\toprule Case & $p$ & $n$ 
& \multicolumn{3}{c}{Normal} 
& \multicolumn{3}{c}{Laplace} 
& \multicolumn{3}{c}{Logistic} 
& \multicolumn{3}{c}{$t_2$} \\ 
\cmidrule(lr){4-6} \cmidrule(lr){7-9} \cmidrule(lr){10-12} \cmidrule(lr){13-15} 
& &
& $T_{\mathrm{CC}}$ & $T_{\mathrm{MAX}}$ & $T_{\mathrm{SUM}}$ 
& $T_{\mathrm{CC}}$ & $T_{\mathrm{MAX}}$ & $T_{\mathrm{SUM}}$
& $T_{\mathrm{CC}}$ & $T_{\mathrm{MAX}}$ & $T_{\mathrm{SUM}}$
& $T_{\mathrm{CC}}$ & $T_{\mathrm{MAX}}$ & $T_{\mathrm{SUM}}$ \\
\midrule
\multirow{4}{*}{Case 1}
& \multirow{2}{*}{120} & $100$   & 5.75 & 5.20 & 5.45 &4.05  & 3.15 & 4.70 & 5.00 & 3.40 & 5.40 & 3.90 &1.10  & 5.10 \\
& & $150$   & 5.80 & 3.95 & 5.90 & 5.40 & 4.00 &5.65  &  4.85 & 2.90 & 5.50  & 3.35 &0.09 & 5.25 \\
\cmidrule(lr){2-15}
& \multirow{2}{*}{240} & $100$   & 5.35 & 4.35 & 5.75 & 5.75 & 3.75 & 6.25 & 5.55 &4.60  & 5.15 & 2.90 & 0.80 & 5.35 \\
& & $150$   &5.95  & 4.55 & 5.75 & 5.50 & 3.85 & 5.40 & 4.90 & 4.90 & 4.65 &3.45& 1.00 & 6.40\\
\midrule
\multirow{4}{*}{Case 2}
& \multirow{2}{*}{120} & $100$   & 5.60 & 4.30 & 5.80 & 5.55 & 2.40 & 6.25 & 5.40 & 3.85 & 6.30 & 3.20 & 1.05 & 5.25 \\
& & $150$   & 5.60 & 4.60 & 5.15 & 5.35 & 4.20 & 5.50 & 5.35 & 3.70 & 5.55 & 4.25 & 1.15 & 6.20 \\
\cmidrule(lr){2-15}
& \multirow{2}{*}{240} & $100$   & 5.20 & 3.70 & 5.45 &5.50  & 3.95 & 6.30 &4.85  & 3.20 & 5.35 & 3.55 & 1.25 & 5.40 \\
& & $150$   & 5.00 & 3.75 & 5.40 &5.40& 3.35 & 6.50 & 5.30 & 4.20 & 5.75 & 4.05 &1.95  & 5.70 \\
\midrule
\multirow{4}{*}{Case 3}
& \multirow{2}{*}{120} & $100$   & 5.15 & 5.00 & 4.75 & 5.00 & 3.40 & 5.85 & 4.95 & 3.90 & 5.55 & 3.15 & 1.20 & 5.20 \\
& &$150$   & 5.05 & 3.15 & 5.30 &6.70& 3.85 & 6.60 & 4.55 & 4.00 & 4.60 & 3.35 & 1.00 & 5.50 \\
\cmidrule(lr){2-15}
& \multirow{2}{*}{240} & $100$   & 5.95 & 4.30 & 5.80 & 5.35 & 3.65 & 6.15 &5.20 & 4.15 & 5.20 & 3.00 & 1.15 & 5.25 \\
& & $150$   & 4.15 & 3.70 & 5.15 & 3.85 & 3.20 & 3.95 & 5.00 &4.25&5.15  &3.25 & 0.65 & 5.25 \\
\bottomrule
\end{tabular}
\end{table}

\begin{table}[htbp] 
\centering 
\caption{Empirical size across different $(n,p)$ configurations and methods for Cases~1--3 ($\tau = 0.25$, $s = 9$; 2000 replications).} 
\label{tab2}
\scriptsize \setlength{\tabcolsep}{3.5pt} \renewcommand{\arraystretch}{1.2} 
\begin{tabular}{c c c | ccc ccc ccc ccc} 
\toprule Case & $p$ & $n$ 
& \multicolumn{3}{c}{Normal} 
& \multicolumn{3}{c}{Laplace} 
& \multicolumn{3}{c}{Logistic} 
& \multicolumn{3}{c}{$t_2$} \\ 
\cmidrule(lr){4-6} \cmidrule(lr){7-9} \cmidrule(lr){10-12} \cmidrule(lr){13-15} 
& &
& $T_{\mathrm{CC}}$ & $T_{\mathrm{MAX}}$ & $T_{\mathrm{SUM}}$ 
& $T_{\mathrm{CC}}$ & $T_{\mathrm{MAX}}$ & $T_{\mathrm{SUM}}$ 
& $T_{\mathrm{CC}}$ & $T_{\mathrm{MAX}}$ & $T_{\mathrm{SUM}}$
& $T_{\mathrm{CC}}$ & $T_{\mathrm{MAX}}$ & $T_{\mathrm{SUM}}$ \\
\midrule
\multirow{4}{*}{Case 1}
& \multirow{2}{*}{120} & $100$   & 6.30 & 6.25 & 7.05 &  6.25 & 5.85 & 6.35 & 6.90 & 5.70 & 6.25 & 4.65 & 2.60 & 5.75 \\
& &  $150$   & 6.70& 5.30 &7.20 & 5.70 & 4.90 & 5.95 & 6.05 & 5.40 & 5.85 &  4.35 & 2.10 & 5.80  \\
\cmidrule(lr){2-15}
& \multirow{2}{*}{240} & $100$   &6.25& 5.20 & 6.05 & 4.55 &4.60  & 4.75 & 6.10 & 5.05 & 6.20 & 4.50 & 2.05 & 5.70 \\
& & $150$   & 5.55 & 4.95 &5.60 &5.05& 4.05 & 5.10 & 5.85 & 5.25 & 5.15 &4.30 & 3.20  & 5.65 \\
\midrule
\multirow{4}{*}{Case 2}
& \multirow{2}{*}{120} & $100$   & 6.50 & 5.70 & 5.65 &5.45  &4.95& 5.80 & 5.20 & 5.00 & 5.00 &5.60& 3.05& 6.60 \\
& &  $150$   & 5.30 & 4.85 & 5.35 & 5.15 & 4.05 & 5.35 & 6.01 & 4.80 & 5.95 & 5.35 & 3.30 & 5.85 \\
\cmidrule(lr){2-15}
& \multirow{2}{*}{240} & $100$   & 5.95 & 5.50 &5.85  &6.60  &5.65  & 6.00 & 6.20 & 4.95 & 5.65 &4.70  & 2.55 & 5.70 \\
& & $150$   & 6.00 & 4.85 & 6.00 & 6.15 & 5.05 & 7.20 & 5.95 & 4.30 & 6.05 & 4.35 & 2.85 & 4.70 \\
\midrule
\multirow{4}{*}{Case 3}
& \multirow{2}{*}{120} & $100$   & 6.55 & 5.55 &6.15  &5.40  &5.25  & 5.40 &5.70  & 5.75 & 5.65 & 4.90 & 2.75 & 5.55 \\
& & $150$   & 5.35 &4.70  & 5.70 & 6.10 & 5.10 & 5.85 & 6.10 & 5.65 & 6.15 & 4.60 & 3.15 & 5.75 \\
\cmidrule(lr){2-15}
& \multirow{2}{*}{240} & $100$   & 6.05 & 5.40 & 6.30 & 6.35 & 4.75 & 5.55 & 6.50 &5.45  & 6.45 &4.40  &3.30  &6.30  \\
& & $150$   & 6.15 &5.50& 6.05 &6.50& 4.40 & 6.50 &6.20& 4.55 & 6.65 &3.25  & 2.45 & 4.85 \\
\bottomrule
\end{tabular}
\end{table}

\begin{table}[htbp] 
\centering 
\caption{Empirical size across different $(n,p)$ configurations and methods for Cases~1--3 ($\tau = 0.75$, $s = 9$; 2000 replications).}
\label{tab3}
\scriptsize \setlength{\tabcolsep}{3.5pt} \renewcommand{\arraystretch}{1.2} 
\begin{tabular}{c c c | ccc ccc ccc ccc} 
\toprule Case & $p$ & $n$ 
& \multicolumn{3}{c}{Normal} 
& \multicolumn{3}{c}{Laplace} 
& \multicolumn{3}{c}{Logistic} 
& \multicolumn{3}{c}{$t_2$} \\ 
\cmidrule(lr){4-6} \cmidrule(lr){7-9} \cmidrule(lr){10-12} \cmidrule(lr){13-15} 
& &
& $T_{\mathrm{CC}}$ & $T_{\mathrm{MAX}}$ & $T_{\mathrm{SUM}}$ 
& $T_{\mathrm{CC}}$ & $T_{\mathrm{MAX}}$ & $T_{\mathrm{SUM}}$ 
& $T_{\mathrm{CC}}$ & $T_{\mathrm{MAX}}$ & $T_{\mathrm{SUM}}$ 
& $T_{\mathrm{CC}}$ & $T_{\mathrm{MAX}}$ & $T_{\mathrm{SUM}}$ \\
\midrule
\multirow{4}{*}{Case 1}
& \multirow{2}{*}{120} & $100$   & 6.15 & 5.55 &7.15 & 6.40 & 5.45 & 6.15 & 5.95 & 5.00 & 5.50 & 5.10 &3.10 & 5.40 \\
& & $150$   & 6.90 & 5.20 &6.30 & 5.00 & 4.70 & 5.15 & 6.40 &5.10  & 5.75 & 4.75 & 3.40 & 5.40 \\
\cmidrule(lr){2-15}
& \multirow{2}{*}{240} & $100$   & 6.30 & 5.70 & 5.90& 6.25 &4.75  & 6.65 & 5.40 & 5.15 & 5.65 & 4.80 & 3.40 & 5.70 \\
& & $150$   &6.30 & 4.95 & 5.35& 6.00 &5.00& 5.65 & 5.70 & 4.50 & 5.85 & 4.45 & 3.35 & 5.20 \\
\midrule
\multirow{4}{*}{Case 2}
& \multirow{2}{*}{120} & $100$   & 6.90 & 5.10 &6.25 & 6.15 & 5.15 & 5.85 & 6.00 & 4.05 & 6.35 & 4.65 & 3.20 & 5.65 \\
& & $150$   & 5.50 &4.20  &5.55 & 6.60 & 4.80 & 6.95 & 6.00 & 4.65 & 5.95 & 5.15 & 3.10 & 5.85 \\
\cmidrule(lr){2-15}
& \multirow{2}{*}{240} & $100$   &6.80& 5.70 &6.15 & 6.20 & 5.00 & 6.75 & 6.15 & 5.15 & 5.65 & 3.95 & 2.65 & 5.20 \\
& & $150$   & 4.95 & 4.25 &5.20 & 6.35 & 4.55 & 6.05 & 5.55 & 4.80 & 5.45 & 4.15 & 3.35 & 4.70 \\
\midrule
\multirow{4}{*}{Case 3}
& \multirow{2}{*}{120} & $100$   & 6.00 & 4.95 &5.70 & 6.35 & 5.15 & 7.05 & 5.90 & 5.05 & 6.15  & 4.45 & 3.45 & 5.50 \\
& &$150$   & 6.35 & 5.30 &5.65 & 6.30 & 4.65 & 5.80 & 5.30 & 4.55 & 5.60 &4.65& 2.60 & 5.50 \\
\cmidrule(lr){2-15}
& \multirow{2}{*}{240} & $100$   & 5.70 & 5.70 &4.90  & 6.10 & 5.60 & 6.40 & 5.70 & 5.35 & 5.55 & 5.05 & 2.65 & 6.80 \\
& & $150$   & 6.50 &5.30  &6.40 & 5.40 & 4.95 & 5.85 & 6.15 & 5.40 & 6.25 & 4.35 & 2.70 & 5.00 \\
\bottomrule
\end{tabular}
\end{table}

We compute empirical power over a range of sparsity levels $s$, covering extremely sparse to moderately dense alternatives. To visualize the comparative performance, we include power curves plotting empirical power against $s$ for each test under all covariance structures and distributions of $\boldsymbol{u}$ 
(see Figures~\ref{fig:power_normal}--\ref{fig:power_t2}). Due to space limitations, we present only the power results for the median quantile ($\tau = 0.5$) in the figures, while the complete results for the other quantiles are provided in Appendix~\ref{sup_fig}.

\begin{figure}[htbp]
\centering
\begin{subfigure}{0.23\textwidth}
    \centering
    \includegraphics[width=\linewidth]{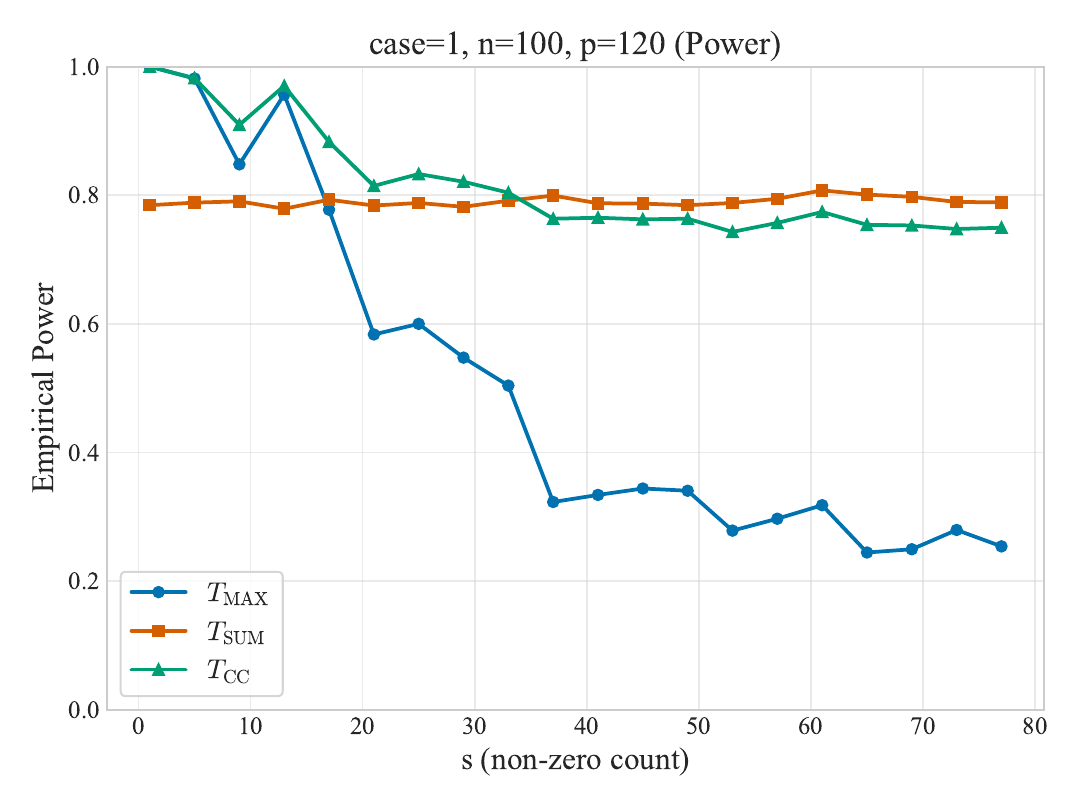}
\end{subfigure}
\hfill
\begin{subfigure}{0.23\textwidth}
    \centering
    \includegraphics[width=\linewidth]{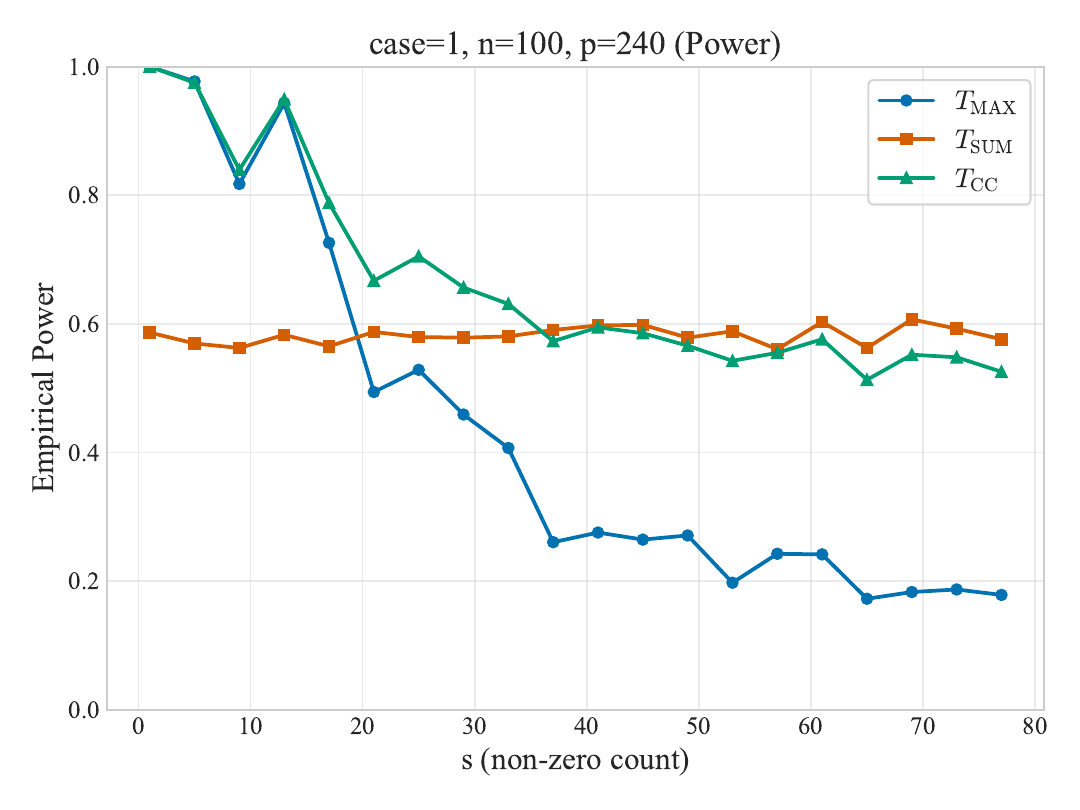}
\end{subfigure}
\hfill
\begin{subfigure}{0.23\textwidth}
    \centering
    \includegraphics[width=\linewidth]{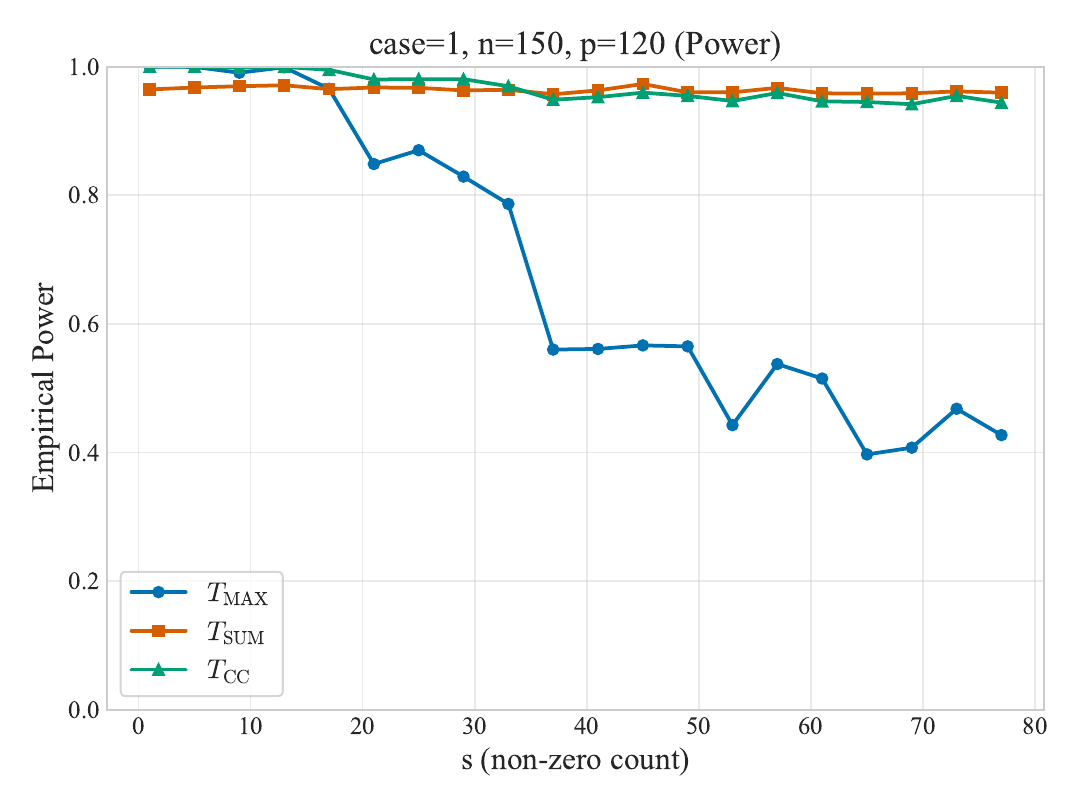}
\end{subfigure}
\hfill
\begin{subfigure}{0.23\textwidth}
    \centering
    \includegraphics[width=\linewidth]{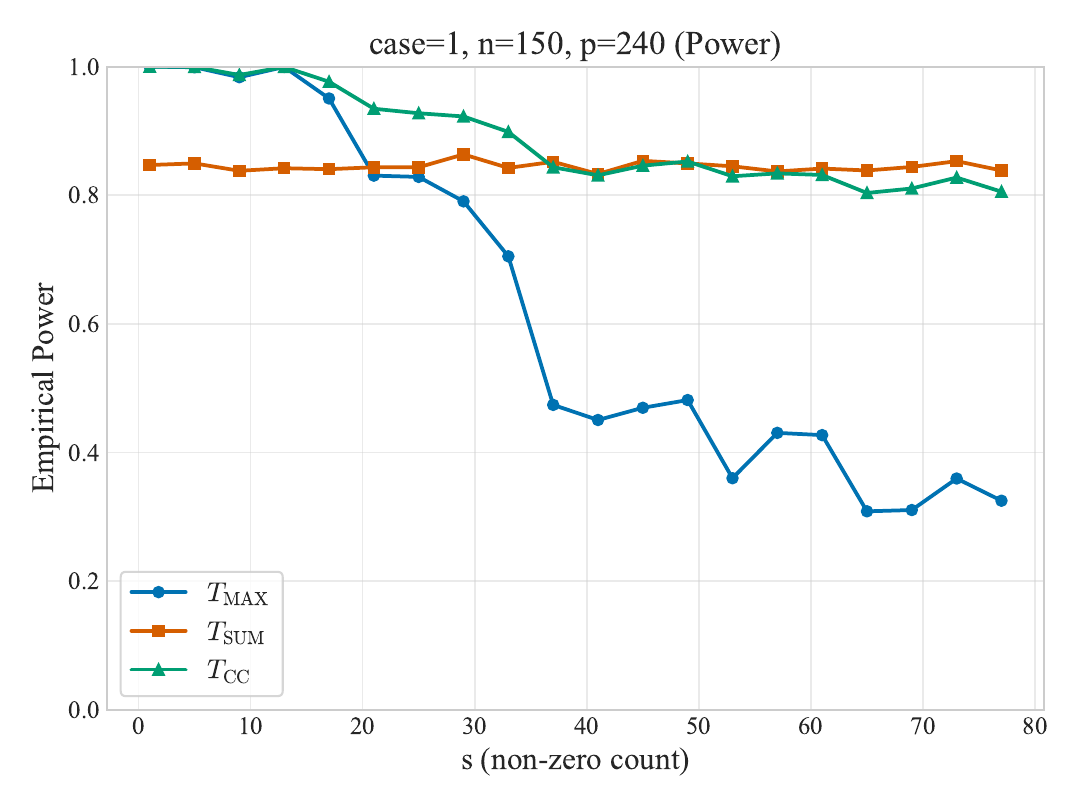}
\end{subfigure}

\vspace{0.3cm}
\begin{subfigure}{0.23\textwidth}
    \centering
    \includegraphics[width=\linewidth]{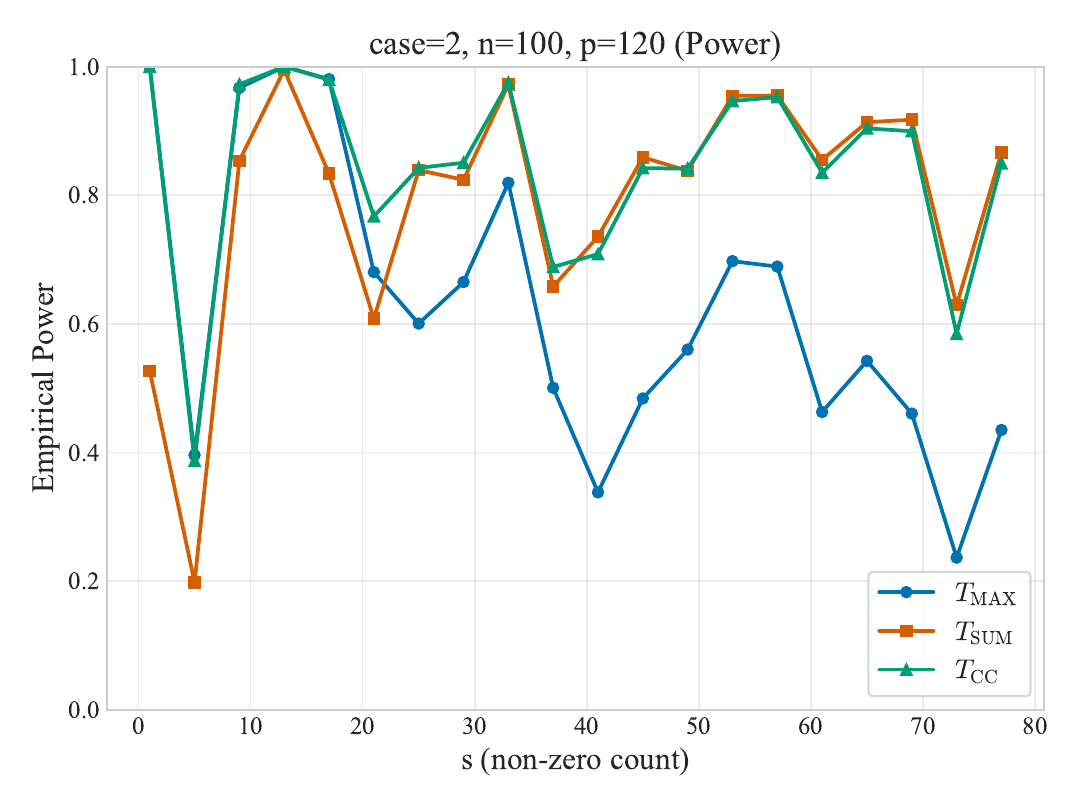}
\end{subfigure}
\hfill
\begin{subfigure}{0.23\textwidth}
    \centering
    \includegraphics[width=\linewidth]{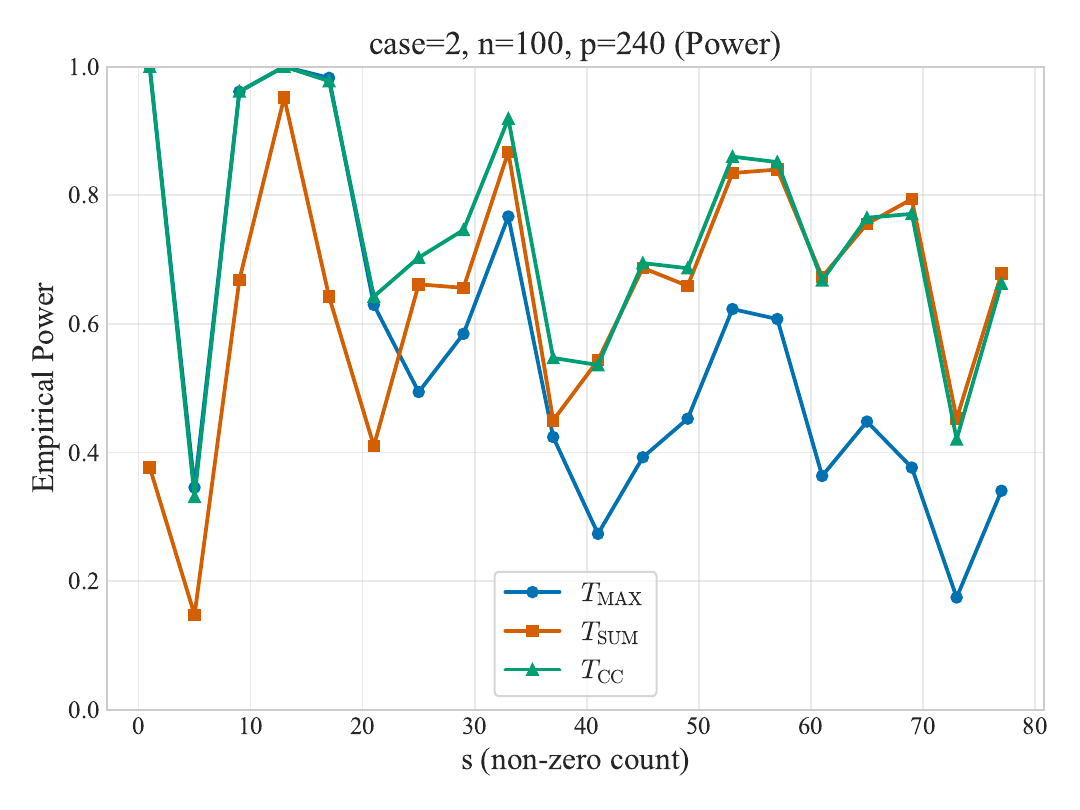}
\end{subfigure}
\hfill
\begin{subfigure}{0.23\textwidth}
    \centering
    \includegraphics[width=\linewidth]{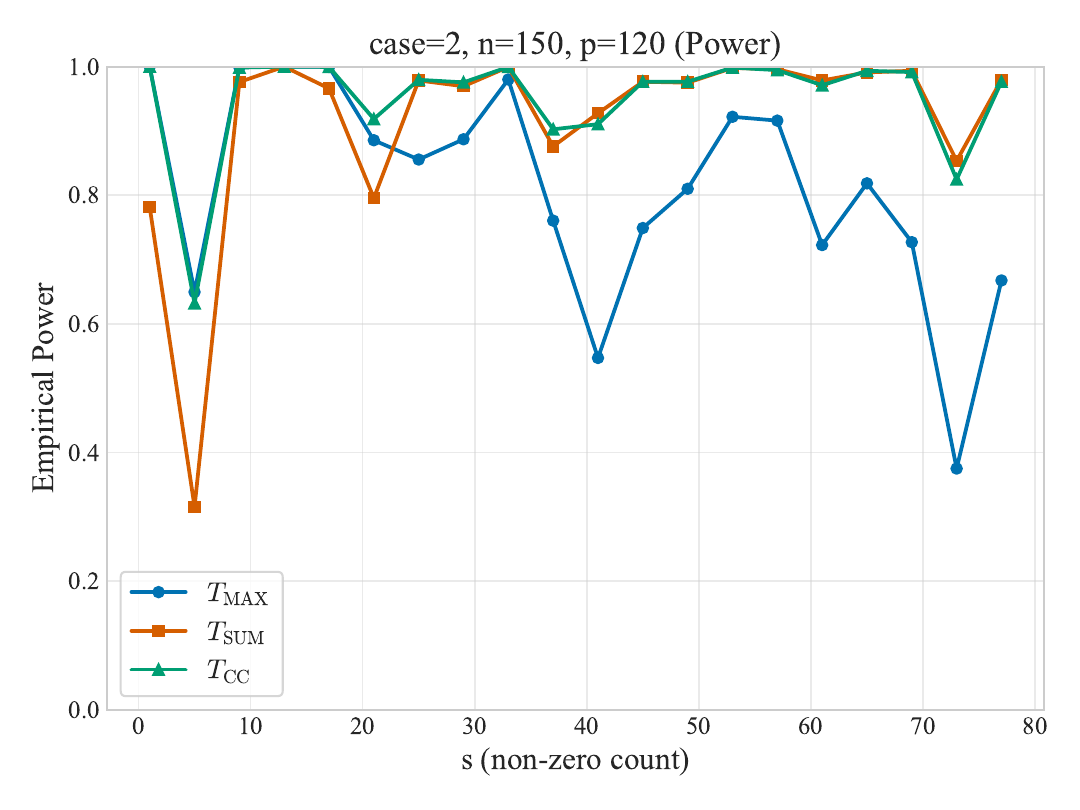}
\end{subfigure}
\hfill
\begin{subfigure}{0.23\textwidth}
    \centering
    \includegraphics[width=\linewidth]{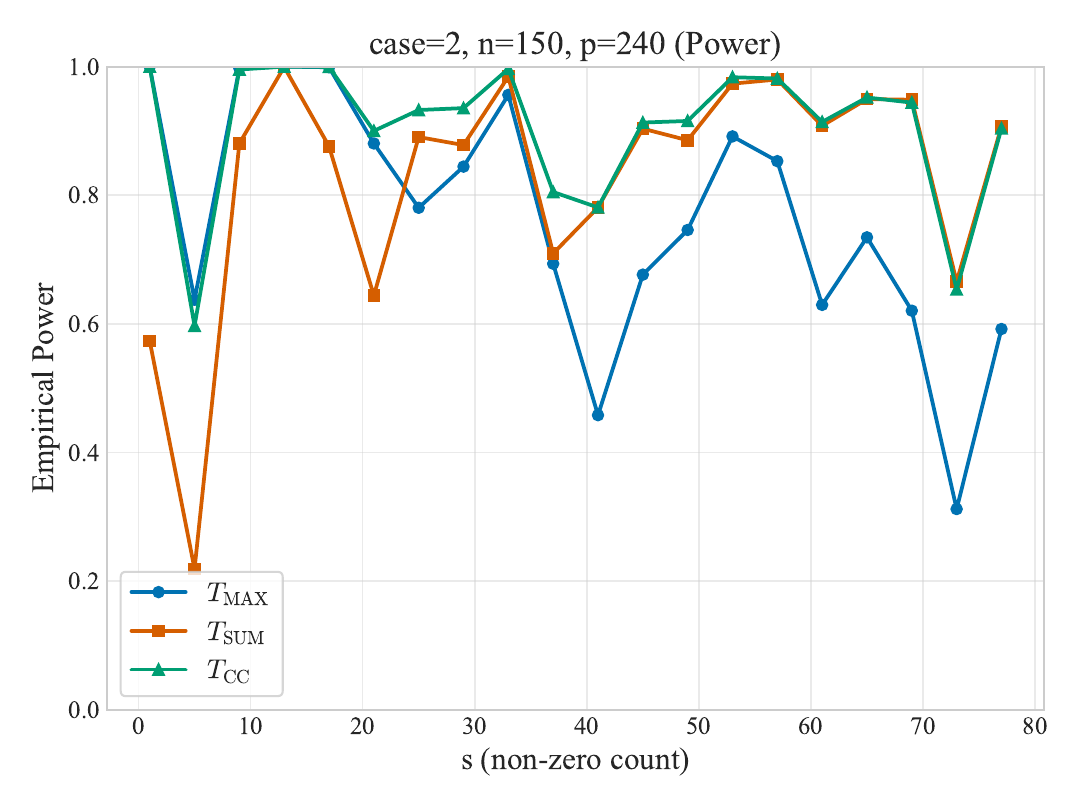}
\end{subfigure}

\vspace{0.3cm}
\begin{subfigure}{0.23\textwidth}
    \centering
    \includegraphics[width=\linewidth]{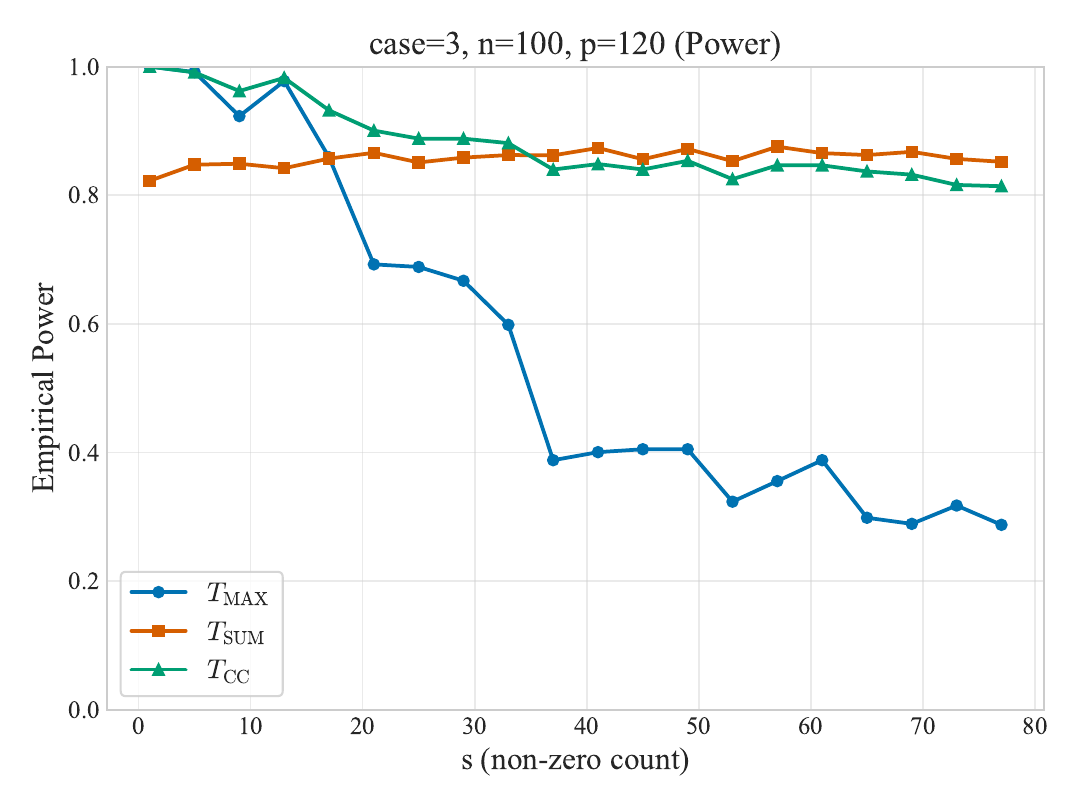}
\end{subfigure}
\hfill
\begin{subfigure}{0.23\textwidth}
    \centering
    \includegraphics[width=\linewidth]{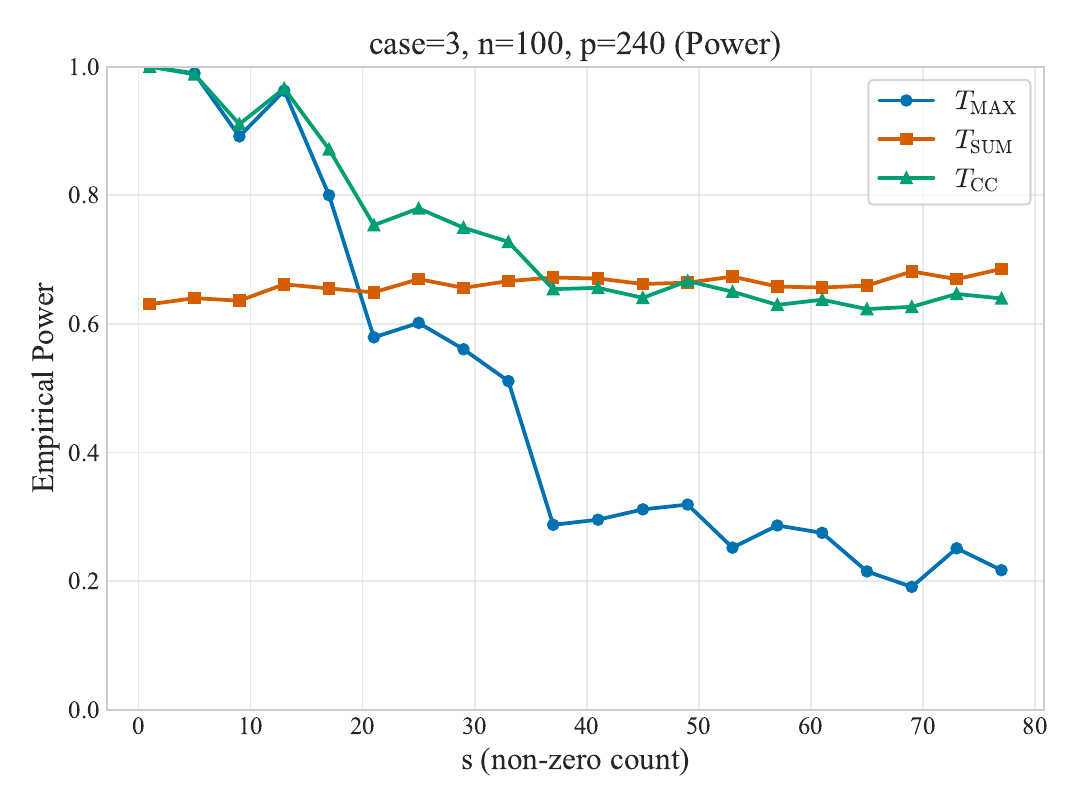}
\end{subfigure}
\hfill
\begin{subfigure}{0.23\textwidth}
    \centering
    \includegraphics[width=\linewidth]{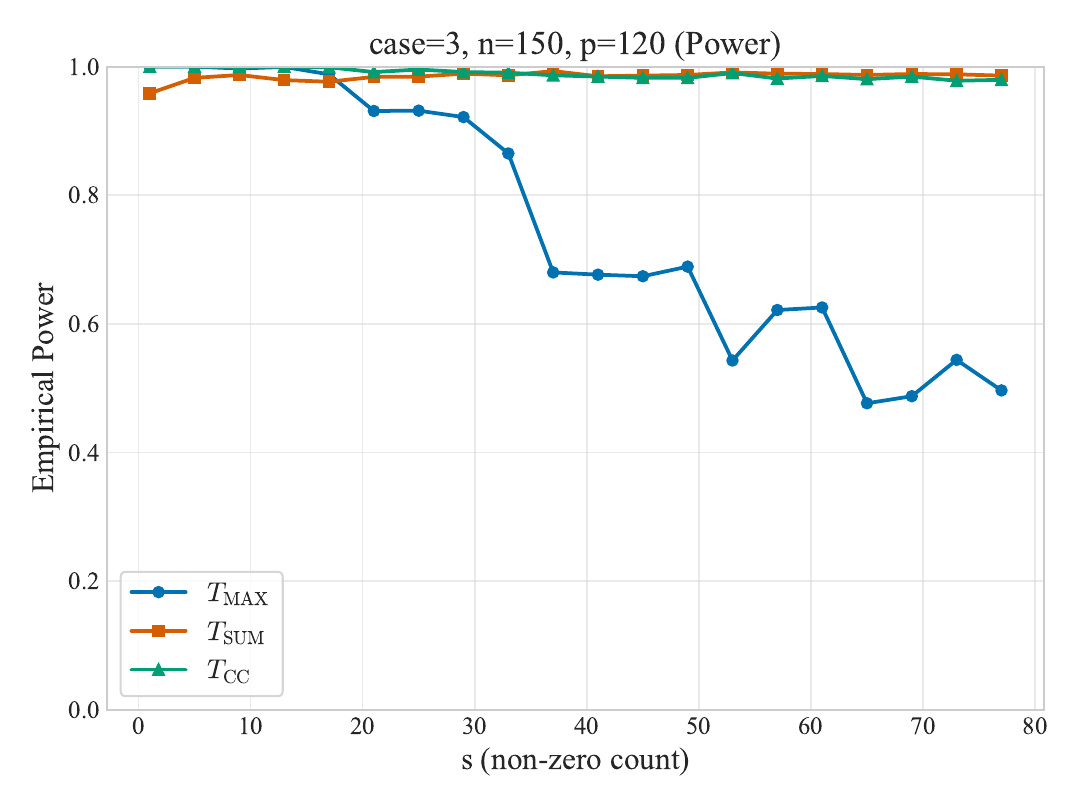}
\end{subfigure}
\hfill
\begin{subfigure}{0.23\textwidth}
    \centering
    \includegraphics[width=\linewidth]{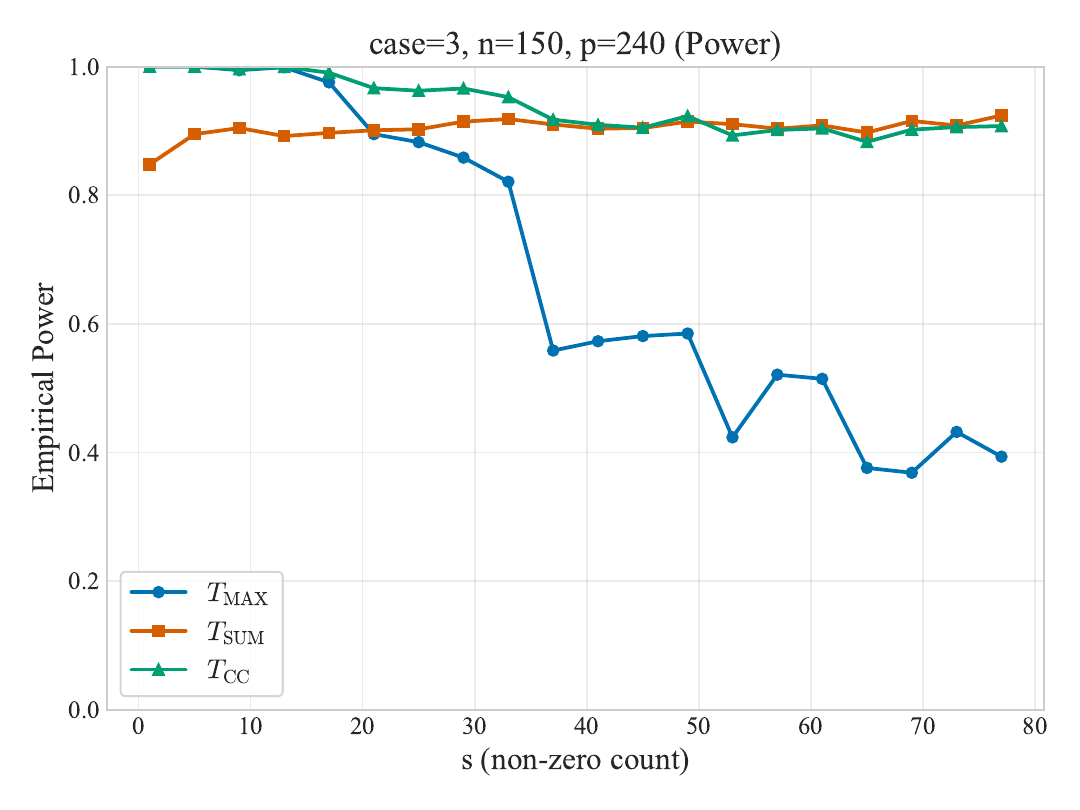}
\end{subfigure}

\caption{Empirical power as a function of $s$ for Cases~1--3 across varying $(n,p)$ 
settings under Normal distribution ($\tau = 0.5$; 2000 replications).
}
\label{fig:power_normal}
\end{figure}
\begin{figure}[htbp]
\centering
\begin{subfigure}{0.23\textwidth}
    \centering
    \includegraphics[width=\linewidth]{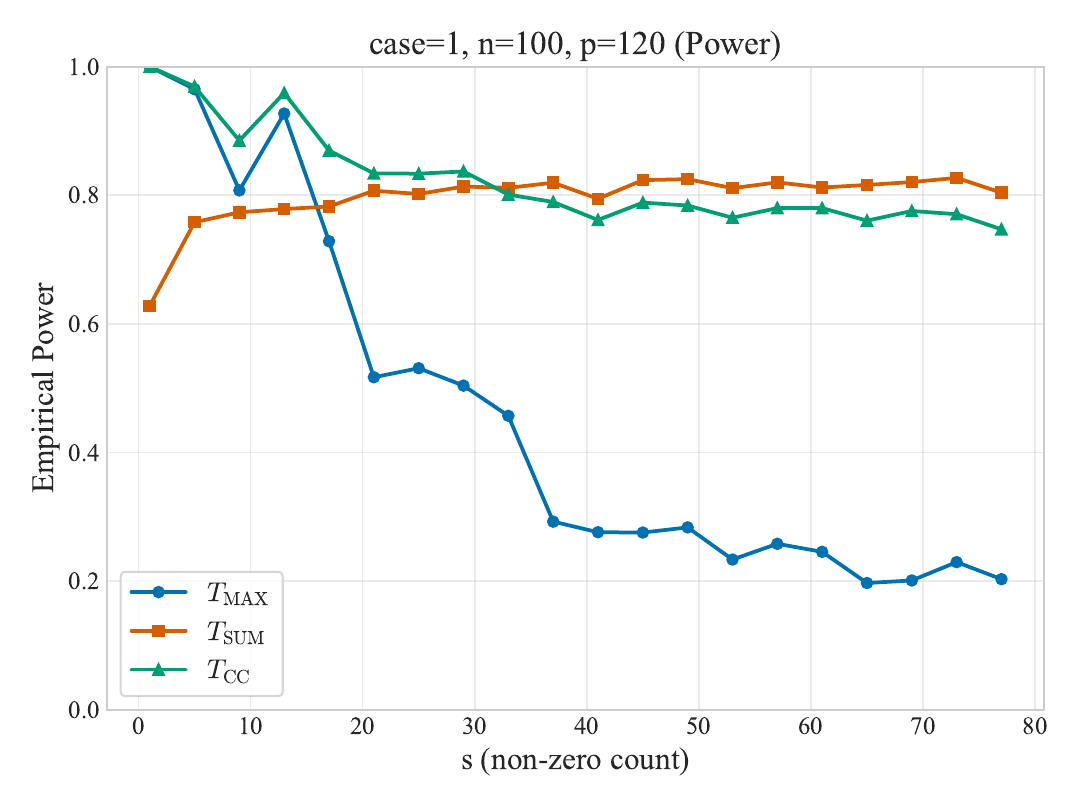}
\end{subfigure}
\hfill
\begin{subfigure}{0.23\textwidth}
    \centering
    \includegraphics[width=\linewidth]{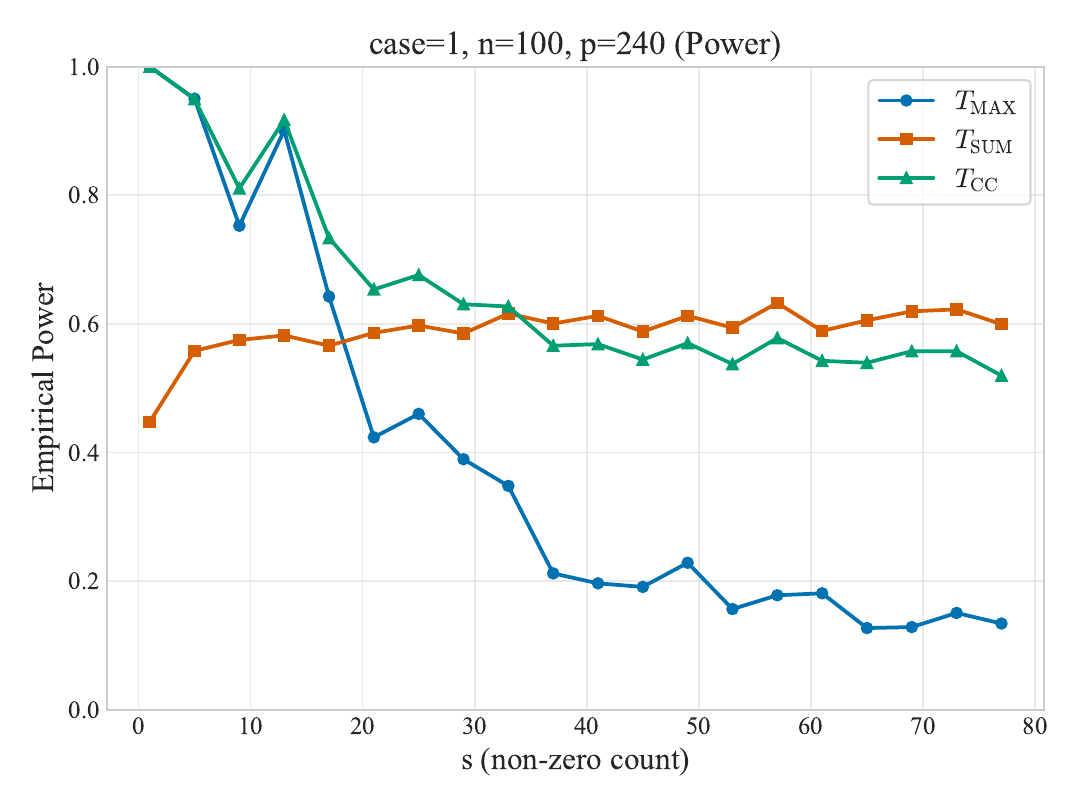}
\end{subfigure}
\hfill
\begin{subfigure}{0.23\textwidth}
    \centering
    \includegraphics[width=\linewidth]{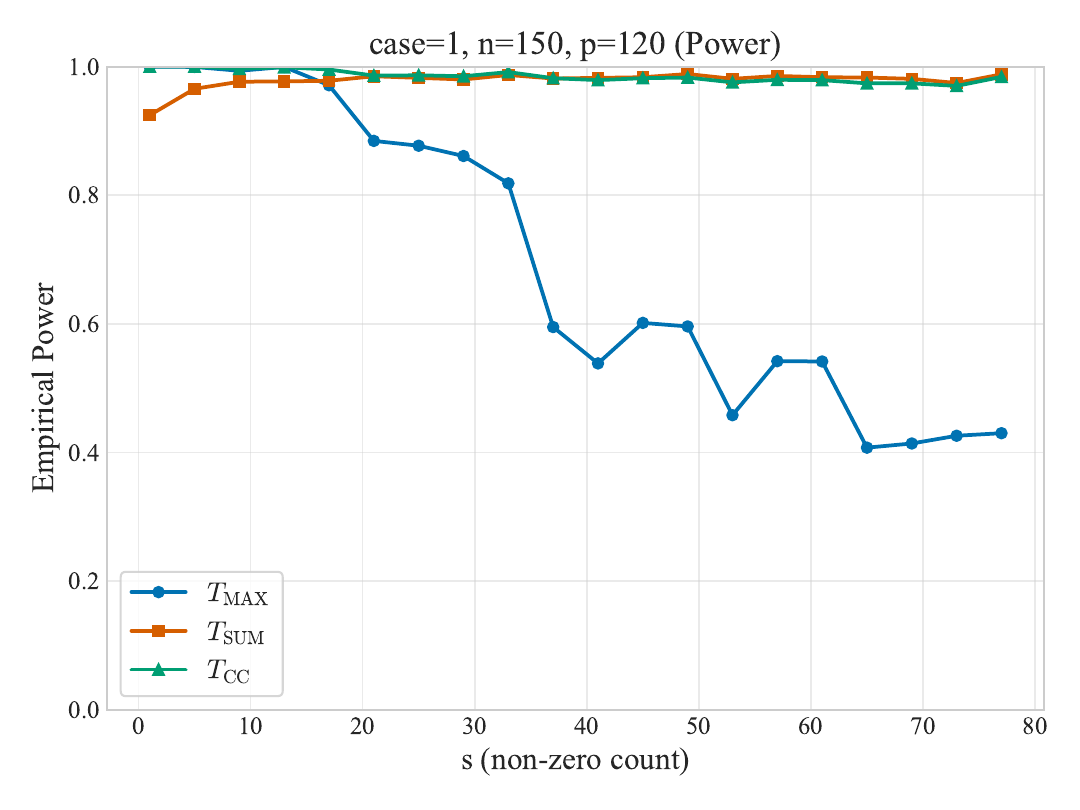}
\end{subfigure}
\hfill
\begin{subfigure}{0.23\textwidth}
    \centering
    \includegraphics[width=\linewidth]{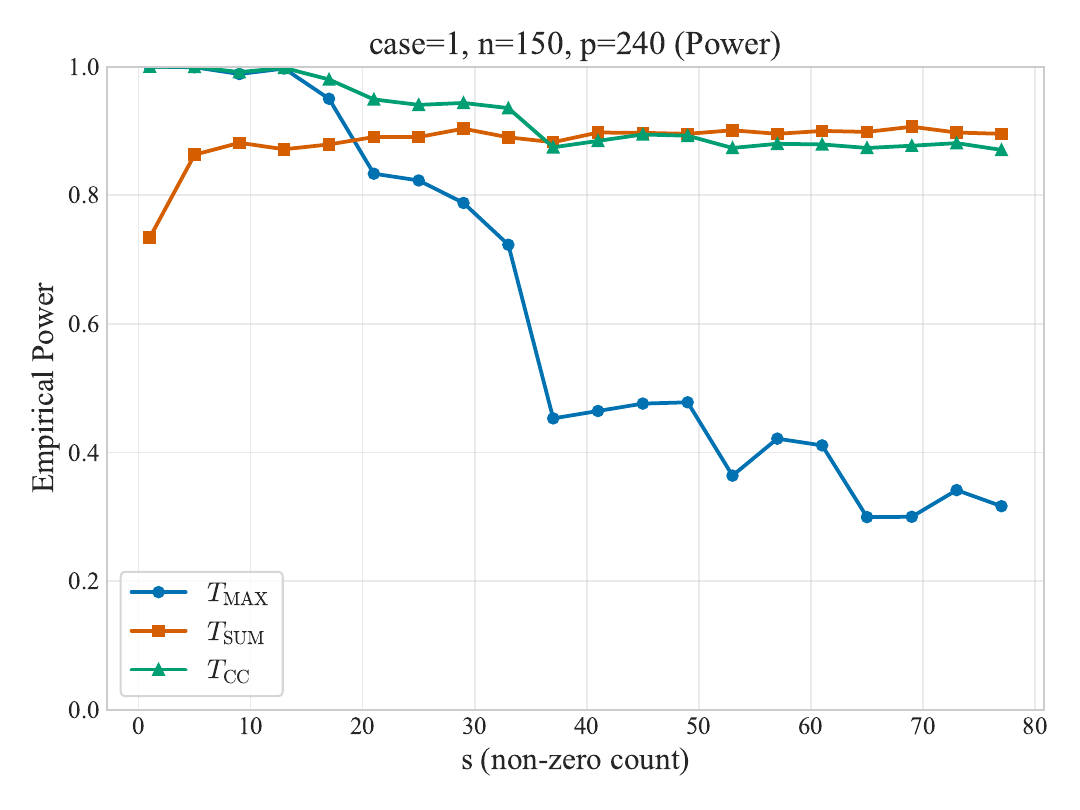}
\end{subfigure}

\vspace{0.3cm}
\begin{subfigure}{0.23\textwidth}
    \centering
    \includegraphics[width=\linewidth]{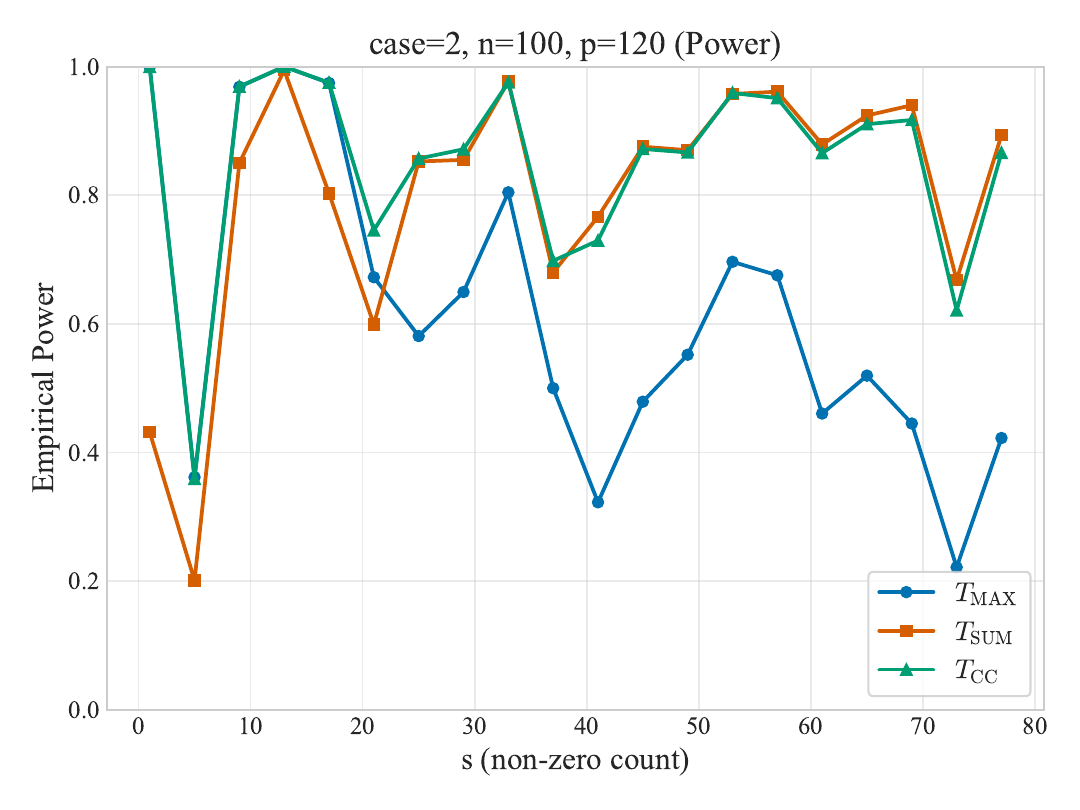}
\end{subfigure}
\hfill
\begin{subfigure}{0.23\textwidth}
    \centering
    \includegraphics[width=\linewidth]{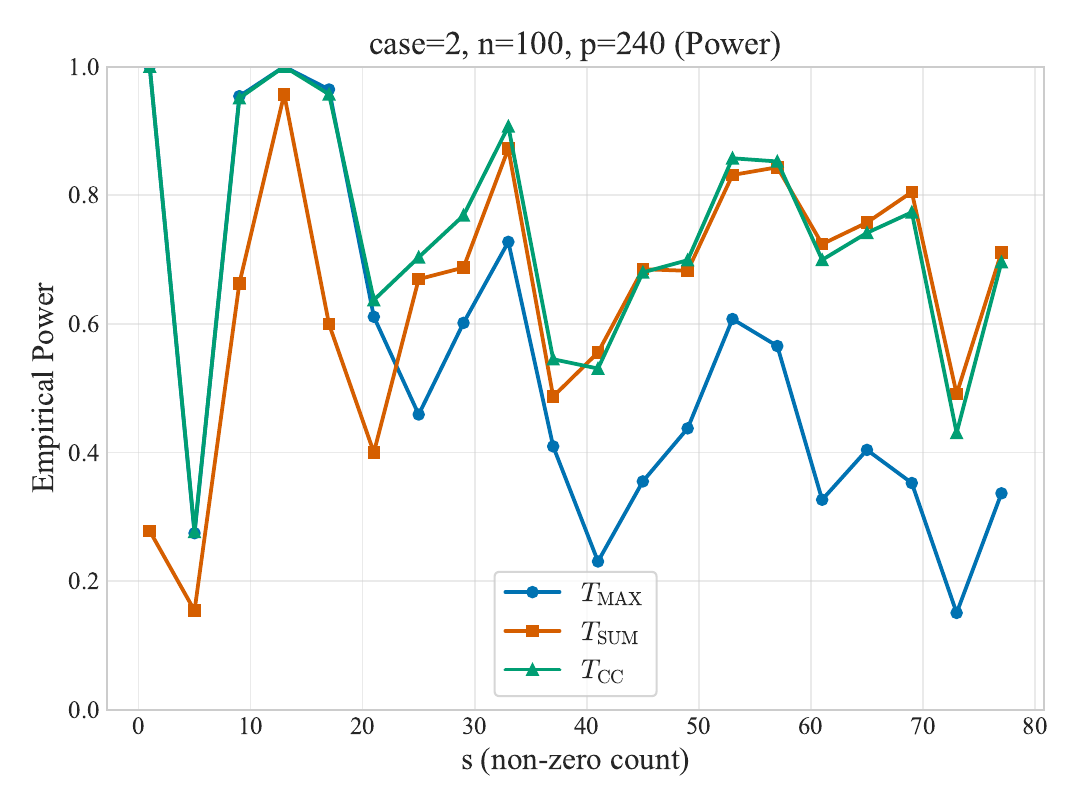}
\end{subfigure}
\hfill
\begin{subfigure}{0.23\textwidth}
    \centering
    \includegraphics[width=\linewidth]{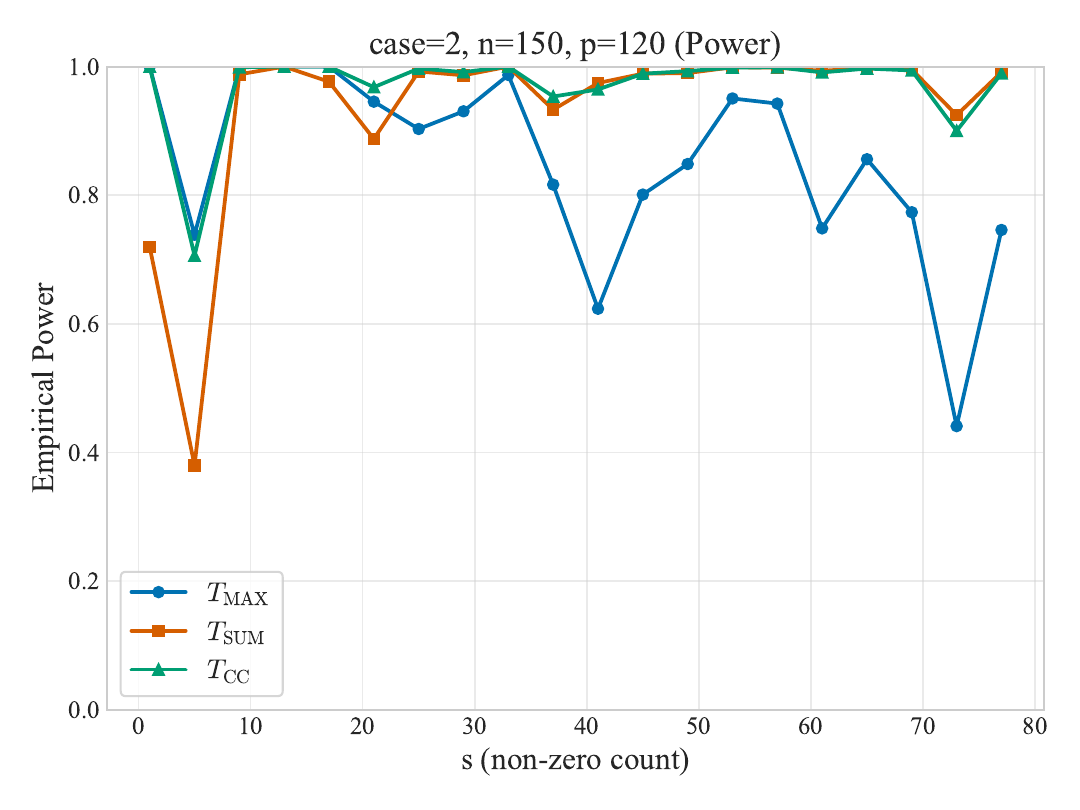}
\end{subfigure}
\hfill
\begin{subfigure}{0.23\textwidth}
    \centering
    \includegraphics[width=\linewidth]{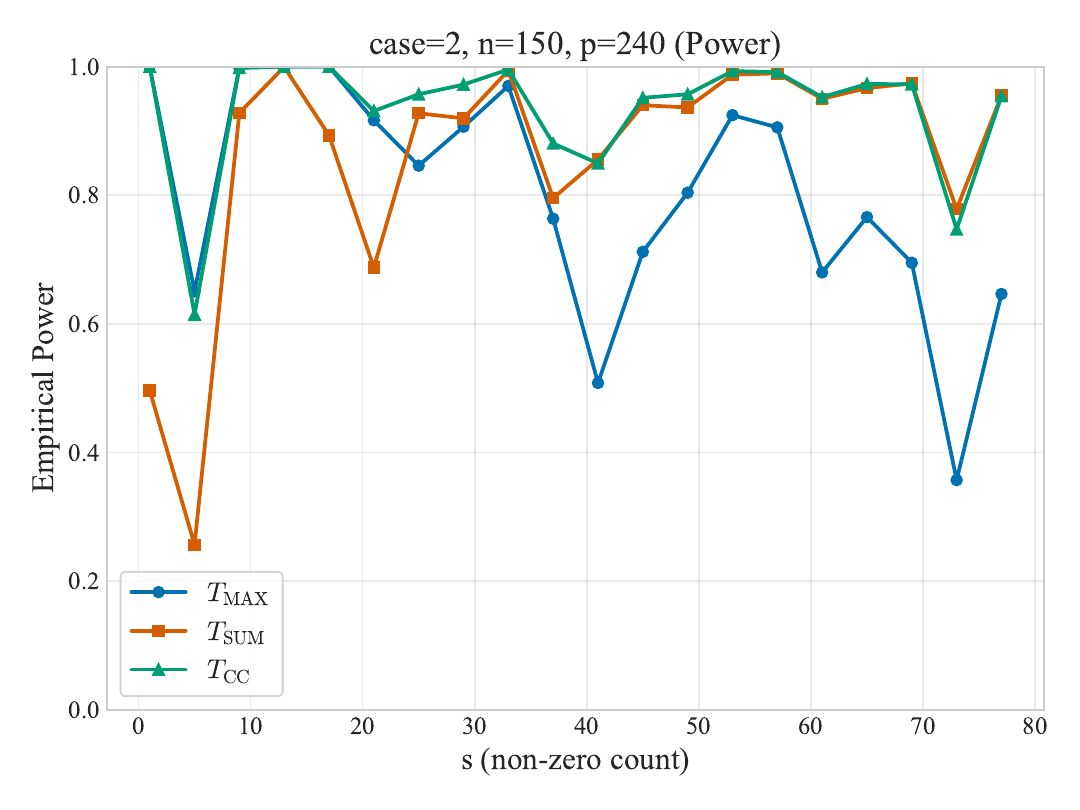}
\end{subfigure}

\vspace{0.3cm}
\begin{subfigure}{0.23\textwidth}
    \centering
    \includegraphics[width=\linewidth]{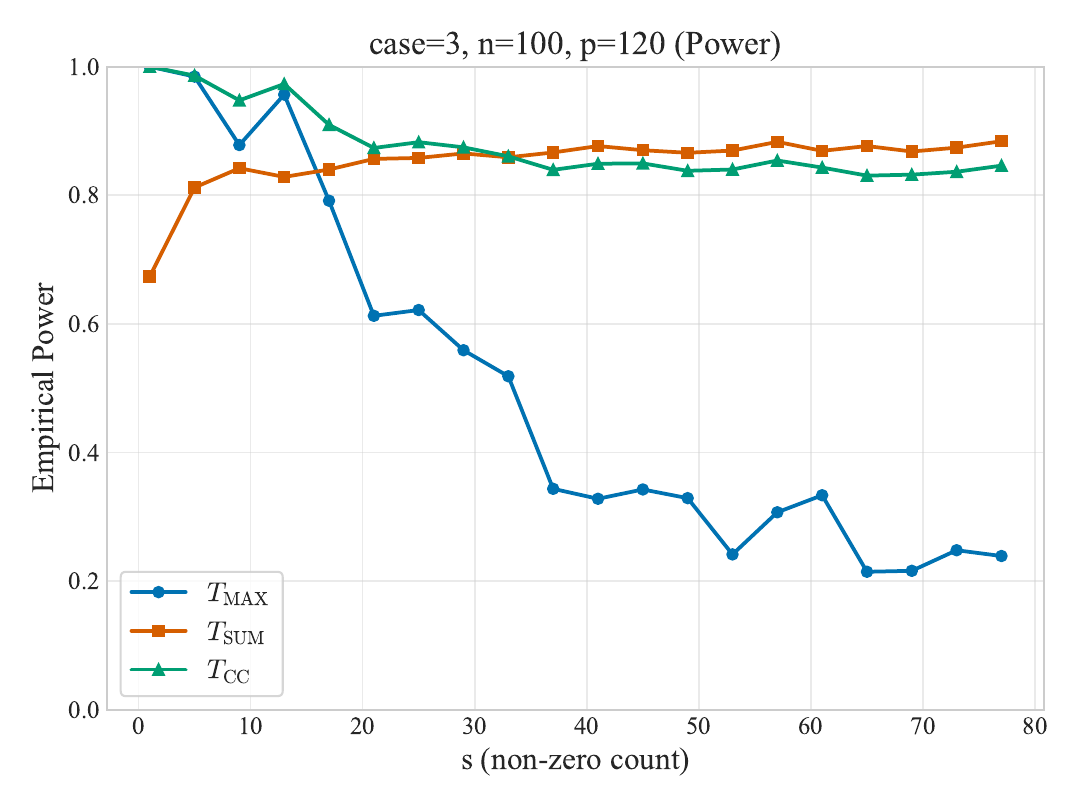}
\end{subfigure}
\hfill
\begin{subfigure}{0.23\textwidth}
    \centering
    \includegraphics[width=\linewidth]{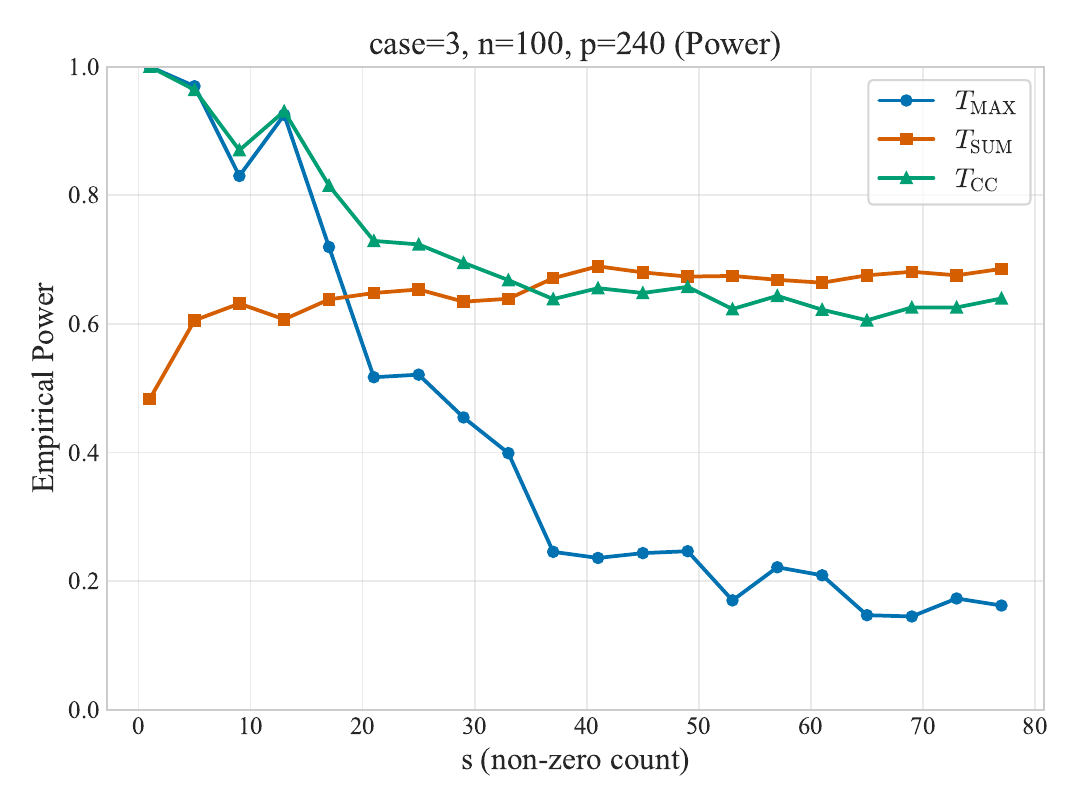}
\end{subfigure}
\hfill
\begin{subfigure}{0.23\textwidth}
    \centering
    \includegraphics[width=\linewidth]{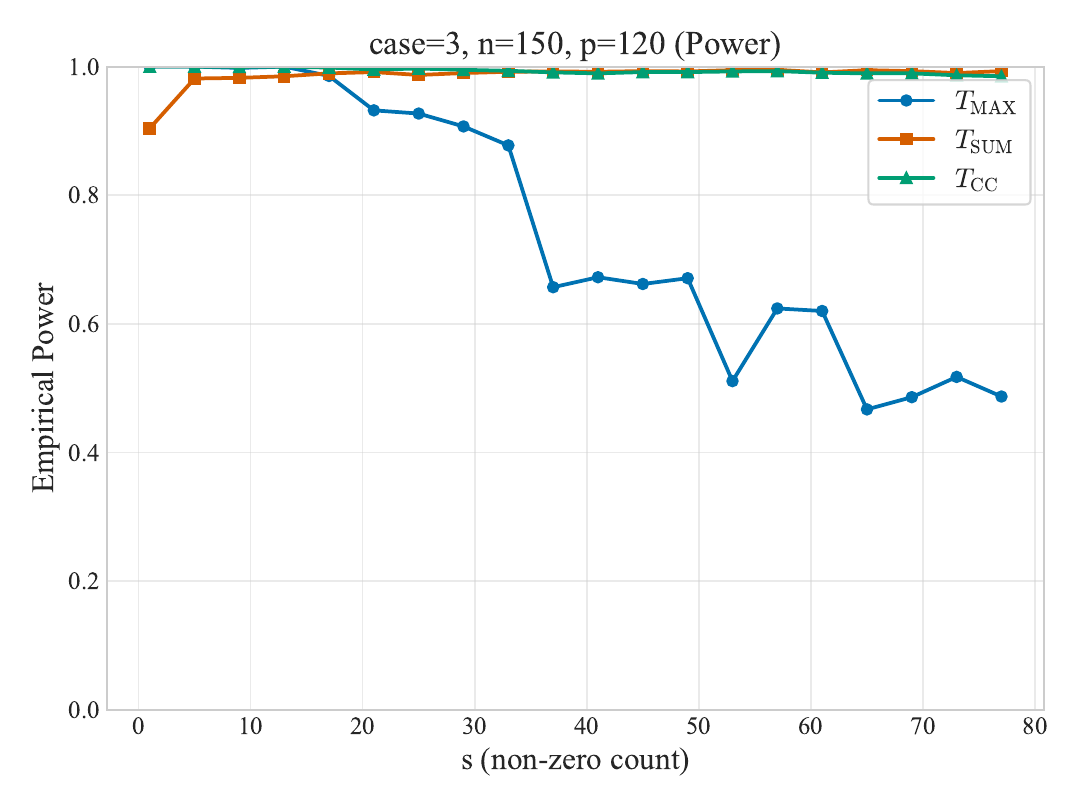}
\end{subfigure}
\hfill
\begin{subfigure}{0.23\textwidth}
    \centering
    \includegraphics[width=\linewidth]{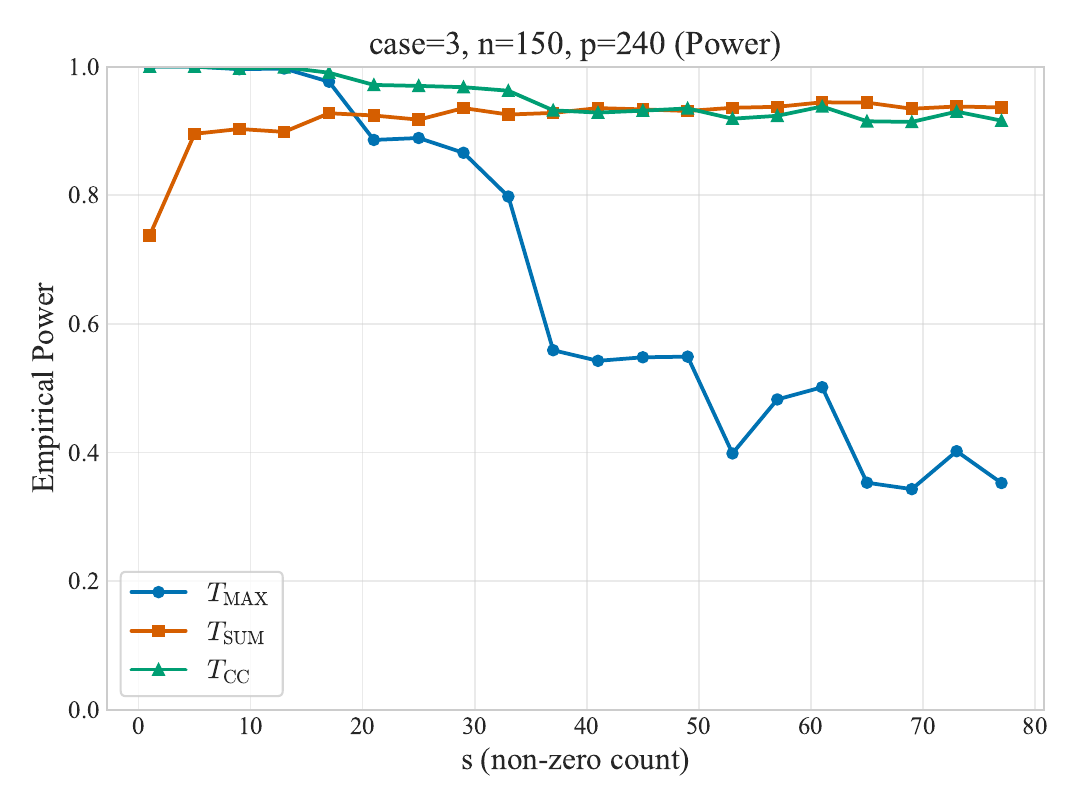}
\end{subfigure}

\caption{Empirical power as a function of $s$ for Cases~1--3 across varying $(n,p)$ 
settings under Laplace distribution ($\tau = 0.5$; 2000 replications).}
\label{fig:power_laplace}
\end{figure}

\begin{figure}[htbp]
\centering
\begin{subfigure}{0.23\textwidth}
    \centering
    \includegraphics[width=\linewidth]{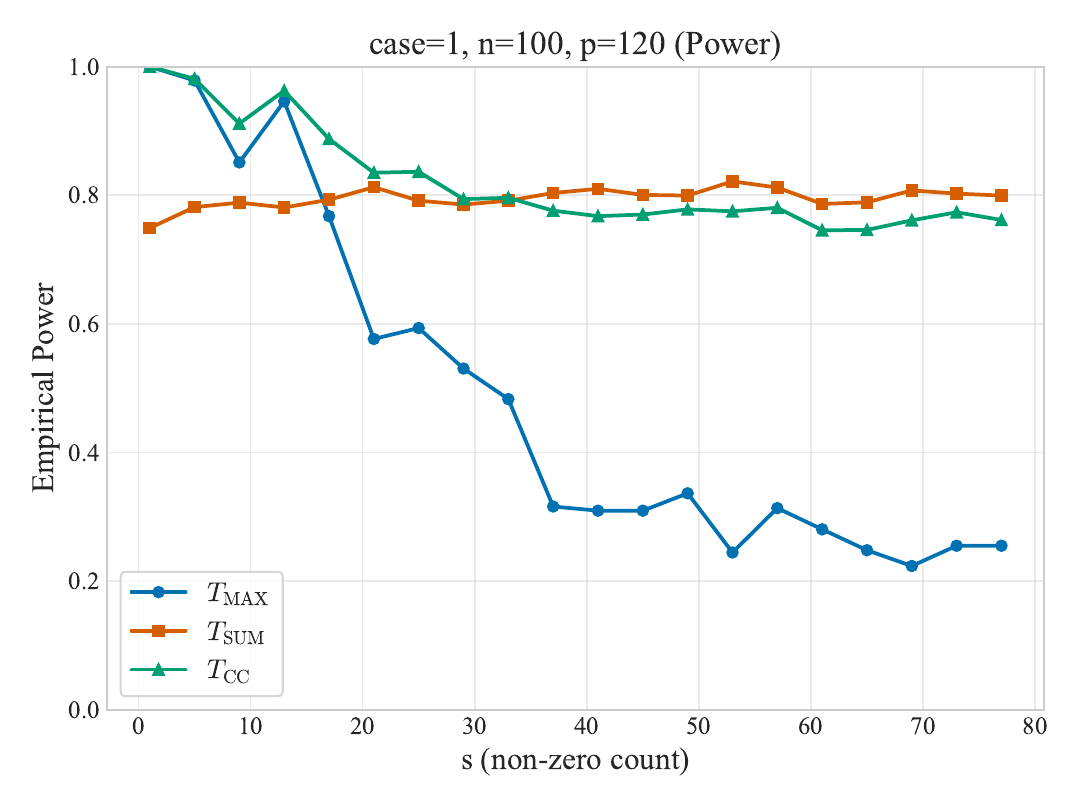}
\end{subfigure}
\hfill
\begin{subfigure}{0.23\textwidth}
    \centering
    \includegraphics[width=\linewidth]{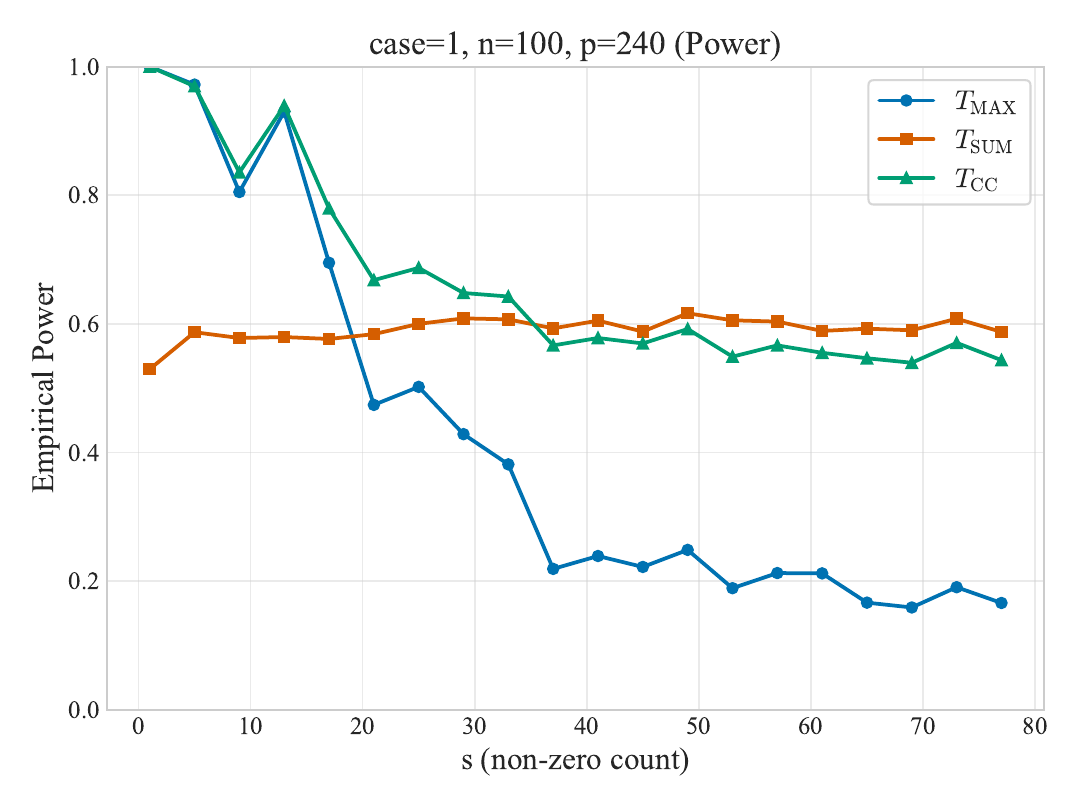}
\end{subfigure}
\hfill
\begin{subfigure}{0.23\textwidth}
    \centering
    \includegraphics[width=\linewidth]{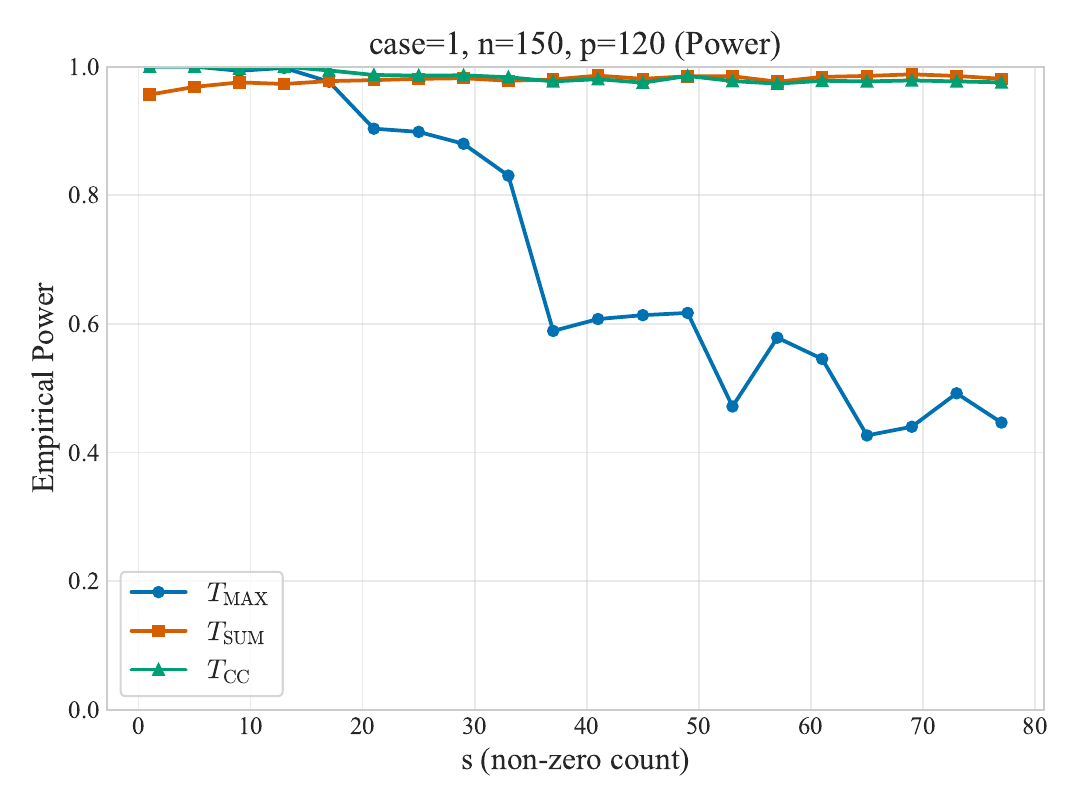}
\end{subfigure}
\hfill
\begin{subfigure}{0.23\textwidth}
    \centering
    \includegraphics[width=\linewidth]{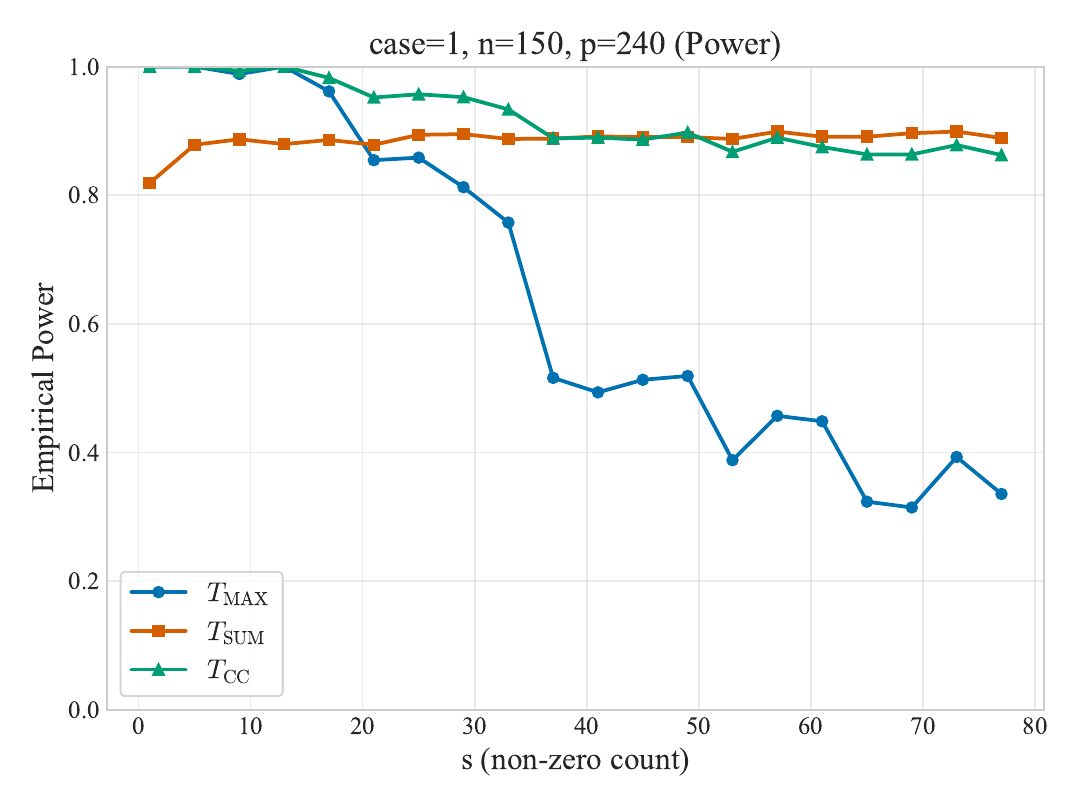}
\end{subfigure}

\vspace{0.3cm}
\begin{subfigure}{0.23\textwidth}
    \centering
    \includegraphics[width=\linewidth]{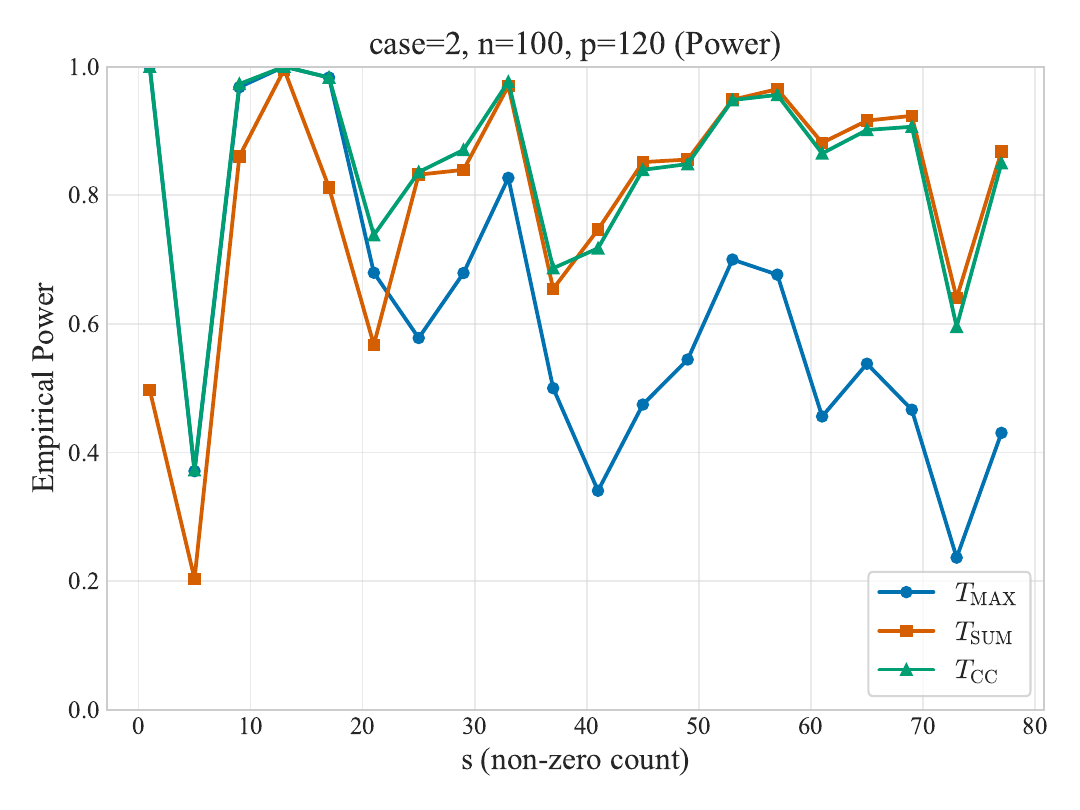}
\end{subfigure}
\hfill
\begin{subfigure}{0.23\textwidth}
    \centering
    \includegraphics[width=\linewidth]{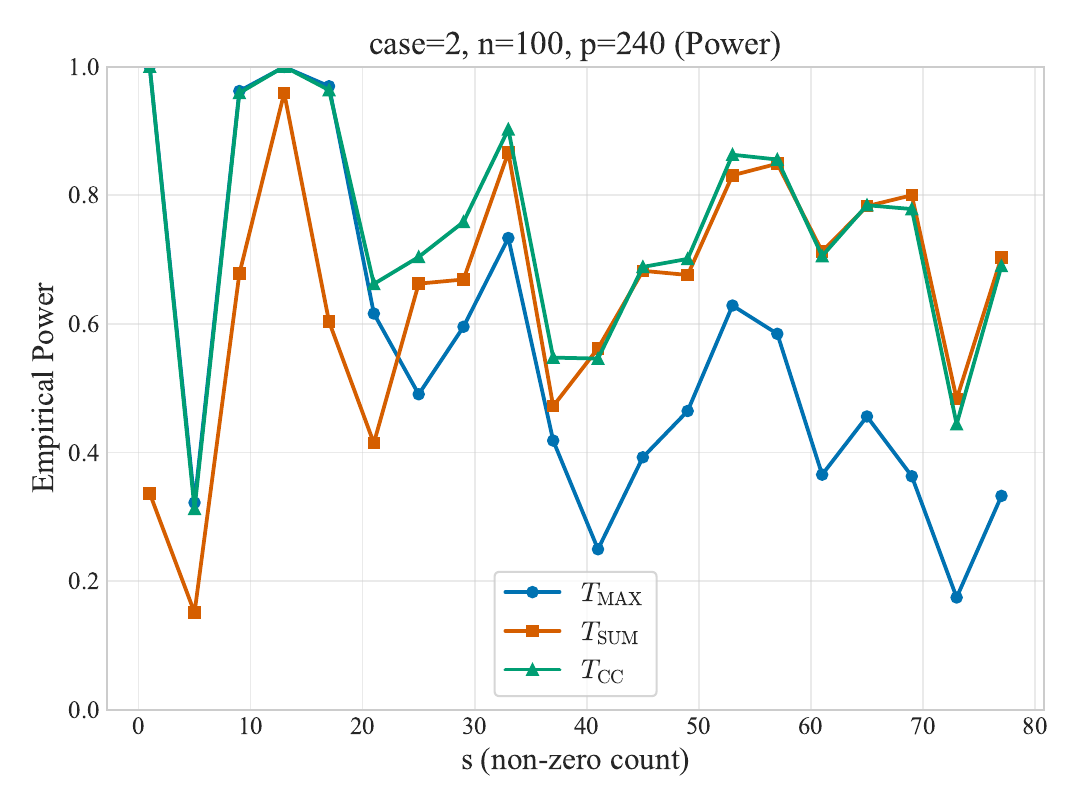}
\end{subfigure}
\hfill
\begin{subfigure}{0.23\textwidth}
    \centering
    \includegraphics[width=\linewidth]{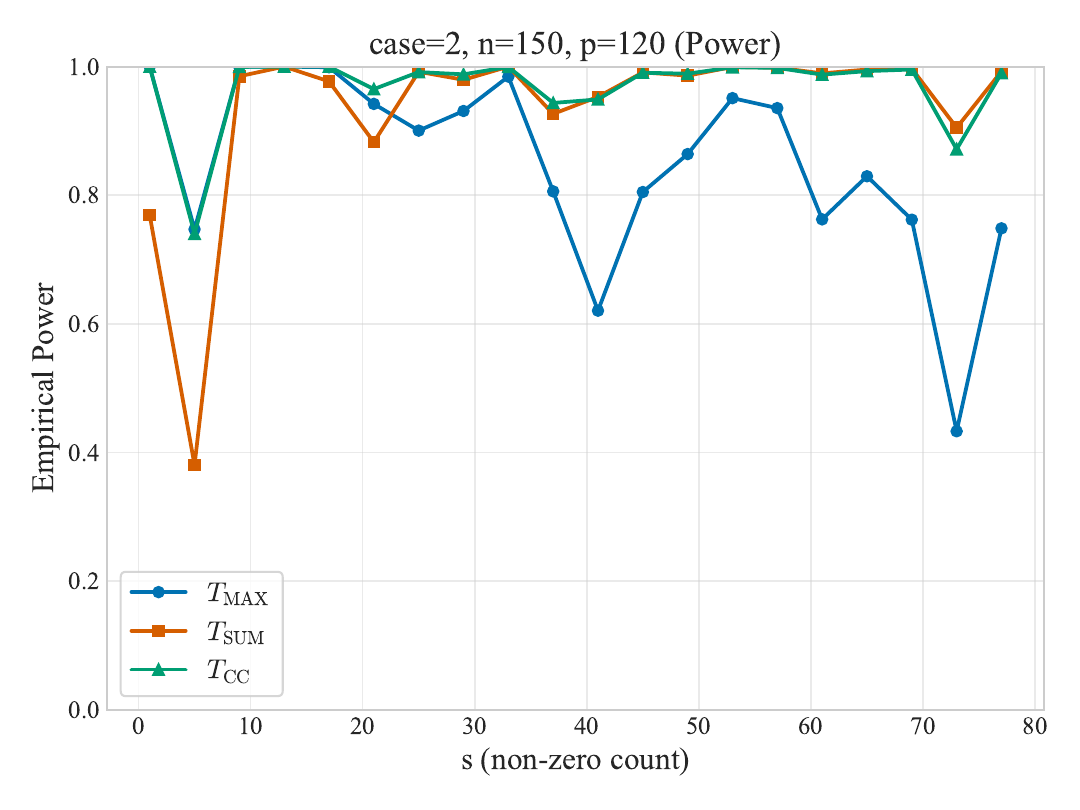}
\end{subfigure}
\hfill
\begin{subfigure}{0.23\textwidth}
    \centering
    \includegraphics[width=\linewidth]{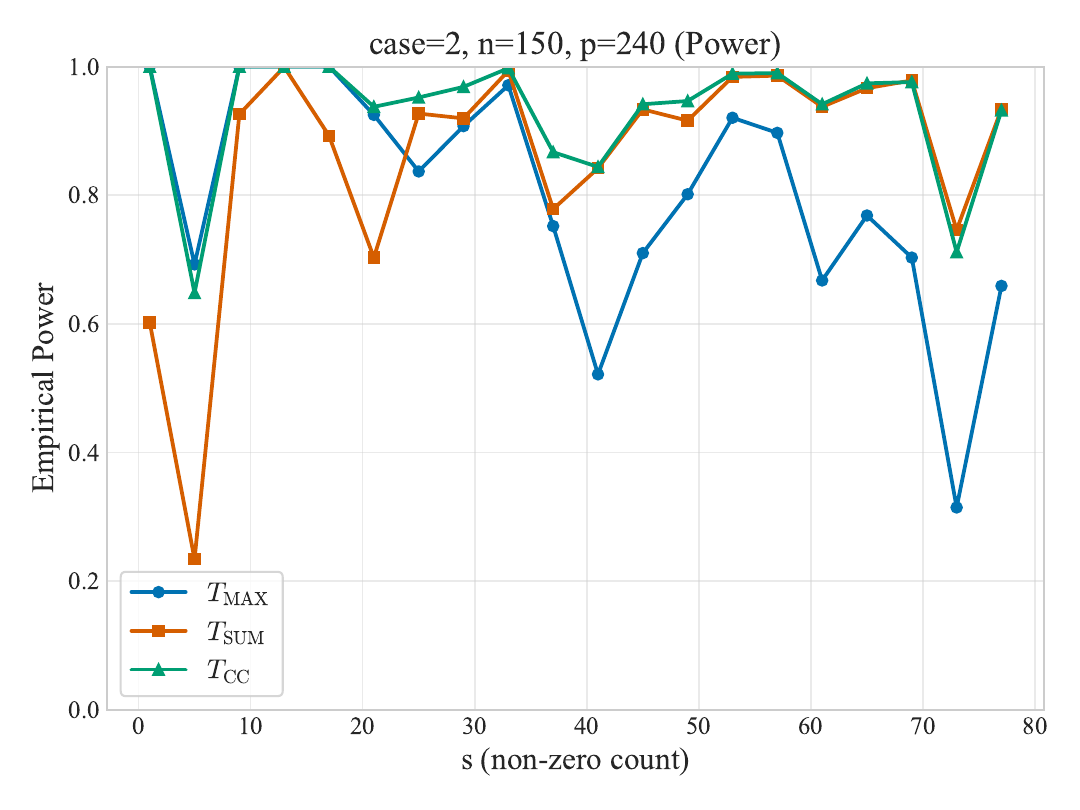}
\end{subfigure}

\vspace{0.3cm}
\begin{subfigure}{0.23\textwidth}
    \centering
    \includegraphics[width=\linewidth]{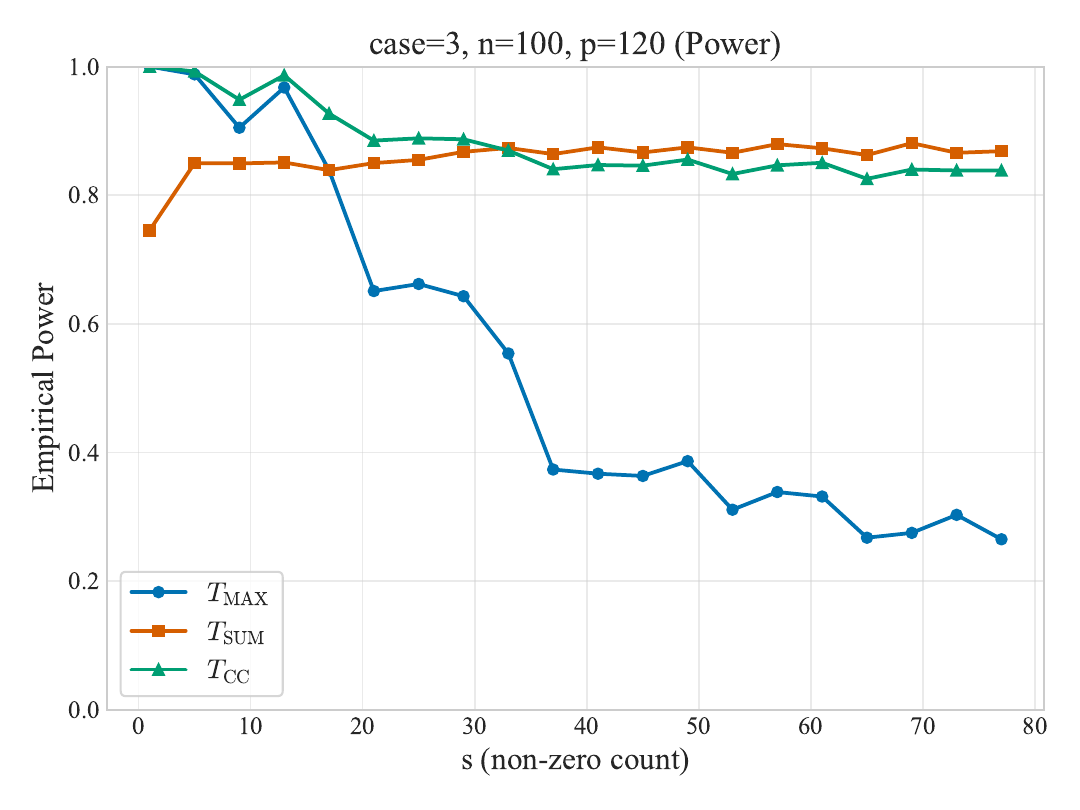}
\end{subfigure}
\hfill
\begin{subfigure}{0.23\textwidth}
    \centering
    \includegraphics[width=\linewidth]{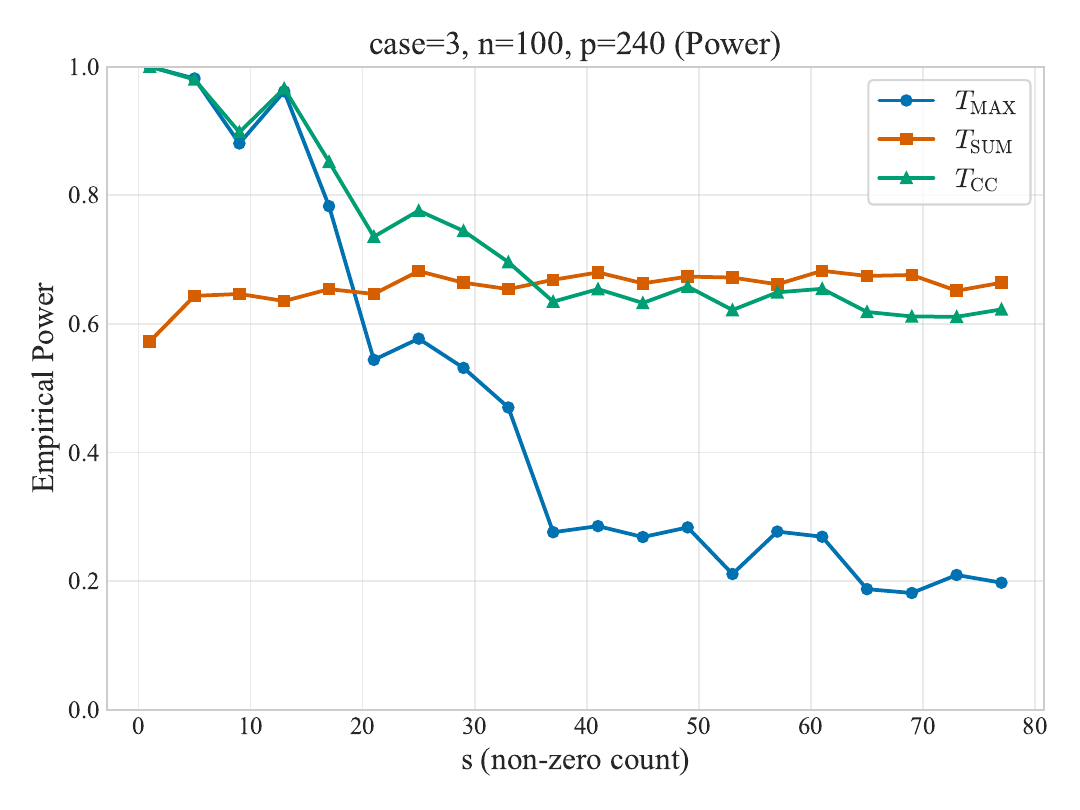}
\end{subfigure}
\hfill
\begin{subfigure}{0.23\textwidth}
    \centering
    \includegraphics[width=\linewidth]{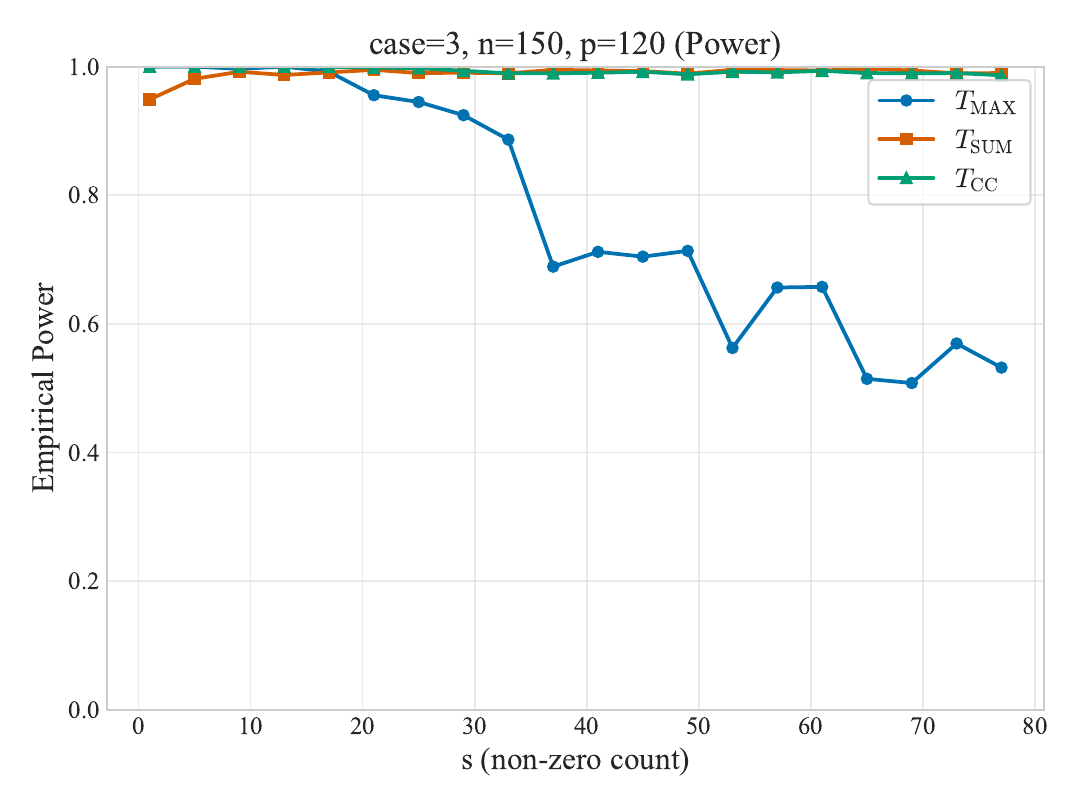}
\end{subfigure}
\hfill
\begin{subfigure}{0.23\textwidth}
    \centering
    \includegraphics[width=\linewidth]{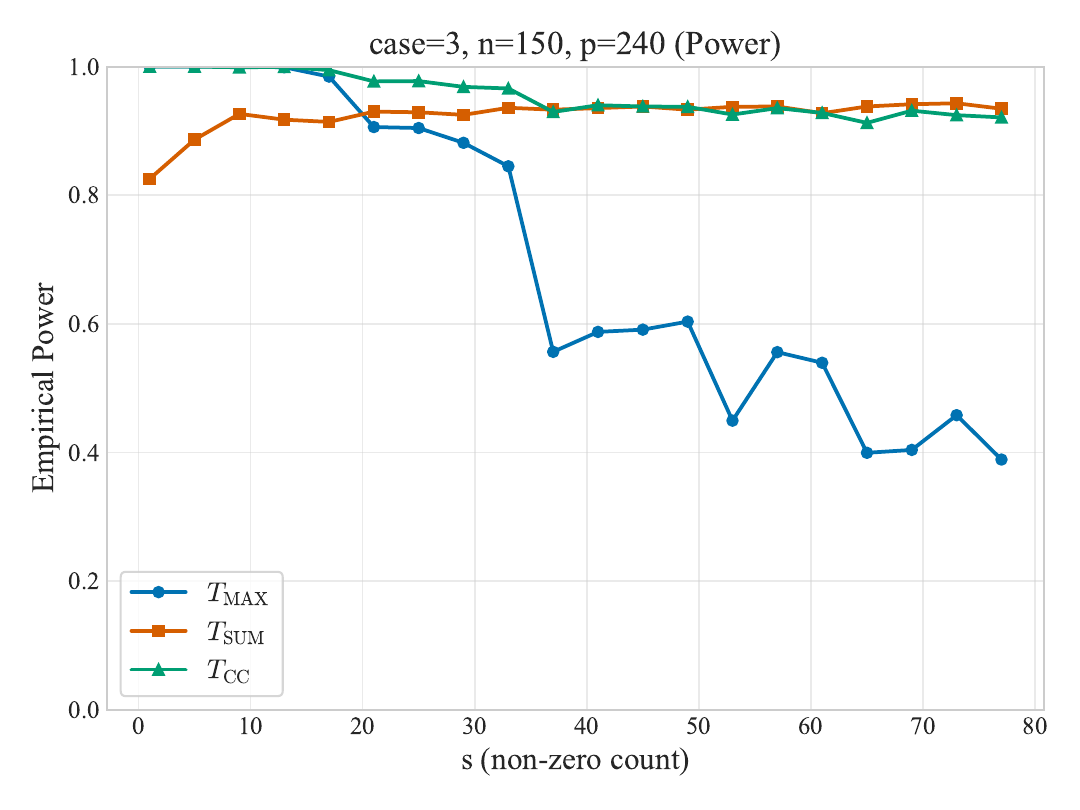}
\end{subfigure}

\caption{Empirical power as a function of $s$ for Cases~1--3 across varying $(n,p)$ 
settings under Logistic distribution ($\tau = 0.5$; 2000 replications).}
\label{fig:power_logistic}
\end{figure}

\begin{figure}[htbp]
\centering
\begin{subfigure}{0.23\textwidth}
    \centering
    \includegraphics[width=\linewidth]{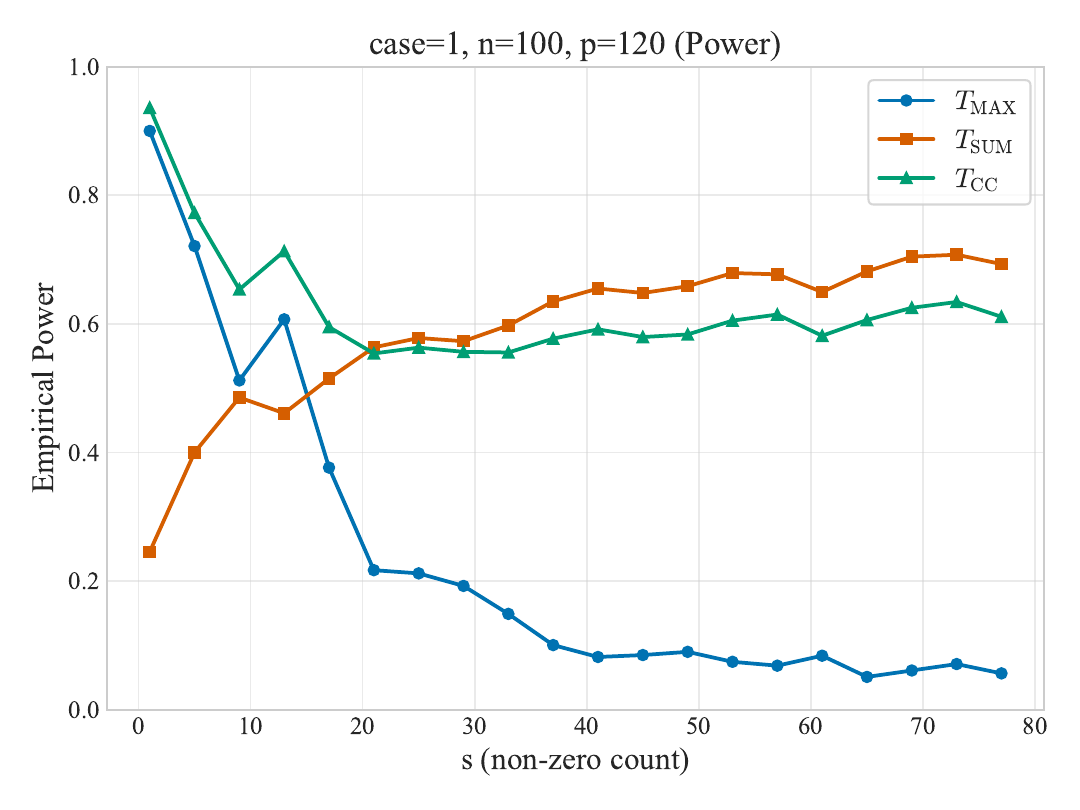}
\end{subfigure}
\hfill
\begin{subfigure}{0.23\textwidth}
    \centering
    \includegraphics[width=\linewidth]{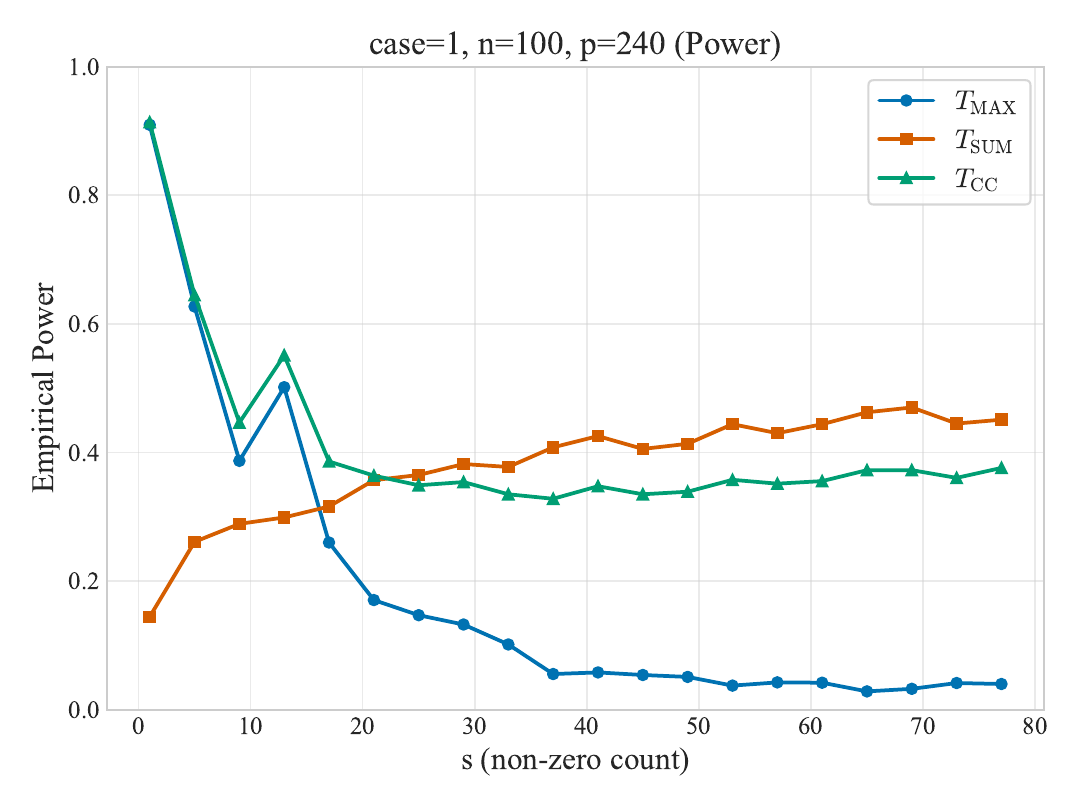}
\end{subfigure}
\hfill
\begin{subfigure}{0.23\textwidth}
    \centering
    \includegraphics[width=\linewidth]{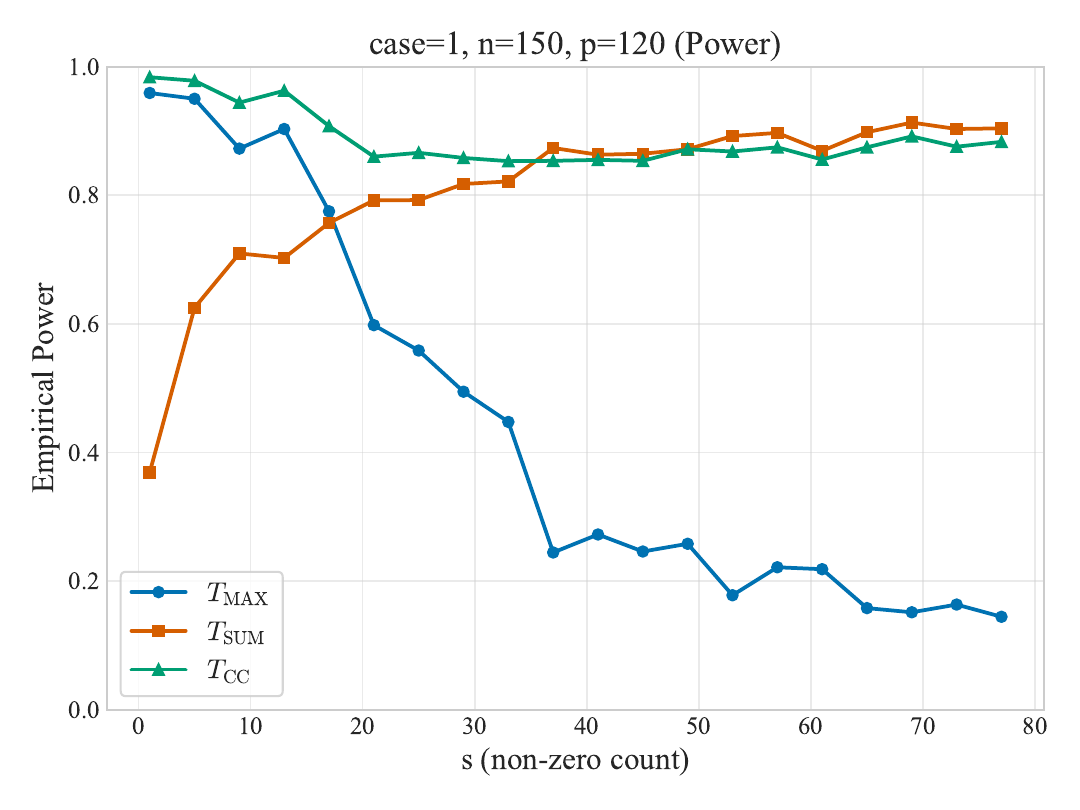}
\end{subfigure}
\hfill
\begin{subfigure}{0.23\textwidth}
    \centering
    \includegraphics[width=\linewidth]{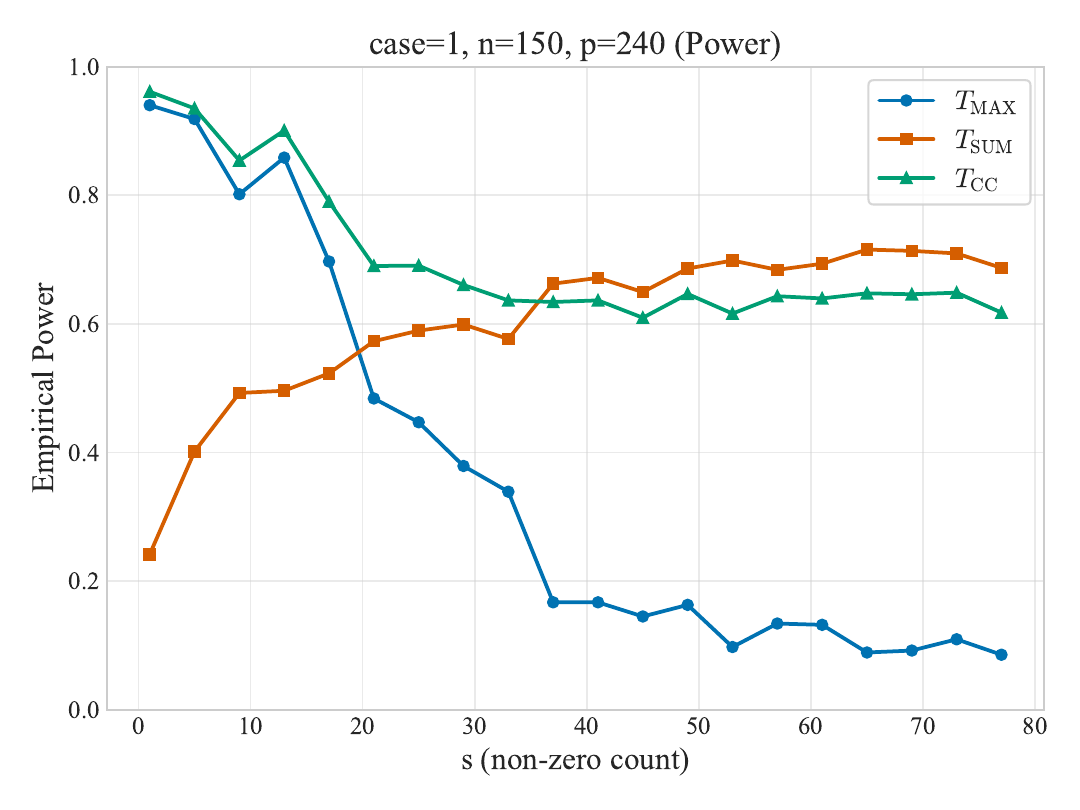}
\end{subfigure}

\vspace{0.3cm}
\begin{subfigure}{0.23\textwidth}
    \centering
    \includegraphics[width=\linewidth]{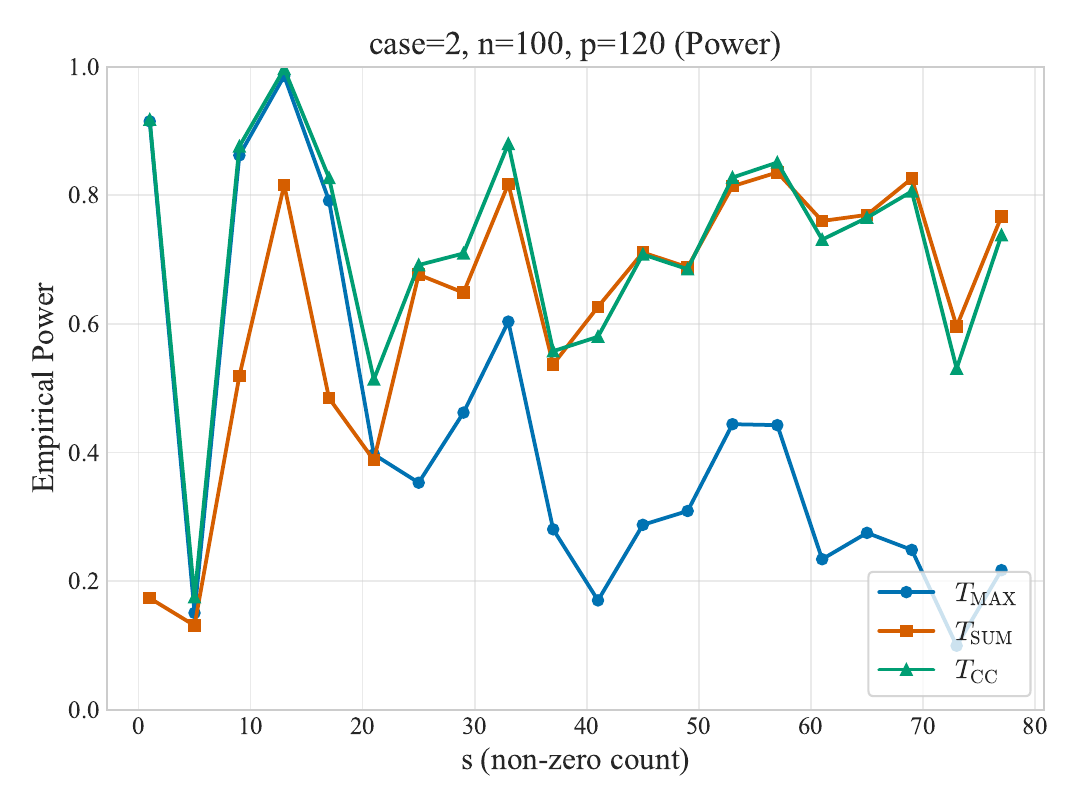}
\end{subfigure}
\hfill
\begin{subfigure}{0.23\textwidth}
    \centering
    \includegraphics[width=\linewidth]{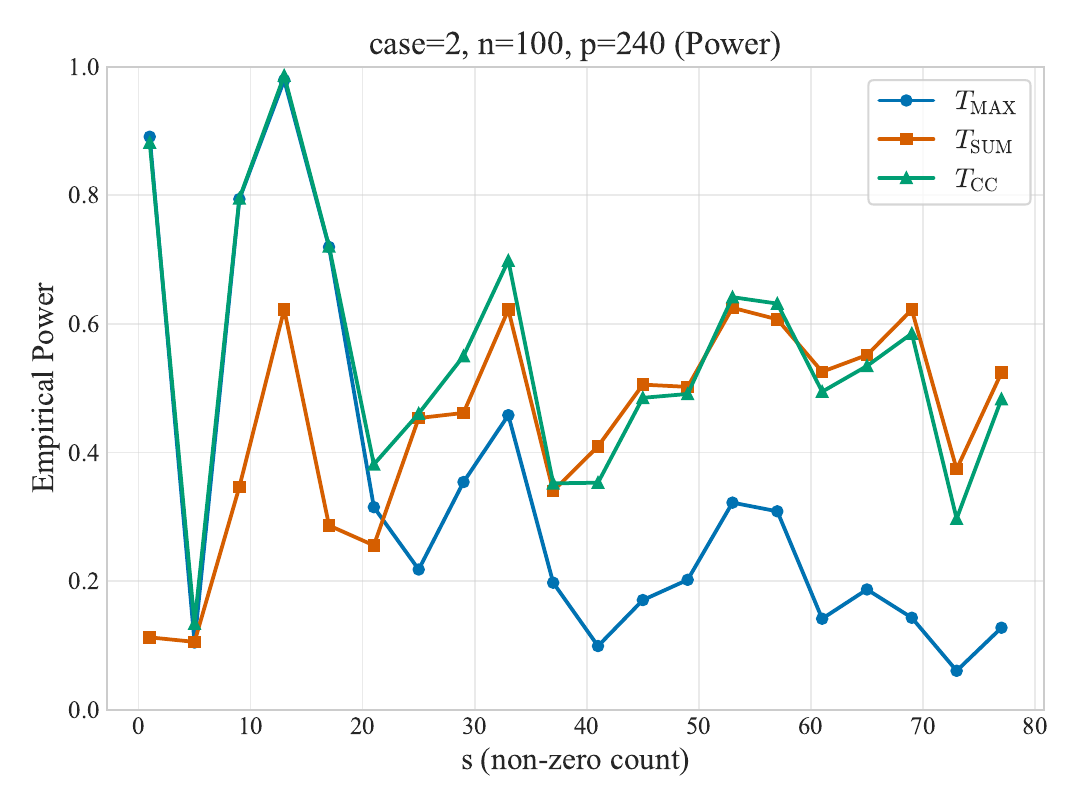}
\end{subfigure}
\hfill
\begin{subfigure}{0.23\textwidth}
    \centering
    \includegraphics[width=\linewidth]{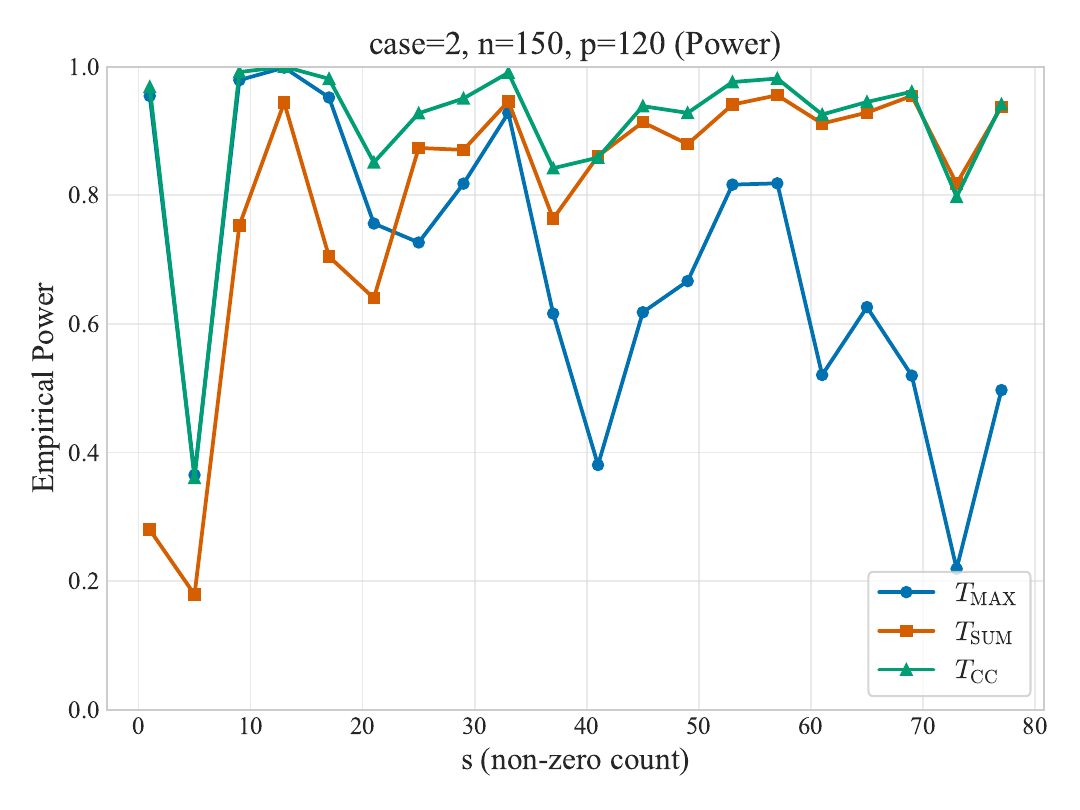}
\end{subfigure}
\hfill
\begin{subfigure}{0.23\textwidth}
    \centering
    \includegraphics[width=\linewidth]{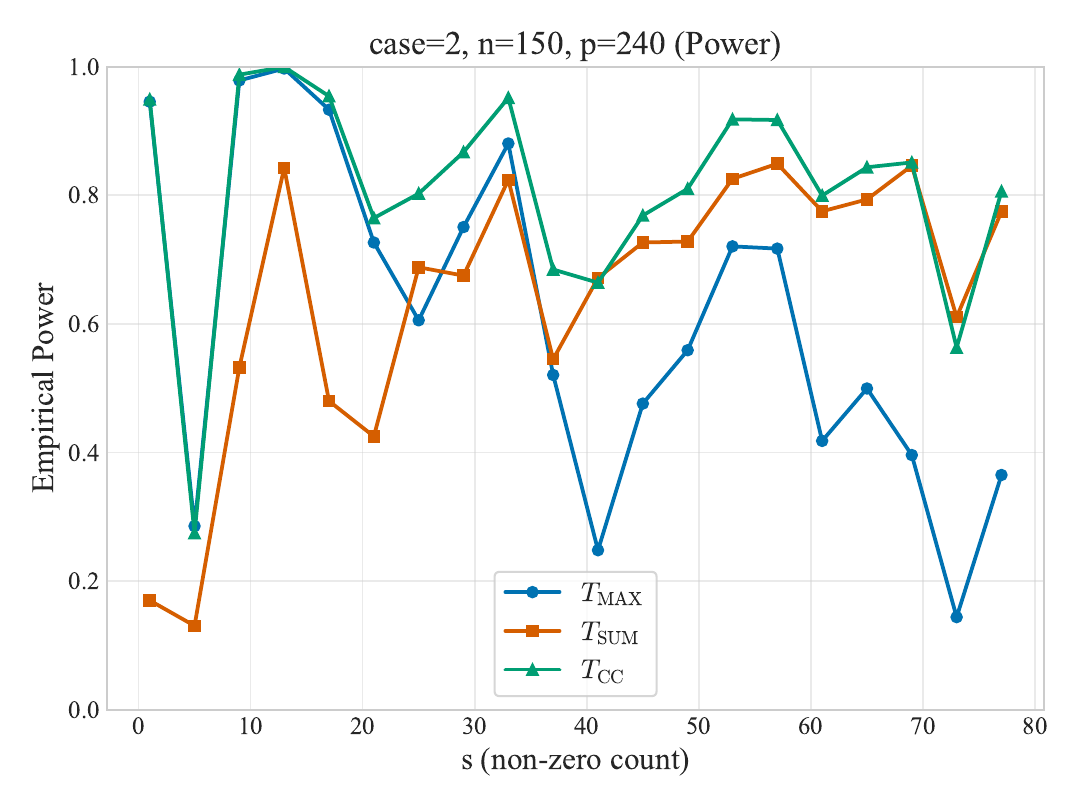}
\end{subfigure}

\vspace{0.3cm}
\begin{subfigure}{0.23\textwidth}
    \centering
    \includegraphics[width=\linewidth]{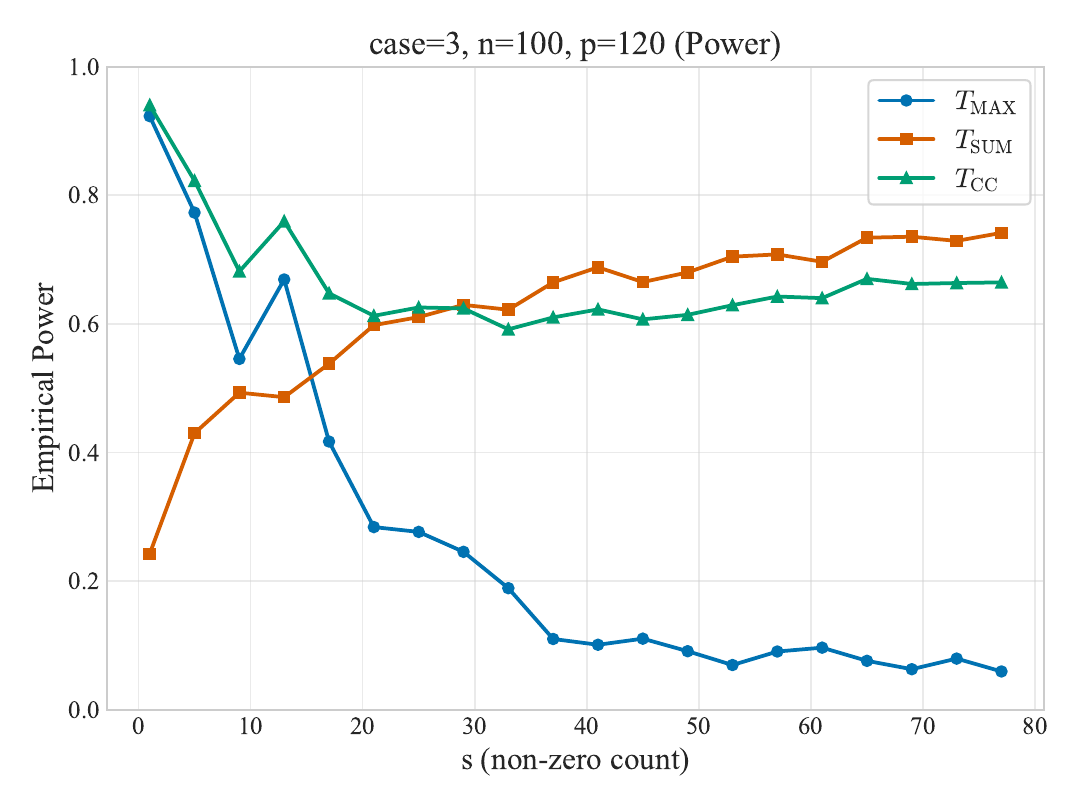}
\end{subfigure}
\hfill
\begin{subfigure}{0.23\textwidth}
    \centering
    \includegraphics[width=\linewidth]{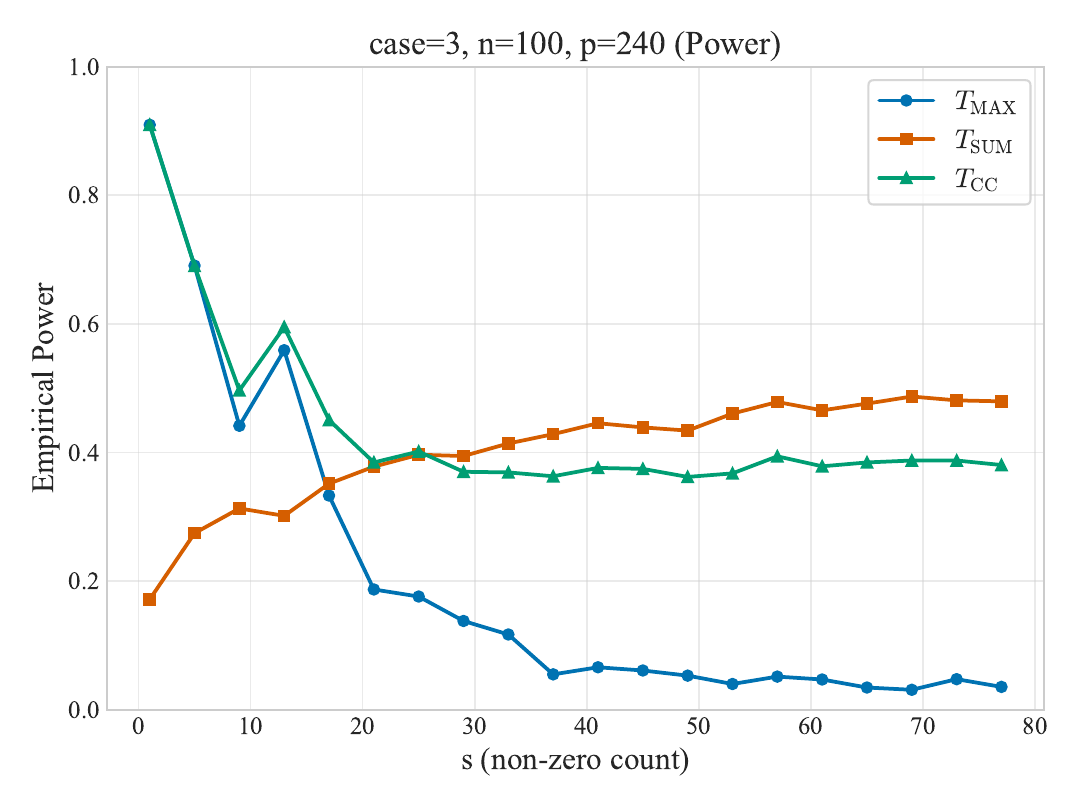}
\end{subfigure}
\hfill
\begin{subfigure}{0.23\textwidth}
    \centering
    \includegraphics[width=\linewidth]{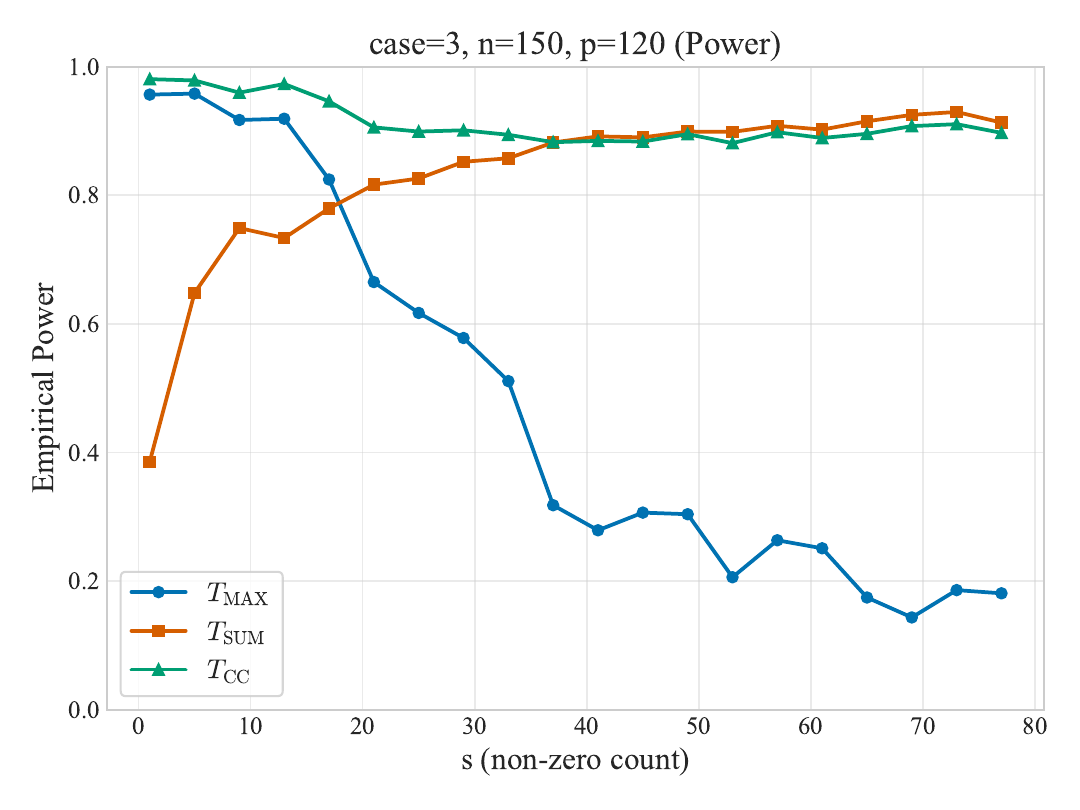}
\end{subfigure}
\hfill
\begin{subfigure}{0.23\textwidth}
    \centering
    \includegraphics[width=\linewidth]{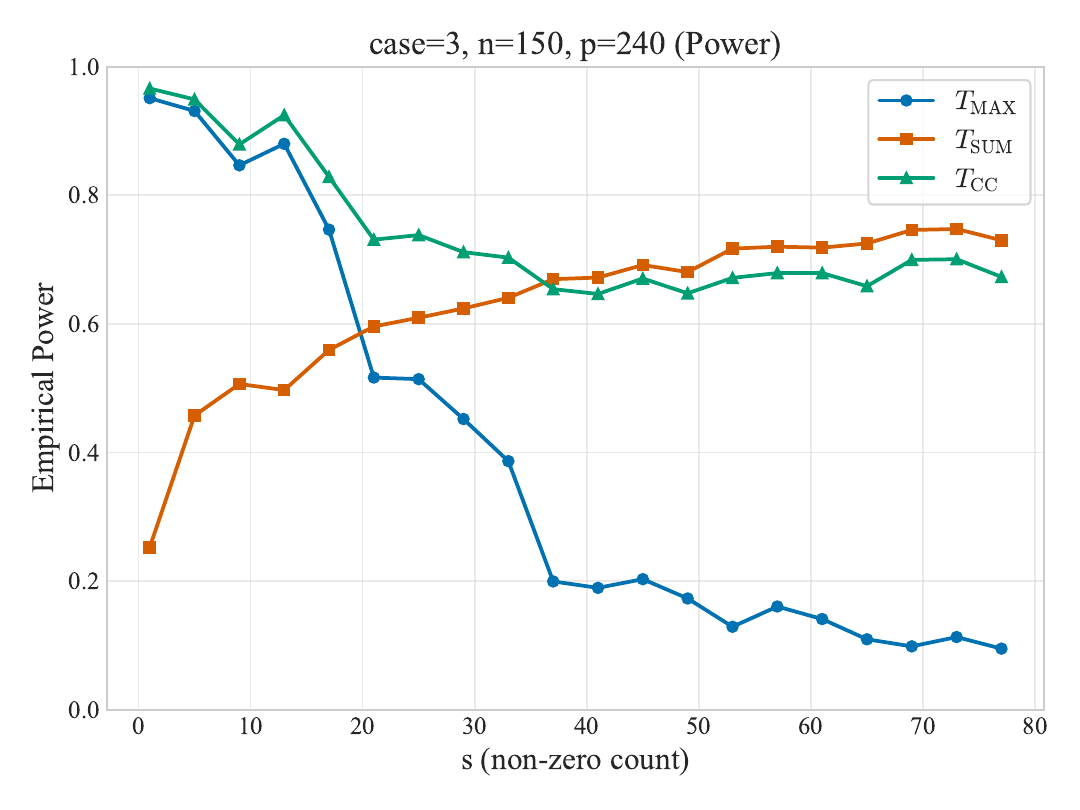}
\end{subfigure}

\caption{Empirical power as a function of $s$ for Cases~1--3 across varying $(n,p)$ 
settings under $t_2$ distribution ($\tau = 0.5$; 2000 replications).}
\label{fig:power_t2}
\end{figure}


The results reveal distinct regimes of optimality: the sum-type test $T_{\mathrm{SUM}}$ performs best when the signals are dense and spread across many coordinates, whereas the max-type test $T_{\mathrm{MAX}}$ is most powerful for extremely sparse alternatives where only a few components carry meaningful deviations. In contrast, the proposed combined test $T_{\mathrm{CC}}$ exhibits uniformly strong performance across the entire sparsity spectrum. It automatically adapts to both dense and sparse regimes, striking a desirable balance between sensitivity to global shifts and the ability to detect isolated but strong signals. Moreover, $T_{\mathrm{CC}}$ demonstrates substantial robustness under non-Gaussian and heavy-tailed designs, maintaining high power even when classical moment-based methods deteriorate. These findings underscore the practical advantage of $T_{\mathrm{CC}}$ as a versatile testing procedure suitable for a broad class of high-dimensional inference problems.

\section{Real data application}\label{sec:realdata}
Wave energy has emerged as a fast-developing and promising renewable resource, and the design of large-scale wave-energy farms is an increasingly important engineering and statistical problem \citep{neshat2020optimisation}. To conduct our empirical analysis, we use the publicly accessible \emph{Large-scale Wave Energy Farm} dataset, available at \href{https://archive.ics.uci.edu/dataset/882/large-scale+wave+energy+farm}{the UCI Machine Learning Repository}.
The dataset contains approximately 54{,}000 configurations of farms consisting of 49 wave-energy converters (WECs), along with an additional 9{,}600 configurations of farms with 100 WECs. 

In this study, we focus on the \texttt{WEC\_Sydney\_49} dataset, which contains 149 covariates and 17{,}964 instances. For each configuration, the dataset provides the Cartesian coordinates $(X_i, Y_i)$ of all 49 WEC units, individual power outputs $\text{Power}_i$, $q$-factors, and a range of geometric and spatial descriptors, including all pairwise inter-device distances and total farm power.
The resulting dataset exhibits high dimensionality, complex nonlinear interactions, and heavy-tailed response behavior, making it an ideal testbed for high-dimensional quantile regression. As illustrated in the right panel of Figure~\ref{fig:combined}, the total power output displays a heavily right-skewed and multimodal distribution, driven by heterogeneous hydrodynamic interactions across WEC configurations. This pronounced non-Gaussian structure highlights the limitations of mean regression for characterizing output variability. In contrast, quantile regression provides a more complete description of the conditional power distribution—particularly in the tails, which correspond to extreme high- and low-output scenarios—while allowing analysis directly on the original power scale. Such modeling is crucial for assessing the reliability, resilience, and extreme-event behavior of large-scale wave-energy farms.

In our analysis of the 49-WEC configuration, we designate the power output from the 13th WEC as the response variable,
\begin{align*}
Y = \text{Power}_{13}.
\end{align*}
This WEC occupies a hydrodynamically strategic interior position (as shown in the left panel of Figure \ref{fig:combined}), and its output reflects aggregated wake interactions, shadowing, and energy propagation patterns generated by surrounding devices. As such, $\text{Power}_{13}$ serves as a representative and scientifically meaningful quantity for modeling WEC-level performance.

To align the empirical setting with the theoretical structure of high-dimensional partially linear quantile regression, we partition the covariates into two groups:
\begin{enumerate}
    \item \textbf{Primary high-dimensional covariates}:
    \begin{align*}
    \mathbf{X} 
    = \bigl( \text{Power}_i : i \in \{1,\ldots,49\}\setminus\{13\} \bigr),
    \end{align*}
    consisting of the remaining 48 individual WEC power outputs.  
    These covariates encode nonlinear wake-interaction patterns among WECs and represent the main predictors of interest.  
    Their dimensionality and potential sparsity motivate the use of penalized high-dimensional quantile regression to perform variable selection and identify the dominant contributors to tail behavior.

    \item \textbf{Secondary physical descriptors}:
    \begin{align*}
    \mathbf{Z}
    = \Bigl\{ \, (X_i, Y_i)_{i=1}^{49},\, qW \Bigr\},
    \end{align*}
    including all WEC spatial coordinates and the farm-level $qW$-factor.  
    These variables provide structural and environmental context, capturing the underlying layout geometry, device spacing, and aggregate hydrodynamic interactions.  
    In our model, $\mathbf{Z}$ is treated as a set of non-focal (nuisance) covariates whose effects are absorbed through a flexible, low-dimensional adjustment without penalization.
\end{enumerate}

This decomposition balances engineering interpretability and statistical efficiency. The vector $\mathbf{Z}$ controls for global geometric and environmental conditions, while the high-dimensional vector $\mathbf{X}$ captures local interaction effects that drive the distribution of WEC-level output across quantiles. The resulting formulation enables a detailed assessment of how spatial layout, device interactions, and farm-level descriptors jointly influence both typical and extreme levels of power production.

By fitting a high-dimensional quantile regression model across a spectrum of quantile levels, we characterize how layout decisions affect not only the central tendency of the output but also the upper-tail efficiency (corresponding to high-yield operating conditions) and lower-tail energy loss scenarios. This analysis provides insights beyond those attainable with conventional optimization-based approaches, such as evolutionary placement strategies, which primarily target the maximization of average output.

\begin{figure}[htbp]
    \centering
    \begin{subfigure}{0.48\textwidth}
        \centering
        \includegraphics[width=\textwidth]{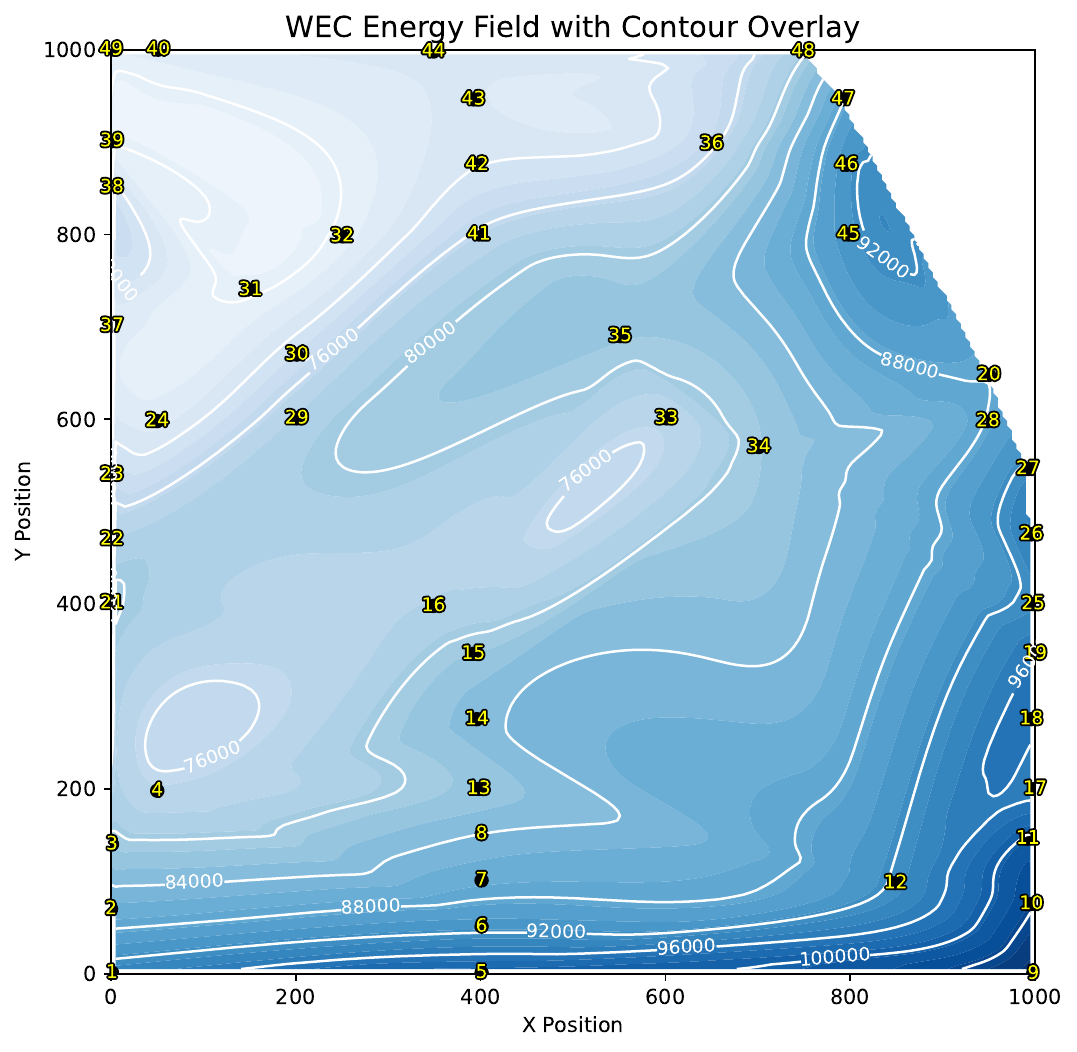}
        \label{fig:1}
    \end{subfigure}
    \hfill
    \begin{subfigure}{0.48\textwidth}
        \centering
        \includegraphics[width=\textwidth]{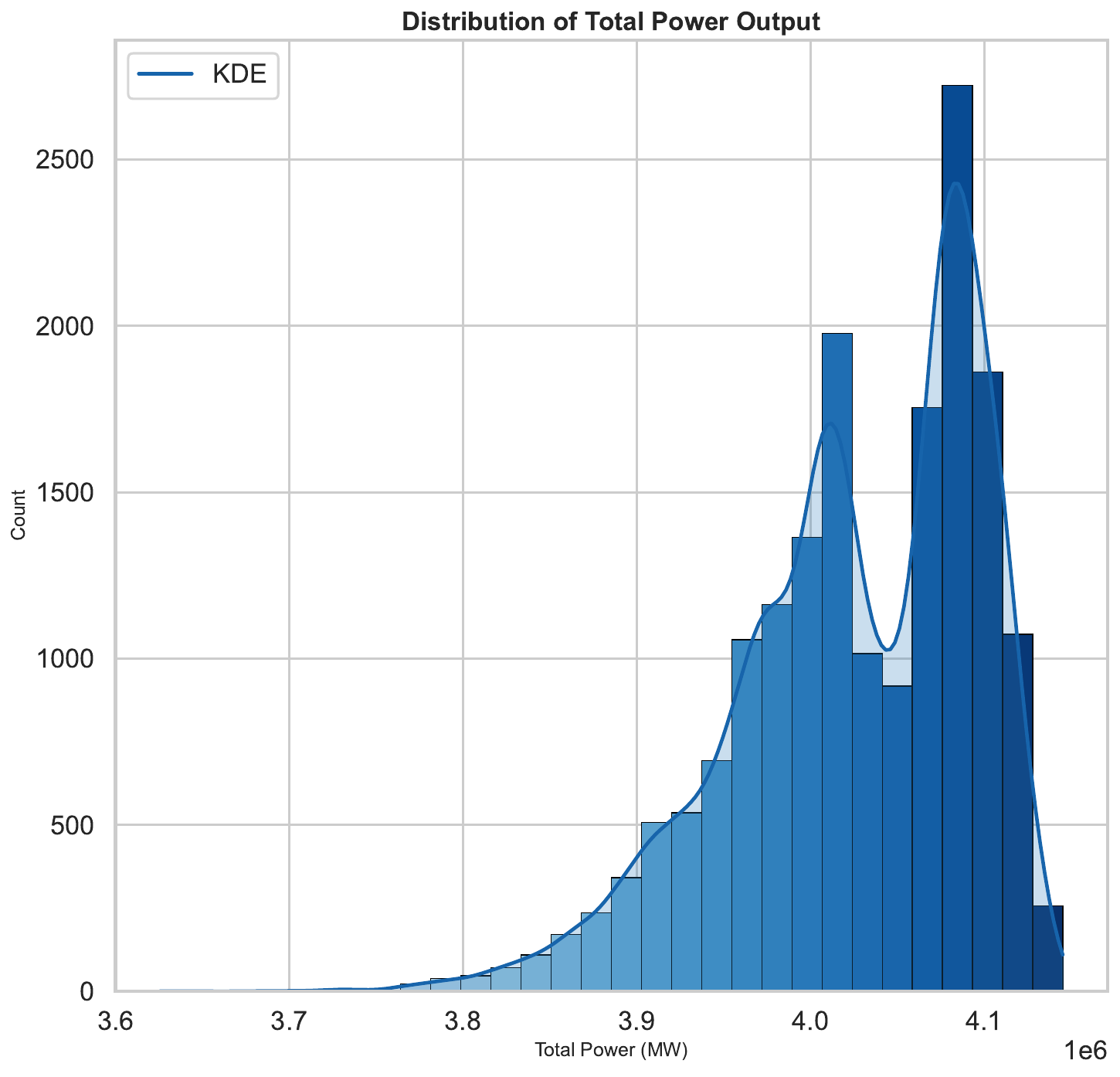}
        \label{fig:2}
    \end{subfigure}
    
    \caption{Visualization of wave energy farm performance. (a) shows the spatial power allocation of individual WEC devices, while (b) depicts the overall power output distribution, reflecting the collective generation behaviour of the farm.}
    \label{fig:combined}
\end{figure}

We conduct 1,000 replications, and in each replication we randomly draw a subsample of size 500 from the full dataset. The empirical rejection proportions of the competing tests across quantile levels are reported in Figure \ref{fig:rejection-rate}. Several findings emerge.

\begin{itemize}
\item Both the $T_{\mathrm{MAX}}$ and $T_{\mathrm{CC}}$ reject the null hypothesis with probability nearly one across all quantile levels. This persistent rejection indicates the presence of at least one highly influential neighboring WEC whose power output exerts a strong hydrodynamic impact on WEC13. Such a phenomenon is consistent with the well-documented inter-device radiation and diffraction coupling observed in large-scale wave-energy arrays. Importantly, the complete rejection achieved by $T_{\mathrm{CC}}$ demonstrates its superior sensitivity to complex local signals relative to traditional aggregation-based tests.

\item The $T_{\mathrm{SUM}}$ statistic exhibits a distinctly quantile-dependent rejection pattern. The highest rejection frequencies occur around the middle quantiles ($\tau\in[0.25,0.60]$, whereas substantially lower rejection rates appear at low and high quantiles. This suggests that aggregate coupling effects are strongest under moderate sea states, where interference and radiation interactions dominate the power redistribution across devices. By contrast, under extreme operating conditions—either very low or very high power output—the variability is driven predominantly by environmental fluctuations rather than inter-device interactions, yielding weaker aggregate effects.
\end{itemize}

These results highlight that while inter-device coupling is consistently detectable by the more sensitive tests such as $T_{\mathrm{MAX}}$ and $T_{\mathrm{CC}}$, the strength and nature of such coupling vary substantially across the conditional distribution of power output.

\begin{figure}[htbp]
    \centering
    \includegraphics[width=0.75\textwidth, height=0.36\textheight]{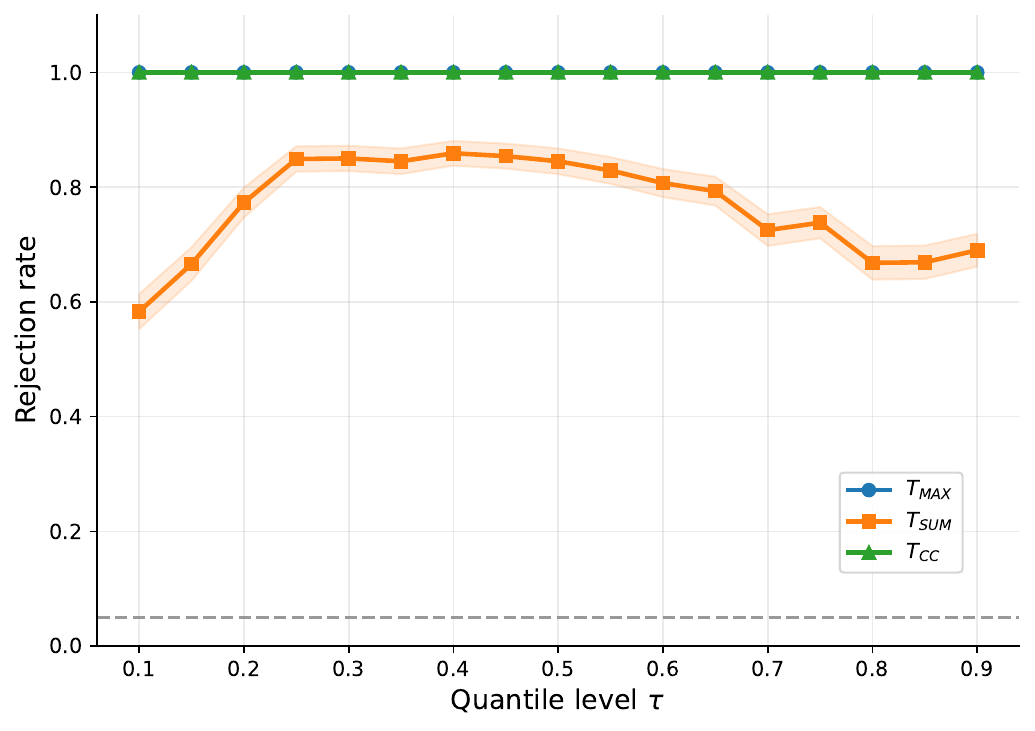}
    \caption{Rejection rates across quantile levels $\tau$ with pointwise 95\% confidence bands for the three test statistics $T_{\mathrm{MAX}}$, $T_{\mathrm{SUM}}$, and $T_{\mathrm{CC}}$.}
    \label{fig:rejection-rate}
\end{figure}

\section{Conclusion}\label{sec:discussion}

In this paper, we develop a new adaptive test for high-dimensional quantile regression by leveraging the asymptotic independence between the max-type and sum-type statistics. The resulting Cauchy combination test provides a simple yet effective way to integrate information across different sparsity regimes. Our theoretical results, together with simulation studies and real data analyses, demonstrate that the proposed method achieves reliable size control and competitive power under a wide range of settings.

Several extensions merit further investigation. One promising direction is to generalize the current framework to more flexible distributions for the covariates $\boldsymbol{U}$ beyond the Gaussian assumption. Another is to explore test statistics that bridge the gap between max-type and sum-type approaches, such as L-type \citep{ma2024adaptive} or other order-statistic-based combinations, which may offer additional adaptivity under intermediate sparsity levels. These developments would help broaden the applicability and robustness of high-dimensional inference in quantile regression.

\bibliographystyle{elsarticle-harv}

\bibliography{reference.bib}

\appendix

\section*{Appendix}
We analyze the joint distribution of the max-type and sum-type statistics. 
Let $\bPsi=(\hat\psi_1,\ldots,\hat\psi_n)^\top$ and denote by $\bX_{\cdot j}$ the $j$th column of $\bX$.
Define
\begin{align}
    M_p 
    &:= \max_{1\le j\le p} \bigl(\bW_{\cdot j}^\top \bPsi\bigr)^2, \label{eq:def_Mp}\\
    U_p 
    &:= \sum_{j=1}^p\left\{\bigl(\bX_{\cdot j}^\top \bPsi\bigr)^2-\sum_{i=1}^n X_{ij}^2\hat\psi_i^2\right\}. \label{eq:def_Up}
\end{align}
Note that
\begin{align}
    \sum_{j=1}^p\bigl(\bX_{\cdot j}^\top \bPsi\bigr)^2
    &= \sum_{j=1}^p\left(\sum_{i=1}^n X_{ij}\hat\psi_i\right)\left(\sum_{\ell=1}^n X_{\ell j}\hat\psi_\ell\right)
      = \sum_{i=1}^n\sum_{\ell=1}^n \bigl(\bX_i^\top\bX_\ell\bigr)\hat\psi_i\hat\psi_\ell, \label{eq:col_to_pair}\\
    \sum_{j=1}^p\sum_{i=1}^n X_{ij}^2\hat\psi_i^2
    &= \sum_{i=1}^n \bigl(\bX_i^\top \bX_i\bigr)\hat\psi_i^2. \label{eq:center_term}
\end{align}
Therefore,
\begin{equation}
    U_p
    = \sum_{i\neq \ell} \bigl(\bX_i^\top \bX_\ell\bigr)\hat\psi_i\hat\psi_\ell. \label{eq:Up_pairwise}
\end{equation}
Recalling that
\[
T_{\mathrm{SUM}}
= \frac{2}{n(n-1)}\sum_{i\neq \ell} \bigl(\bX_i^\top \bX_\ell\bigr)\hat\psi_i\hat\psi_\ell,
\]
we obtain the exact identity
\begin{equation}
    T_{\mathrm{SUM}}=\frac{2}{n(n-1)}\,U_p. \label{eq:Tsum_Up_identity}
\end{equation}
In the sequel, we work with the centered-and-scaled version of $U_p$:
\begin{equation}
    \widetilde{U_p} := \frac{U_p-\Esp(U_p)}{v_p},
    \qquad v_p := \sqrt{\Var(U_p)}. \label{eq:Up_standardized}
\end{equation}
By \eqref{eq:Tsum_Up_identity}, $\widetilde U_p$ is equivalent to the standardized $T_{\mathrm{SUM}}$ up to the deterministic factor $2/\{n(n-1)\}$, and hence the asymptotic behavior of $(T_{\mathrm{SUM}},T_{\mathrm{MAX}})$ can be studied through $(U_p,M_p)$.

\begin{lemma}[Conditional structure of $(\boldsymbol{\Psi},\boldsymbol{W},\boldsymbol{V})$]\label{lem:cond_struct}
Let
$$
\mathcal{F}=\sigma(\mathbf{Z},\boldsymbol{\varepsilon})
$$
be the $\sigma$-field generated by the covariate matrix
$$
\mathbf{Z}=(\boldsymbol{Z}_1^\top,\ldots,\boldsymbol{Z}_n^\top)^\top
$$
and the noise vector
$$
\boldsymbol{\varepsilon}=(\varepsilon_1,\ldots,\varepsilon_n)^\top.
$$
Under the null hypothesis $H_0:\boldsymbol{\beta}_\tau=\boldsymbol{0}$ and Assumptions \textnormal{(A1)}--\textnormal{(A3)}, the following statements hold.
\begin{enumerate}
\item[(i)] The score vector
$$
\boldsymbol{\Psi}=(\hat{\psi}_1,\ldots,\hat{\psi}_n)^\top
$$
is $\mathcal{F}$-measurable.

\item[(ii)] Conditional on $\mathcal{F}$ (equivalently, conditional on $\mathbf{Z}$ since $\boldsymbol{\varepsilon}$ is independent of $\mathbf{X}$), the design matrix
$$
\mathbf{X}=(\boldsymbol{X}_1^\top,\ldots,\boldsymbol{X}_n^\top)^\top
$$
admits the decomposition
\begin{equation}\label{eq:X_decomposition_cond_struct_std}
\mathbf{X}=\mathbf{M}(\mathbf{Z})+\mathbf{E},
\end{equation}
where $\mathbf{M}(\mathbf{Z})$ is deterministic given $\mathbf{Z}$, and the rows of
$$
\mathbf{E}=(\boldsymbol{E}_1^\top,\ldots,\boldsymbol{E}_n^\top)^\top
$$
are independent Gaussian vectors satisfying
$$
\boldsymbol{E}_i\mid \mathbf{Z}\sim N(\boldsymbol{0},\boldsymbol{\Sigma}_{x|z}),
\qquad i=1,\ldots,n,
$$
with $\boldsymbol{\Sigma}_{x|z}=\boldsymbol{\Sigma}_x-\boldsymbol{\Sigma}_{xz}\boldsymbol{\Sigma}_z^{-1}\boldsymbol{\Sigma}_{zx}$.

Let
$$
P_Z=\mathbf{Z}(\mathbf{Z}^\top\mathbf{Z})^{-1}\mathbf{Z}^\top,
\qquad
P_Z^\perp=\mathbf{I}_n-P_Z,
$$
and define the projected matrices
$$
\mathbf{V}=P_Z\mathbf{X},
\qquad
\mathbf{W}=P_Z^\perp\mathbf{X}.
$$
Then, conditional on $\mathcal{F}$,
\begin{equation}\label{eq:WV_decomposition_cond_struct_std}
\mathbf{W}=P_Z^\perp \mathbf{E},
\qquad
\mathbf{V}=P_Z\mathbf{M}(\mathbf{Z})+P_Z\mathbf{E},
\end{equation}
and $\mathbf{W}$ and $\mathbf{V}$ are independent Gaussian matrices.
\end{enumerate}
\end{lemma}

\begin{proof}
(i) Under $H_0\colon \beta = 0$, the $\tau$-quantile regression model reduces to
$$Y_i = Z_i^\trans \alpha + \varepsilon_i, \qquad i = 1,\dots,n,$$
where $Z_i^\trans$ denotes the $i$th row of $\bZ$ and $\alpha \in \mathbb{R}^q$ is the nuisance parameter. Let $\hat{\alpha}$ be the estimator of $\alpha$ obtained from the $\tau$-quantile regression of $Y$ on $\bZ$ under $H_0$. By standard quantile regression theory and Assumption (A1) on the density $f_{\varepsilon}$ (positivity, differentiability, and bounded derivative), $\hat{\alpha}$ is a measurable function of $(\bZ,Y)$. Therefore, the residuals
$$r_i = Y_i - Z_i^\trans \hat{\alpha}, \qquad i=1,\dots,n,$$
are $\mathcal{F}$-measurable.
By the construction of the score,
\begin{align*}
\hat{\psi}_i =
\begin{cases}
1-\tau, & r_i \ge 0,\\[2pt]
-\tau, & r_i < 0.
\end{cases}
\end{align*}
Since $r_i$ is a measurable function of $(Z_i,\varepsilon_i)$ and $\mathcal{F}=\sigma(\bZ,\varepsilon)$, it follows that $\hat{\psi}_i$ is $\mathcal{F}$-measurable.
Thus, the vector $\bPsi = (\hat{\psi}_1,\dots,\hat{\psi}_n)^\trans$ is $\mathcal{F}$-measurable and is therefore fixed conditional on $\mathcal{F}$. This proves (i).

(ii) By Assumption (A2), for each $i$,
$$
\begin{pmatrix}
\widetilde{\boldsymbol{Z}}_i\\
\boldsymbol{X}_i
\end{pmatrix}
\sim N\!\left(
\begin{pmatrix}
\boldsymbol{0}\\
\boldsymbol{0}
\end{pmatrix},
\begin{pmatrix}
\boldsymbol{\Sigma}_z & \boldsymbol{\Sigma}_{zx}\\
\boldsymbol{\Sigma}_{xz} & \boldsymbol{\Sigma}_x
\end{pmatrix}
\right),
$$
where $\boldsymbol{Z}_i=(1,\widetilde{\boldsymbol{Z}}_i^\top)^\top$.
By the multivariate normal conditioning formula,
$$
\boldsymbol{X}_i\mid \widetilde{\boldsymbol{Z}}_i
\sim
N\!\left(
\boldsymbol{\mu}_{x|z}(\widetilde{\boldsymbol{Z}}_i),
\boldsymbol{\Sigma}_{x|z}
\right),
\qquad
\boldsymbol{\mu}_{x|z}(\widetilde{\boldsymbol{Z}}_i)
=
\boldsymbol{\Sigma}_{xz}\boldsymbol{\Sigma}_z^{-1}\widetilde{\boldsymbol{Z}}_i,
$$
where
$$
\boldsymbol{\Sigma}_{x|z}
=\boldsymbol{\Sigma}_x-\boldsymbol{\Sigma}_{xz}\boldsymbol{\Sigma}_z^{-1}\boldsymbol{\Sigma}_{zx}.
$$
Therefore we can write
$$
\boldsymbol{X}_i=\boldsymbol{\mu}_{x|z}(\widetilde{\boldsymbol{Z}}_i)+\boldsymbol{E}_i,
$$
where, conditional on $\mathbf{Z}$, the vectors $\boldsymbol{E}_1,\ldots,\boldsymbol{E}_n$ are independent and satisfy
$$
\boldsymbol{E}_i\mid \mathbf{Z}\sim N(\boldsymbol{0},\boldsymbol{\Sigma}_{x|z}).
$$
Stacking over $i$ yields \eqref{eq:X_decomposition_cond_struct_std} with
$$
\mathbf{M}(\mathbf{Z})
=
\left(
\boldsymbol{\mu}_{x|z}(\widetilde{\boldsymbol{Z}}_1)^\top,\ldots,
\boldsymbol{\mu}_{x|z}(\widetilde{\boldsymbol{Z}}_n)^\top
\right)^\top
$$
and $\mathbf{E}=(\boldsymbol{E}_1^\top,\ldots,\boldsymbol{E}_n^\top)^\top$.

Next, for each coordinate $j\in\{1,\ldots,p\}$, the $j$th column of $\mathbf{M}(\mathbf{Z})$ has the form $\mathbf{Z}\boldsymbol{c}_j$ for some vector $\boldsymbol{c}_j\in\mathbb{R}^q$ (because $\boldsymbol{\mu}_{x|z}(\widetilde{\boldsymbol{Z}}_i)$ is linear in $\widetilde{\boldsymbol{Z}}_i$ and the first component of $\boldsymbol{Z}_i$ equals $1$).
Hence each column of $\mathbf{M}(\mathbf{Z})$ lies in the column space of $\mathbf{Z}$, and thus
$$
P_Z^\perp \mathbf{M}(\mathbf{Z})=\boldsymbol{0}.
$$
Using $\mathbf{W}=P_Z^\perp\mathbf{X}$ and $\mathbf{V}=P_Z\mathbf{X}$ together with $\mathbf{X}=\mathbf{M}(\mathbf{Z})+\mathbf{E}$ gives \eqref{eq:WV_decomposition_cond_struct_std}.

Finally, conditional on $\mathbf{Z}$ (and hence on $\mathcal{F}$), $\mathbf{E}$ is a Gaussian matrix and $(P_Z^\perp\mathbf{E},P_Z\mathbf{E})$ is jointly Gaussian as a linear transformation of $\mathbf{E}$.
Moreover, since $P_Z$ and $P_Z^\perp$ are orthogonal projections, we have $P_Z^\perp P_Z=\boldsymbol{0}$, and therefore the cross-covariance between $P_Z^\perp\mathbf{E}$ and $P_Z\mathbf{E}$ is zero conditional on $\mathbf{Z}$.
Hence $P_Z^\perp\mathbf{E}$ and $P_Z\mathbf{E}$ are conditionally uncorrelated jointly Gaussian and thus conditionally independent given $\mathbf{Z}$.
Because $P_Z\mathbf{M}(\mathbf{Z})$ is deterministic given $\mathbf{Z}$, adding it to $P_Z\mathbf{E}$ does not affect independence.
Therefore, $\mathbf{W}$ and $\mathbf{V}$ are independent conditional on $\mathcal{F}$.
\end{proof}

To analyze the dependence between the Max-type and Sum-type statistics, we fix an index set $\Ind \subset \{1,\dots,p\}$ with $|\Ind| = d$, that corresponds to the locations of potential extreme values, and let $\Jnd = \{1, \dots, p\} \setminus \Ind$. For any $j \in \Jnd$, we employ the Gaussian linear decomposition of $\bW_{\cdot j}$ conditional on $\bW_{\Ind}$:
\begin{equation} \label{eq:decomp}
    \bW_{\cdot j} = \sum_{k \in \Ind} \gamma_{jk} \bW_{\cdot k} + \bW_{\cdot j}^*, 
    \qquad \text{where } \bW_{\cdot j}^* \perp \{\bW_{\cdot k}\}_{k \in \Ind}.
\end{equation}
The coefficients $\{\gamma_{jk}\}_{k \in \Ind}$ are determined by the Gaussian structure with covariance matrix $\boldsymbol{\Sigma}_{x|z}$. Assumption \textnormal{(A3)} implies that the correlation matrix $\mathbf{R}_{x|z} = (r_{ij})_{1\le i,j\le p}$ satisfies $\max_{1\le i,j\le p} |r_{ij}| \le r_0 < 1$ and the eigenvalues of $\boldsymbol{\Sigma}_{x|z}$ are uniformly bounded. Consequently, there exists a constant $C_{\lambda} < \infty$, depending only on $r_0$ and the eigenvalue bounds in \textnormal{(A3)}, such that
$$\sum_{j \in \Jnd} \gamma_{jk}^2 \le C_{\lambda}
\quad \text{for all } k \in \Ind.$$

\begin{lemma}[Uniform negligibility of remainder terms] \label{lem:negligibility}
    Using $\bX_{\cdot j} = \bW_{\cdot j} + \bV_{\cdot j}$, we decompose
    \begin{align*}
        U_p &= \sum_{j=1}^p \left\{ (\bX_{\cdot j}^\trans \bPsi)^2 - \sum_{i=1}^n X_{ij}^2 \hat{\psi}_i^2 \right\} \\
        &= A_1 + 2A_2 + A_3 - (A_{4,W} + 2A_{4,WV} + A_{4,V}),
    \end{align*}
    where
    \begin{align*}
    A_{1} &= \sum_{j=1}^p (\bW_{\cdot j}^\trans \bPsi)^2, \\
    A_{2} &= \sum_{j=1}^p (\bW_{\cdot j}^\trans \bPsi)(\bV_{\cdot j}^\trans \bPsi), \\
    A_{3} &= \sum_{j=1}^p(\bV_{\cdot j}^\trans \bPsi)^2, \\
    A_{4,W} &= - \sum_{j=1}^p \sum_{i=1}^n (W_{ij}^2 \hat{\psi}_i^2), \\
    A_{4,WV} &= - \sum_{j=1}^p \sum_{i=1}^n (W_{ij}V_{ij} \hat{\psi}_i^2), \\
    A_{4,V} &= - \sum_{j=1}^p \sum_{i=1}^n (V_{ij}^2 \hat{\psi}_i^2).
    \end{align*}
    Under Assumptions \textnormal{(A1)}–\textnormal{(A4)}, $v_p = \sqrt{\Var(U_p)} \asymp \sqrt{p}$.

    For any index set $\Ind \subset \{1,\dots,p\}$ with $|\Ind|=d$, let $\Jnd = \{1,\dots,p\}\setminus\Ind$ and consider the decomposition
    $$
        \bW_{\cdot j} = \sum_{k\in\Ind}\gamma_{jk}\bW_{\cdot k} + \bW_{\cdot j}^*, 
        \qquad j\in\Jnd,
    $$
    where $\bW_{\cdot j}^*$ is independent of $\{\bW_{\cdot k}\}_{k\in\Ind}$ and $\sum_{j\in\Jnd}\gamma_{jk}^2 \le C_\lambda$ for all $k\in\Ind$, with $C_\lambda<\infty$ depending only on \textnormal{(A3)}. Denote $Q_k = \bW_{\cdot k}^\trans\bPsi$ for $k\in\Ind$.

    For each component $A_1, A_2, A_3, A_{4,W}, A_{4,WV}, A_{4,V}$, we further decompose
    $$
        A_\ell = A_{\ell,\Ind}^\perp + \Theta_{\ell,\Ind},
    $$
    where $A_{\ell,\Ind}^\perp$ depends only on $\{\bW_{\cdot j}^*\}_{j\in\Jnd}$, $\bV$ and $\bPsi$. We then define
    $$
    \begin{aligned}
    A_{\Ind}^\perp 
    &:= A_{1,\Ind}^\perp + 2A_{2,\Ind}^\perp + A_{3,\Ind}^\perp
        - \bigl(A_{4,W,\Ind}^\perp + 2A_{4,WV,\Ind}^\perp + A_{4,V,\Ind}^\perp\bigr),\\
    \Theta_{\Ind}
    &:= \Theta_{1,\Ind} + 2\Theta_{2,\Ind}
        - \bigl(\Theta_{4,W,\Ind} + 2\Theta_{4,WV,\Ind}\bigr),
    \end{aligned}
    $$
    so that $U_p = A_{\Ind}^\perp + \Theta_{\Ind}$.

    Then there exist constants $C>0$ and $p_0\ge 3$ such that, with
    $$
        \epsilon_p := \frac{(\log p)^C}{\sqrt{p}} \to 0,
    $$
    for every $t\ge 1$ and all $p\ge p_0$,
    $$
        \sup_{\Ind:|\Ind|=d} P\bigl(|\Theta_{\Ind}| \ge \epsilon_p v_p\bigr) 
        \le \frac{1}{p^t}.
    $$
\end{lemma}
\begin{proof}
Fix an arbitrary index set $\Ind$ with $|\Ind|=d$ and denote its complement by $\Jnd$.
All constants below may depend on $(\Sigma_z,\Sigma_x,\Sigma_{xz})$, the correlation bound $r_0<1$ in (A3), and the density constants in (A1), but not on $p$ or on the choice of $\Ind$. Recall $Q_k:=\bW_{\cdot k}^\top\bPsi$ for $k\in\Ind$.
We do \emph{not} assume that $\Ind$ contains the maximizer of $T_{\mathrm{MAX}}$.
Instead, since $d$ is fixed, we establish a uniform (in $\Ind$) high-probability bound for $\max_{k\in\Ind}|Q_k|$ and $\sum_{k\in\Ind}Q_k^2$.

Conditional on $\mathcal F=\sigma(\bZ,\varepsilon)$, $\bPsi$ is fixed by Lemma~\ref{lem:cond_struct}(i), and $\bW_{\cdot k}$ is a centered Gaussian vector (as a linear transform of the Gaussian design) with conditional covariance matrix bounded in operator norm by a constant multiple of $I_n$ under (A3).
Hence, for each fixed $k$,
\begin{equation}\label{eq:Qk_subG}
Q_k\mid \mathcal F \ \text{is centered sub-Gaussian with} \ 
\Var(Q_k\mid \mathcal F)\le C\,\|\bPsi\|^2 \asymp C n,
\end{equation}
where we used $\|\bPsi\|^2=\sum_{i=1}^n \hat\psi_i^2\asymp n$ since $\hat\psi_i\in\{-\tau,1-\tau\}$.

Therefore, for any $u>0$, by a conditional Gaussian tail bound and then taking expectation,
\begin{equation}\label{eq:Qk_tail}
\sup_{1\le k\le p} P\bigl(|Q_k|\ge u\sqrt{n}\bigr)\le 2\exp(-c u^2)
\end{equation}
for some constant $c>0$.

Now fix any $t\ge 1$. Let $u_p:=\sqrt{(t+d+2)\log p}$.
Using a union bound over all $k\in\Ind$ and then over all $\Ind$ with $|\Ind|=d$ (note that the number of such sets is at most $p^d$), we obtain
\begin{align}
\sup_{\Ind:|\Ind|=d}
P\Bigl(\max_{k\in\Ind}|Q_k|\ge u_p\sqrt{n}\Bigr)
&\le \sup_{\Ind:|\Ind|=d}\sum_{k\in\Ind} P\bigl(|Q_k|\ge u_p\sqrt{n}\bigr)\nonumber\\
&\le d \cdot 2\exp(-c u_p^2)
\le \frac{1}{p^{t+d}}, \label{eq:uniform_maxQ}
\end{align}
for all sufficiently large $p$.
Consequently, with probability at least $1-p^{-(t+d)}$ uniformly over $\Ind$,
\begin{equation}\label{eq:uniform_sumQ}
\sum_{k\in\Ind}Q_k^2 \le d\max_{k\in\Ind}Q_k^2 \le d\,u_p^2\,n \asymp n\log p.
\end{equation}

In what follows, we work on the high-probability event in \eqref{eq:uniform_sumQ} when bounding terms involving $\{Q_k\}_{k\in\Ind}$, and all resulting bounds will be uniform over $\Ind$ with $|\Ind|=d$.
We now prove each component separately.

\medskip
\noindent\textbf{1. Remainder from $A_1$.}
For $j\in\Jnd$,
$$
    \bW_{\cdot j}^\trans\bPsi
    = \bW_{\cdot j}^{*\,\trans}\bPsi + \sum_{k\in\Ind}\gamma_{jk}Q_k.
$$
Hence
$$
\begin{aligned}
    (\bW_{\cdot j}^\trans\bPsi)^2
    &= (\bW_{\cdot j}^{*\,\trans}\bPsi)^2
      + \Bigl(\sum_{k\in\Ind}\gamma_{jk}Q_k\Bigr)^2
      + 2(\bW_{\cdot j}^{*\,\trans}\bPsi)\Bigl(\sum_{k\in\Ind}\gamma_{jk}Q_k\Bigr).
\end{aligned}
$$
Define
$$
    A_{1,\Ind}^\perp := \sum_{j\in\Jnd}(\bW_{\cdot j}^{*\,\trans}\bPsi)^2,
$$
and
$$
    \Theta_{1,\Ind}^{\mathrm{drift}} 
    := \sum_{j\in\Jnd}\Bigl(\sum_{k\in\Ind}\gamma_{jk}Q_k\Bigr)^2,
    \qquad
    \Theta_{1,\Ind}^{\mathrm{cross}} 
    := 2\sum_{j\in\Jnd} (\bW_{\cdot j}^{*\,\trans}\bPsi)
                         \Bigl(\sum_{k\in\Ind}\gamma_{jk}Q_k\Bigr).
$$
The remaining part $\sum_{j\in\Ind}Q_j^2$ is also absorbed into the remainder. Thus
$$
    A_1 = A_{1,\Ind}^\perp + \Theta_{1,\Ind},
    \qquad
    \Theta_{1,\Ind} := \Theta_{1,\Ind}^{\mathrm{drift}} + \Theta_{1,\Ind}^{\mathrm{cross}} + \sum_{j\in\Ind}Q_j^2.
$$

By Cauchy–Schwarz inequality, we have
$$
\begin{aligned}
    \Theta_{1,\Ind}^{\mathrm{drift}}
    &= \sum_{j\in\Jnd}\Bigl(\sum_{k\in\Ind}\gamma_{jk}Q_k\Bigr)^2\\
    &\le \sum_{j\in\Jnd}\Bigl(\sum_{k\in\Ind}\gamma_{jk}^2\Bigr)\Bigl(\sum_{k\in\Ind}Q_k^2\Bigr)\\
    &= \Bigl(\sum_{k\in\Ind}Q_k^2\Bigr)
       \Bigl(\sum_{k\in\Ind}\sum_{j\in\Jnd}\gamma_{jk}^2\Bigr)
     \le d C_\lambda \sum_{k\in\Ind}Q_k^2.
\end{aligned}
$$
where the last inequality uses the bound $\sum_{j\in\Jnd} \gamma_{jk}^2 \le C_\lambda$ for all $k\in\Ind$.

Using $\sum_{k\in\Ind}Q_k^2 = O_{P}(\log p)$, we obtain
$$
    \Theta_{1,\Ind}^{\mathrm{drift}} = O_{P}(\log p).
$$
By standard Gaussian concentration inequalities for $Q_k$, 
for any fixed $t \ge 1$ there exists a constant $C_1>0$ such that
$$
    \sup_{\Ind} P\Bigl(|\Theta_{1,\Ind}^{\mathrm{drift}}| \ge C_1 \log p\Bigr)
    \le \frac{1}{p^t}
$$
for all $p$ sufficiently large.

Define
$$
    Z_j^* := \bW_{\cdot j}^{*\,\trans}\bPsi,
    \qquad
    L_j := \sum_{k\in\Ind}\gamma_{jk}Q_k,
    \quad j\in\Jnd.
$$
Conditional on $\mathcal{F}$ and $\{\bW_{\cdot k}\}_{k\in\Ind}$, the coefficients $L_j$ are fixed and $\{Z_j^*\}_{j\in\Jnd}$ are centered Gaussian variables with
$$
    \Var(Z_j^* \mid \mathcal{F}) \le C_2 \|\bPsi\|^2 \asymp C_2 n
$$
for some constant $C_2>0$. Using the bounded eigenvalues of $\boldsymbol{\Sigma}_{x|z}$ and the correlation bound in (A3), we further have
$$
    \sum_{j\in\Jnd}L_j^2 \le C_3 \sum_{k\in\Ind}Q_k^2
    = O_{P}(\log p).
$$
It follows that
$$
    \Var\bigl(\Theta_{1,\Ind}^{\mathrm{cross}} \mid \mathcal{F}, \{\bW_{\cdot k}\}_{k\in\Ind}\bigr)
    \le C_4 \|\bPsi\|^2 \sum_{j\in\Jnd}L_j^2
    = O_{P}(n\log p),
$$
and therefore
$$
    \Theta_{1,\Ind}^{\mathrm{cross}} = O_{P}(\sqrt{n\log p}).
$$
By a Gaussian tail bound for linear forms, for every $t\ge 1$ there exists $C_5>0$ such that
$$
    \sup_{\Ind} P\Bigl(|\Theta_{1,\Ind}^{\mathrm{cross}}| \ge C_5\sqrt{n\log p}\Bigr)
    \le \frac{1}{p^{t}}
$$
for all sufficiently large $p$.

Since $d$ is fixed,
$$
    \sum_{j\in\Ind}Q_j^2 = O_{P}(\log p),
$$
and a union bound with the usual Gaussian tail yields, for some $C_6>0$,
$$
    \sup_{\Ind} P\Bigl(\sum_{j\in\Ind}Q_j^2 \ge C_6 \log p\Bigr)
    \le \frac{1}{p^t}
$$
for all large $p$.

Combining the three bounds and using $v_p\asymp\sqrt{p}$, 
there exists a sufficiently large constant $C>0$ such that
$$
    C_1\log p + C_5\sqrt{n\log p} + C_6\log p
    \le \frac{\epsilon_p}{6} v_p
$$
for all sufficiently large $p$, which implies
$$
    \sup_{\Ind} P\Bigl(|\Theta_{1,\Ind}| \ge \tfrac{\epsilon_p}{6} v_p\Bigr)
    \le \frac{1}{p^t}.
$$

\medskip
\noindent\textbf{2. Remainder from $A_2$.}
We have
$$
    A_2 = \sum_{j=1}^p (\bW_{\cdot j}^\trans\bPsi)(\bV_{\cdot j}^\trans\bPsi)
        = \sum_{j\in\Jnd} (\bW_{\cdot j}^\trans\bPsi)(\bV_{\cdot j}^\trans\bPsi)
          + \sum_{j\in\Ind} (\bW_{\cdot j}^\trans\bPsi)(\bV_{\cdot j}^\trans\bPsi).
$$
For $j\in\Jnd$,
$$
    \bW_{\cdot j}^\trans\bPsi
    = \bW_{\cdot j}^{*\,\trans}\bPsi + \sum_{k\in\Ind}\gamma_{jk}Q_k,
$$
so
$$
\begin{aligned}
    (\bW_{\cdot j}^\trans\bPsi)(\bV_{\cdot j}^\trans\bPsi)
    &= (\bW_{\cdot j}^{*\,\trans}\bPsi)(\bV_{\cdot j}^\trans\bPsi)
       + \Bigl(\sum_{k\in\Ind}\gamma_{jk}Q_k\Bigr)(\bV_{\cdot j}^\trans\bPsi).
\end{aligned}
$$
Define
$$
    A_{2,\Ind}^\perp := \sum_{j\in\Jnd} (\bW_{\cdot j}^{*\,\trans}\bPsi)(\bV_{\cdot j}^\trans\bPsi),
$$
and
$$
    \Theta_{2,\Ind}
    := \sum_{j\in\Jnd}\Bigl(\sum_{k\in\Ind}\gamma_{jk}Q_k\Bigr)(\bV_{\cdot j}^\trans\bPsi)
       + \sum_{j\in\Ind}(\bW_{\cdot j}^\trans\bPsi)(\bV_{\cdot j}^\trans\bPsi),
$$
so that $A_2 = A_{2,\Ind}^\perp + \Theta_{2,\Ind}$.

Let
$$
    G_k := \sum_{j\in\Jnd}\gamma_{jk}\bV_{\cdot j}^\trans\bPsi + \bV_{\cdot k}^\trans\bPsi,
    \qquad k\in\Ind.
$$
Then $\Theta_{2,\Ind}$ can be written as
$$
    \Theta_{2,\Ind} = \sum_{k\in\Ind} Q_k G_k.
$$
Conditional on $\mathcal{F}$, the $G_k$ are Gaussian with mean zero and variance
$$
    \Var(G_k\mid\mathcal{F}) \le C_7 \|\bPsi\|^2 = O(n),
$$
for some $C_7>0$, using Assumption (A2) and the bounded eigenvalues of $\boldsymbol{\Sigma}_{x|z}$ in (A3). Thus $G_k = O_{P}(\sqrt{n})$ uniformly in $k,\Ind$. Combining this with $|Q_k|=O_{P}(\sqrt{\log p})$ gives
$$
    |\Theta_{2,\Ind}| 
    \le \sum_{k\in\Ind}|Q_k||G_k|
    = O_{P}(\sqrt{n\log p}).
$$
Moreover, by a Gaussian tail bound and a union bound over $k\in\Ind$, for any fixed $t \ge 1$ there exists a constant $C_8>0$ such that
$$
    \sup_{\Ind} P\Bigl(|\Theta_{2,\Ind}| \ge C_8\sqrt{n\log p}\Bigr)
    \le \frac{1}{p^{t}}
$$
for all $p$ sufficiently large. 
By taking $C$ sufficiently large and using $v_p \asymp \sqrt{p}$ and (A4) as before, we then obtain
$$
    \sup_{\Ind} P\Bigl(|\Theta_{2,\Ind}| \ge \tfrac{\epsilon_p}{6}v_p\Bigr)
    \le \frac{1}{p^{t}}.
$$

\medskip
\noindent\textbf{3. Remainder from $A_3$ and $A_{4,V}$.}
The components
$$
    S_3 = \sum_{j=1}^p (\bV_{\cdot j}^\trans\bPsi)^2,
    \qquad
    S_{4,V} = \sum_{j=1}^p\sum_{i=1}^n V_{ij}^2\hat{\psi}_i^2
$$
do not involve $\bW$. Therefore we simply set
$$
    S_{3,\Ind}^\perp := S_3,\quad \Theta_{3,\Ind} := 0,
    \qquad
    S_{4,V,\Ind}^\perp := S_{4,V},\quad \Theta_{4,V,\Ind} := 0,
$$
and they do not contribute to $\Theta_{\Ind}$.

\medskip
\noindent\textbf{4. Remainder from $A_{4,W}$ and $A_{4,WV}$.}
We focus on explicitly bounding the remainder arising from the centering term involving $\bW$, as this contains the dependencies on the extreme set $\Ind$. Recall
$$
    A_{4,W} = -\sum_{j=1}^p \sum_{i=1}^n W_{ij}^2 \hat{\psi}_i^2.
$$
Decompose the sum over indices $j$ into $\Ind$ and $\Jnd$. The sum over $\Ind$ consists of only $d$ terms (where $d$ is fixed) and is trivially bounded by $O_{P}(n \log p)$, which is negligible compared to $v_p \asymp \sqrt{p}n$ (assuming unscaled inputs) or $v_p \asymp \sqrt{p}$ (if inputs are scaled). We focus on the sum over $\Jnd$.
For $j \in \Jnd$, substitute the decomposition $W_{ij} = W_{ij}^* + \delta_{ij}$, where $\delta_{ij} := \sum_{k \in \Ind} \gamma_{jk} W_{ik}$. Then,
$$
    W_{ij}^2 = (W_{ij}^*)^2 + 2 W_{ij}^* \delta_{ij} + \delta_{ij}^2.
$$
Define the \emph{ideal} component independent of $\Ind$ as
$$
    A_{4,W,\Ind}^\perp := -\sum_{j \in \Jnd} \sum_{i=1}^n (W_{ij}^*)^2 \hat{\psi}_i^2.
$$
The remainder is $\Theta_{4,W,\Ind} := A_{4,W} - A_{4,W,\Ind}^\perp = \Theta_{4,W}^{\mathrm{cross}} + \Theta_{4,W}^{\mathrm{drift}} + R_{\Ind}$, where $R_{\Ind}$ collects the negligible terms for $j \in \Ind$. The dominant remainder terms are:
$$
    \Theta_{4,W}^{\mathrm{cross}} := -2 \sum_{j \in \Jnd} \sum_{i=1}^n \hat{\psi}_i^2 W_{ij}^* \delta_{ij},
    \qquad
    \Theta_{4,W}^{\mathrm{drift}} := -\sum_{j \in \Jnd} \sum_{i=1}^n \hat{\psi}_i^2 \delta_{ij}^2.
$$
We bound these two terms separately.

\textbf{(i) Bound for the Drift Term $\Theta_{4,W}^{\mathrm{drift}}$.}
Note that $|\hat{\psi}_i| \le \max(\tau, 1-\tau) < 1$. Thus,
$$
    |\Theta_{4,W}^{\mathrm{drift}}|
    \le \sum_{i=1}^n \sum_{j \in \Jnd} \delta_{ij}^2
    = \sum_{i=1}^n \sum_{j \in \Jnd} \left( \sum_{k \in \Ind} \gamma_{jk} W_{ik} \right)^2.
$$
By the Cauchy-Schwarz inequality and the uniform bound on correlations from Assumption (A3) ($\sum_{j \in \Jnd} \gamma_{jk}^2 \le C_{\lambda}$), we have
$$
    \sum_{j \in \Jnd} \left( \sum_{k \in \Ind} \gamma_{jk} W_{ik} \right)^2
    \le d \sum_{k \in \Ind} W_{ik}^2 \left( \sum_{j \in \Jnd} \gamma_{jk}^2 \right)
    \le d \, C_{\lambda} \sum_{k \in \Ind} W_{ik}^2.
$$
Summing over $i=1,\dots,n$ yields
$$
    |\Theta_{4,W}^{\mathrm{drift}}|
    \le d \, C_{\lambda} \sum_{k \in \Ind} \|\bW_{\cdot k}\|^2.
$$
Recall that $\|\bW_{\cdot k}\|^2 \sim \chi^2_n$ (scaled by variance). Thus $\|\bW_{\cdot k}\|^2 = O_{P}(n)$. Consequently, $|\Theta_{4,W}^{\mathrm{drift}}| = O_{P}(n)$.
Given Assumption (A4), $p$ grows exponentially relative to $n$ (specifically $\log p \ll n^{1/4}$ implies $p$ is large, but usually in high-dimensional testing $v_p$ scales with $\sqrt{p}$).
Regardless of the specific scaling of entries (whether $W_{ij} \asymp 1$ or $W_{ij} \asymp n^{-1/2}$), the ratio of the drift term (involving $d$ columns) to the variance of the sum statistic (involving $p$ columns) vanishes.
Specifically, since $v_p$ scales with the aggregate variance of $p$ terms, we have $|\Theta_{4,W}^{\mathrm{drift}}| / v_p \asymp n / (n\sqrt{p}) = p^{-1/2} \to 0$.
Standard $\chi^2$ tail bounds imply that for any $t \ge 1$,
$$
    \sup_{\Ind} P\left( |\Theta_{4,W}^{\mathrm{drift}}| \ge C_{9,a} \, n \right) \le \frac{1}{p^t},
$$
which is sufficiently small compared to $\epsilon_p v_p$.

\textbf{(ii) Bound for the Cross Term $\Theta_{4,W}^{\mathrm{cross}}$.}
We write
$$
    \Theta_{4,W}^{\mathrm{cross}} = -2 \sum_{j \in \Jnd} Z_j^*,
    \quad \text{where } Z_j^* := \sum_{i=1}^n (\hat{\psi}_i^2 \delta_{ij}) W_{ij}^*.
$$
Conditional on $\mathcal{F}$ and $\Ind$ (and thus on $\bW_{\Ind}$ and $\delta_{ij}$), the variables $\{W_{ij}^*\}_{j \in \Jnd, i=1\dots n}$ are independent Gaussians with mean zero.
Thus, conditional on $\bW_{\Ind}$, $\Theta_{4,W}^{\mathrm{cross}}$ is a sum of independent Gaussian variables with mean zero and conditional variance
$$
    \sigma_{\text{cross}}^2
    := \Var\left( \Theta_{4,W}^{\mathrm{cross}} \,\middle|\, \bW_{\Ind} \right)
    = 4 \sum_{j \in \Jnd} \sum_{i=1}^n (\hat{\psi}_i^2 \delta_{ij})^2 \Var(W_{ij}^*).
$$
Using $\hat{\psi}_i^4 \le 1$, $\Var(W_{ij}^*) \le \Sigma_{jj} \le C$, and the bound for $\delta_{ij}^2$ derived above:
$$
    \sigma_{\text{cross}}^2
    \le 4 C \sum_{i=1}^n \sum_{j \in \Jnd} \delta_{ij}^2
    \le 4 C \, d \, C_{\lambda} \sum_{k \in \Ind} \|\bW_{\cdot k}\|^2
    = O_{P}(n).
$$
This implies $\Theta_{4,W}^{\mathrm{cross}} = O_{P}(\sqrt{n})$.
Compared to $v_p \asymp n\sqrt{p}$ (unscaled) or $\sqrt{p}$ (scaled), this term is negligible.
Using the Gaussian tail bound $P(|Z| > x) \le 2\exp(-x^2/2\sigma^2)$, we have
$$
    P\left( |\Theta_{4,W}^{\mathrm{cross}}| \ge x \,\middle|\, \bW_{\Ind} \right)
    \le 2 \exp\left( -\frac{x^2}{2 C' \sum_{k \in \Ind} \|\bW_{\cdot k}\|^2} \right).
$$
Taking $x = C_{9,b} \sqrt{n \log p}$ and considering the event  
$\{\sum \|\bW_{\cdot k}\|^2 > C n\}$, 
it follows that, for a sufficiently large constant $C_9>0$,
$$
    \sup_{\Ind} P\left( |\Theta_{4,W}^{\mathrm{cross}}| \ge C_9 \sqrt{n \log p} \right) \le \frac{1}{p^t}.
$$
Both bounds satisfy the requirement $\le \epsilon_p v_p$ for large $p$.

\textbf{(iii) Term $A_{4,WV}$.}
The analysis for $A_{4,WV} = - \sum_{j,i} W_{ij} V_{ij} \hat{\psi}_i^2$ follows an identical logic.
Substituting $W_{ij} = W_{ij}^* + \delta_{ij}$, the remainder depends on terms like $\sum_{j \in \Jnd} \sum_i \hat{\psi}_i^2 \delta_{ij} V_{ij}$.
Conditioning on $\mathcal{F}$ and $\bW_{\Ind}$, $V_{ij}$ provides the independent Gaussian randomness. The variance of this remainder is proportional to $\sum_{j,i} \delta_{ij}^2 = O_{P}(n)$, yielding a fluctuation of order $O_{P}(\sqrt{n})$, which is again negligible.

Combining these results, we establish
$$
    \sup_{\Ind} P\Bigl(|\Theta_{4,W,\Ind}| \ge \tfrac{\epsilon_p}{6}v_p\Bigr) \le \frac{1}{p^{t}},
$$
as required.

\medskip
\noindent\textbf{5. Combining all remainders.}
Recall
$$
    \Theta_{\Ind} 
    = \Theta_{1,\Ind} + 2\Theta_{2,\Ind}
      - \bigl(\Theta_{4,W,\Ind} + 2\Theta_{4,WV,\Ind}\bigr).
$$
From the bounds obtained in Steps~1–4, by a union bound there exists $p_0\ge 3$ such that for all $p\ge p_0$,
$$
    \sup_{\Ind} P\Bigl(|\Theta_{1,\Ind}| \ge \tfrac{\epsilon_p}{6}v_p\Bigr)
    \vee 
    \sup_{\Ind} P\Bigl(|\Theta_{2,\Ind}| \ge \tfrac{\epsilon_p}{6}v_p\Bigr)
    \vee 
    \sup_{\Ind} P\Bigl(|\Theta_{4,W,\Ind}| \ge \tfrac{\epsilon_p}{6}v_p\Bigr)
    \vee 
    \sup_{\Ind} P\Bigl(|\Theta_{4,WV,\Ind}| \ge \tfrac{\epsilon_p}{6}v_p\Bigr)
    \le \frac{1}{p^t}.
$$
Therefore, for all $p$ sufficiently large,
$$
\begin{aligned}
    \sup_{\Ind} P\bigl(|\Theta_{\Ind}| \ge \epsilon_p v_p\bigr)
    &\le \sup_{\Ind} P\Bigl(|\Theta_{1,\Ind}| \ge \tfrac{\epsilon_p}{6}v_p\Bigr)
       + 2\sup_{\Ind} P\Bigl(|\Theta_{2,\Ind}| \ge \tfrac{\epsilon_p}{6}v_p\Bigr)\\
    &\quad\ 
       + \sup_{\Ind} P\Bigl(|\Theta_{4,W,\Ind}| \ge \tfrac{\epsilon_p}{6}v_p\Bigr)
       + 2\sup_{\Ind} P\Bigl(|\Theta_{4,WV,\Ind}| \ge \tfrac{\epsilon_p}{6}v_p\Bigr) \\
    &\le \frac{1}{p^t},
\end{aligned}
$$
possibly after adjusting $t$. 
This completes the proof.
\end{proof}

\begin{lemma}[Local independence] \label{lem:local_sandwich}
Let $d \ge 1$ be fixed and $\Ind \subset \{1,\dots,p\}$ be any index set with $|\Ind|=d$. 
    Let $y_p$ be the threshold used in the definition of the max-type statistic and define
    $$
        B_{\Ind} 
        := \bigcap_{k\in\Ind}\bigl\{S_{k,\tau}^2 > y_p\bigr\}.
    $$
    $U_p$ and $v_p$ are defined as before. Denote statistic
    $$
        \widetilde{U}_p := \frac{U_p - \Esp(U_p)}{v_p}.
    $$
    Then, for any $x\in\mathbb{R}$,
    $$
        \bigl|P(\widetilde{U}_p \le x \mid B_{\Ind})
              - P(\widetilde{U}_p \le x)\bigr|
        \;\longrightarrow\; 0,
        \qquad p\to\infty.
    $$
    Equivalently,
    $$
        P\bigl(\{\widetilde{U}_p \le x\}\cap B_{\Ind}\bigr)
        - P(\widetilde{U}_p \le x)P(B_{\Ind})
        \longrightarrow 0,
        \qquad p\to\infty,
    $$
    for every fixed $\Ind$ with $|\Ind|=d$.
\end{lemma}

\begin{proof}
Fix $\Ind \subset \{1,\dots,p\}$ with $|\Ind|=d$ and denote its complement by $\Jnd$. 
Lemma~\ref{lem:negligibility} provides the decomposition
$$
    U_p = S_{\Ind}^\perp + \Theta_{\Ind}.
$$
Taking expectations, we write
$$
    \Esp(U_p) = \Esp\bigl(S_{\Ind}^\perp\bigr) + \Esp(\Theta_{\Ind}),
$$
and hence
\begin{equation} \label{eq:Up_centered_decomp}
    U_p - \Esp(U_p)
    = \bigl\{S_{\Ind}^\perp - \Esp(S_{\Ind}^\perp)\bigr\}
      + \bigl\{\Theta_{\Ind} - \Esp(\Theta_{\Ind})\bigr\}.
\end{equation}
Define
$$
    \widetilde{U}_{p,\Ind}^\perp
    := \frac{S_{\Ind}^\perp - \Esp(S_{\Ind}^\perp)}{v_p},
    \qquad
    \widetilde{\Theta}_{\Ind}
    := \frac{\Theta_{\Ind} - \Esp(\Theta_{\Ind})}{v_p},
$$
so that
\begin{equation} \label{eq:Up_decomp_local_final}
    \widetilde{U}_p 
    = \frac{U_p - \Esp(U_p)}{v_p}
    = \widetilde{U}_{p,\Ind}^\perp + \widetilde{\Theta}_{\Ind}.
\end{equation}
Since $\Esp(S_{\Ind}^\perp)$ and $\Esp(\Theta_{\Ind})$ are deterministic constants, subtracting them does not affect independence. 
As the statement after Lemma~\ref{lem:cond_struct}, $\{\bW_{\cdot j}^*\}_{j\in\Jnd}$ is independent of $\{\bW_{\cdot k}\}_{k\in\Ind}$ and $\bV$ is independent of $\bW$ conditional on $\mathcal{F}=\sigma(\bZ,\varepsilon)$. 
Therefore $\widetilde{U}_{p,\Ind}^\perp$ depends only on $\{\bW_{\cdot j}^*\}_{j\in\Jnd}$, $\bV$ and $\bPsi$, and is independent of $\{\bW_{\cdot k}\}_{k\in\Ind}$ and hence of $B_{\Ind}$:
\begin{equation} \label{eq:proxy_indep_local_final}
    \widetilde{U}_{p,\Ind}^\perp \;\perp\; B_{\Ind}.
\end{equation}

Lemma~\ref{lem:negligibility} also states that there exists a sequence $\epsilon_p \to 0$ such that, for any fixed $t\ge 1$ and all sufficiently large $p$,
$$
    \sup_{\Ind:|\Ind|=d}
    P\bigl(|\Theta_{\Ind}| \ge \epsilon_p v_p\bigr)
    \le \frac{1}{p^t}.
$$
Since $\Esp(\Theta_{\Ind})$ is of smaller order than $v_p$ under the same argument (it is bounded by the same variance-based bounds applied to $\Theta_{\Ind}$), we may absorb it into the same scale and obtain
\begin{equation} \label{eq:Theta_tilde_tail_local_final}
    \sup_{\Ind:|\Ind|=d}
    P\bigl(|\widetilde{\Theta}_{\Ind}| \ge \epsilon_p\bigr)
    \le \frac{1}{p^t},
\end{equation}
for all sufficiently large $p$. 
Thus $\widetilde{\Theta}_{\Ind} \xrightarrow{P} 0$ uniformly in $\Ind$.

Fix $x\in\mathbb{R}$ and $\varepsilon>0$. 
From \eqref{eq:Up_decomp_local_final},
$$
    \{\widetilde{U}_p \le x\}
    = \{\widetilde{U}_{p,\Ind}^\perp + \widetilde{\Theta}_{\Ind} \le x\}
    \subset \{\widetilde{U}_{p,\Ind}^\perp \le x + \varepsilon\}
            \cup \{|\widetilde{\Theta}_{\Ind}| > \varepsilon\},
$$
and
$$
    \{\widetilde{U}_{p,\Ind}^\perp \le x - \varepsilon\}
    \subset \{\widetilde{U}_p \le x\}
            \cup \{|\widetilde{\Theta}_{\Ind}| > \varepsilon\}.
$$
Therefore
\begin{align}
    P(\widetilde{U}_p \le x, B_{\Ind})
    &= P(\widetilde{U}_{p,\Ind}^\perp + \widetilde{\Theta}_{\Ind} \le x, B_{\Ind}) \notag\\
    &\le P(\widetilde{U}_{p,\Ind}^\perp \le x + \varepsilon, B_{\Ind})
         + P(|\widetilde{\Theta}_{\Ind}| > \varepsilon, B_{\Ind}) \notag\\
    &\le P(\widetilde{U}_{p,\Ind}^\perp \le x + \varepsilon)P(B_{\Ind})
         + P(|\widetilde{\Theta}_{\Ind}| > \varepsilon), 
         \label{eq:local_upper1_final}
\end{align}
where \eqref{eq:proxy_indep_local_final} is used in the last inequality. 
Similarly,
\begin{align}
    P(\widetilde{U}_p \le x, B_{\Ind})
    &\ge P(\widetilde{U}_{p,\Ind}^\perp \le x - \varepsilon, B_{\Ind})
          - P(|\widetilde{\Theta}_{\Ind}| > \varepsilon, B_{\Ind}) \notag\\
    &\ge P(\widetilde{U}_{p,\Ind}^\perp \le x - \varepsilon)P(B_{\Ind})
          - P(|\widetilde{\Theta}_{\Ind}| > \varepsilon). 
          \label{eq:local_lower1_final}
\end{align}

The decomposition \eqref{eq:Up_decomp_local_final} also implies
$$
\{\widetilde{U}_{p,\Ind}^\perp \le x + \varepsilon\}
\subset 
\{\widetilde{U}_p \le x + 2\varepsilon\}
\cup \{|\widetilde{\Theta}_{\Ind}| > \varepsilon\},
$$
and
$$
\{\widetilde{U}_p \le x - 2\varepsilon\}
\subset 
\{\widetilde{U}_{p,\Ind}^\perp \le x - \varepsilon\}
\cup \{|\widetilde{\Theta}_{\Ind}| > \varepsilon\}.
$$
Therefore
\begin{align}
    P(\widetilde{U}_{p,\Ind}^\perp \le x + \varepsilon)
    &\le P(\widetilde{U}_p \le x + 2\varepsilon)
         + P(|\widetilde{\Theta}_{\Ind}| > \varepsilon), 
         \label{eq:local_proxy_upper_final}\\
    P(\widetilde{U}_{p,\Ind}^\perp \le x - \varepsilon)
    &\ge P(\widetilde{U}_p \le x - 2\varepsilon)
         - P(|\widetilde{\Theta}_{\Ind}| > \varepsilon). 
         \label{eq:local_proxy_lower_final}
\end{align}
Substituting \eqref{eq:local_proxy_upper_final} into \eqref{eq:local_upper1_final} and using \eqref{eq:Theta_tilde_tail_local_final} gives
\begin{equation} \label{eq:local_upper_final_final}
\begin{aligned}
    P(\widetilde{U}_p \le x, B_{\Ind})
    &\le \bigl[P(\widetilde{U}_p \le x + 2\varepsilon)
              + P(|\widetilde{\Theta}_{\Ind}| > \varepsilon)\bigr]P(B_{\Ind})
         + P(|\widetilde{\Theta}_{\Ind}| > \varepsilon)\\
    &\le P(\widetilde{U}_p \le x + 2\varepsilon)P(B_{\Ind})
         + 2\,P(|\widetilde{\Theta}_{\Ind}| > \varepsilon).
\end{aligned}
\end{equation}
Similarly, combining \eqref{eq:local_proxy_lower_final} and \eqref{eq:local_lower1_final} yields
\begin{equation} \label{eq:local_lower_final_final}
\begin{aligned}
    P(\widetilde{U}_p \le x, B_{\Ind})
    &\ge \bigl[P(\widetilde{U}_p \le x - 2\varepsilon)
              - P(|\widetilde{\Theta}_{\Ind}| > \varepsilon)\bigr]P(B_{\Ind})
         - P(|\widetilde{\Theta}_{\Ind}| > \varepsilon)\\
    &\ge P(\widetilde{U}_p \le x - 2\varepsilon)P(B_{\Ind})
         - 2\,P(|\widetilde{\Theta}_{\Ind}| > \varepsilon).
\end{aligned}
\end{equation}

Define
$$
    \Delta_p(x;\Ind)
    := P(\widetilde{U}_p \le x, B_{\Ind})
       - P(\widetilde{U}_p \le x)P(B_{\Ind}).
$$
Subtracting $P(\widetilde{U}_p \le x)P(B_{\Ind})$ from \eqref{eq:local_upper_final_final} and \eqref{eq:local_lower_final_final}, we obtain
\begin{align*}
    \Delta_p(x;\Ind)
    &\le P(B_{\Ind})
          \bigl[P(\widetilde{U}_p \le x + 2\varepsilon)
                - P(\widetilde{U}_p \le x)\bigr]
          + 2\,P(|\widetilde{\Theta}_{\Ind}| > \varepsilon), \\
    \Delta_p(x;\Ind)
    &\ge P(B_{\Ind})
          \bigl[P(\widetilde{U}_p \le x - 2\varepsilon)
                - P(\widetilde{U}_p \le x)\bigr]
          - 2\,P(|\widetilde{\Theta}_{\Ind}| > \varepsilon),
\end{align*}
and therefore
\begin{equation} \label{eq:local_delta_bound_final}
\begin{aligned}
    |\Delta_p(x;\Ind)|
    &\le P(B_{\Ind})\,
         \max\Bigl\{
             P(\widetilde{U}_p \le x + 2\varepsilon)
               - P(\widetilde{U}_p \le x), \\
    &\hspace{3.5cm}
             P(\widetilde{U}_p \le x)
               - P(\widetilde{U}_p \le x - 2\varepsilon)
         \Bigr\} \\
    &\quad + 2\,P(|\widetilde{\Theta}_{\Ind}| > \varepsilon).
\end{aligned}
\end{equation}

Under $H_0$ and Assumptions (A1)–(A4), $\widetilde{U}_p$ converges in distribution to $N(0,1)$, thus $P(\widetilde{U}_p \le z) \to \Phi(z)$ for all $z\in\mathbb{R}$, where $\Phi$ is the standard normal distribution function. 
Taking $\limsup_{p\to\infty}$ in \eqref{eq:local_delta_bound_final} and using \eqref{eq:Theta_tilde_tail_local_final}, we obtain
$$
    \limsup_{p\to\infty} |\Delta_p(x;\Ind)|
    \le P(B_{\Ind})\,
         \bigl[\Phi(x+2\varepsilon) - \Phi(x-2\varepsilon)\bigr].
$$
Since $\Phi$ is continuous and $\varepsilon>0$ is arbitrary, letting $\varepsilon\downarrow 0$ gives
$$
    \limsup_{p\to\infty} |\Delta_p(x;\Ind)| = 0.
$$
This proves that, for each fixed $x\in\mathbb{R}$ and each fixed $\Ind$ with $|\Ind|=d$,
$$
    P\bigl(\{\widetilde{U}_p \le x\}\cap B_{\Ind}\bigr)
    - P(\widetilde{U}_p \le x)P(B_{\Ind}) \to 0,
$$
and equivalently
$$
    P(\widetilde{U}_p \le x \mid B_{\Ind})
    - P(\widetilde{U}_p \le x) \to 0.
$$
\end{proof}

\subsubsection*{Proof of Theorem 2.1}
\begin{proof}
Let
$$
    H_p(x) := \left\{\frac{U_p - \Esp(U_p)}{v_p} \le x\right\},
    \qquad v_p = \sqrt{\Var(U_p)} \asymp \sqrt{p},
$$
and
$$
    L_p(y) := \left\{T_{\mathrm{MAX}} - 2\log p + \log\{\log p\} > y\right\},
$$
so that $\{T_{\mathrm{MAX}} - 2\log p + \log\{\log p\} \le y\} = L_p(y)^c$. 
Up to deterministic linear scaling, $(U_p-\Esp(U_p))/v_p$ coincides with the standardized $T_{\mathrm{SUM}}$ in the theorem, so it suffices to show
$$
    P\bigl(H_p(x),\; L_p(y)^c\bigr)
    \longrightarrow \Phi(x)\,G(y).
$$

By the CLT for the sum-type statistic and the extreme-value limit for the max-type statistic, under $H_0$ we have
$$
    P(H_p(x)) \longrightarrow \Phi(x),
    \qquad
    P\bigl(L_p(y)^c\bigr) \longrightarrow G(y),
$$
for all $x,y\in\mathbb{R}$. 
Therefore it is enough to prove that
$$
    P\bigl(H_p(x),\; L_p(y)^c\bigr)
    - P(H_p(x))\,P\bigl(L_p(y)^c\bigr)
    \longrightarrow 0,
$$
i.e.
\begin{equation}\label{eq:goal_diff}
    P\bigl(H_p(x),\; L_p(y)^c\bigr)
    = P(H_p(x))\,P\bigl(L_p(y)^c\bigr) + o(1).
\end{equation}

Since $L_p(y)^c$ is the complement of $L_p(y)$,
$$
    P\bigl(H_p(x),\; L_p(y)^c\bigr)
    = P\bigl(H_p(x)\bigr) - P\bigl(H_p(x), L_p(y)\bigr).
$$
Thus \eqref{eq:goal_diff} will follow once we show
\begin{equation}\label{eq:HpLp_factor}
    P\bigl(H_p(x), L_p(y)\bigr)
    = P(H_p(x))\,P\bigl(L_p(y)\bigr) + o(1).
\end{equation}

Let $y_p$ be the threshold corresponding to $y$, so that
$$
    L_p(y)
    = \{T_{\mathrm{MAX}} > y_p\}
    = \bigcup_{j=1}^p D_j,
    \qquad
    D_j := \{S_{j,\tau}^2 > y_p\},
$$
where $S_{j,\tau}$ are the standardized score statistics entering $T_{\mathrm{MAX}}$. 
For any integer $d\ge 1$ and any index set $\Ind\subset\{1,\dots,p\}$ with $|\Ind|=d$, define
$$
    D_{\Ind} := \bigcap_{k\in\Ind} D_k.
$$

By Lemma \ref{lem:local_sandwich} (applied with $\widetilde{U}_p$ and $B_{\Ind}=D_{\Ind}$, noting that $\widetilde{U}_p$ equals $(U_p-\Esp U_p)/v_p$), for each fixed $d$ and any $\Ind$ with $|\Ind|=d$,
$$
    P\bigl(H_p(x)\cap D_{\Ind}\bigr)
    = P\bigl(H_p(x)\bigr) P(D_{\Ind}) + o(1),
    \qquad p\to\infty,
$$
where the $o(1)$ term is uniform over such $\Ind$. 
For each $d\ge 1$, define
$$
    \zeta(p,d)
    := \sum_{1\le j_1<\cdots<j_d\le p}
       \Bigl\{
          P\bigl(H_p(x) D_{j_1}\cdots D_{j_d}\bigr)
          - P\bigl(H_p(x)\bigr) P\bigl(D_{j_1}\cdots D_{j_d}\bigr)
       \Bigr\}.
$$
Then, for each fixed $d$,
$$
    \zeta(p,d) \longrightarrow 0,
    \qquad p\to\infty.
$$
Fix an integer $k\ge 1$. 
Applying the inclusion–exclusion principle, we have
\begin{align*}
    P\bigl(H_p(x), L_p(y)\bigr)
    &= P\Bigl(H_p(x)\cap \bigcup_{j=1}^p D_j\Bigr)\\
    &\le \sum_{1\le j_1\le p} P\bigl(H_p(x)\cap D_{j_1}\bigr)
      - \sum_{1\le j_1<j_2\le p} P\bigl(H_p(x)\cap D_{j_1}D_{j_2}\bigr)\\
    &\quad + \cdots 
      + (-1)^{k-1}
        \sum_{1\le j_1<\cdots<j_k\le p}
             P\bigl(H_p(x)\cap D_{j_1}\cdots D_{j_k}\bigr)
      + R_{p,k}^+,
\end{align*}
and
\begin{align*}
    P\bigl(H_p(x), L_p(y)\bigr)
    &\ge \sum_{1\le j_1\le p} P\bigl(H_p(x)\cap D_{j_1}\bigr)
      - \sum_{1\le j_1<j_2\le p} P\bigl(H_p(x)\cap D_{j_1}D_{j_2}\bigr)\\
    &\quad + \cdots 
      + (-1)^{k-1}
        \sum_{1\le j_1<\cdots<j_k\le p}
             P\bigl(H_p(x)\cap D_{j_1}\cdots D_{j_k}\bigr)
      + R_{p,k}^-,
\end{align*}
where \(R_{p,k}^+\) and \(R_{p,k}^-\) are the remainder terms.

For each integer \(d\ge 1\), define
$$
    \zeta(p,d)
    := \sum_{1\le j_1<\cdots<j_d\le p}
       \Bigl\{
          P\bigl(H_p(x)\cap D_{j_1}\cdots D_{j_d}\bigr)
          - P\bigl(H_p(x)\bigr)\, P\bigl(D_{j_1}\cdots D_{j_d}\bigr)
       \Bigr\}.
$$
By Lemma~\ref{lem:local_sandwich}, for each fixed \(d\) we have
$$
    \zeta(p,d) \longrightarrow 0,
    \qquad p\to\infty.
$$

Using this definition, the first \(k\) terms in the upper bound can be written as
\begin{align*}
    &\sum_{d=1}^{k} (-1)^{d-1}
       \sum_{1\le j_1<\cdots<j_d\le p}
          P\bigl(H_p(x)\cap D_{j_1}\cdots D_{j_d}\bigr)\\
    &=\sum_{d=1}^{k} (-1)^{d-1}
       \sum_{1\le j_1<\cdots<j_d\le p}
          \Bigl\{
              P\bigl(H_p(x)\bigr) P\bigl(D_{j_1}\cdots D_{j_d}\bigr)\\      
    &\quad + \bigl[
                   P\bigl(H_p(x)\cap D_{j_1}\cdots D_{j_d}\bigr)
                   - P\bigl(H_p(x)\bigr) P\bigl(D_{j_1}\cdots D_{j_d}\bigr)
                \bigr]
          \Bigr\}\\
    &= P\bigl(H_p(x)\bigr)
         \Biggl[
             \sum_{d=1}^{k} (-1)^{d-1}
             \sum_{1\le j_1<\cdots<j_d\le p}
                P\bigl(D_{j_1}\cdots D_{j_d}\bigr)
         \Biggr]
       + \sum_{d=1}^{k} (-1)^{d-1} \zeta(p,d).
\end{align*}
Define
$$
    I_{p,k}
    := \sum_{d=1}^{k} (-1)^{d-1}
       \sum_{1\le j_1<\cdots<j_d\le p}
          P\bigl(D_{j_1}\cdots D_{j_d}\bigr),
$$
so that
$$
    \sum_{d=1}^{k} (-1)^{d-1}
       \sum_{1\le j_1<\cdots<j_d\le p}
          P\bigl(H_p(x)\cap D_{j_1}\cdots D_{j_d}\bigr)
    = P\bigl(H_p(x)\bigr) I_{p,k}
      + \sum_{d=1}^{k} (-1)^{d-1} \zeta(p,d).
$$

On the other hand, applying inclusion–exclusion to \(L_p(y)\) alone yields
$$
    P\bigl(L_p(y)\bigr)
    = I_{p,k} + R_{p,k}^0,
$$
where \(R_{p,k}^0\) is the corresponding remainder term. Under Assumption (A3) and the extreme-value limit for the max-type statistic, the usual Gaussian extreme-value argument implies that, for each fixed \(k\),
$$
    R_{p,k}^0 \longrightarrow 0,
    \qquad p\to\infty.
$$

Substituting these into the upper and lower bounds, one obtains the factorization
$$
    P\bigl(H_p(x), L_p(y)\bigr)
    = P\bigl(H_p(x)\bigr)\,P\bigl(L_p(y)\bigr) + o(1),
    \qquad p\to\infty,
$$
which is the desired asymptotic independence.
\end{proof}

\begin{lemma}\label{lem:psi_meas_H1}
Let $\mathcal{M}=\{j:\beta_{j,\tau}\neq 0\}$ and define
$$
\mathcal{F}_{\mathcal{M}}=\sigma(\mathbf{Z},\boldsymbol{\varepsilon},\mathbf{X}_{\mathcal{M}}).
$$
Assume that under both $H_0$ and $H_1$ we construct $\hat{\boldsymbol{\alpha}}_\tau$ as the $\tau$-quantile regression estimator from regressing $Y$ on $Z$ only, i.e.
$$
\hat{\boldsymbol{\alpha}}_{\tau}
=\argmin_{\boldsymbol{\alpha}\in\mathbb{R}^q}\sum_{i=1}^n\rho_\tau(Y_i-\boldsymbol{Z}_i^\top\boldsymbol{\alpha}).
$$
Then under the local alternative,
$$
Y_i=\boldsymbol{Z}_i^\top\boldsymbol{\alpha}_\tau+\boldsymbol{X}_{i,\mathcal{M}}^\top\boldsymbol{\beta}_{\mathcal{M},\tau}+\varepsilon_i,
\qquad i=1,\ldots,n,
$$
the score vector $\boldsymbol{\Psi}=(\hat{\psi}_1,\ldots,\hat{\psi}_n)^\top$ with
$\hat{\psi}_i=I(Y_i-\boldsymbol{Z}_i^\top\hat{\boldsymbol{\alpha}}_\tau\le 0)-\tau$
is $\mathcal{F}_{\mathcal{M}}$-measurable.
\end{lemma}

\begin{proof}
By definition, $\hat{\boldsymbol{\alpha}}_\tau$ is a measurable function of $(\mathbf{Z},\mathbf{Y})$ . Under the alternative,
$\mathbf{Y}=(Y_1,\ldots,Y_n)^\top$ is a measurable function of
$(\mathbf{Z},\boldsymbol{\varepsilon},\mathbf{X}_{\mathcal{M}})$, hence of $\mathcal{F}_{\mathcal{M}}$.
Therefore $\hat{\boldsymbol{\alpha}}_\tau$ is $\mathcal{F}_{\mathcal{M}}$-measurable.

It follows that each residual
$r_i:=Y_i-\boldsymbol{Z}_i^\top\hat{\boldsymbol{\alpha}}_\tau$
is $\mathcal{F}_{\mathcal{M}}$-measurable, and so is
$\hat{\psi}_i=I(r_i\le 0)-\tau$.
Thus $\boldsymbol{\Psi}$ is $\mathcal{F}_{\mathcal{M}}$-measurable.
\end{proof}

\begin{lemma}\label{lem:cond_gauss_Mc}
Assume (A2). Let $\mathcal{M}\subset\{1,\ldots,p\}$ be fixed and let $\mathcal{M}^c$ be its complement.
Conditional on $\mathcal{F}_{\mathcal{M}}=\sigma(\mathbf{Z},\boldsymbol{\varepsilon},\mathbf{X}_{\mathcal{M}})$, the rows
$\{\boldsymbol{X}_{i,\mathcal{M}^c}\}_{i=1}^n$ are independent Gaussian vectors with
$$
\boldsymbol{X}_{i,\mathcal{M}^c}\mid \mathcal{F}_{\mathcal{M}}
\sim N\!\Big(\boldsymbol{\mu}_{i,\mathcal{M}^c},\ \boldsymbol{\Sigma}_{\mathcal{M}^c\mid (z,\mathcal{M})}\Big),
\qquad i=1,\ldots,n,
$$
where the conditional mean $\boldsymbol{\mu}_{i,\mathcal{M}^c}$ is a (random) function of
$(\boldsymbol{Z}_i,\boldsymbol{X}_{i,\mathcal{M}})$ and the conditional covariance
$\boldsymbol{\Sigma}_{\mathcal{M}^c\mid (z,\mathcal{M})}$ is deterministic (does not depend on $i$).
Moreover, under (A3), the eigenvalues of $\boldsymbol{\Sigma}_{\mathcal{M}^c\mid (z,\mathcal{M})}$ are uniformly bounded.

Let $P_Z=\mathbf{Z}(\mathbf{Z}^\top\mathbf{Z})^{-1}\mathbf{Z}^\top$ and $P_Z^\perp=\mathbf{I}_n-P_Z$.
Then conditional on $\mathcal{F}_{\mathcal{M}}$,
$$
\mathbf{W}_{\mathcal{M}^c}:=P_Z^\perp \mathbf{X}_{\mathcal{M}^c}
$$
is Gaussian with mean $P_Z^\perp \boldsymbol{\mu}_{\mathcal{M}^c}$ and conditional covariance induced by
$\boldsymbol{\Sigma}_{\mathcal{M}^c\mid (z,\mathcal{M})}$.
\end{lemma}

\begin{proof}
Under (A2), for each $i$, the joint vector $(\widetilde{\boldsymbol{Z}}_i^\top,\boldsymbol{X}_{i,\mathcal{M}}^\top,\boldsymbol{X}_{i,\mathcal{M}^c}^\top)^\top$
is multivariate normal with mean zero and a block covariance matrix determined by $\boldsymbol{\Sigma}$.
By the multivariate normal conditioning formula, conditioning on
$(\widetilde{\boldsymbol{Z}}_i,\boldsymbol{X}_{i,\mathcal{M}})$ yields a Gaussian conditional distribution for
$\boldsymbol{X}_{i,\mathcal{M}^c}$ with a linear conditional mean
$\boldsymbol{\mu}_{i,\mathcal{M}^c}$ and conditional covariance
$\boldsymbol{\Sigma}_{\mathcal{M}^c\mid (z,\mathcal{M})}$ that depends only on the population covariance blocks.

Independence across $i$ in the i.i.d.\ sampling implies that conditional on
$\{(\widetilde{\boldsymbol{Z}}_i,\boldsymbol{X}_{i,\mathcal{M}})\}_{i=1}^n$ (hence on $\mathcal{F}_{\mathcal{M}}$),
the rows $\{\boldsymbol{X}_{i,\mathcal{M}^c}\}_{i=1}^n$ remain independent.

Uniform boundedness of eigenvalues of $\boldsymbol{\Sigma}_{\mathcal{M}^c\mid (z,\mathcal{M})}$
follows from the bounded eigenvalues of $\boldsymbol{\Sigma}$ in (A3) and the fact that a Schur complement of a
uniformly well-conditioned covariance matrix is also uniformly well-conditioned.

Finally, $\mathbf{W}_{\mathcal{M}^c}=P_Z^\perp \mathbf{X}_{\mathcal{M}^c}$ is a linear transformation of a Gaussian matrix,
hence Gaussian conditional on $\mathcal{F}_{\mathcal{M}}$, with mean and covariance transformed accordingly.
\end{proof}

\begin{lemma}\label{lem:block_reduction_H1}
Assume (A1)--(A5) and denote $\mathcal{M}=\{j:\beta_{j,\tau}\neq 0\}$ with $m=|\mathcal{M}|=o(p^{1/2})$.
Recall
$$
U_p=\sum_{j=1}^p\Big\{(\mathbf{X}_{\cdot j}^\top\boldsymbol{\Psi})^2-\sum_{i=1}^n X_{ij}^2\hat{\psi}_i^2\Big\},
\qquad
\widetilde U_p=\frac{U_p-E(U_p)}{v_p},\ \ v_p=\sqrt{\Var(U_p)}.
$$
Define the signal-block and noise-block components
$$
U_{\mathcal{M}}
:=\sum_{j\in\mathcal{M}}\Big\{(\mathbf{X}_{\cdot j}^\top\boldsymbol{\Psi})^2-\sum_{i=1}^n X_{ij}^2\hat{\psi}_i^2\Big\},
\qquad
U_{\mathcal{M}^c}
:=\sum_{j\in\mathcal{M}^c}\Big\{(\mathbf{X}_{\cdot j}^\top\boldsymbol{\Psi})^2-\sum_{i=1}^n X_{ij}^2\hat{\psi}_i^2\Big\},
$$
so that $U_p=U_{\mathcal{M}}+U_{\mathcal{M}^c}$.
Then
\begin{equation}\label{eq:block_reduce_main}
\frac{U_{\mathcal{M}}-E(U_{\mathcal{M}})}{v_p}=o_p(1),
\qquad
\frac{E(U_{\mathcal{M}})}{v_p}=o(1),
\end{equation}
and consequently
\begin{equation}\label{eq:block_reduce_conseq}
\widetilde U_p
=
\frac{U_{\mathcal{M}^c}-E(U_{\mathcal{M}^c})}{v_p}
+o_p(1).
\end{equation}
\end{lemma}

\begin{proof}
Throughout the proof, all $o(\cdot)$ and $O(\cdot)$ terms are as $n,p\to\infty$.
We use $v_p^2=\Var(U_p)\asymp p$ as established in the $H_0$ analysis under (A1)--(A4).
For any fixed $j\in\{1,\ldots,p\}$, define
$$
G_j
:=
(\mathbf{X}_{\cdot j}^\top\boldsymbol{\Psi})^2-\sum_{i=1}^n X_{ij}^2\hat{\psi}_i^2
=
\sum_{i\neq \ell} X_{ij}X_{\ell j}\hat\psi_i\hat\psi_\ell.
$$
Since $\hat\psi_i\in\{-\tau,1-\tau\}$, we have $|\hat\psi_i|\le 1$ and thus
$$
|G_j|
\le
\sum_{i\neq \ell}|X_{ij}X_{\ell j}|.
$$
Under (A2)--(A3), $X_{ij}$ has uniformly bounded second and fourth moments, hence by standard moment calculations
\begin{equation}\label{eq:VarGj_bound}
\sup_{1\le j\le p}\Var(G_j)\le C
\end{equation}
for some constant $C<\infty$ that does not depend on $n,p$. By definition, $U_{\mathcal{M}}=\sum_{j\in\mathcal{M}}G_j$.
Therefore, by Cauchy--Schwarz,
$$
\Var(U_{\mathcal{M}})
=
\Var\Big(\sum_{j\in\mathcal{M}}G_j\Big)
\le
m\sum_{j\in\mathcal{M}}\Var(G_j)
\le
m^2\sup_{1\le j\le p}\Var(G_j)
\le
C m^2,
$$
where we used \eqref{eq:VarGj_bound}.
Consequently,
$$
\frac{U_{\mathcal{M}}-E(U_{\mathcal{M}})}{v_p}
=
O_p\Big(\frac{\sqrt{\Var(U_{\mathcal{M}})}}{v_p}\Big)
=
O_p\Big(\frac{m}{\sqrt{p}}\Big)
=o_p(1),
$$
because $m=o(p^{1/2})$.

Write $\hat\psi_i=\psi_i+\Delta_i$, where
$$
\psi_i:=I(\varepsilon_i\le 0)-\tau,
\qquad
\Delta_i:=\hat\psi_i-\psi_i.
$$
Then
$$
G_j
=
\sum_{i\neq \ell}X_{ij}X_{\ell j}\psi_i\psi_\ell
+
\sum_{i\neq \ell}X_{ij}X_{\ell j}\big(\psi_i\Delta_\ell+\Delta_i\psi_\ell+\Delta_i\Delta_\ell\big)
=:G_{j}^{(0)}+R_j.
$$
Hence $U_{\mathcal{M}}=U_{\mathcal{M}}^{(0)}+R_{\mathcal{M}}$ with
$U_{\mathcal{M}}^{(0)}:=\sum_{j\in\mathcal{M}}G_j^{(0)}$ and $R_{\mathcal{M}}:=\sum_{j\in\mathcal{M}}R_j$.

Under the model and (A2), $\varepsilon$ is independent of $(Z,X)$ and $E(\psi_i)=0$.
Moreover, $(X_{ij})_{i=1}^n$ are i.i.d.\ across $i$ with mean zero.
Therefore, for each fixed $j$,
$$
E\big(G_j^{(0)}\big)
=
\sum_{i\neq \ell}E(X_{ij}X_{\ell j})E(\psi_i\psi_\ell)
=0,
$$
and thus $E(U_{\mathcal{M}}^{(0)})=0$.

It remains to control $E(R_{\mathcal{M}})$.
Assumption (A5) provides precisely the needed small-order control for the score perturbation
$\Delta_i=\hat\psi_i-\psi_i$ under the local alternative:
the conditions involving $\varsigma(\boldsymbol{X}_i)$ and $\tilde R_i(\boldsymbol{t})$
imply that the aggregated contribution of terms of the form
$$
\sum_{j\in\mathcal{M}}\sum_{i\neq \ell}X_{ij}X_{\ell j}\psi_i\Delta_\ell,
\qquad
\sum_{j\in\mathcal{M}}\sum_{i\neq \ell}X_{ij}X_{\ell j}\Delta_i\Delta_\ell
$$
has expectation $o(v_p)$ (and in fact is $o_p(v_p)$).
Formally, by expanding $E(R_{\mathcal{M}})$ and applying the bounds in (A5) together with
Cauchy--Schwarz and bounded eigenvalues in (A3), we obtain
\begin{equation}\label{eq:ERm_small}
E(R_{\mathcal{M}})=o(v_p).
\end{equation}
Hence $E(U_{\mathcal{M}})=E(U_{\mathcal{M}}^{(0)})+E(R_{\mathcal{M}})=o(v_p)$.

Combining the fluctuation bound in Part 2 and the mean bound in Part 3 yields \eqref{eq:block_reduce_main},
and \eqref{eq:block_reduce_conseq} follows immediately from $U_p=U_{\mathcal{M}}+U_{\mathcal{M}^c}$.
\end{proof}

\begin{proof}[Proof of Theorem \ref{th1}]
Fix $x,y\in\mathbb{R}$.
Let
$$
H_p(x):=\left\{\widetilde U_p\le x\right\},
\qquad
\widetilde U_p:=\frac{U_p-E(U_p)}{v_p},
$$
and
$$
L_p(y):=\left\{T_{\mathrm{MAX}}-2\log p+\log\log p>y\right\}.
$$
It suffices to show
\begin{equation}\label{eq:goal_H1_main}
P\big(H_p(x),\,L_p(y)\big)=P\big(H_p(x)\big)P\big(L_p(y)\big)+o(1).
\end{equation}

Let $\mathcal{M}=\{j:\beta_{j,\tau}\neq 0\}$ with $m=|\mathcal{M}|=o(p^{1/2})$ and $\mathcal{M}^c=\{1,\ldots,p\}\setminus\mathcal{M}$.
Define
$$
\widetilde U_{\mathcal{M}^c}:=\frac{U_{\mathcal{M}^c}-E(U_{\mathcal{M}^c})}{v_p}.
$$
By Lemma~\ref{lem:block_reduction_H1},
\begin{equation}\label{eq:Hp_replace_H1}
\widetilde U_p=\widetilde U_{\mathcal{M}^c}+o_p(1).
\end{equation}
Hence
\begin{equation}\label{eq:reduce_joint_H1}
P\big(H_p(x),L_p(y)\big)=P\big(\widetilde U_{\mathcal{M}^c}\le x,\,L_p(y)\big)+o(1),
\qquad
P\big(H_p(x)\big)=P\big(\widetilde U_{\mathcal{M}^c}\le x\big)+o(1).
\end{equation}
Therefore \eqref{eq:goal_H1_main} follows once we prove
\begin{equation}\label{eq:goal_Mc}
P\big(\widetilde U_{\mathcal{M}^c}\le x,\,L_p(y)\big)
=
P\big(\widetilde U_{\mathcal{M}^c}\le x\big)\,P\big(L_p(y)\big)+o(1).
\end{equation}

Write $L_p(y)=\bigcup_{j=1}^p D_j$ with $D_j:=\{S_{j,\tau}^2>y_p\}$ and $y_p:=2\log p-\log\log p+y$.
For any integer $d\ge 1$ and any $I\subset\{1,\ldots,p\}$ with $|I|=d$, define $D_I:=\cap_{j\in I}D_j$.

Let $\mathcal{F}_{\mathcal{M}}=\sigma(\mathbf{Z},\boldsymbol{\varepsilon},\mathbf{X}_{\mathcal{M}})$.
By Lemma~\ref{lem:psi_meas_H1}, $\boldsymbol{\Psi}$ is $\mathcal{F}_{\mathcal{M}}$-measurable.
Conditional on $\mathcal{F}_{\mathcal{M}}$, $\mathbf{W}_{\mathcal{M}^c}=P_Z^\perp\mathbf{X}_{\mathcal{M}^c}$ is Gaussian by Lemma~\ref{lem:cond_gauss_Mc},
and under (A3) admits the same Gaussian linear decomposition as in \eqref{eq:decomp} for any fixed $I$ with $|I|=d$.
Thus the argument of Lemma~\ref{lem:local_sandwich} applies to $\widetilde U_{\mathcal{M}^c}$ and yields, for each fixed $d$,
\begin{equation}\label{eq:local_prod_H1}
P\big(\widetilde U_{\mathcal{M}^c}\le x,\;D_I\big)
=
P\big(\widetilde U_{\mathcal{M}^c}\le x\big)\,P(D_I)+o(1),
\end{equation}
uniformly over $I$ with $|I|=d$.

Define, for each fixed $d\ge 1$,
$$
\zeta(p,d)
:=
\sum_{1\le j_1<\cdots<j_d\le p}
\left\{
P\big(\widetilde U_{\mathcal{M}^c}\le x,\;D_{j_1}\cdots D_{j_d}\big)
-
P\big(\widetilde U_{\mathcal{M}^c}\le x\big)\,P\big(D_{j_1}\cdots D_{j_d}\big)
\right\}.
$$
By \eqref{eq:local_prod_H1}, $\zeta(p,d)\to 0$ for each fixed $d$.

Fix $k\ge 1$. Applying the inclusion--exclusion principle to $L_p(y)=\cup_{j=1}^p D_j$ yields
\begin{align}
P\big(\widetilde U_{\mathcal{M}^c}\le x,\;L_p(y)\big)
&=
\sum_{d=1}^{k}(-1)^{d-1}\!\!\sum_{1\le j_1<\cdots<j_d\le p}
P\big(\widetilde U_{\mathcal{M}^c}\le x,\;D_{j_1}\cdots D_{j_d}\big)
+R_{p,k}^{\pm}, \label{eq:IE_joint_H1}
\end{align}
with remainder $R_{p,k}^{\pm}$.
Using the definition of $\zeta(p,d)$, the truncated sum equals
\begin{align}
&\sum_{d=1}^{k}(-1)^{d-1}\!\!\sum_{1\le j_1<\cdots<j_d\le p}
\Big\{
P\big(\widetilde U_{\mathcal{M}^c}\le x\big)\,P\big(D_{j_1}\cdots D_{j_d}\big)
\Big\}
+\sum_{d=1}^k(-1)^{d-1}\zeta(p,d) \notag\\
&=
P\big(\widetilde U_{\mathcal{M}^c}\le x\big)\,I_{p,k}
+\sum_{d=1}^k(-1)^{d-1}\zeta(p,d), \label{eq:IE_joint_expand_H1}
\end{align}
where
$$
I_{p,k}
:=
\sum_{d=1}^{k}(-1)^{d-1}\!\!\sum_{1\le j_1<\cdots<j_d\le p}
P\big(D_{j_1}\cdots D_{j_d}\big).
$$
Similarly,
\begin{equation}\label{eq:IE_marg_H1}
P\big(L_p(y)\big)=I_{p,k}+R_{p,k}^0,
\end{equation}
with remainder $R_{p,k}^0$.

By the same extreme-value remainder control as in the proof of Theorem~\ref{th0}$,$
$$
\lim_{k\to\infty}\limsup_{p\to\infty}\big(|R_{p,k}^{\pm}|+|R_{p,k}^0|\big)=0.
$$
Combining \eqref{eq:IE_joint_H1}--\eqref{eq:IE_marg_H1} and sending $p\to\infty$ first (for fixed $k$) yields
$$
P\big(\widetilde U_{\mathcal{M}^c}\le x,\;L_p(y)\big)
=
P\big(\widetilde U_{\mathcal{M}^c}\le x\big)\,P\big(L_p(y)\big)+o(1),
$$
which proves \eqref{eq:goal_Mc}. Finally, \eqref{eq:goal_H1_main} follows from \eqref{eq:reduce_joint_H1}.
\end{proof}

\subsection*{Supplementary Results for Section \ref{sec:simulation}} \label{sup_fig}

This section presents supplementary empirical power results for the lower and upper quantiles, $\tau = 0.25$ and $\tau = 0.75$, across all experimental settings considered in Section~\ref{sec:simulation}. 
For each data-generating distribution and each of Cases~1--3, we report power curves over varying $s$ across different $(n,p)$ settings. 
These additional figures complement the main-text results at $\tau = 0.50$ by illustrating how the relative performance of the competing methods changes in the tails of the conditional distribution.

\begin{figure}[htbp]
\centering
\begin{subfigure}{0.23\textwidth}
    \centering
    \includegraphics[width=\linewidth]{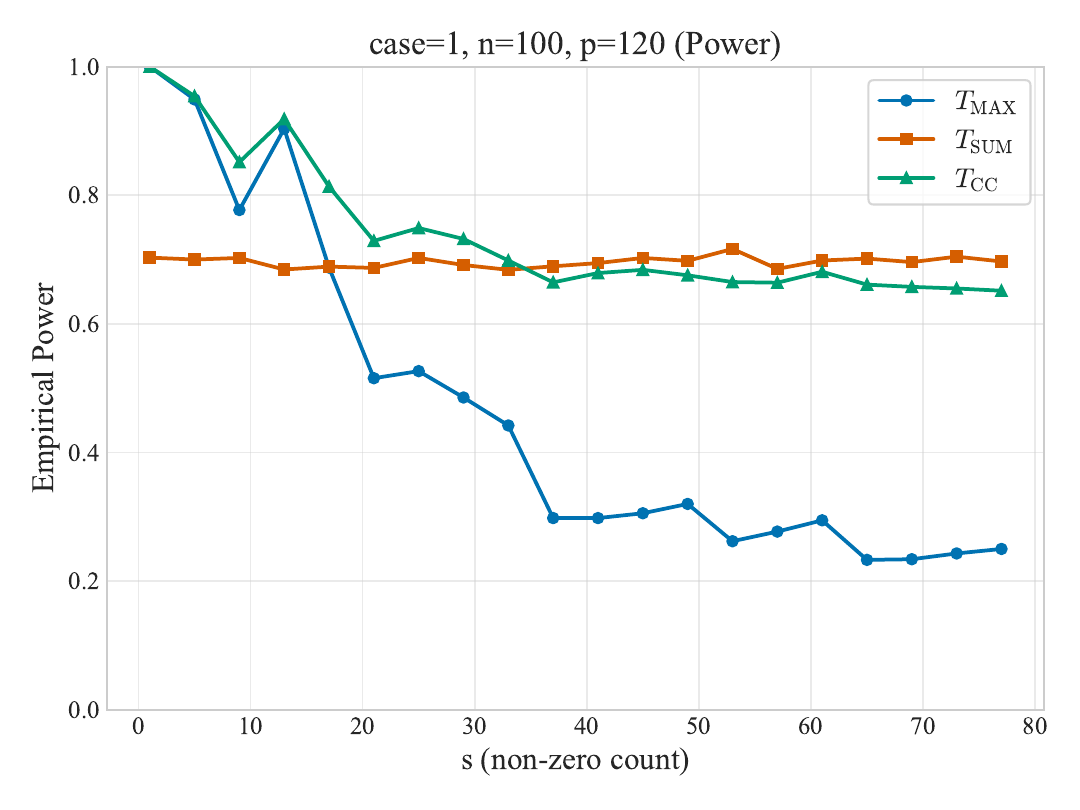}
\end{subfigure}
\hfill
\begin{subfigure}{0.23\textwidth}
    \centering
    \includegraphics[width=\linewidth]{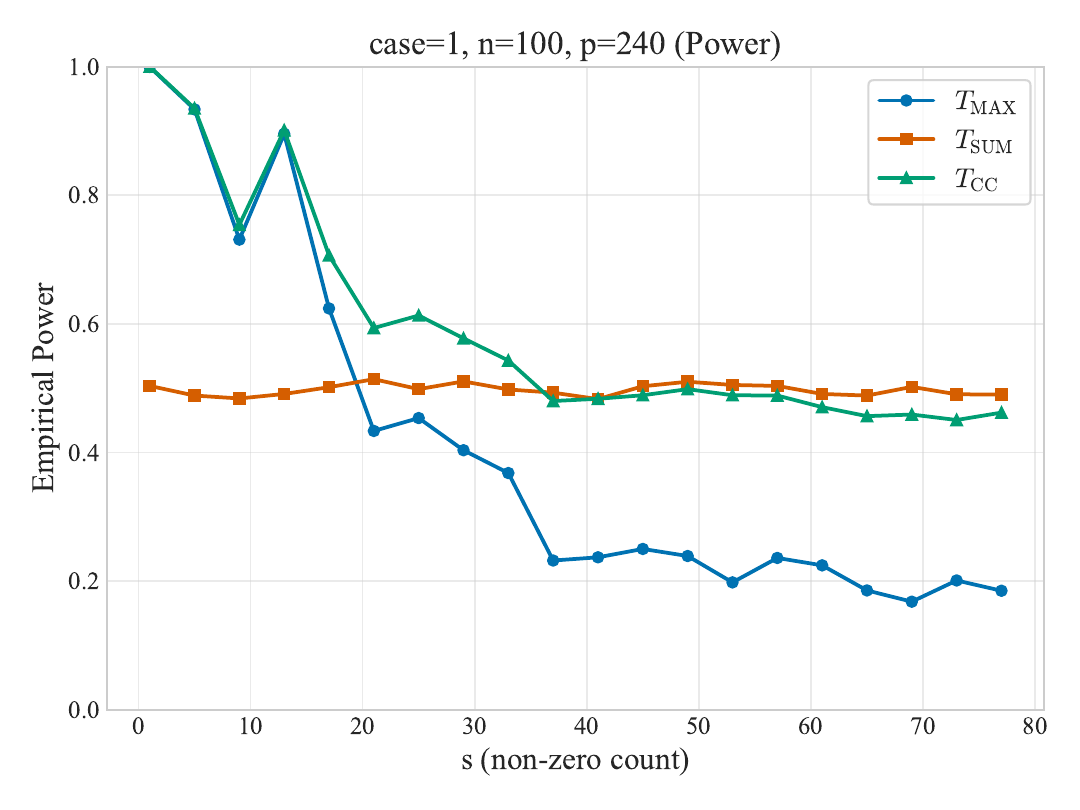}
\end{subfigure}
\hfill
\begin{subfigure}{0.23\textwidth}
    \centering
    \includegraphics[width=\linewidth]{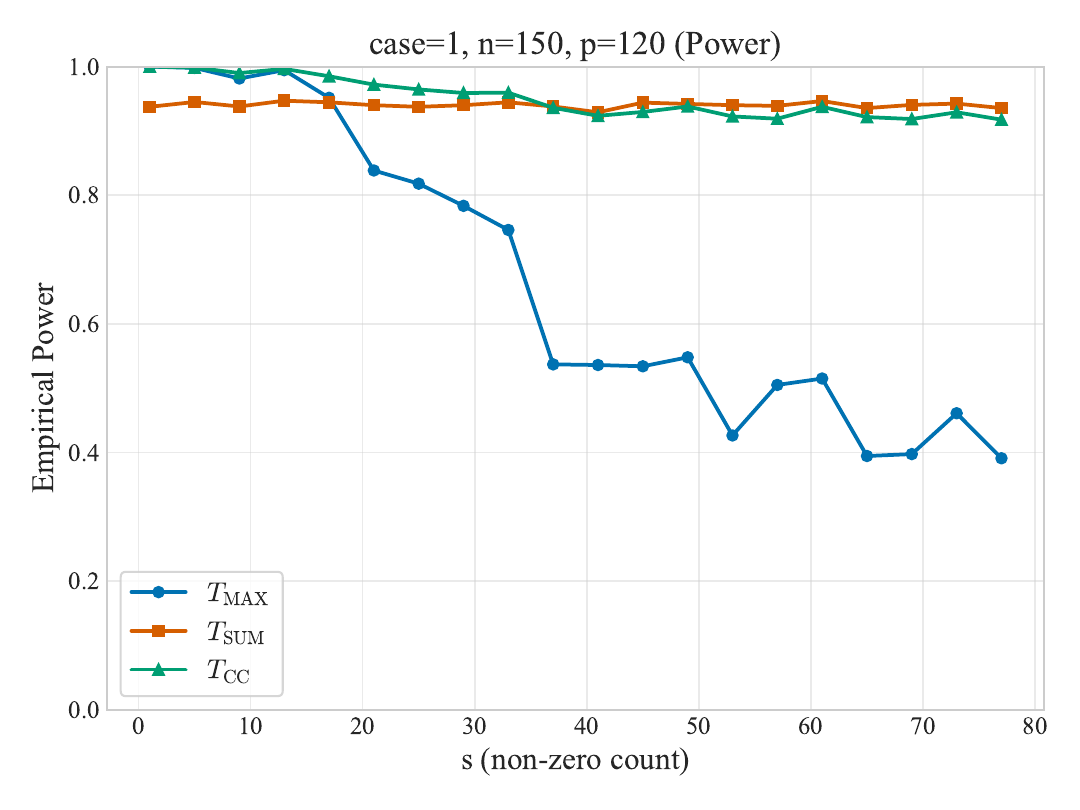}
\end{subfigure}
\hfill
\begin{subfigure}{0.23\textwidth}
    \centering
    \includegraphics[width=\linewidth]{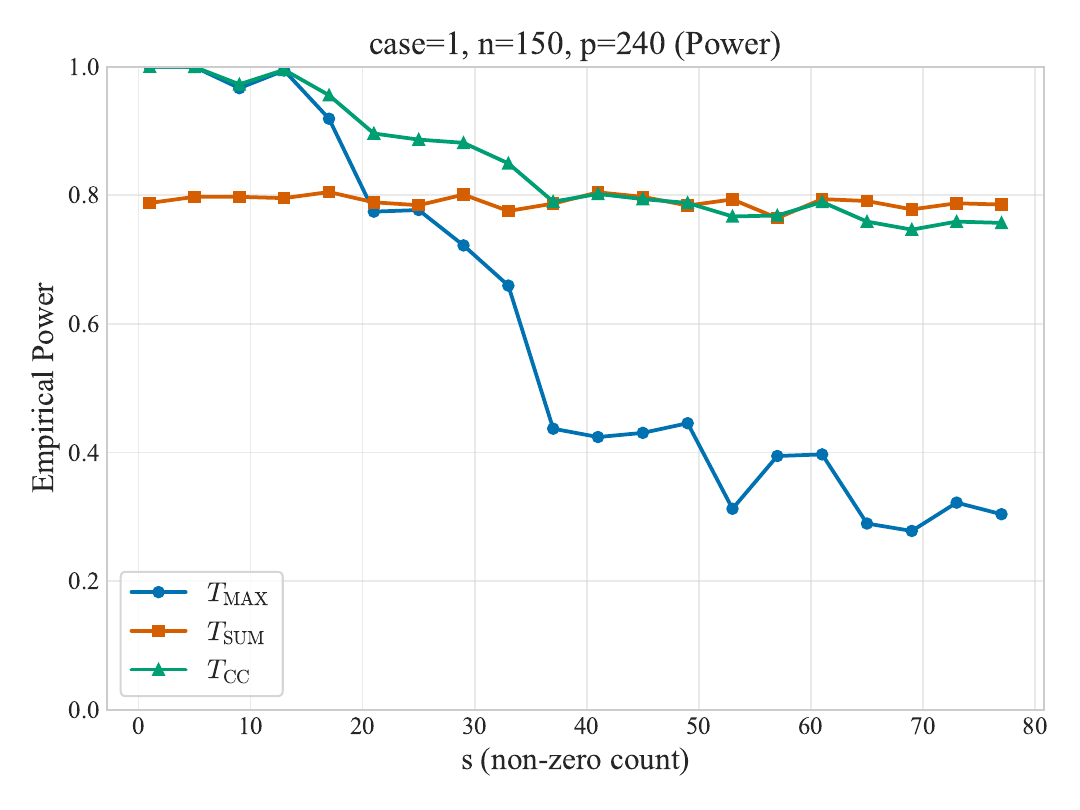}
\end{subfigure}

\vspace{0.3cm}
\begin{subfigure}{0.23\textwidth}
    \centering
    \includegraphics[width=\linewidth]{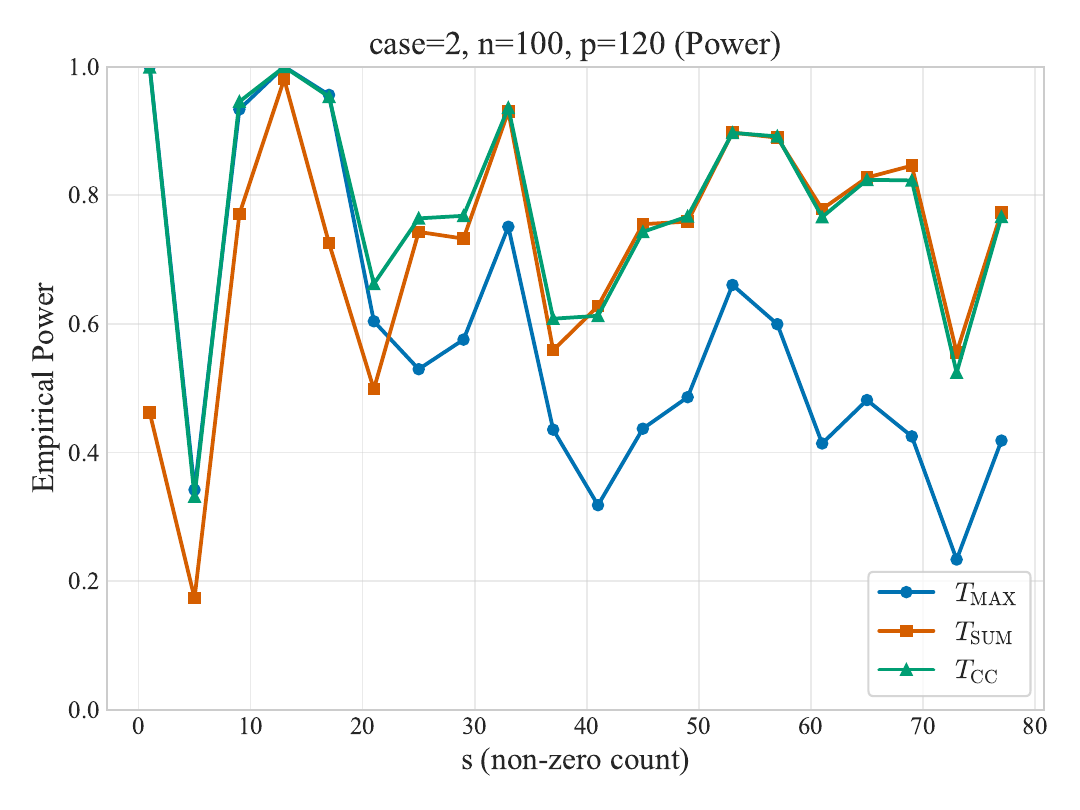}
\end{subfigure}
\hfill
\begin{subfigure}{0.23\textwidth}
    \centering
    \includegraphics[width=\linewidth]{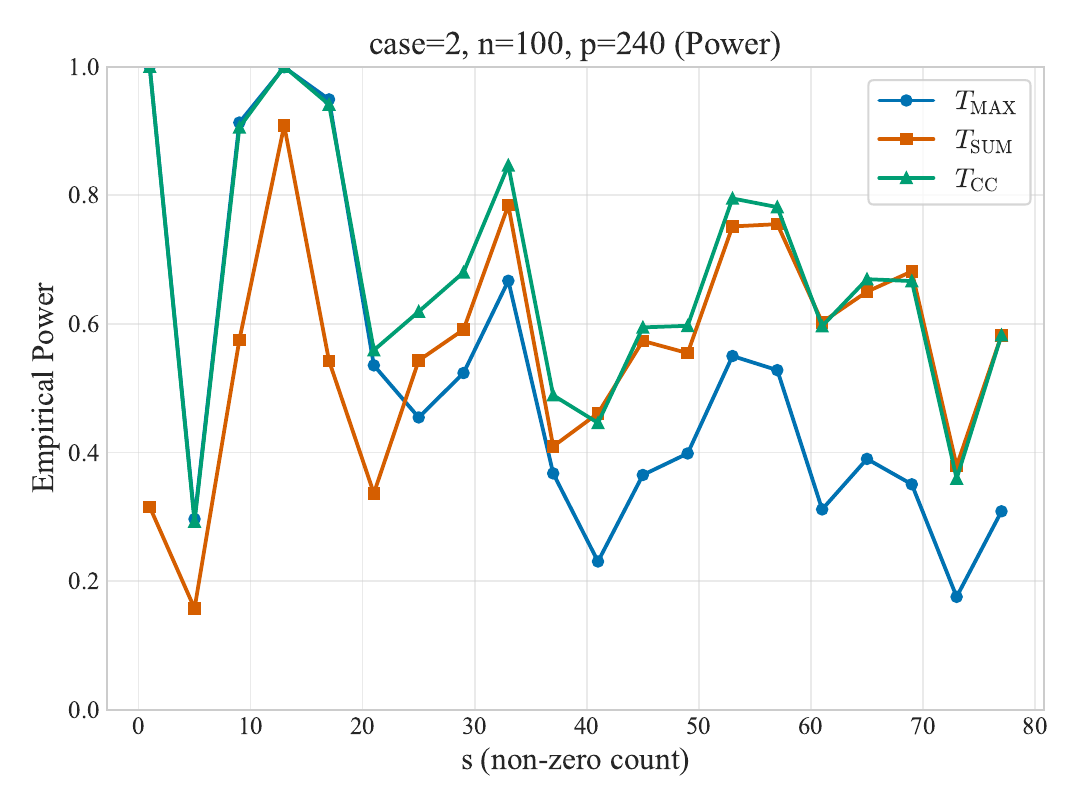}
\end{subfigure}
\hfill
\begin{subfigure}{0.23\textwidth}
    \centering
    \includegraphics[width=\linewidth]{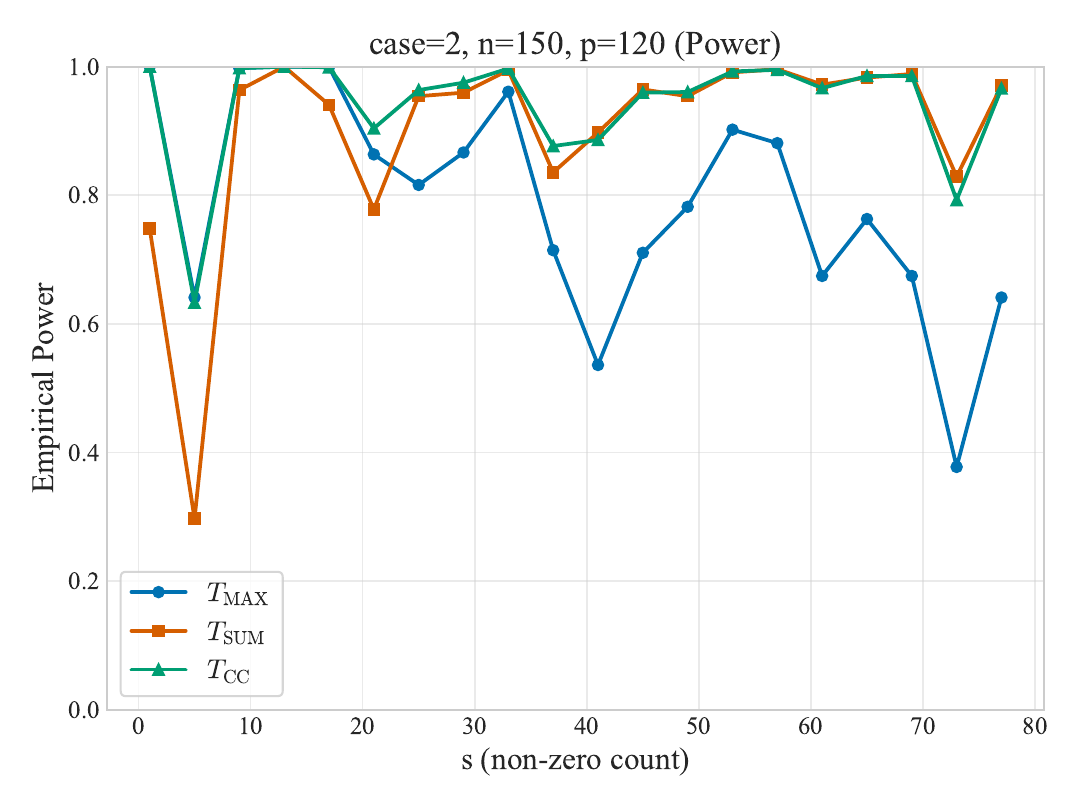}
\end{subfigure}
\hfill
\begin{subfigure}{0.23\textwidth}
    \centering
    \includegraphics[width=\linewidth]{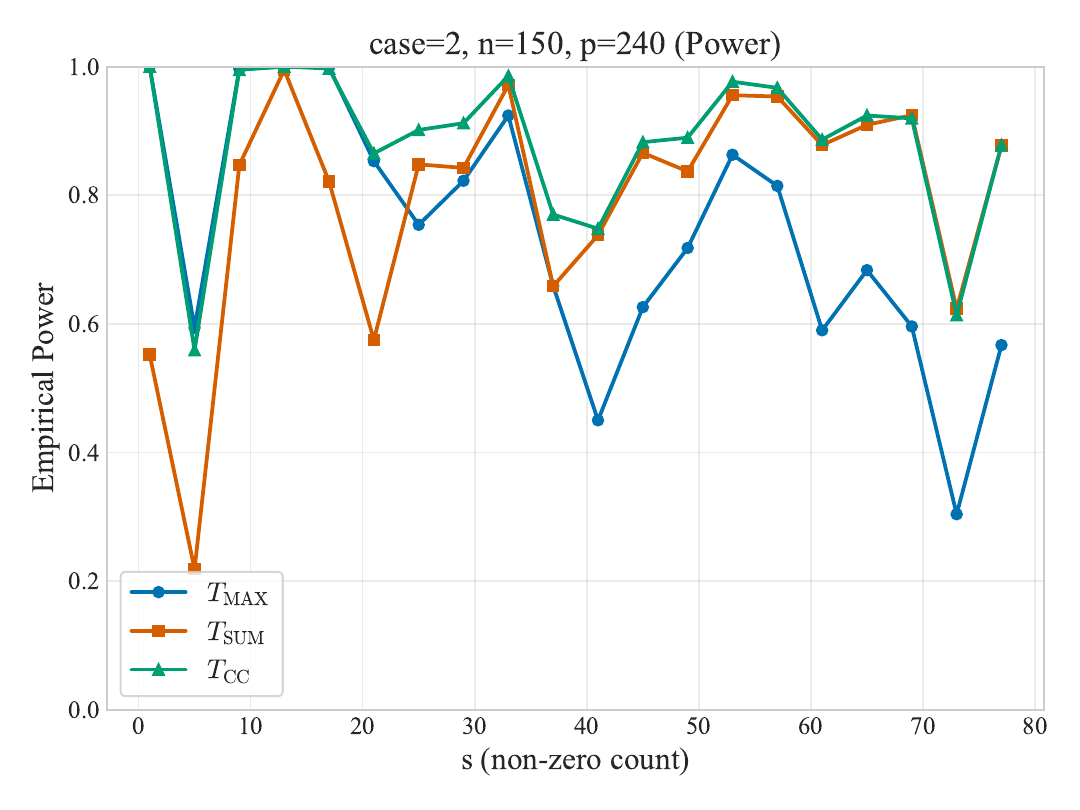}
\end{subfigure}

\vspace{0.3cm}
\begin{subfigure}{0.23\textwidth}
    \centering
    \includegraphics[width=\linewidth]{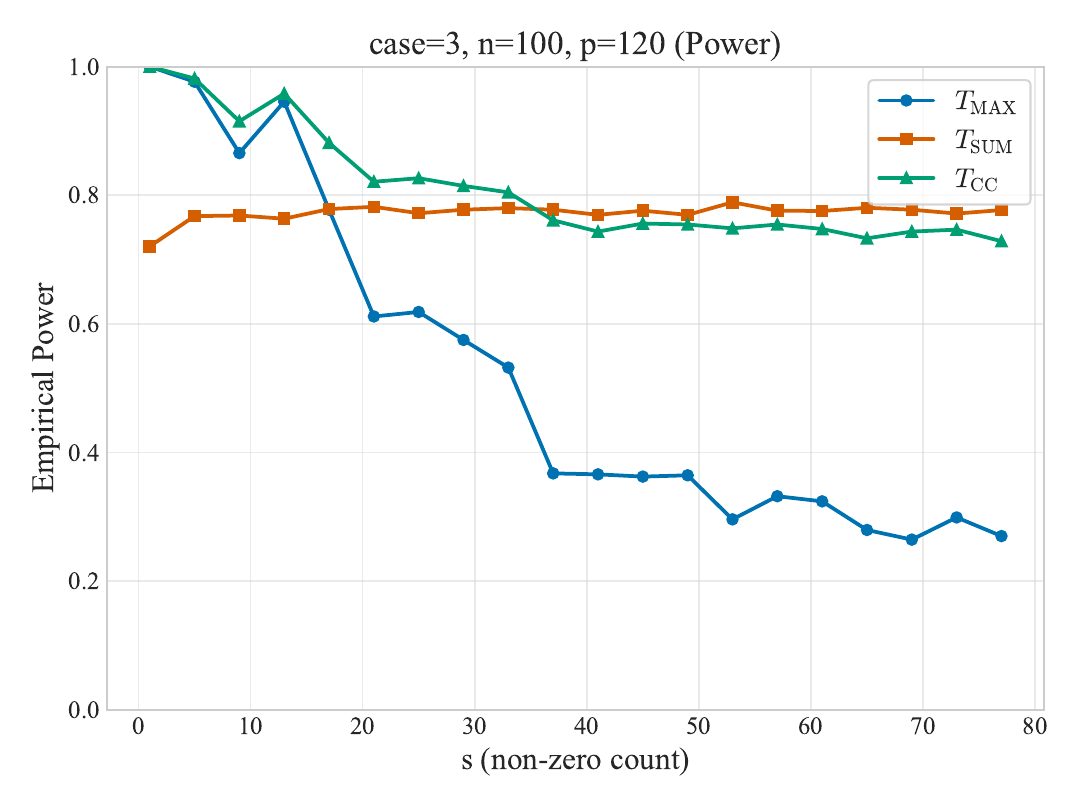}
\end{subfigure}
\hfill
\begin{subfigure}{0.23\textwidth}
    \centering
    \includegraphics[width=\linewidth]{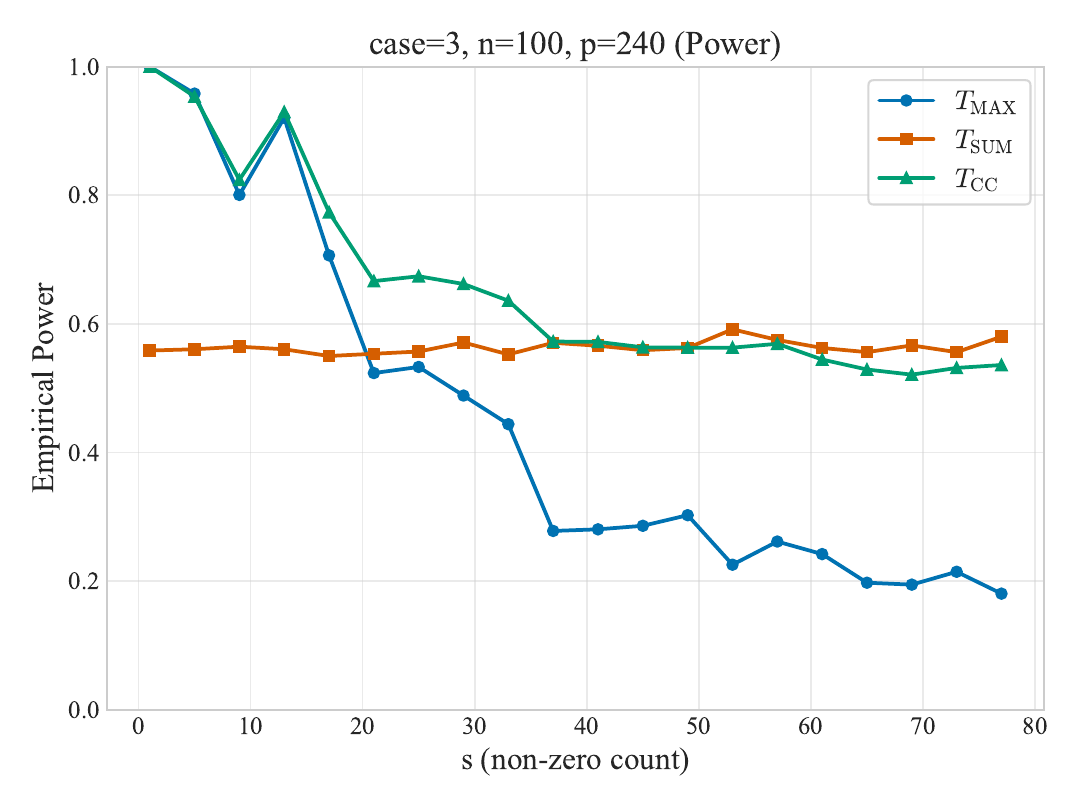}
\end{subfigure}
\hfill
\begin{subfigure}{0.23\textwidth}
    \centering
    \includegraphics[width=\linewidth]{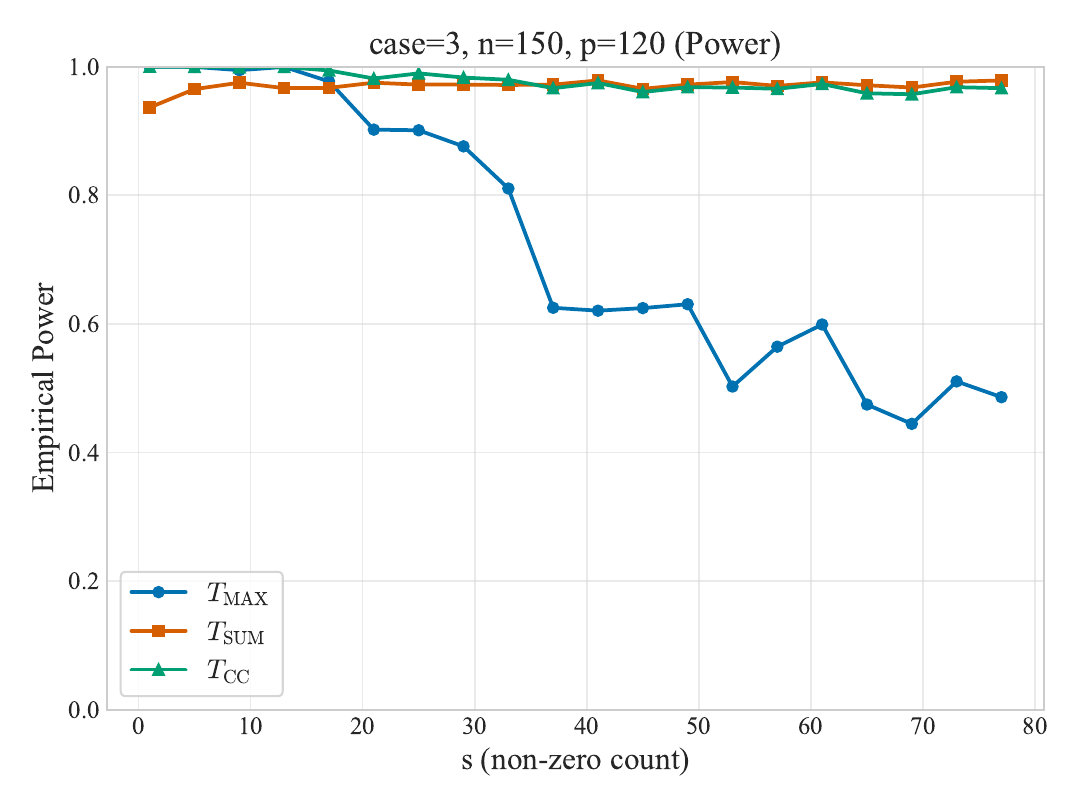}
\end{subfigure}
\hfill
\begin{subfigure}{0.23\textwidth}
    \centering
    \includegraphics[width=\linewidth]{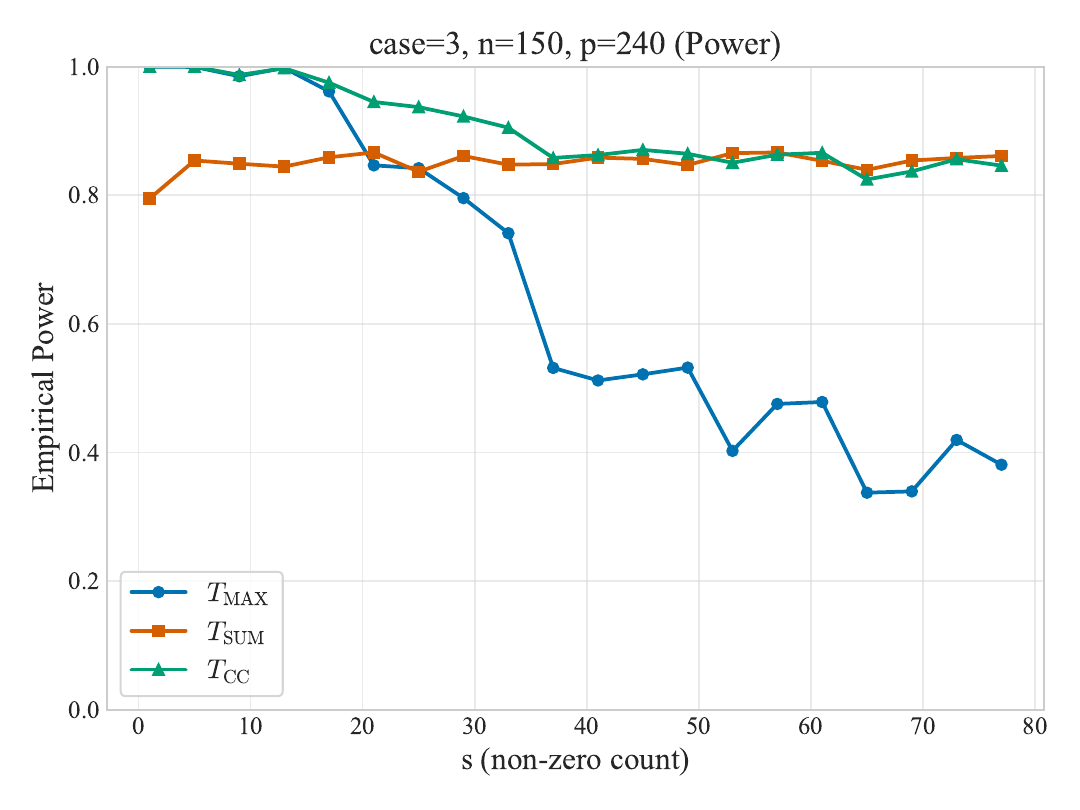}
\end{subfigure}

\caption{Empirical power as a function of $s$ for Cases~1--3 across varying $(n,p)$ 
settings under Normal distribution ($\tau = 0.25$; 2000 replications).}
\label{fig:power25}
\end{figure}


\begin{figure}[htbp]
\centering
\begin{subfigure}{0.23\textwidth}
    \centering
    \includegraphics[width=\linewidth]{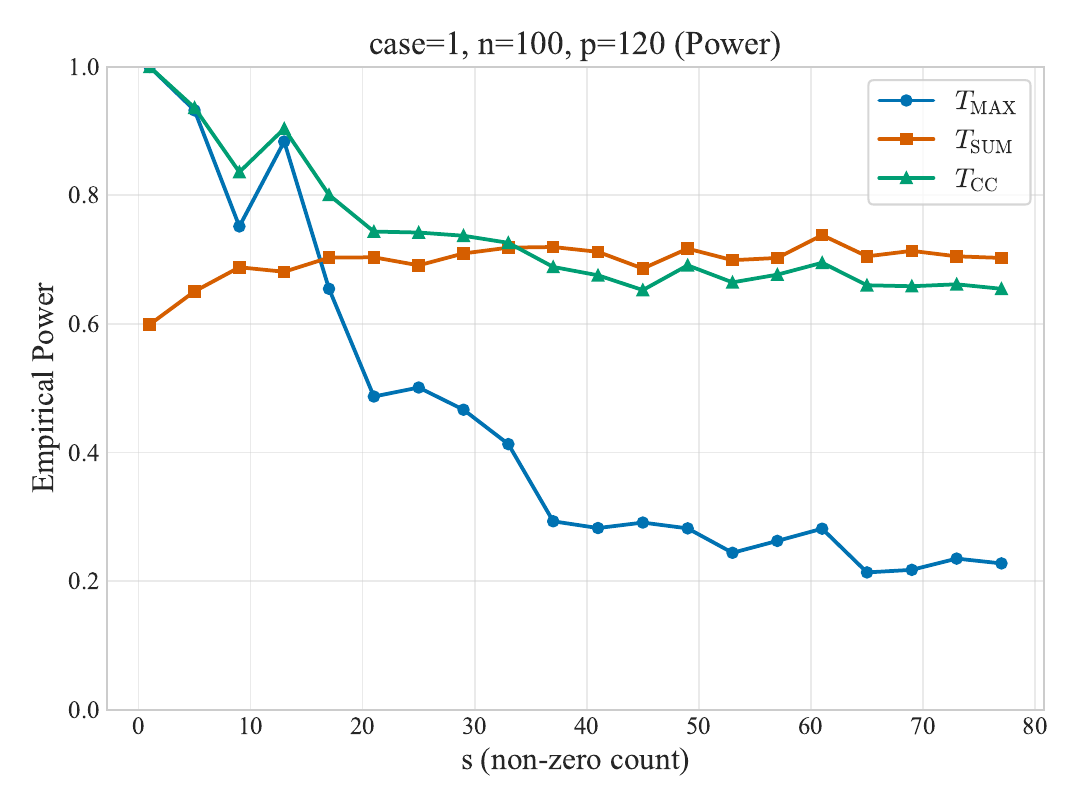}
\end{subfigure}
\hfill
\begin{subfigure}{0.23\textwidth}
    \centering
    \includegraphics[width=\linewidth]{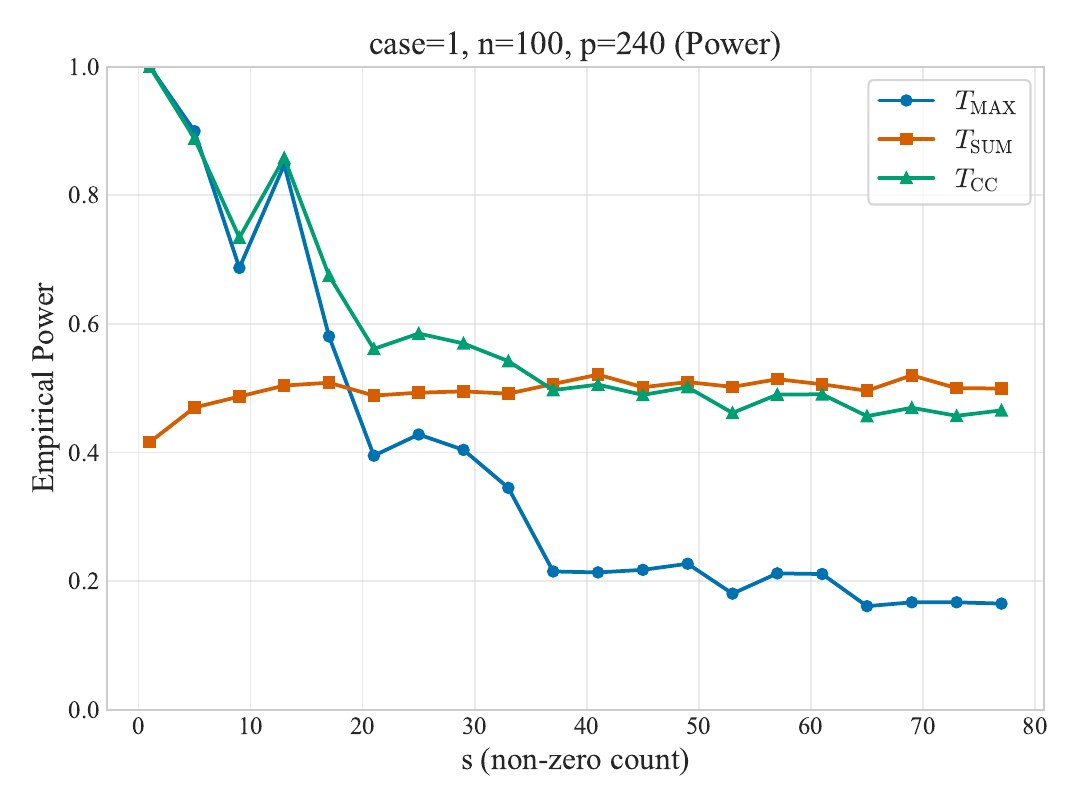}
\end{subfigure}
\hfill
\begin{subfigure}{0.23\textwidth}
    \centering
    \includegraphics[width=\linewidth]{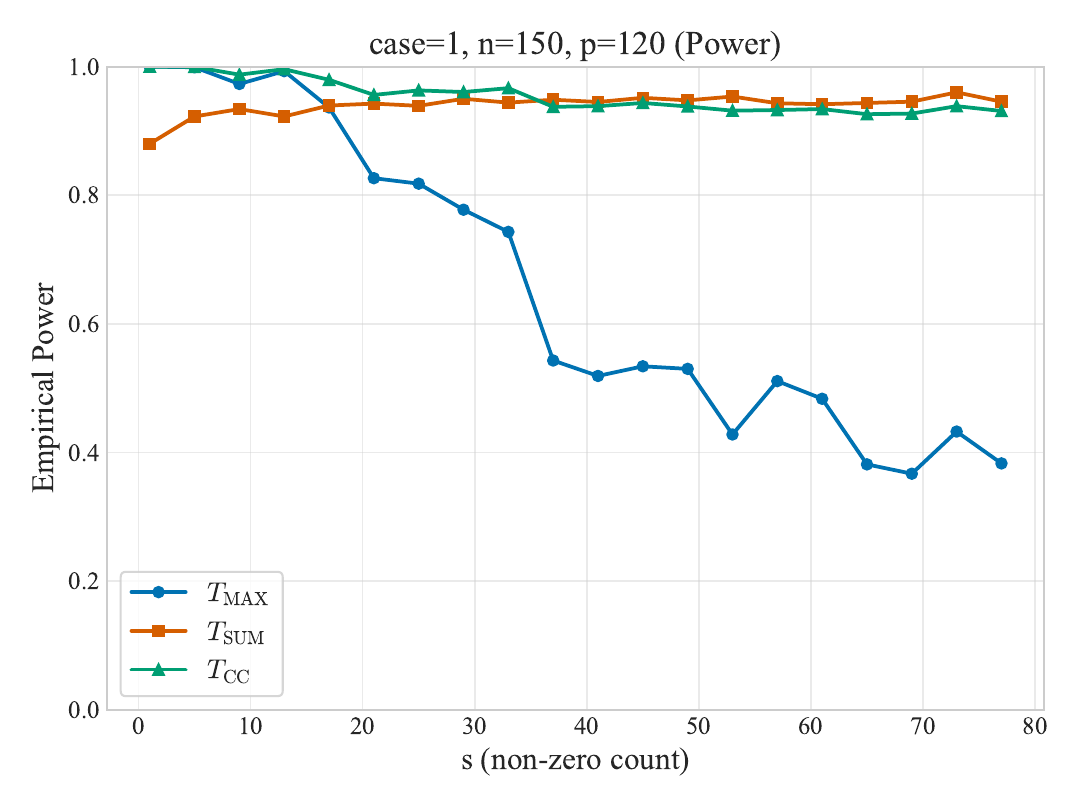}
\end{subfigure}
\hfill
\begin{subfigure}{0.23\textwidth}
    \centering
    \includegraphics[width=\linewidth]{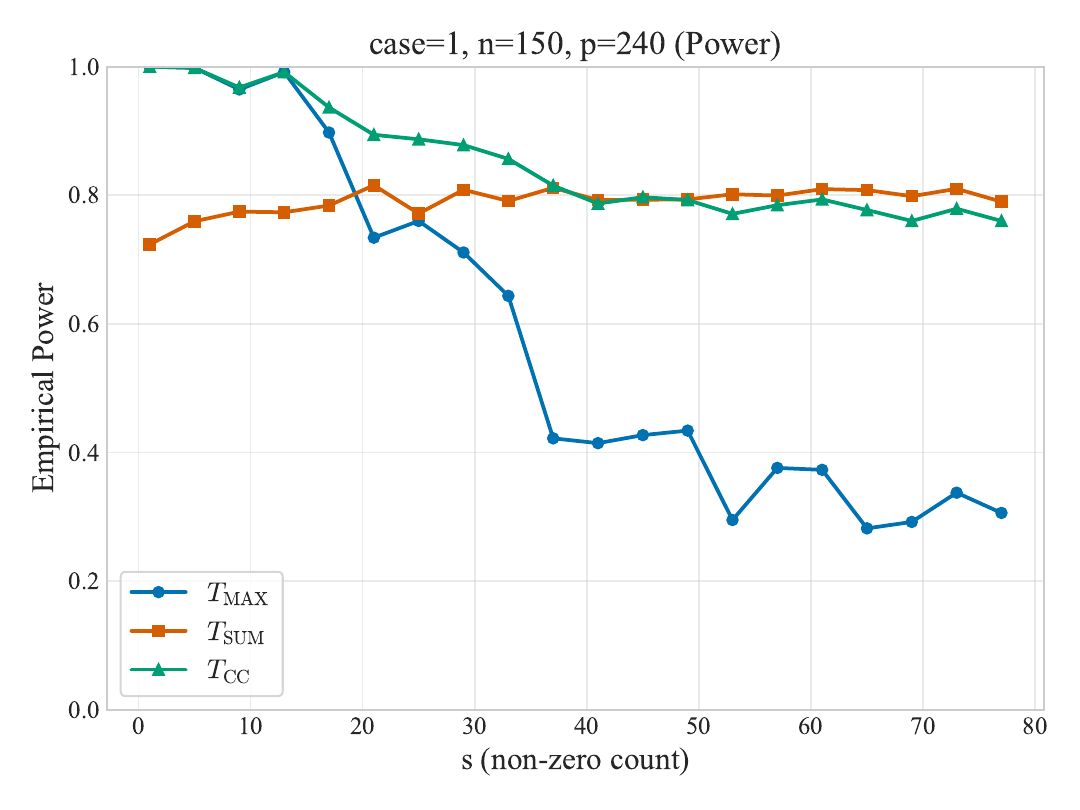}
\end{subfigure}

\vspace{0.3cm}
\begin{subfigure}{0.23\textwidth}
    \centering
    \includegraphics[width=\linewidth]{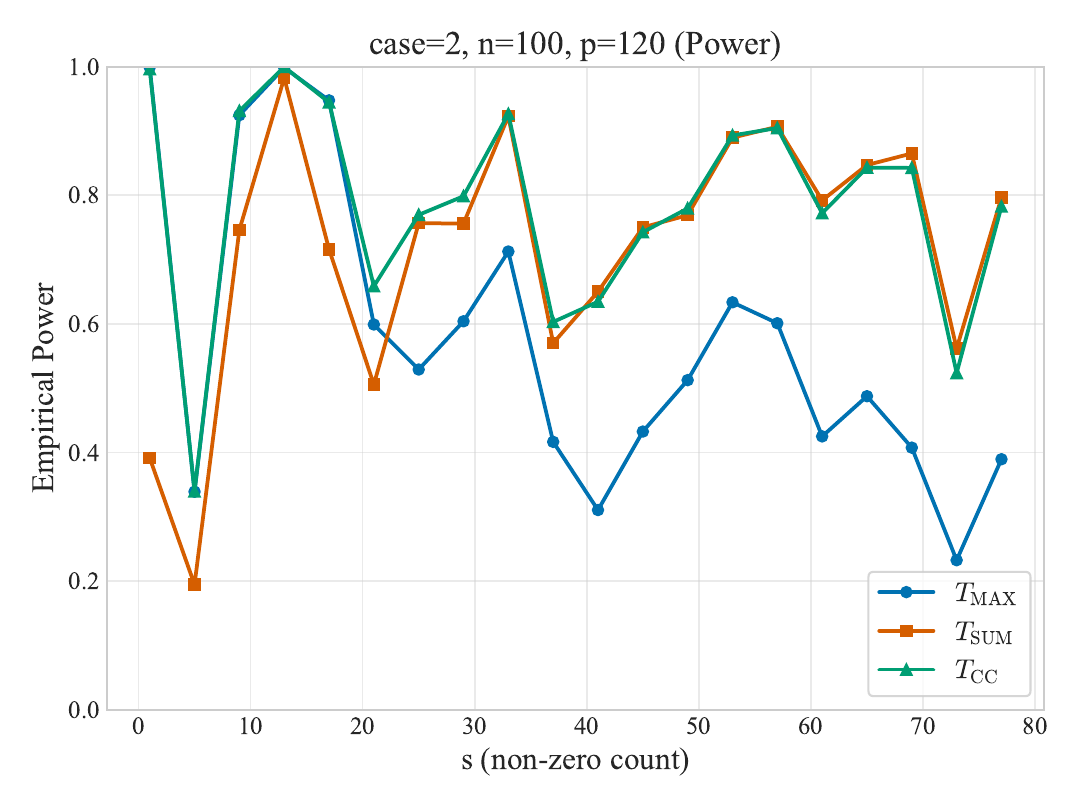}
\end{subfigure}
\hfill
\begin{subfigure}{0.23\textwidth}
    \centering
    \includegraphics[width=\linewidth]{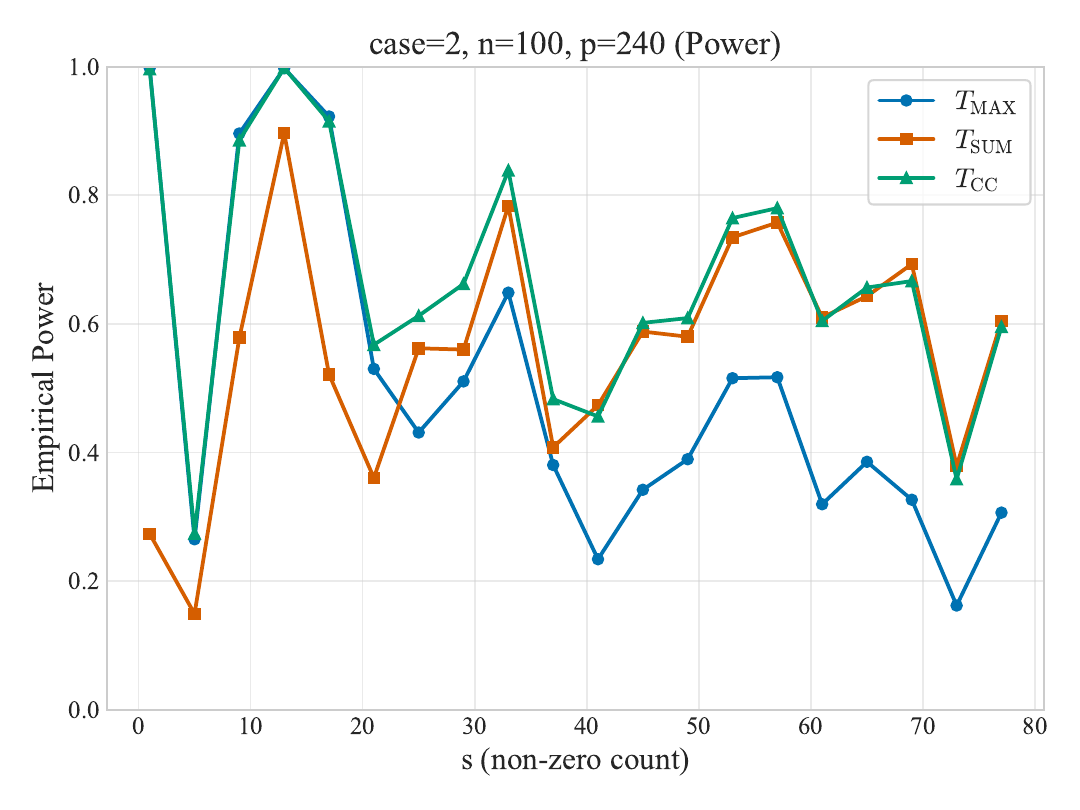}
\end{subfigure}
\hfill
\begin{subfigure}{0.23\textwidth}
    \centering
    \includegraphics[width=\linewidth]{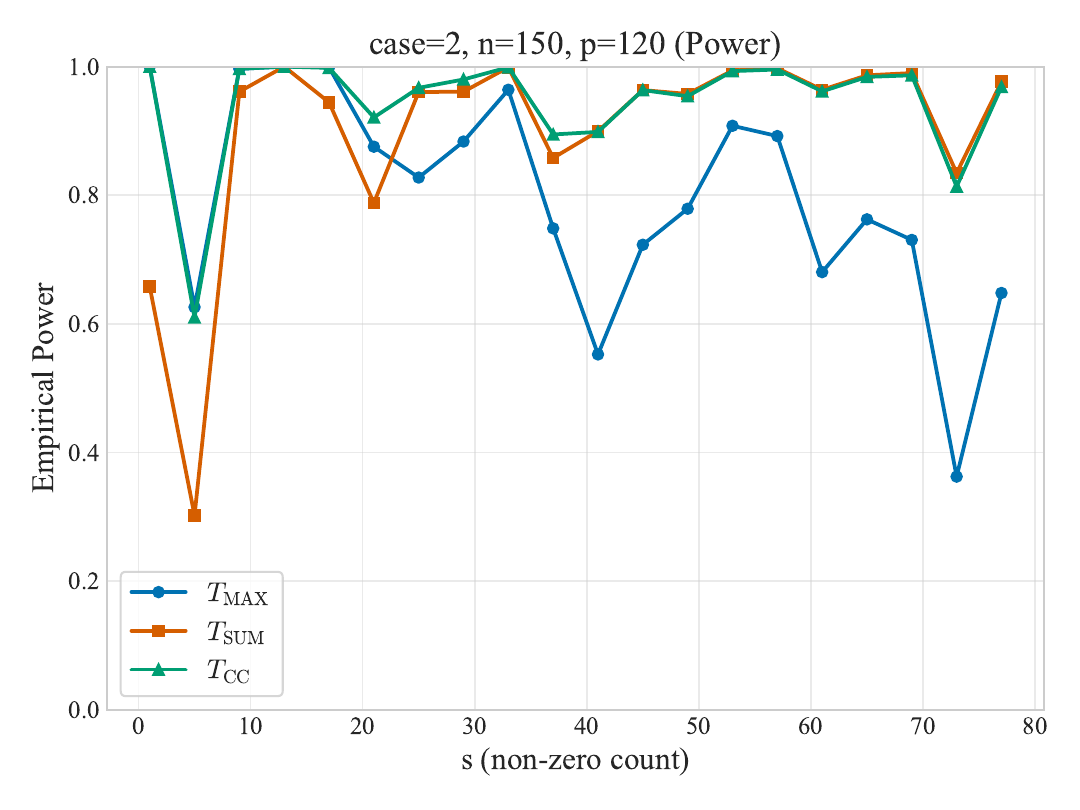}
\end{subfigure}
\hfill
\begin{subfigure}{0.23\textwidth}
    \centering
    \includegraphics[width=\linewidth]{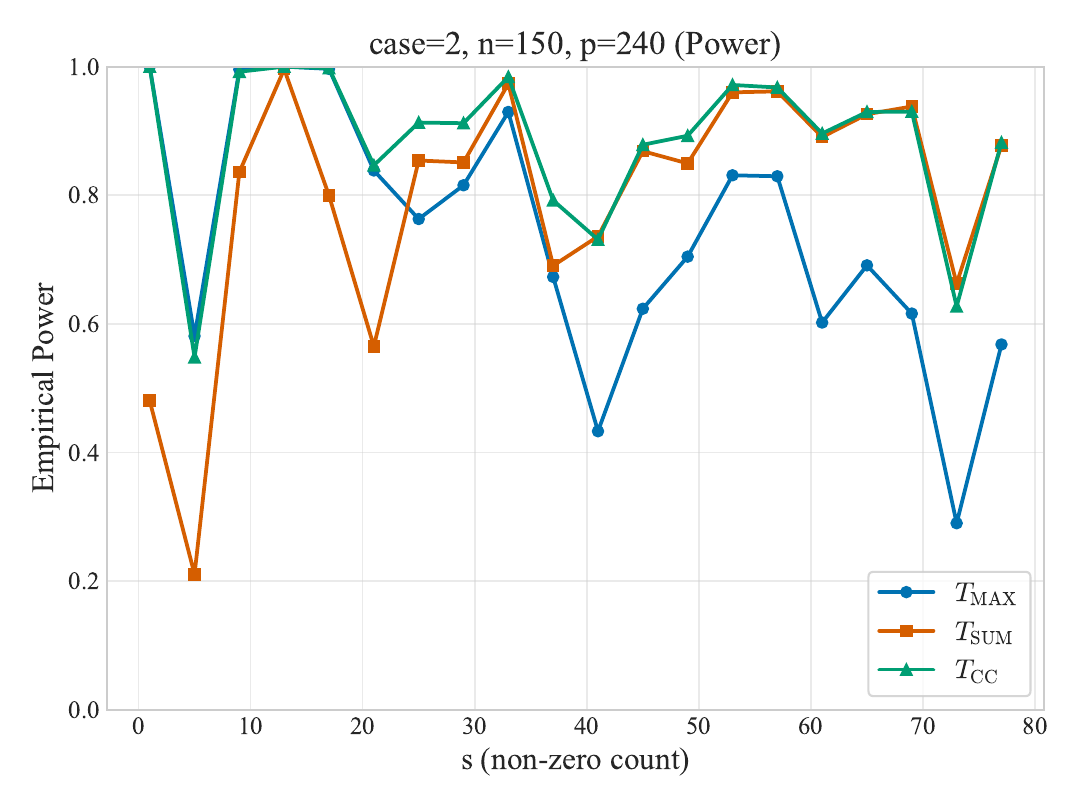}
\end{subfigure}

\vspace{0.3cm}
\begin{subfigure}{0.23\textwidth}
    \centering
    \includegraphics[width=\linewidth]{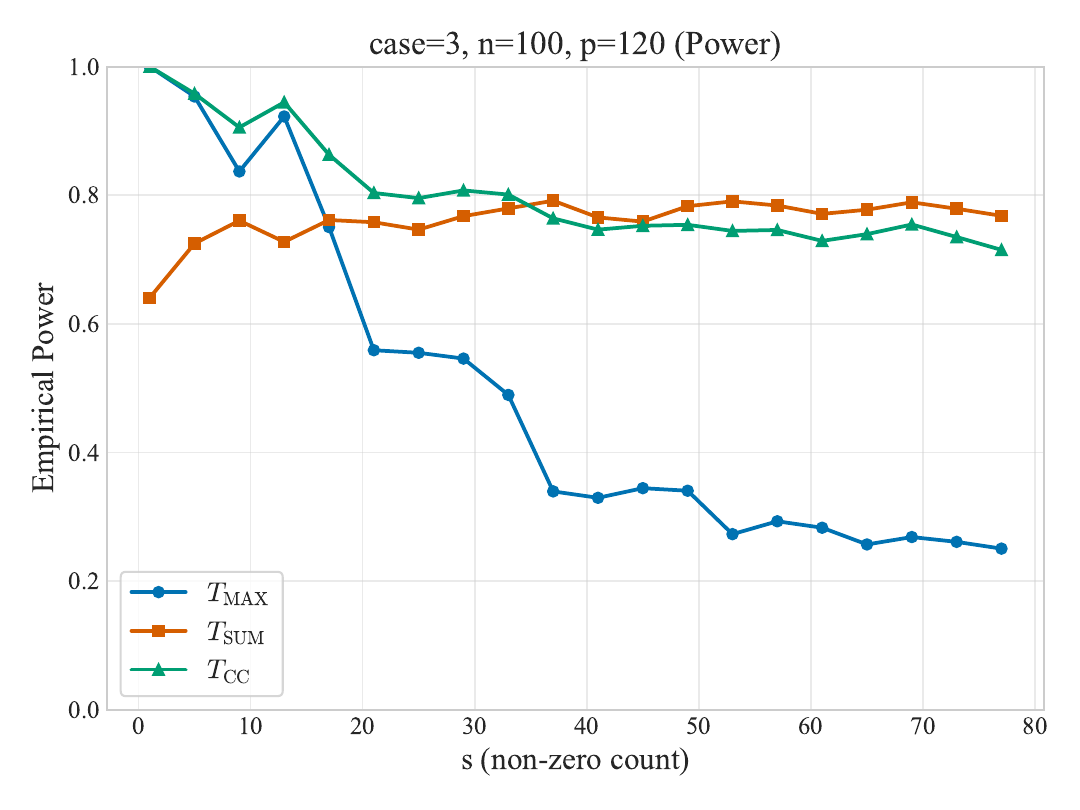}
\end{subfigure}
\hfill
\begin{subfigure}{0.23\textwidth}
    \centering
    \includegraphics[width=\linewidth]{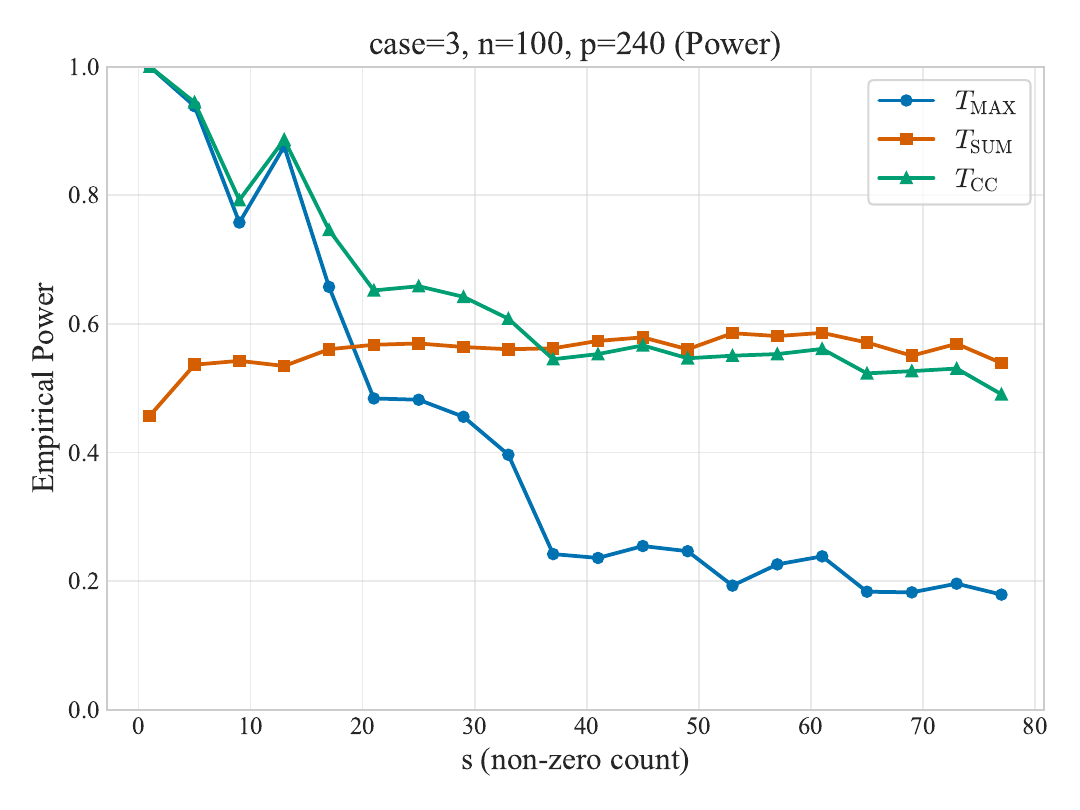}
\end{subfigure}
\hfill
\begin{subfigure}{0.23\textwidth}
    \centering
    \includegraphics[width=\linewidth]{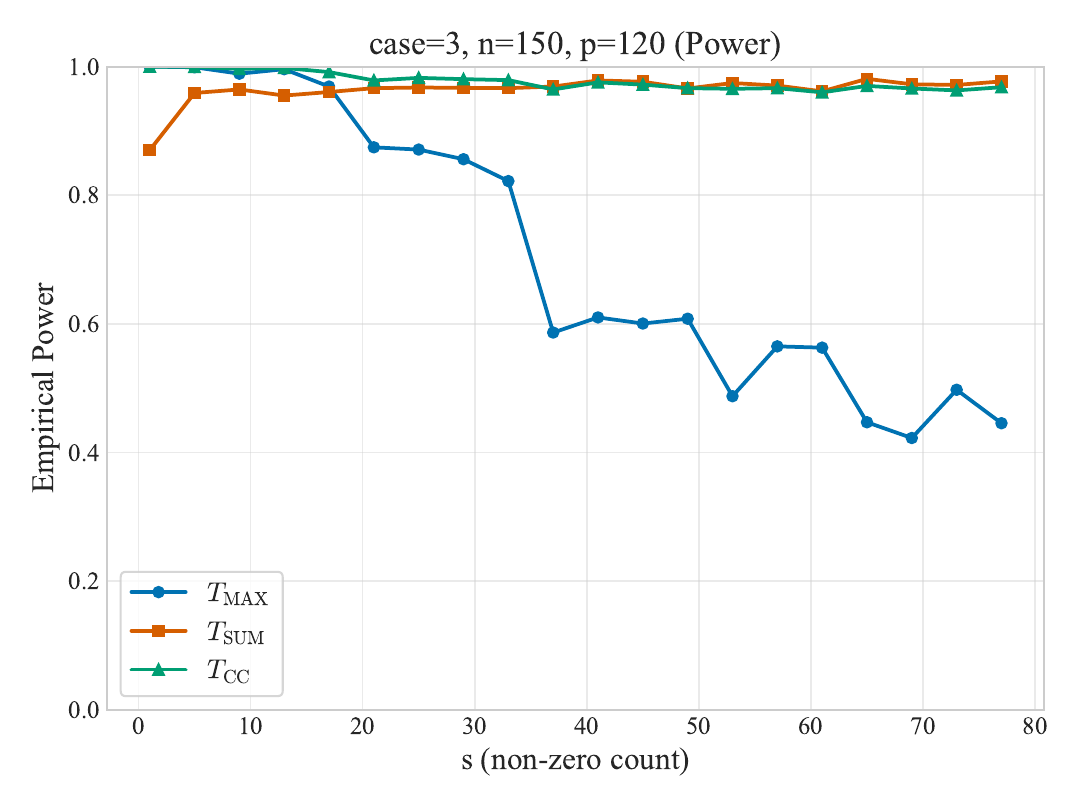}
\end{subfigure}
\hfill
\begin{subfigure}{0.23\textwidth}
    \centering
    \includegraphics[width=\linewidth]{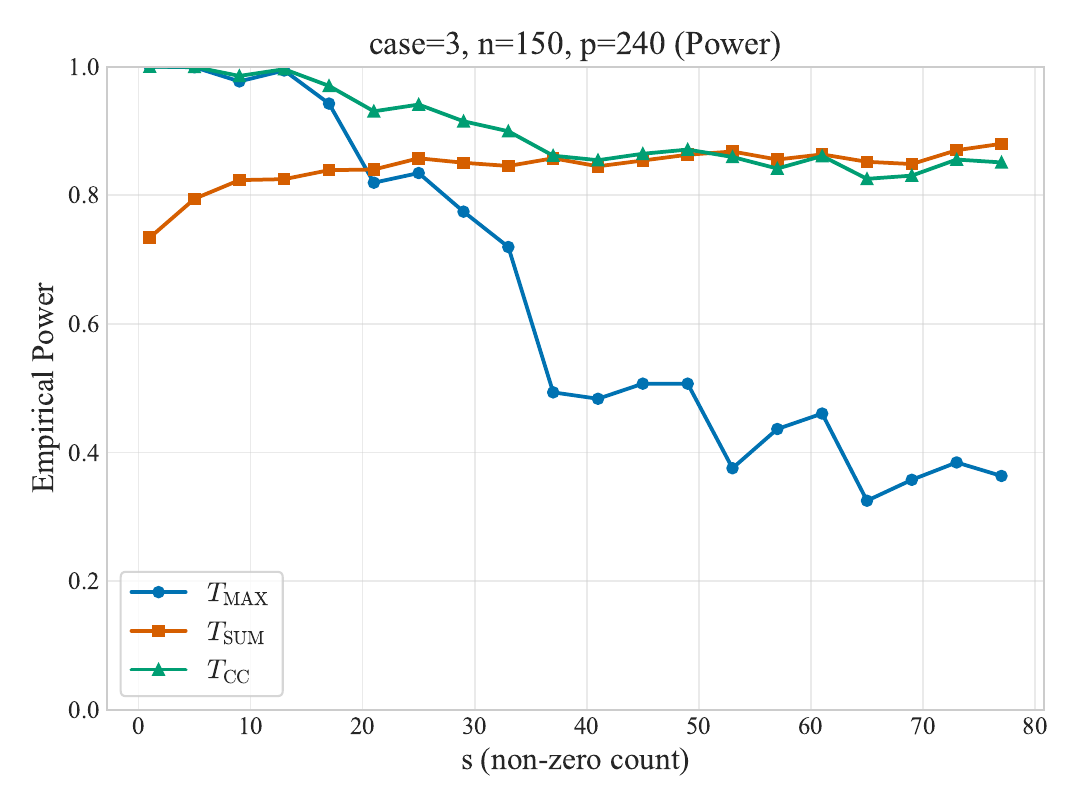}
\end{subfigure}

\caption{Empirical power as a function of $s$ for Cases~1--3 across varying $(n,p)$ 
settings under Laplace distribution ($\tau = 0.25$; 2000 replications).}
\label{fig:power25_laplace}
\end{figure}



\begin{figure}[htbp]
\centering
\begin{subfigure}{0.23\textwidth}
    \centering
    \includegraphics[width=\linewidth]{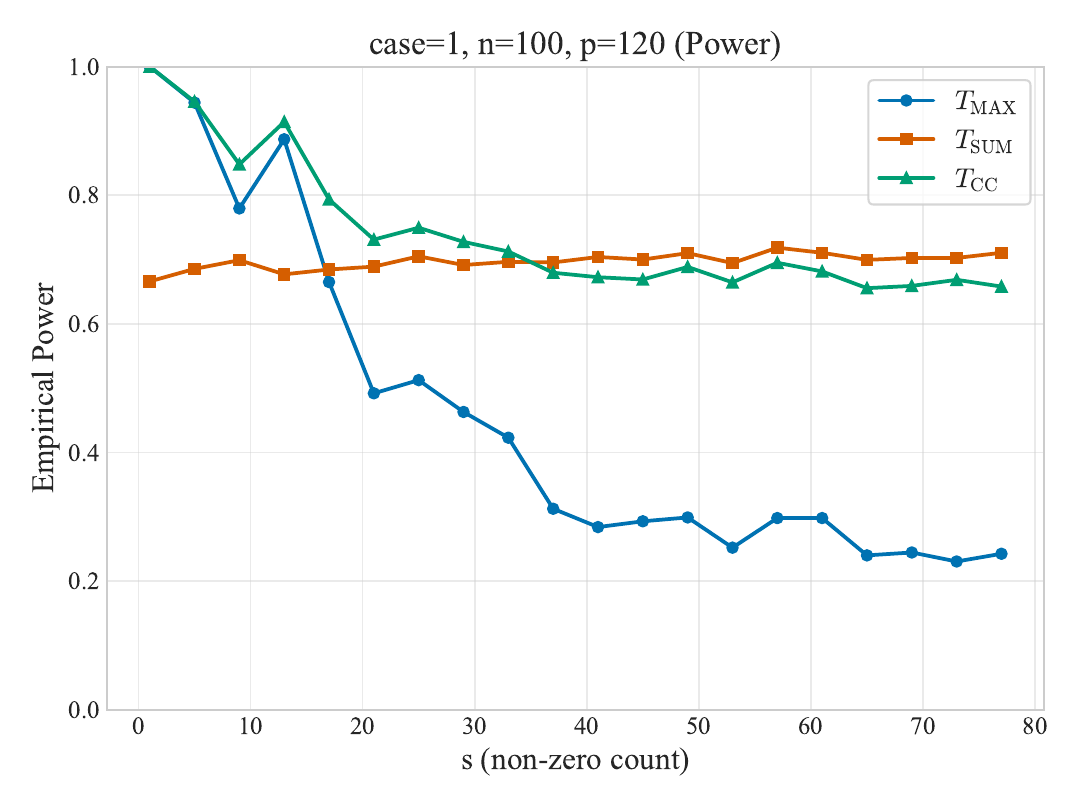}
\end{subfigure}
\hfill
\begin{subfigure}{0.23\textwidth}
    \centering
    \includegraphics[width=\linewidth]{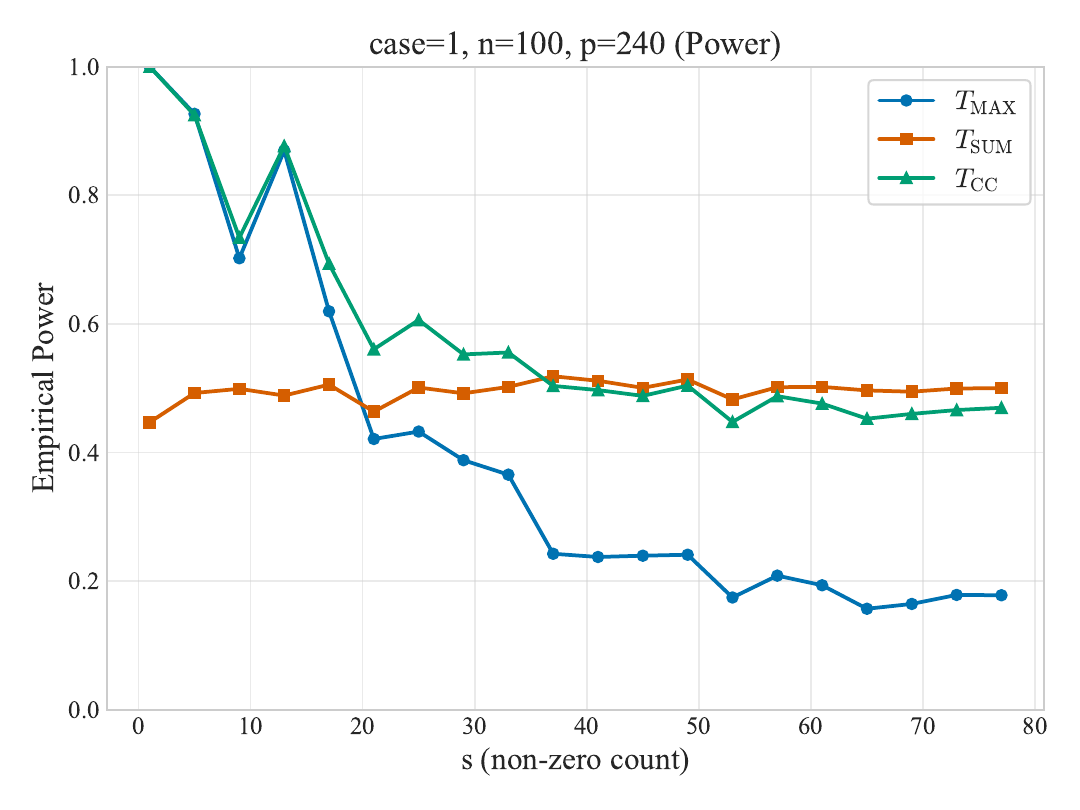}
\end{subfigure}
\hfill
\begin{subfigure}{0.23\textwidth}
    \centering
    \includegraphics[width=\linewidth]{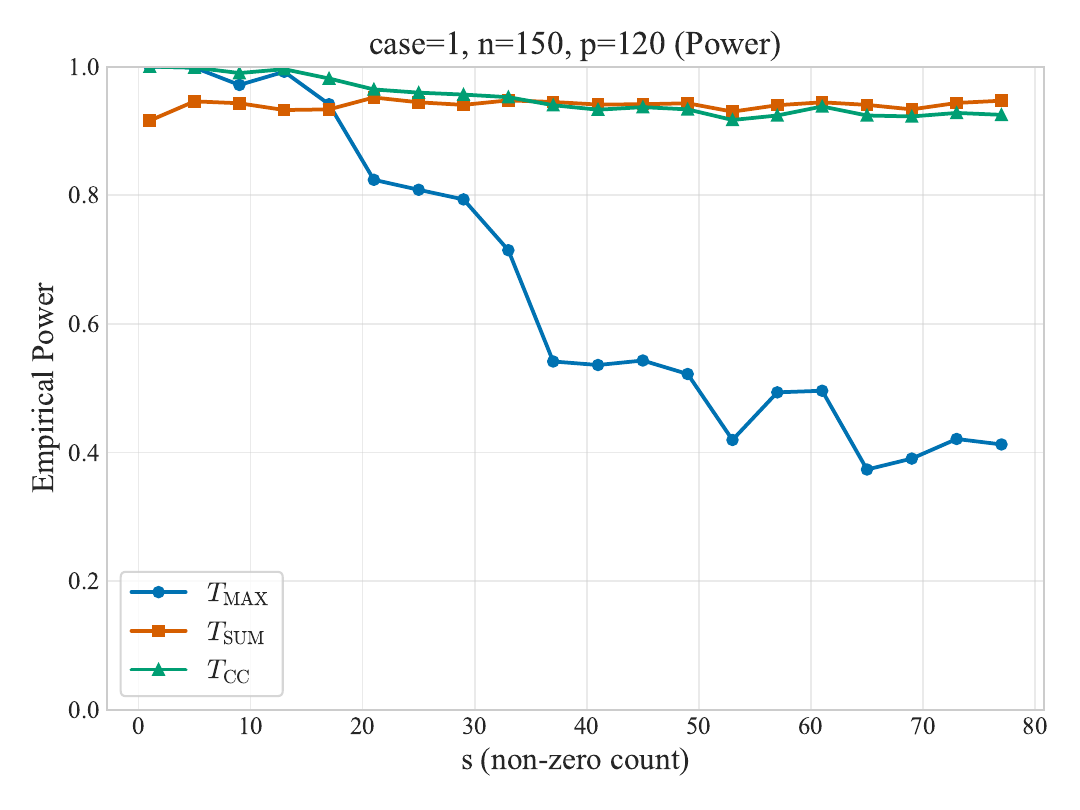}
\end{subfigure}
\hfill
\begin{subfigure}{0.23\textwidth}
    \centering
    \includegraphics[width=\linewidth]{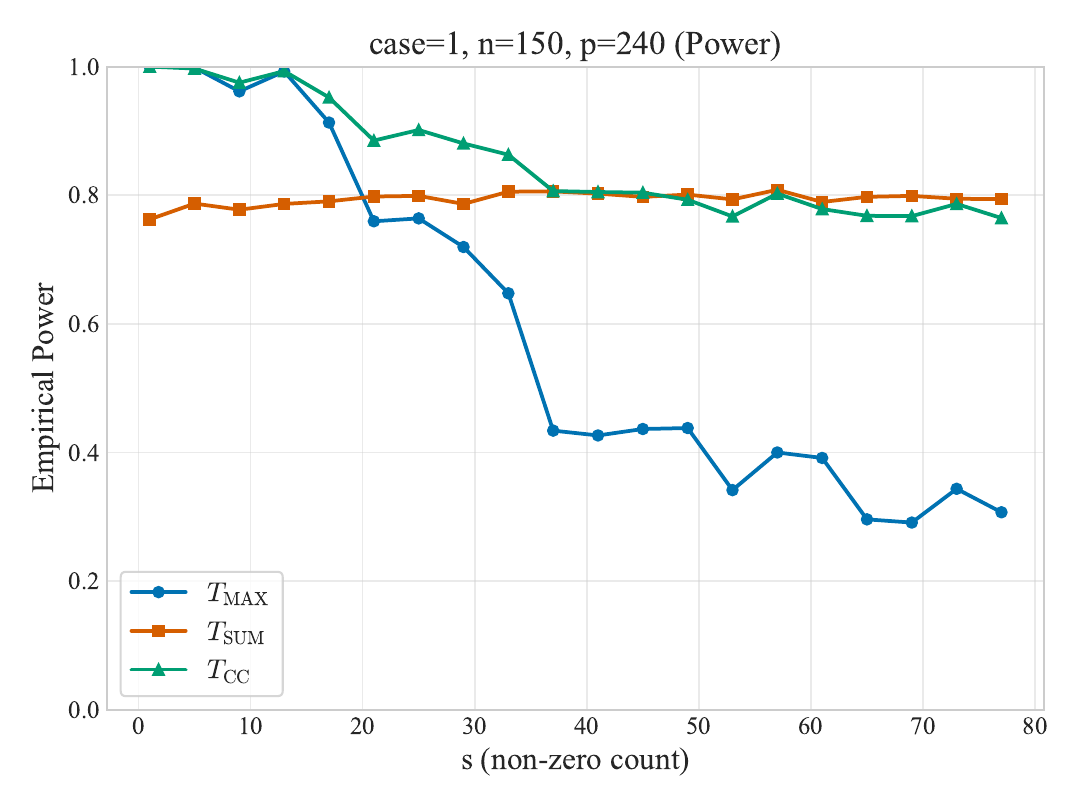}
\end{subfigure}

\vspace{0.3cm}
\begin{subfigure}{0.23\textwidth}
    \centering
    \includegraphics[width=\linewidth]{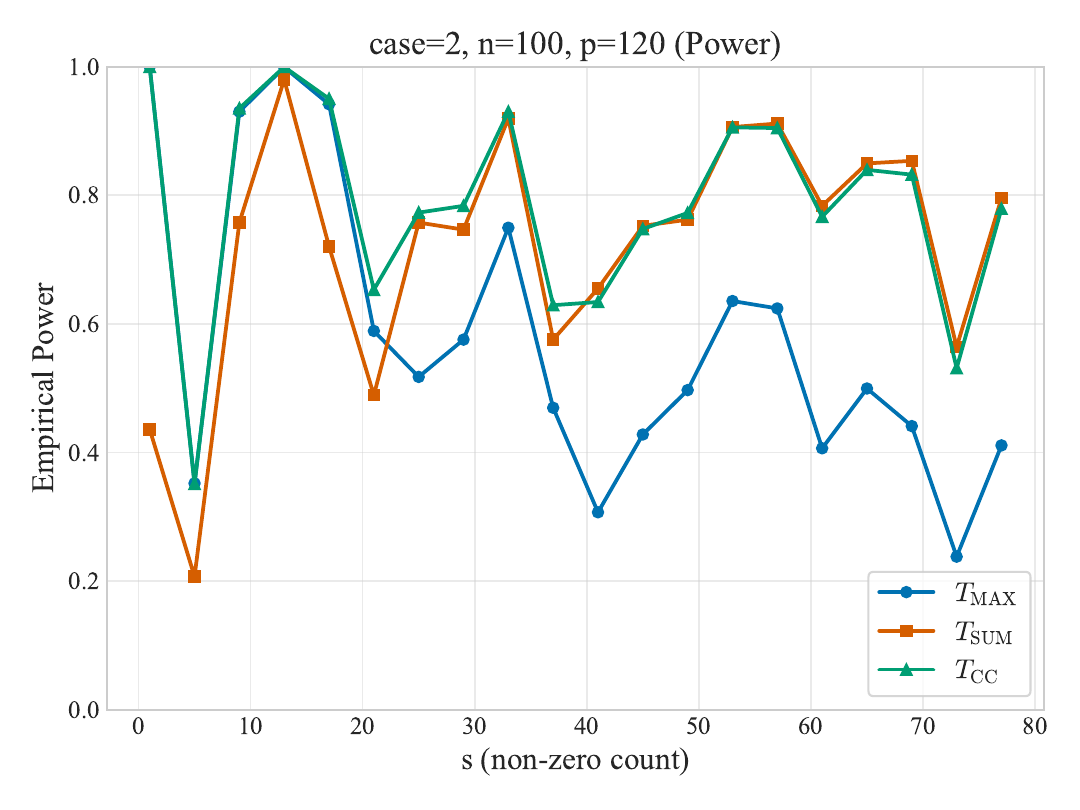}
\end{subfigure}
\hfill
\begin{subfigure}{0.23\textwidth}
    \centering
    \includegraphics[width=\linewidth]{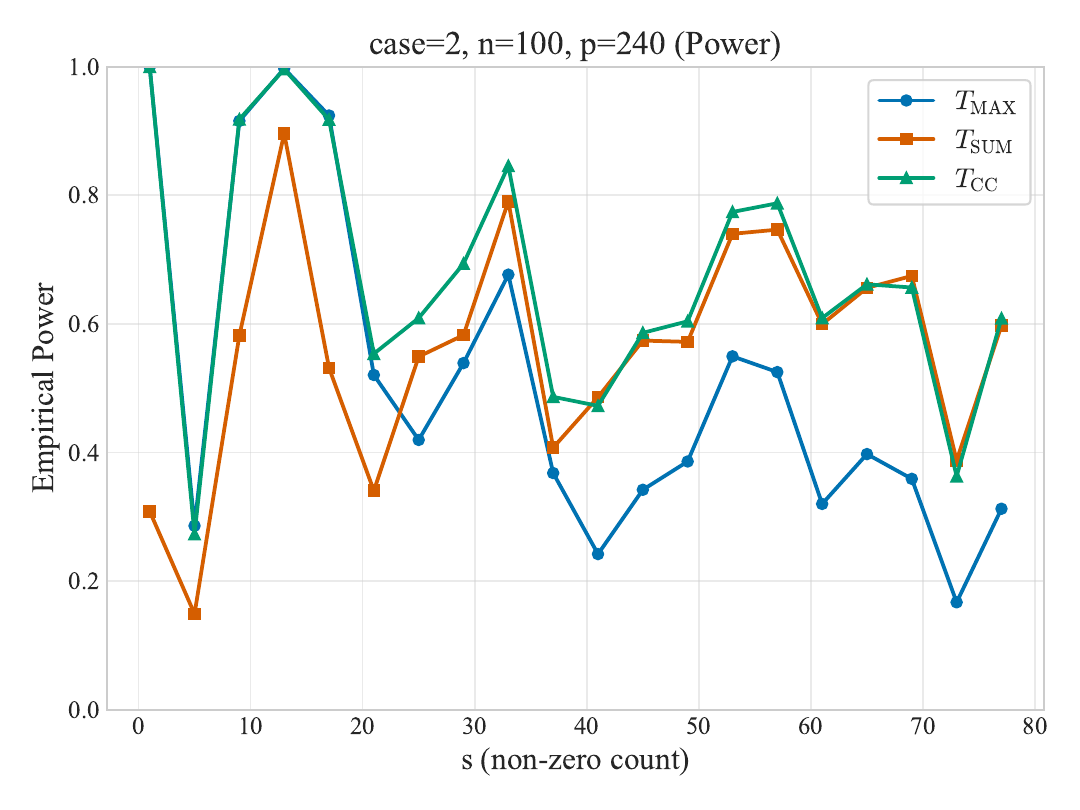}
\end{subfigure}
\hfill
\begin{subfigure}{0.23\textwidth}
    \centering
    \includegraphics[width=\linewidth]{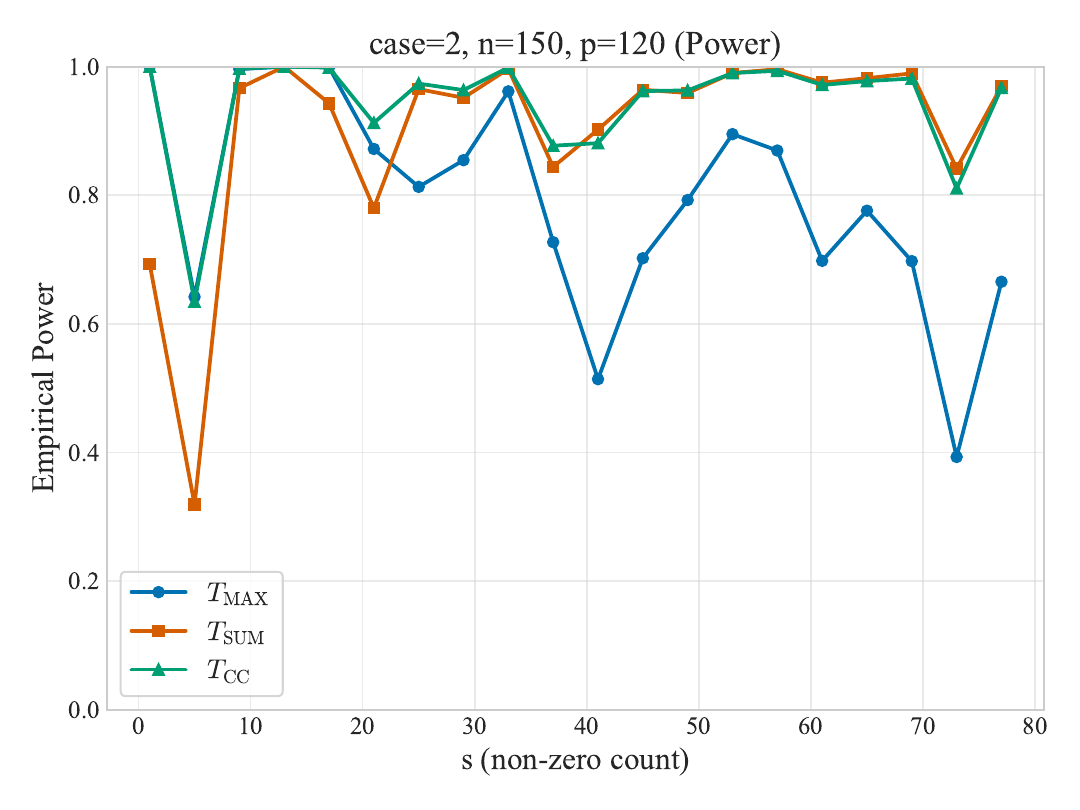}
\end{subfigure}
\hfill
\begin{subfigure}{0.23\textwidth}
    \centering
    \includegraphics[width=\linewidth]{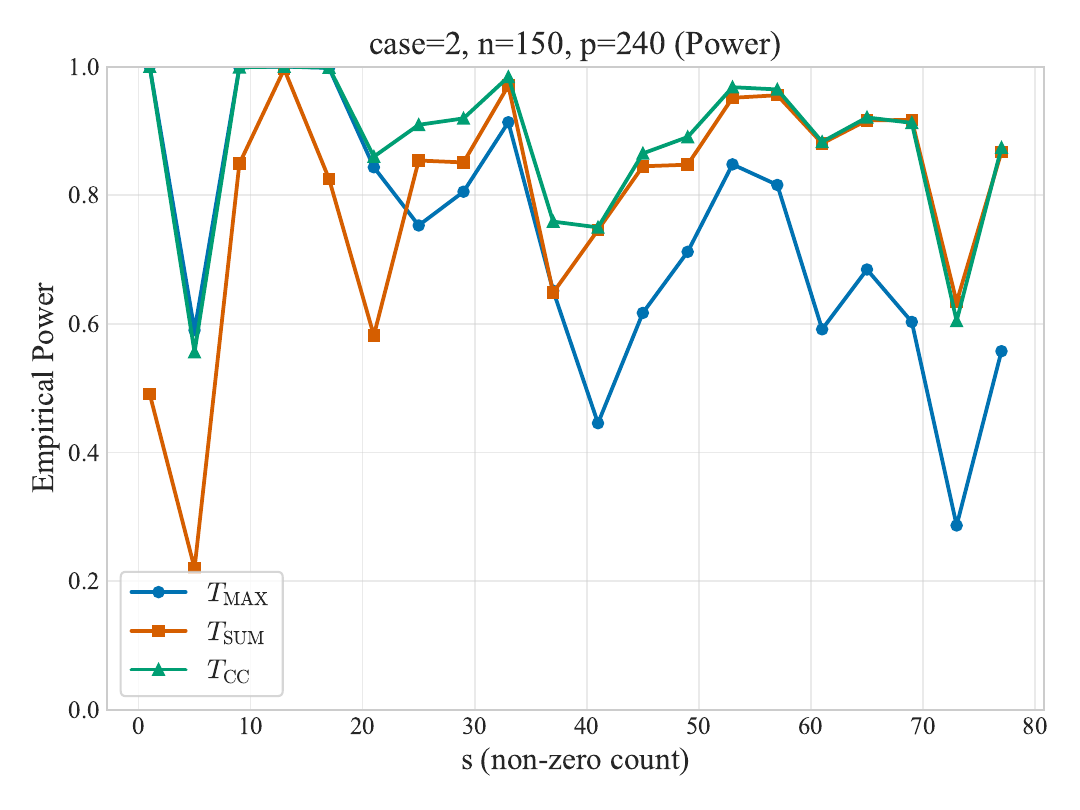}
\end{subfigure}

\vspace{0.3cm}
\begin{subfigure}{0.23\textwidth}
    \centering
    \includegraphics[width=\linewidth]{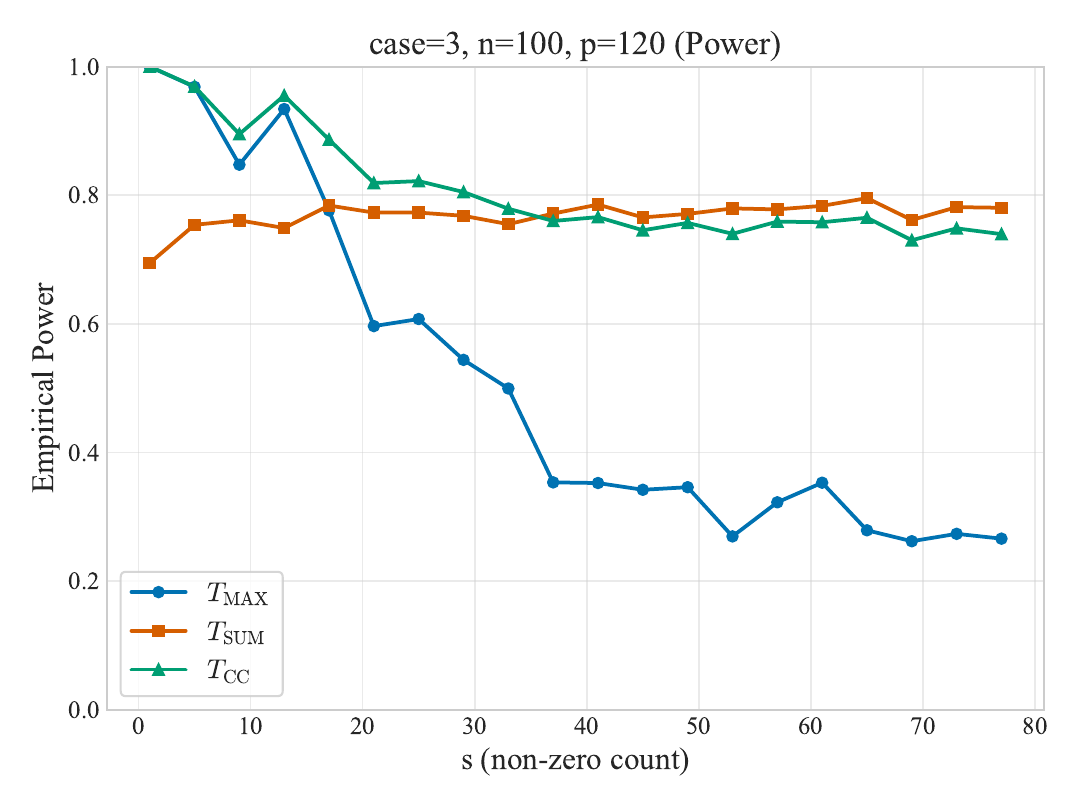}
\end{subfigure}
\hfill
\begin{subfigure}{0.23\textwidth}
    \centering
    \includegraphics[width=\linewidth]{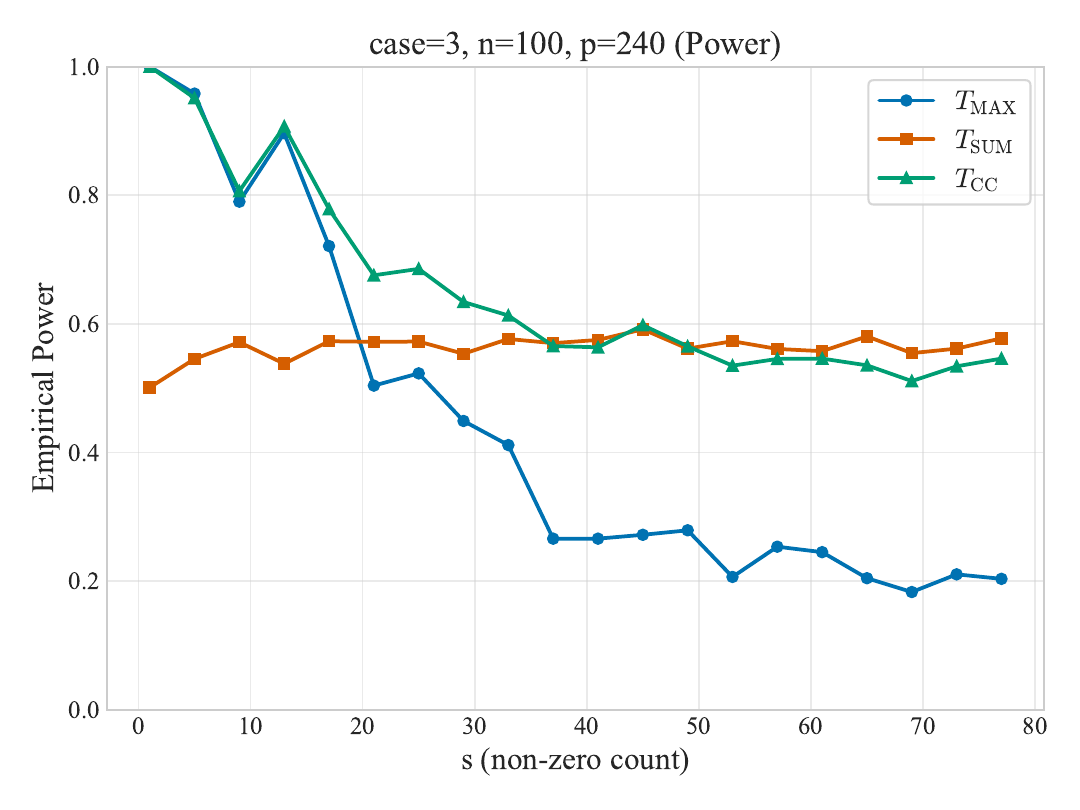}
\end{subfigure}
\hfill
\begin{subfigure}{0.23\textwidth}
    \centering
    \includegraphics[width=\linewidth]{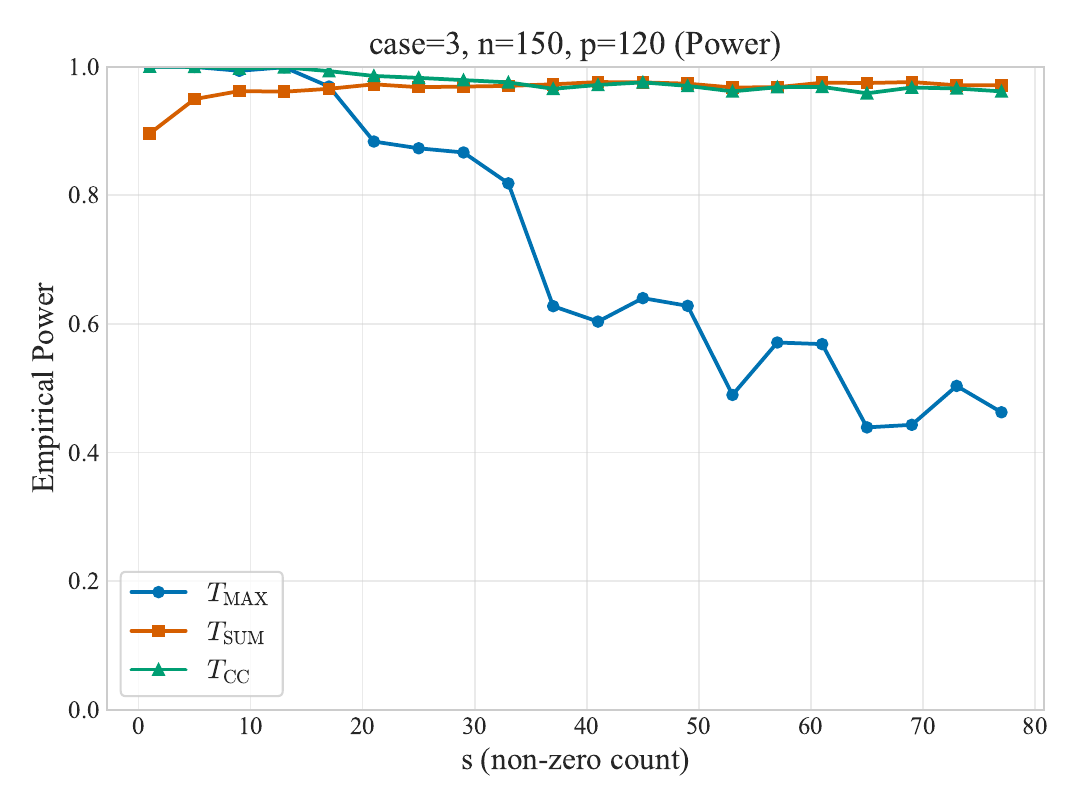}
\end{subfigure}
\hfill
\begin{subfigure}{0.23\textwidth}
    \centering
    \includegraphics[width=\linewidth]{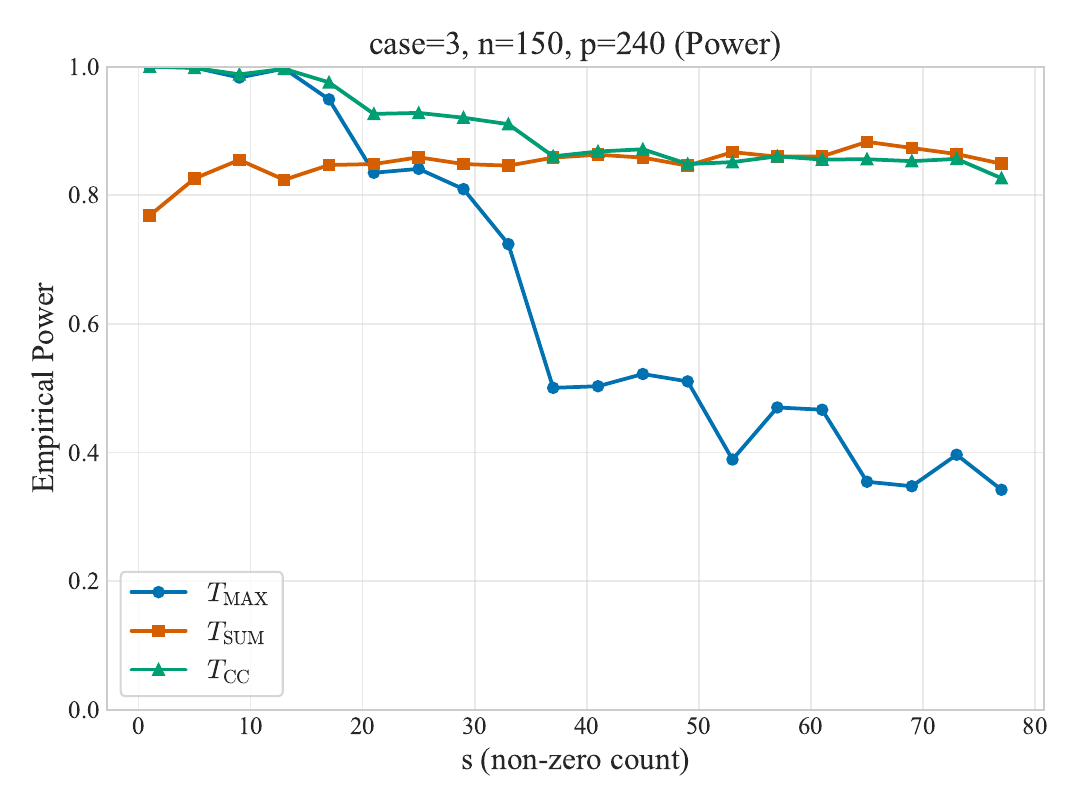}
\end{subfigure}

\caption{Empirical power as a function of $s$ for Cases~1--3 across varying $(n,p)$ 
settings under Logistic distribution ($\tau = 0.25$; 2000 replications).}
\label{fig:power25_logistic}
\end{figure}


\begin{figure}[htbp]
\centering
\begin{subfigure}{0.23\textwidth}
    \centering
    \includegraphics[width=\linewidth]{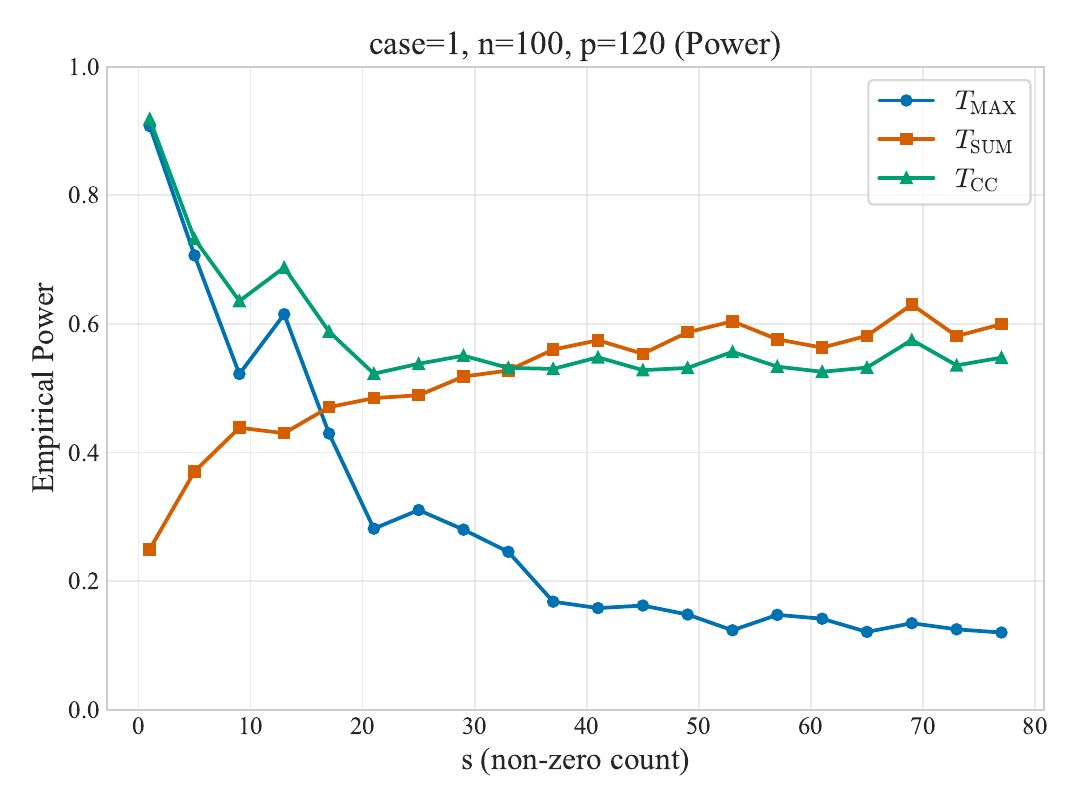}
\end{subfigure}
\hfill
\begin{subfigure}{0.23\textwidth}
    \centering
    \includegraphics[width=\linewidth]{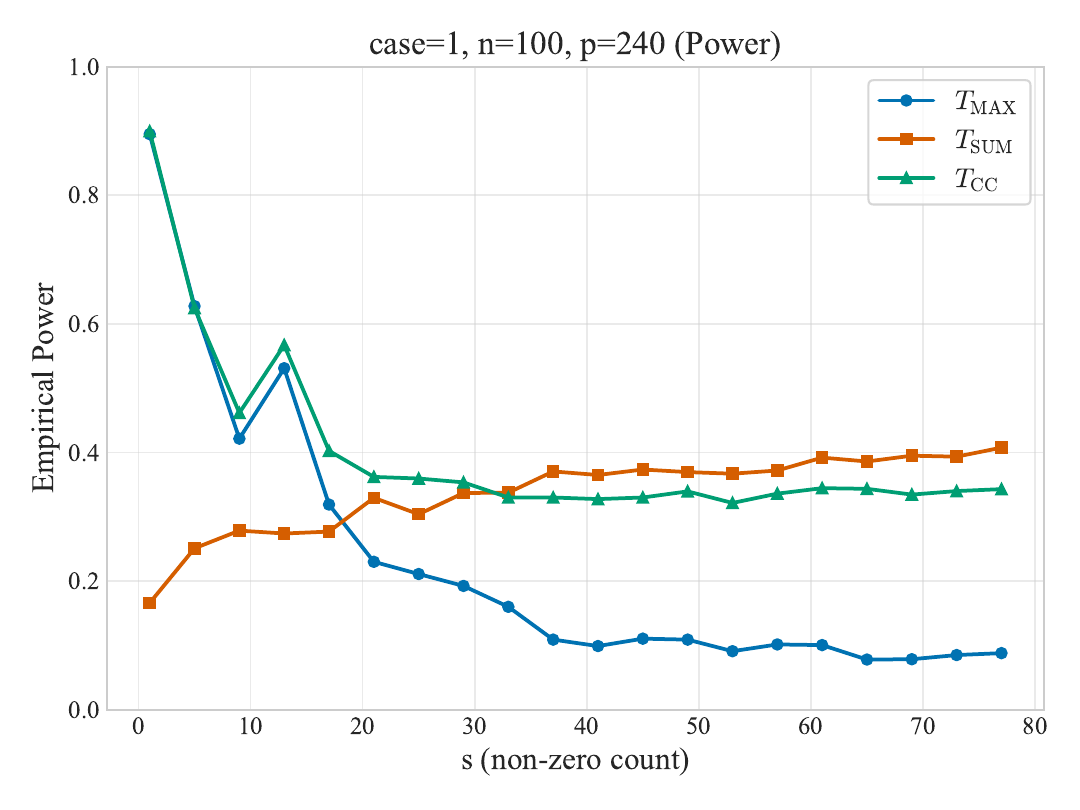}
\end{subfigure}
\hfill
\begin{subfigure}{0.23\textwidth}
    \centering
    \includegraphics[width=\linewidth]{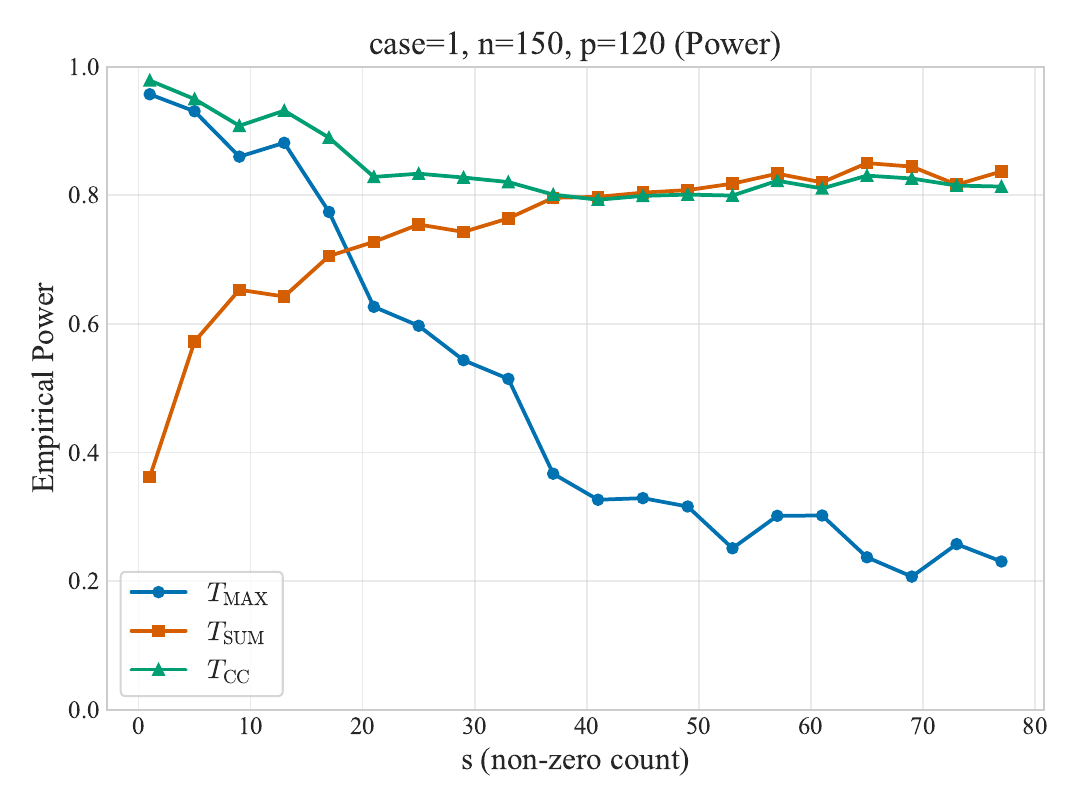}
\end{subfigure}
\hfill
\begin{subfigure}{0.23\textwidth}
    \centering
    \includegraphics[width=\linewidth]{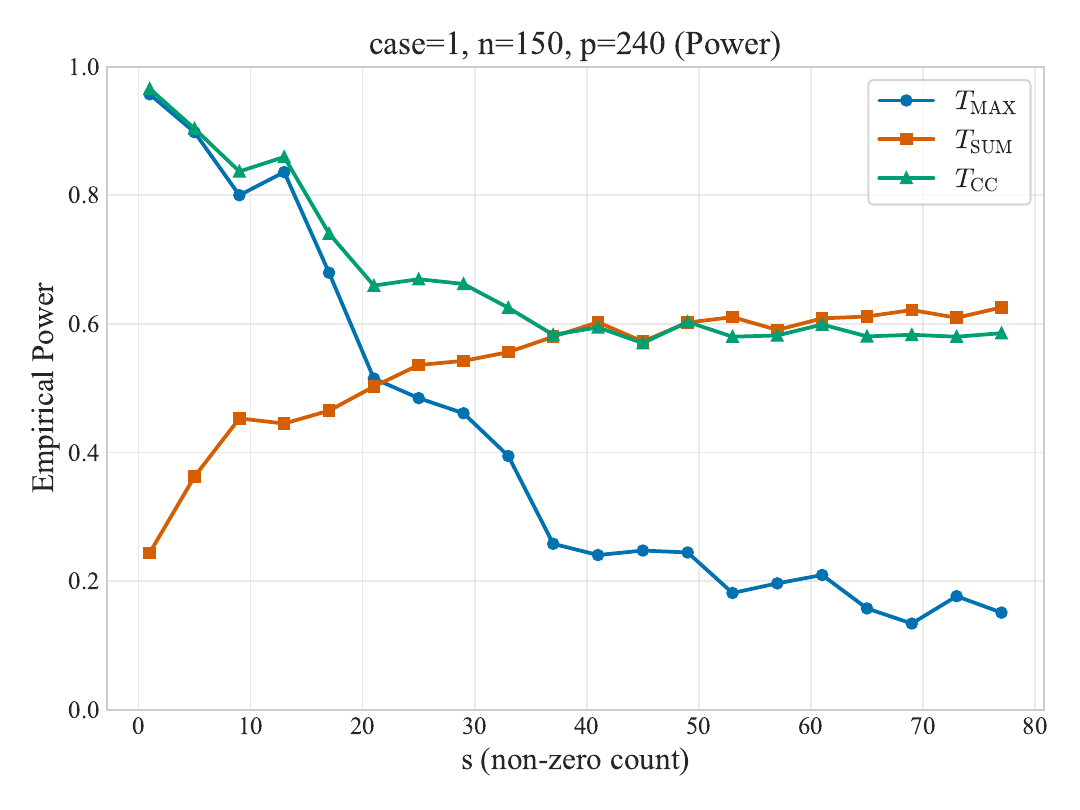}
\end{subfigure}

\vspace{0.3cm}
\begin{subfigure}{0.23\textwidth}
    \centering
    \includegraphics[width=\linewidth]{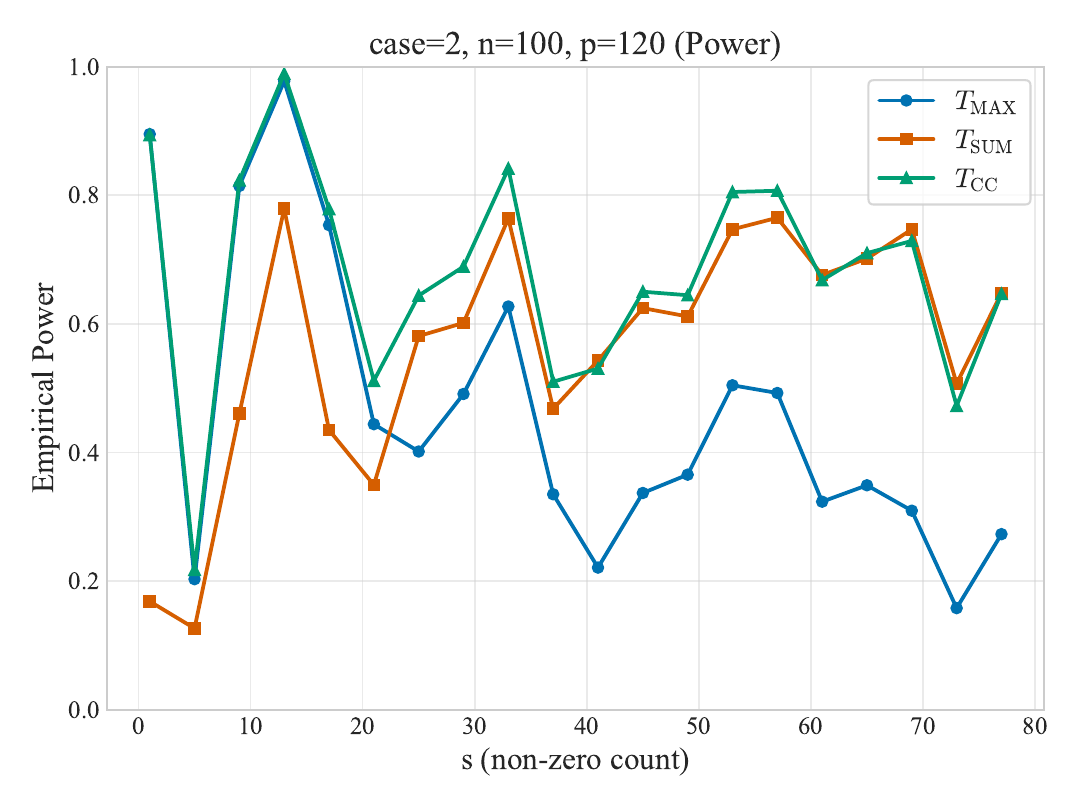}
\end{subfigure}
\hfill
\begin{subfigure}{0.23\textwidth}
    \centering
    \includegraphics[width=\linewidth]{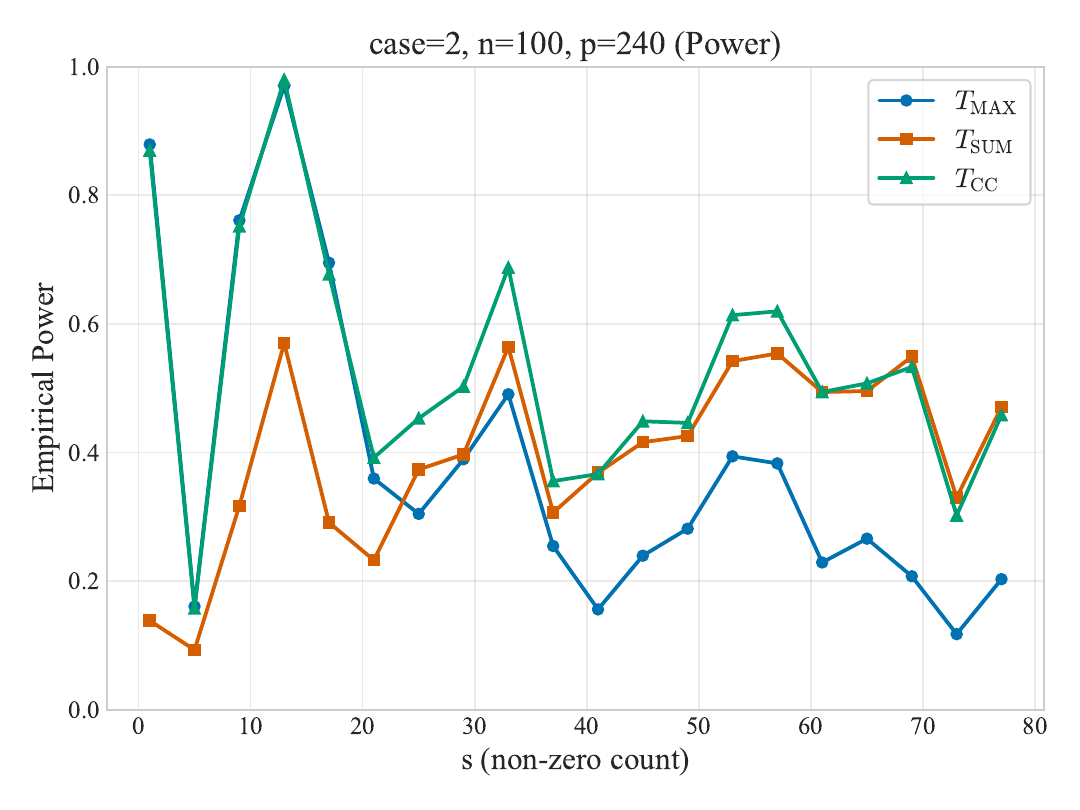}
\end{subfigure}
\hfill
\begin{subfigure}{0.23\textwidth}
    \centering
    \includegraphics[width=\linewidth]{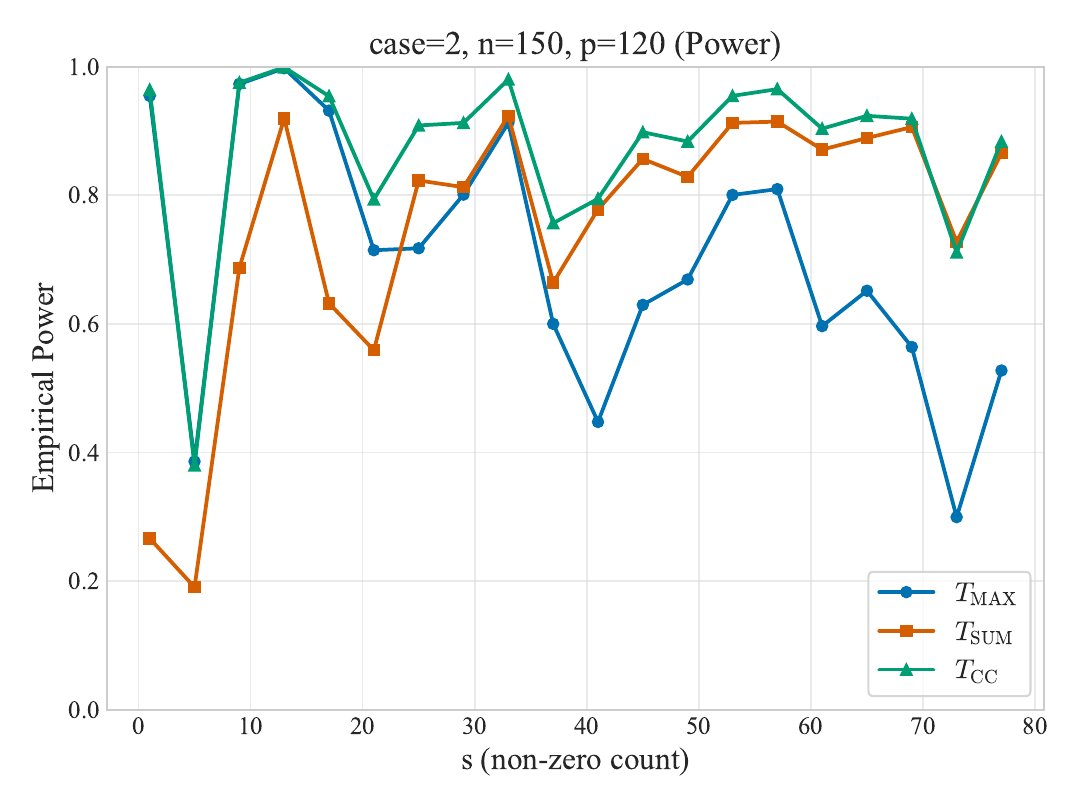}
\end{subfigure}
\hfill
\begin{subfigure}{0.23\textwidth}
    \centering
    \includegraphics[width=\linewidth]{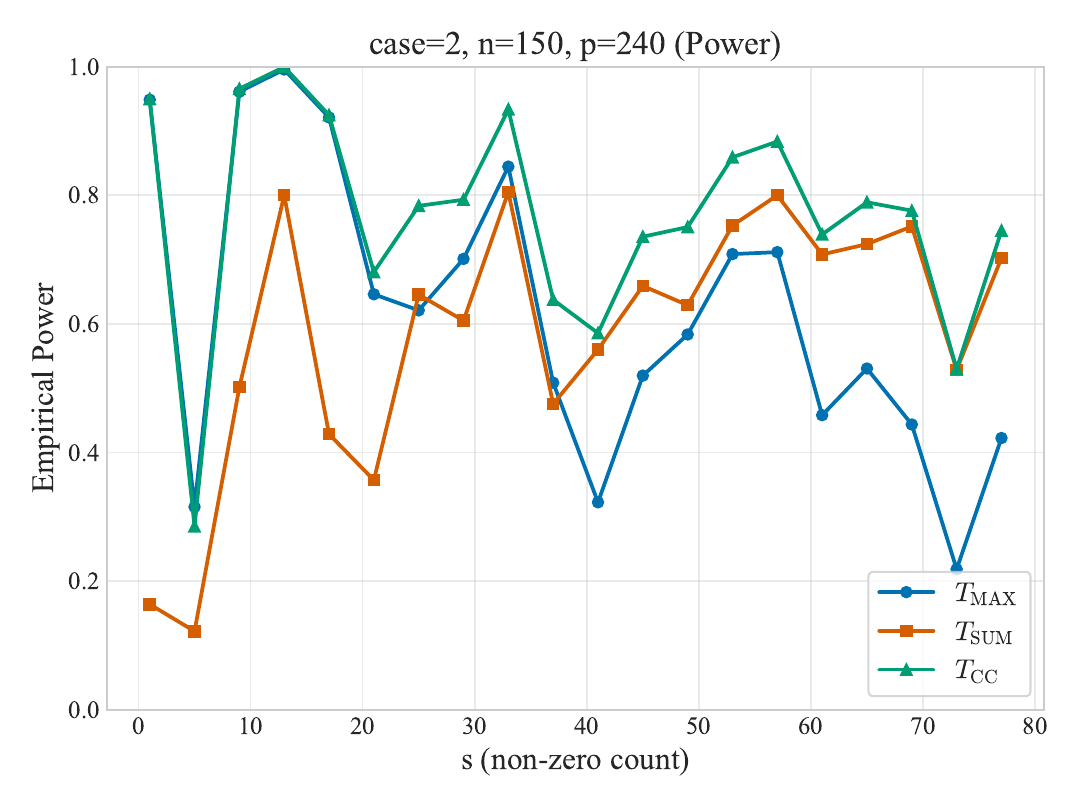}
\end{subfigure}

\vspace{0.3cm}
\begin{subfigure}{0.23\textwidth}
    \centering
    \includegraphics[width=\linewidth]{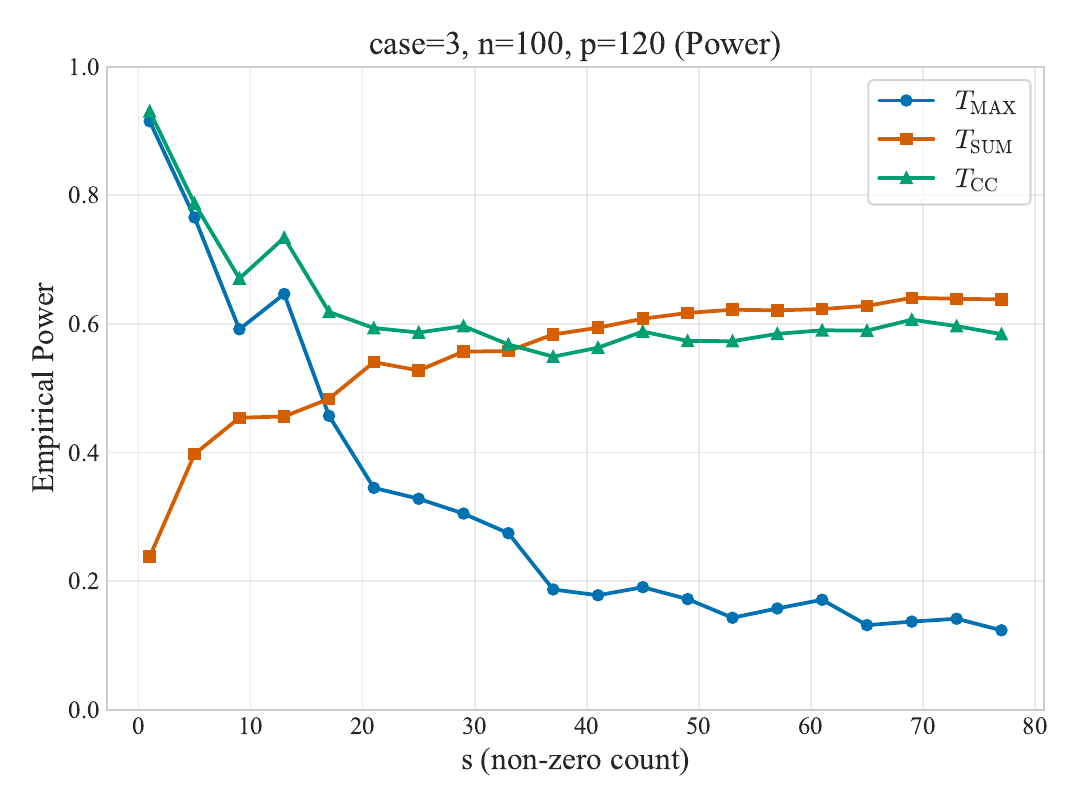}
\end{subfigure}
\hfill
\begin{subfigure}{0.23\textwidth}
    \centering
    \includegraphics[width=\linewidth]{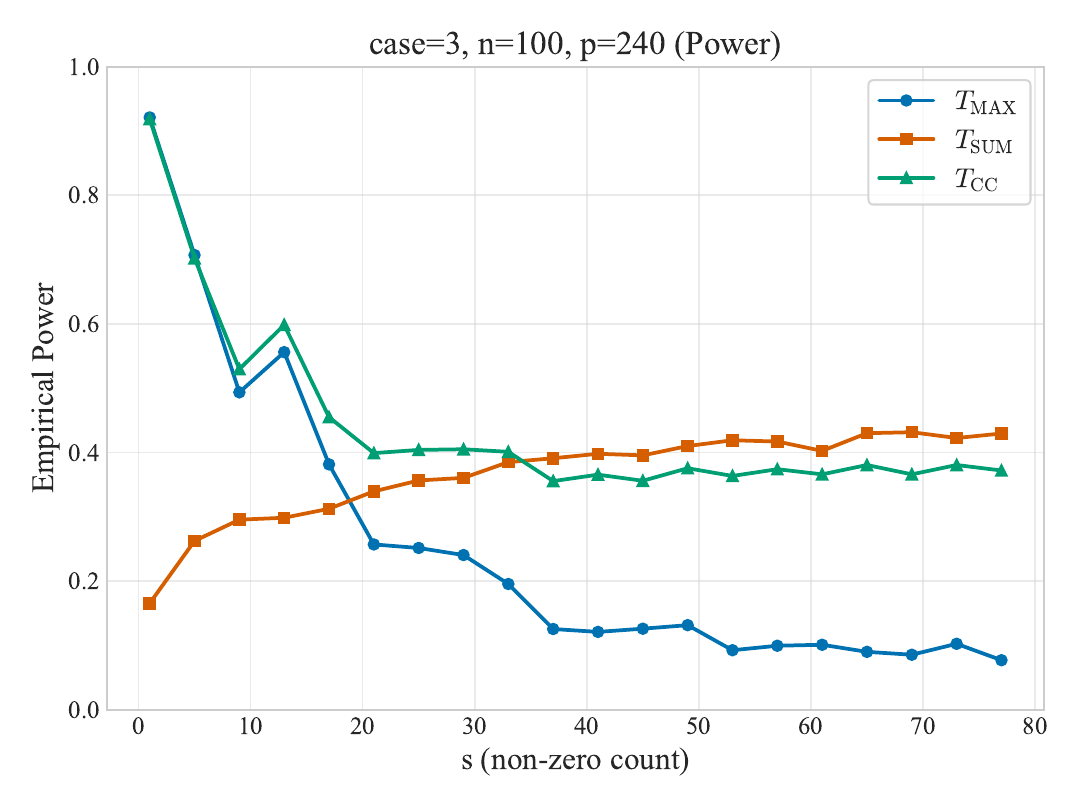}
\end{subfigure}
\hfill
\begin{subfigure}{0.23\textwidth}
    \centering
    \includegraphics[width=\linewidth]{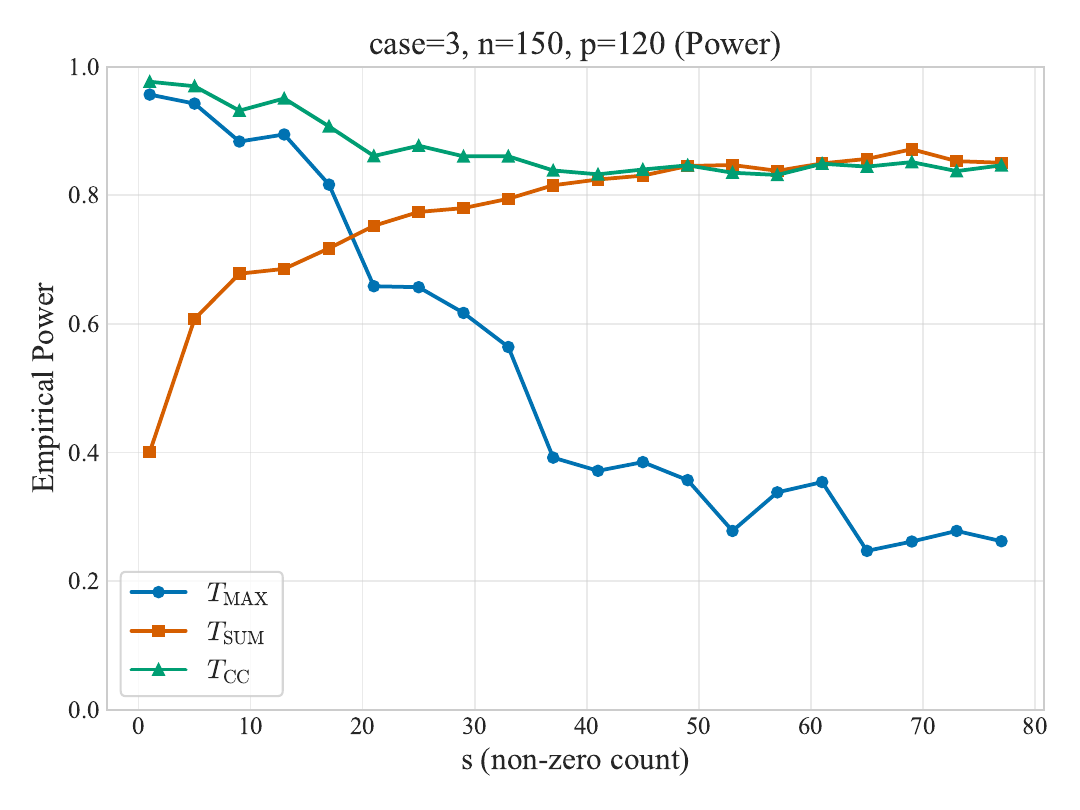}
\end{subfigure}
\hfill
\begin{subfigure}{0.23\textwidth}
    \centering
    \includegraphics[width=\linewidth]{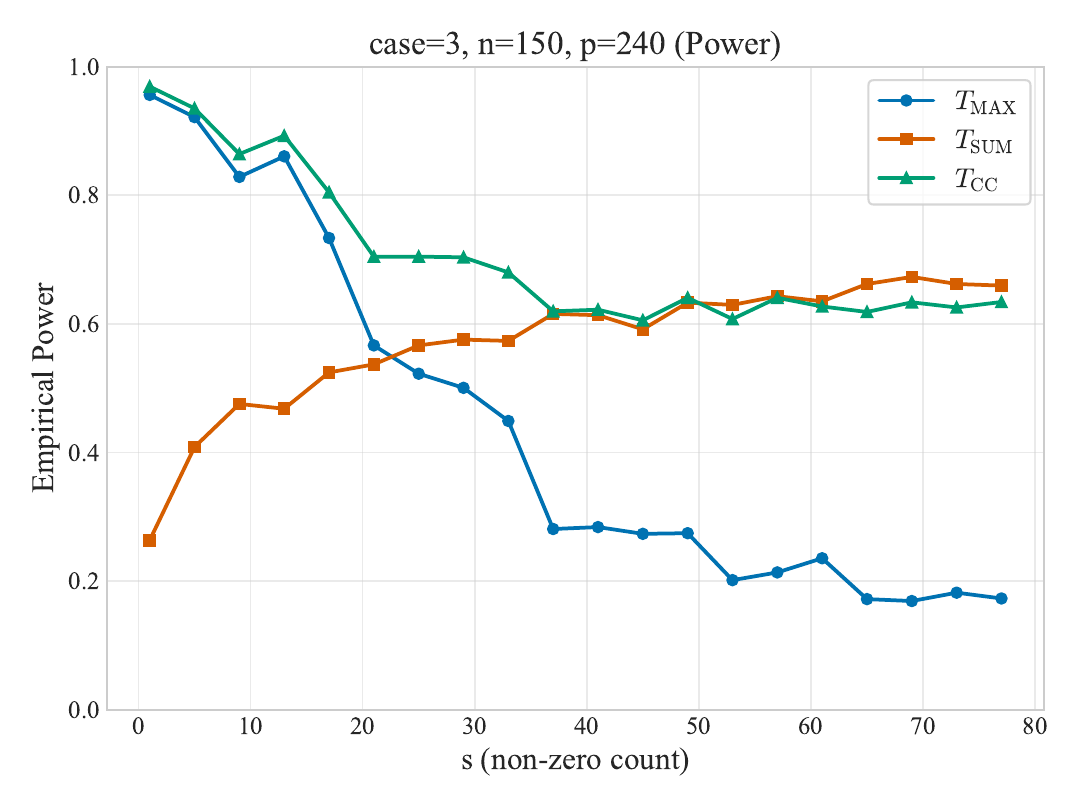}
\end{subfigure}

\caption{Empirical power as a function of $s$ for Cases~1--3 across varying $(n,p)$ 
settings under $t_2$ distribution ($\tau = 0.25$; 2000 replications).}
\label{fig:power25_t2}
\end{figure}

\begin{figure}[htbp]
\centering
\begin{subfigure}{0.23\textwidth}
    \centering
    \includegraphics[width=\linewidth]{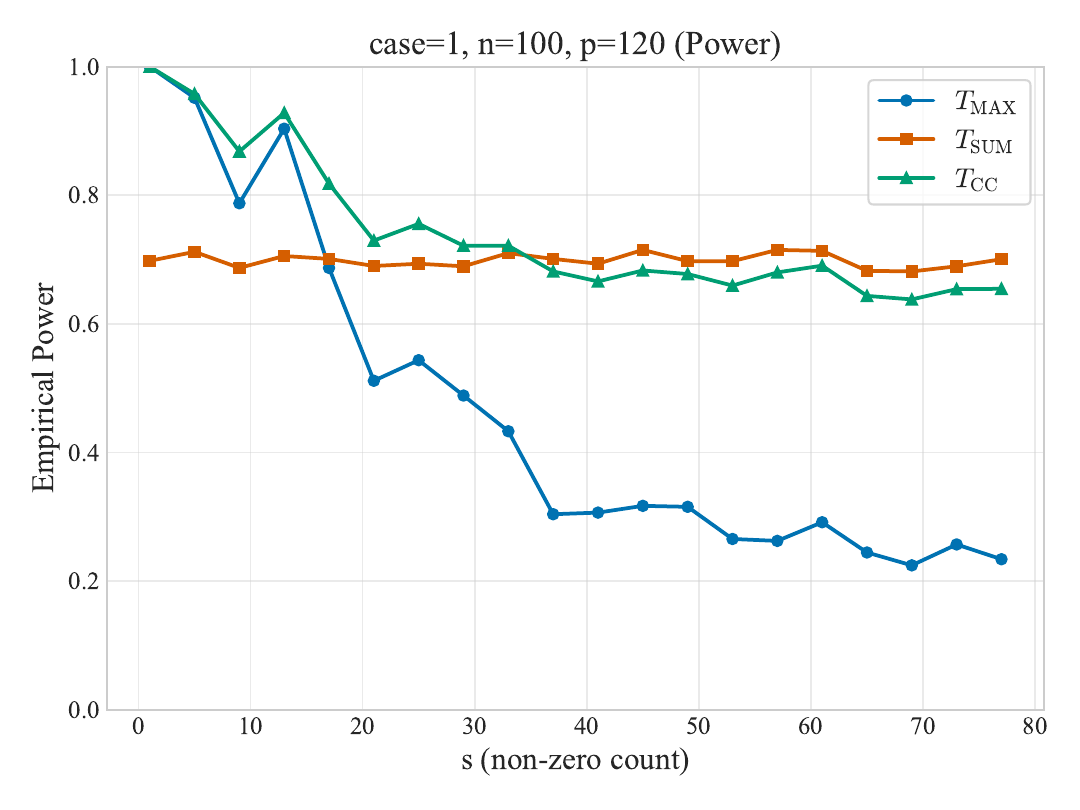}
\end{subfigure}
\hfill
\begin{subfigure}{0.23\textwidth}
    \centering
    \includegraphics[width=\linewidth]{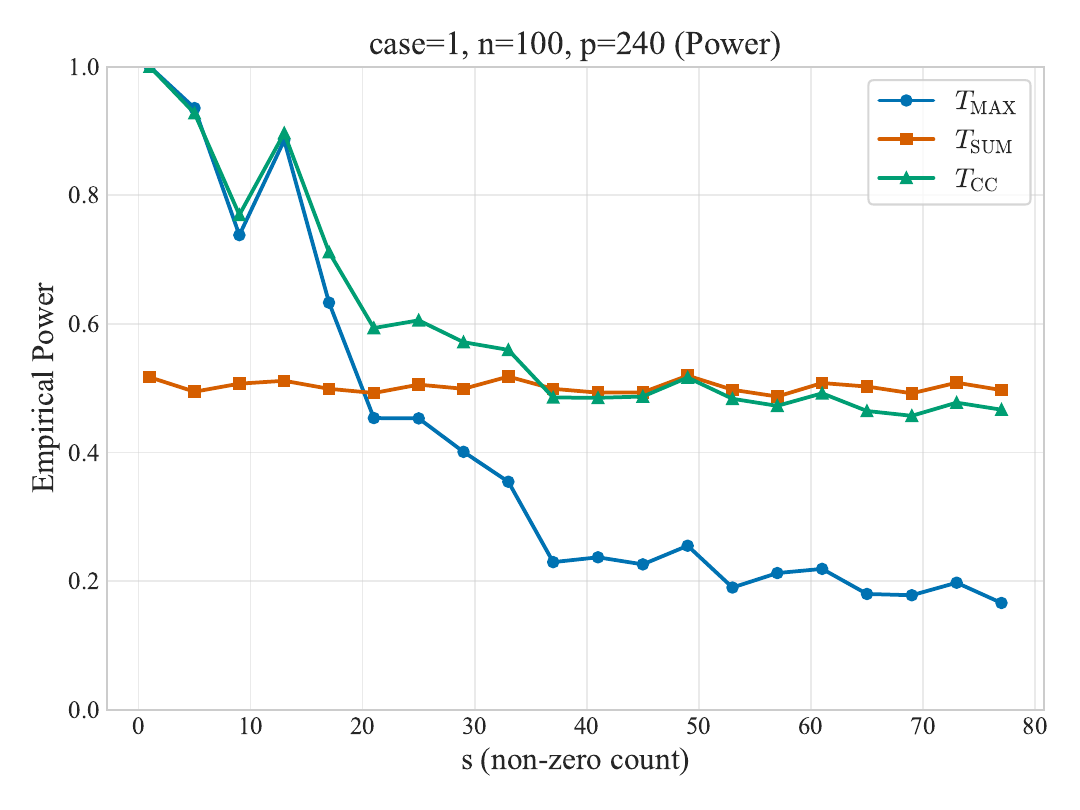}
\end{subfigure}
\hfill
\begin{subfigure}{0.23\textwidth}
    \centering
    \includegraphics[width=\linewidth]{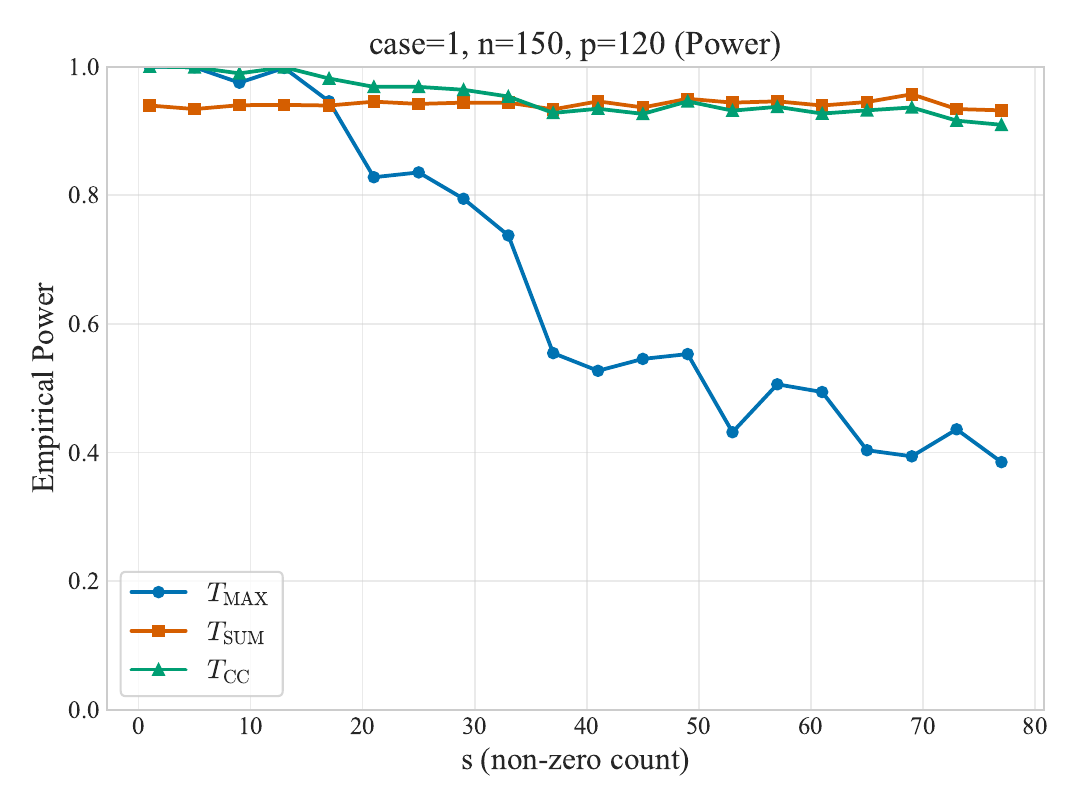}
\end{subfigure}
\hfill
\begin{subfigure}{0.23\textwidth}
    \centering
    \includegraphics[width=\linewidth]{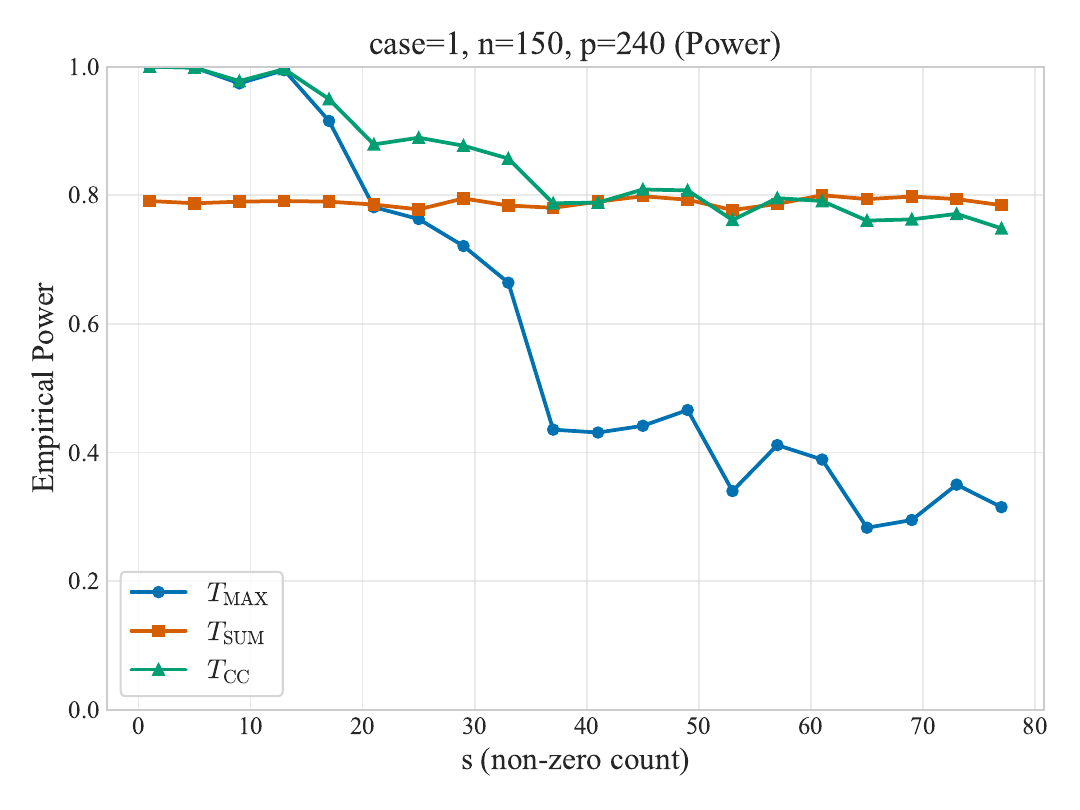}
\end{subfigure}

\vspace{0.3cm}
\begin{subfigure}{0.23\textwidth}
    \centering
    \includegraphics[width=\linewidth]{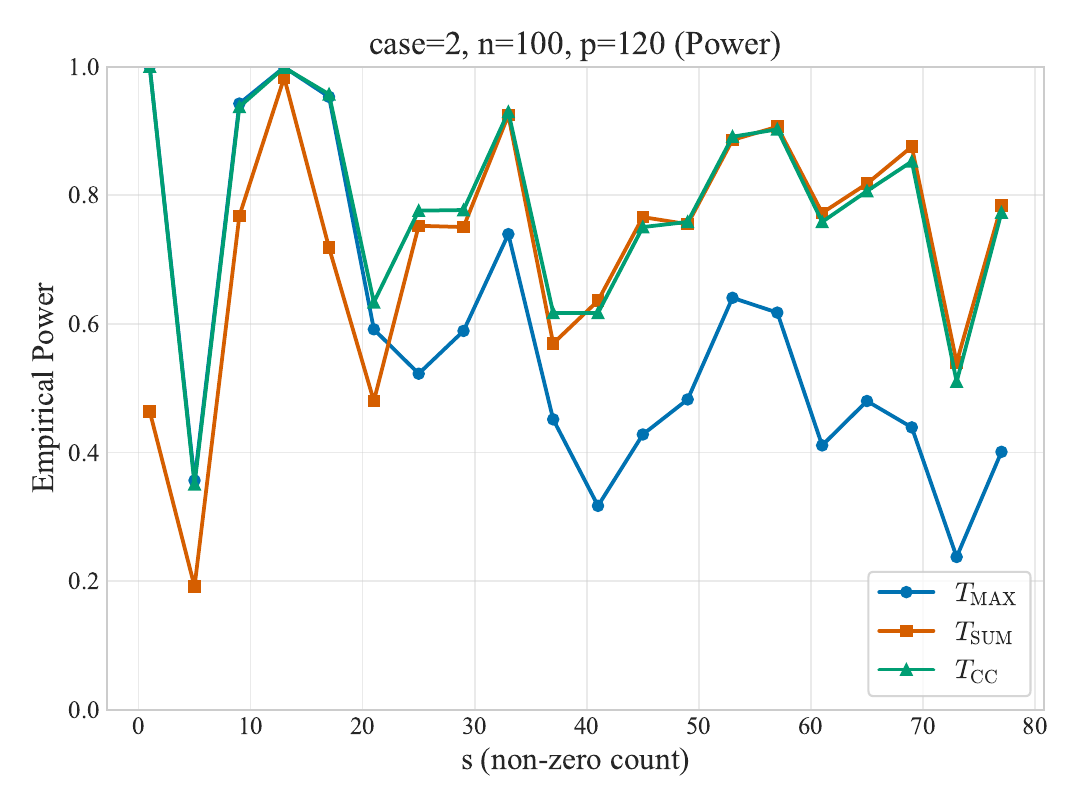}
\end{subfigure}
\hfill
\begin{subfigure}{0.23\textwidth}
    \centering
    \includegraphics[width=\linewidth]{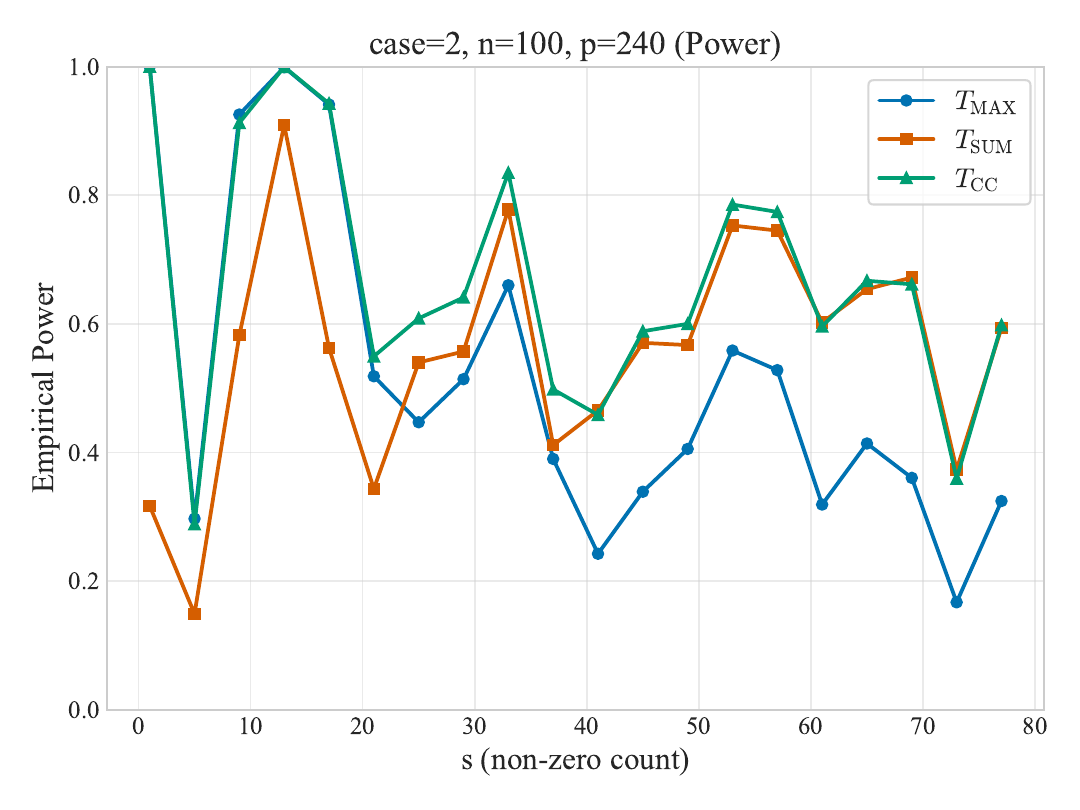}
\end{subfigure}
\hfill
\begin{subfigure}{0.23\textwidth}
    \centering
    \includegraphics[width=\linewidth]{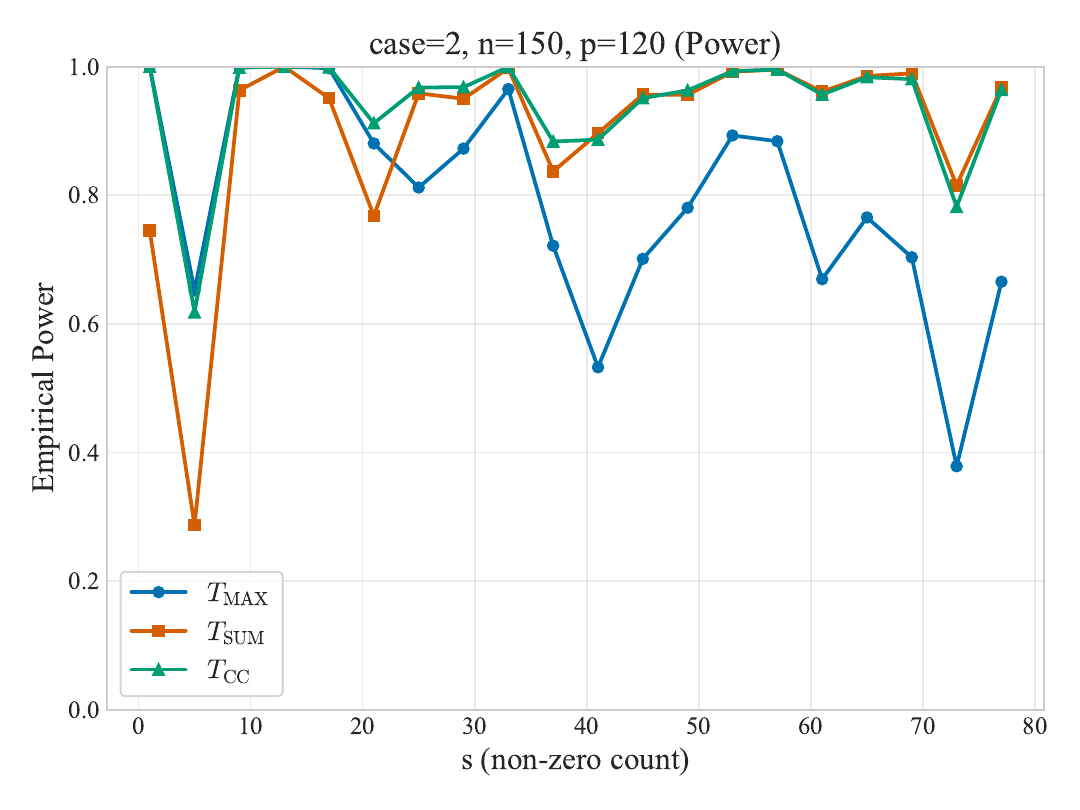}
\end{subfigure}
\hfill
\begin{subfigure}{0.23\textwidth}
    \centering
    \includegraphics[width=\linewidth]{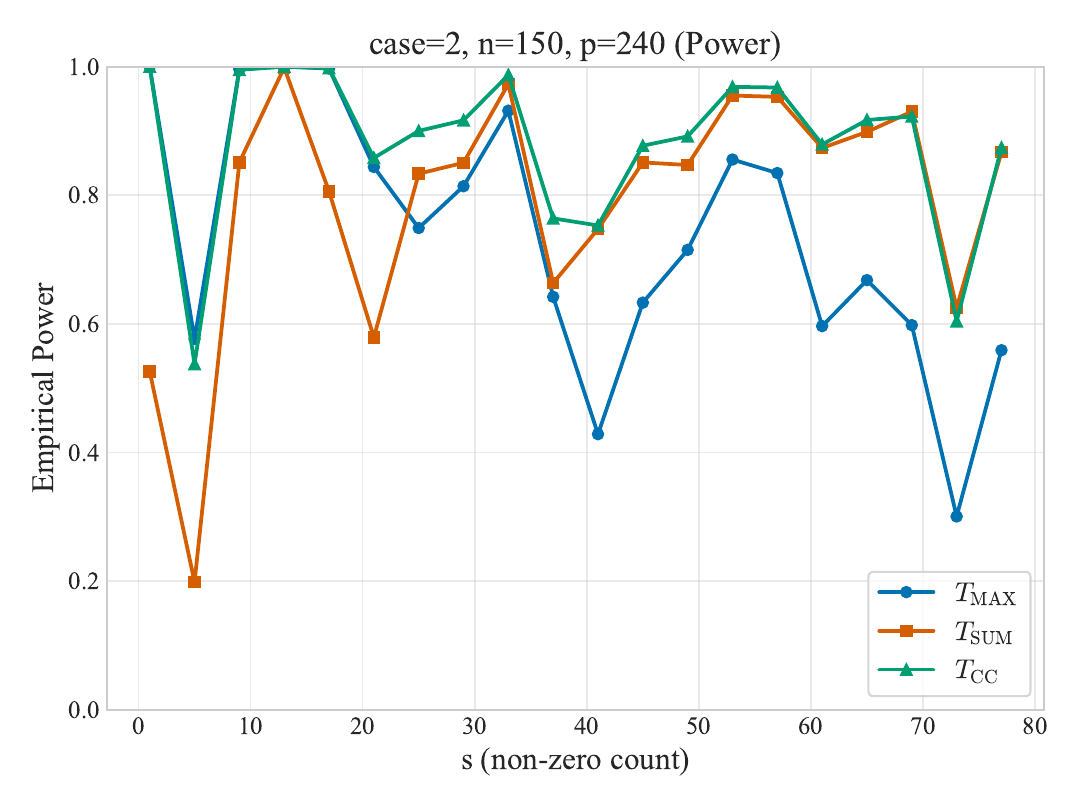}
\end{subfigure}

\vspace{0.3cm}
\begin{subfigure}{0.23\textwidth}
    \centering
    \includegraphics[width=\linewidth]{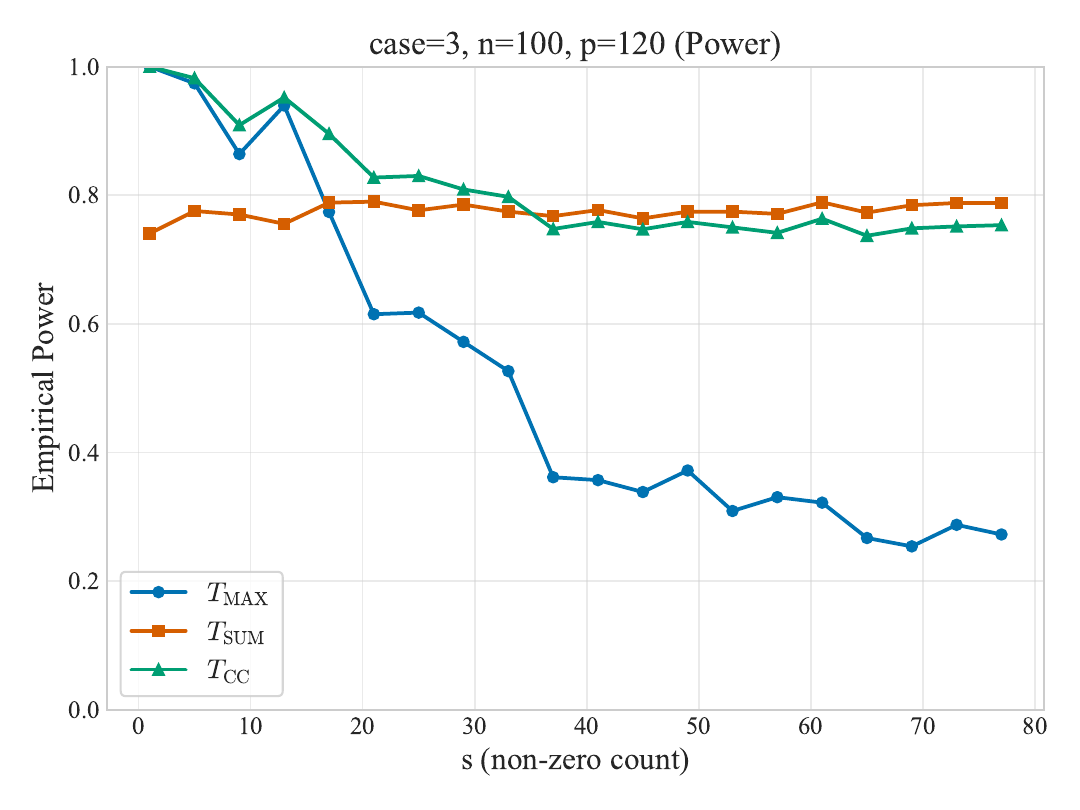}
\end{subfigure}
\hfill
\begin{subfigure}{0.23\textwidth}
    \centering
    \includegraphics[width=\linewidth]{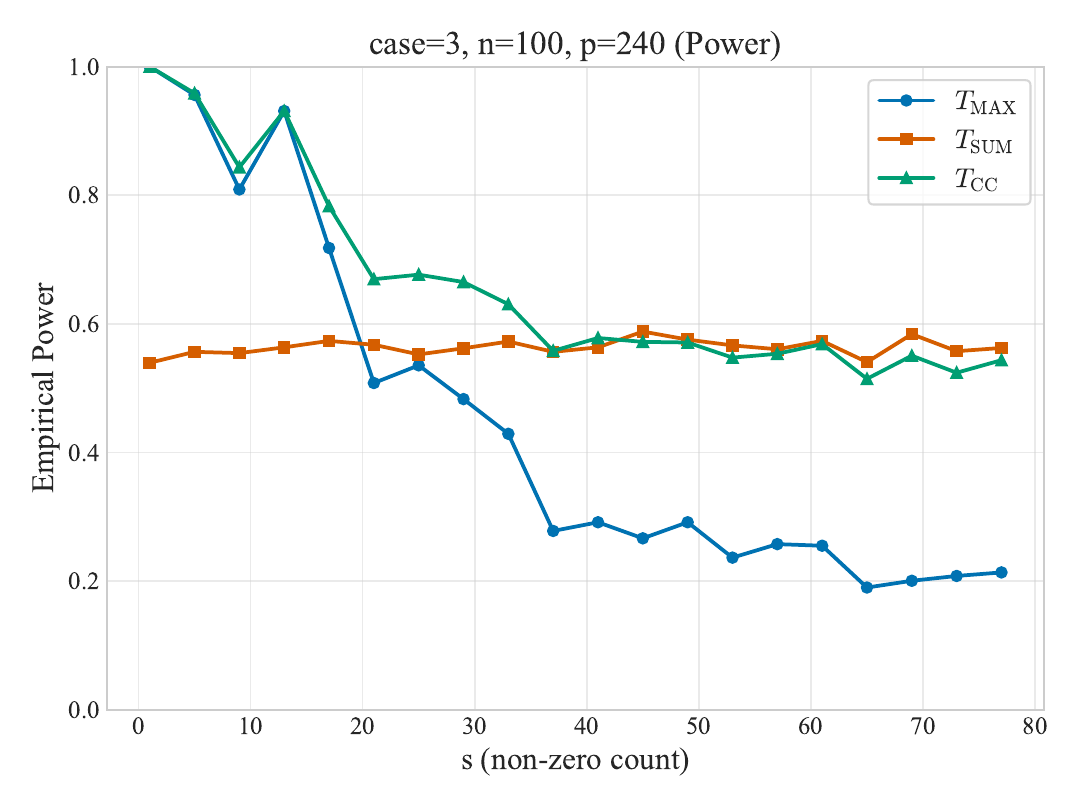}
\end{subfigure}
\hfill
\begin{subfigure}{0.23\textwidth}
    \centering
    \includegraphics[width=\linewidth]{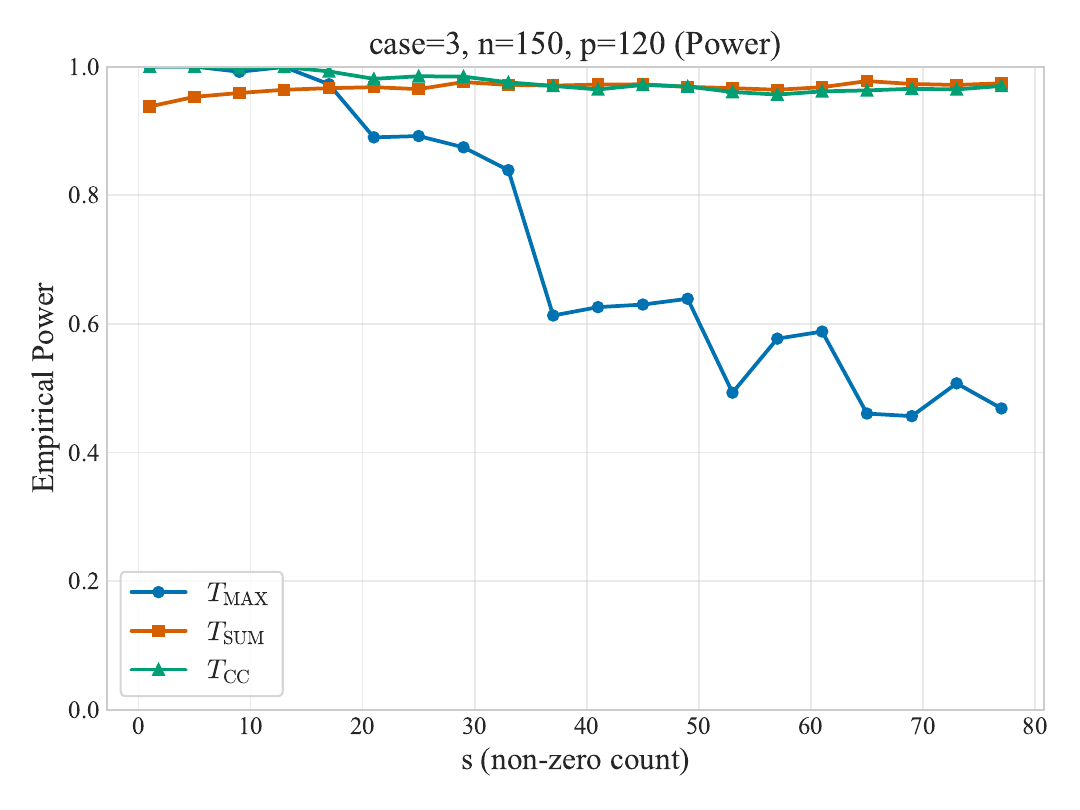}
\end{subfigure}
\hfill
\begin{subfigure}{0.23\textwidth}
    \centering
    \includegraphics[width=\linewidth]{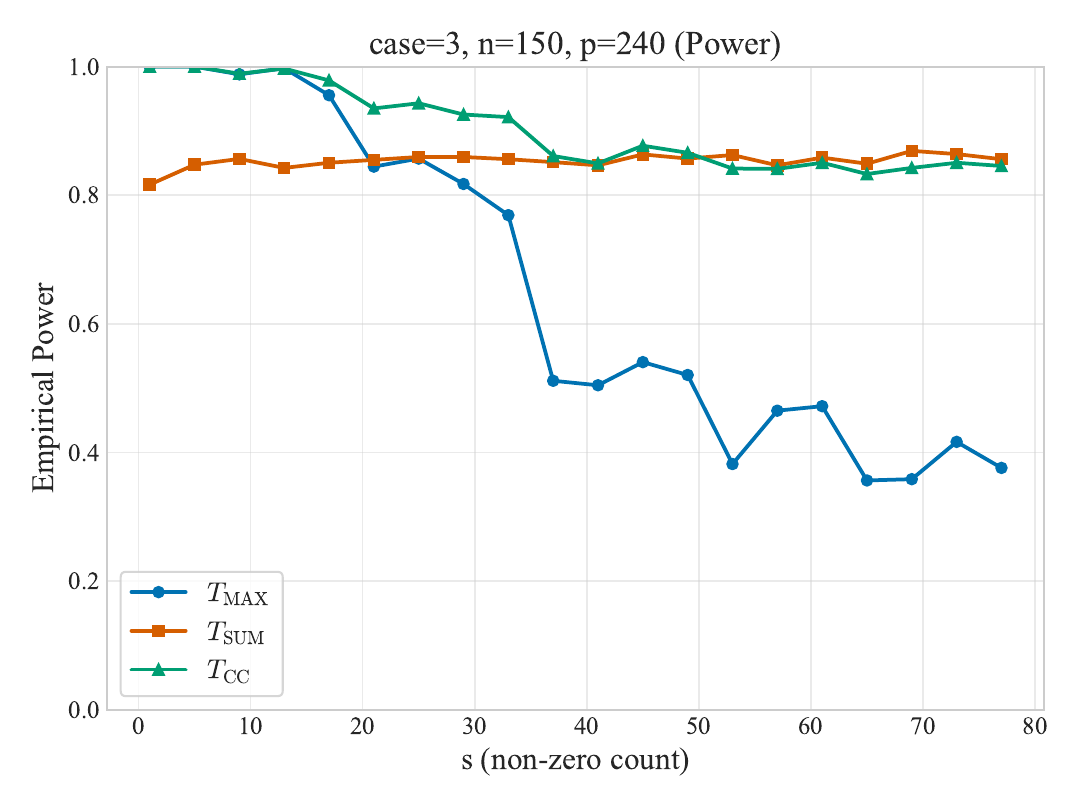}
\end{subfigure}

\caption{Empirical power as a function of $s$ for Cases~1--3 across varying $(n,p)$ 
settings under Normal distribution ($\tau = 0.75$; 2000 replications)..}
\label{fig:power75}
\end{figure}

\begin{figure}[htbp]
\centering
\begin{subfigure}{0.23\textwidth}
    \centering
    \includegraphics[width=\linewidth]{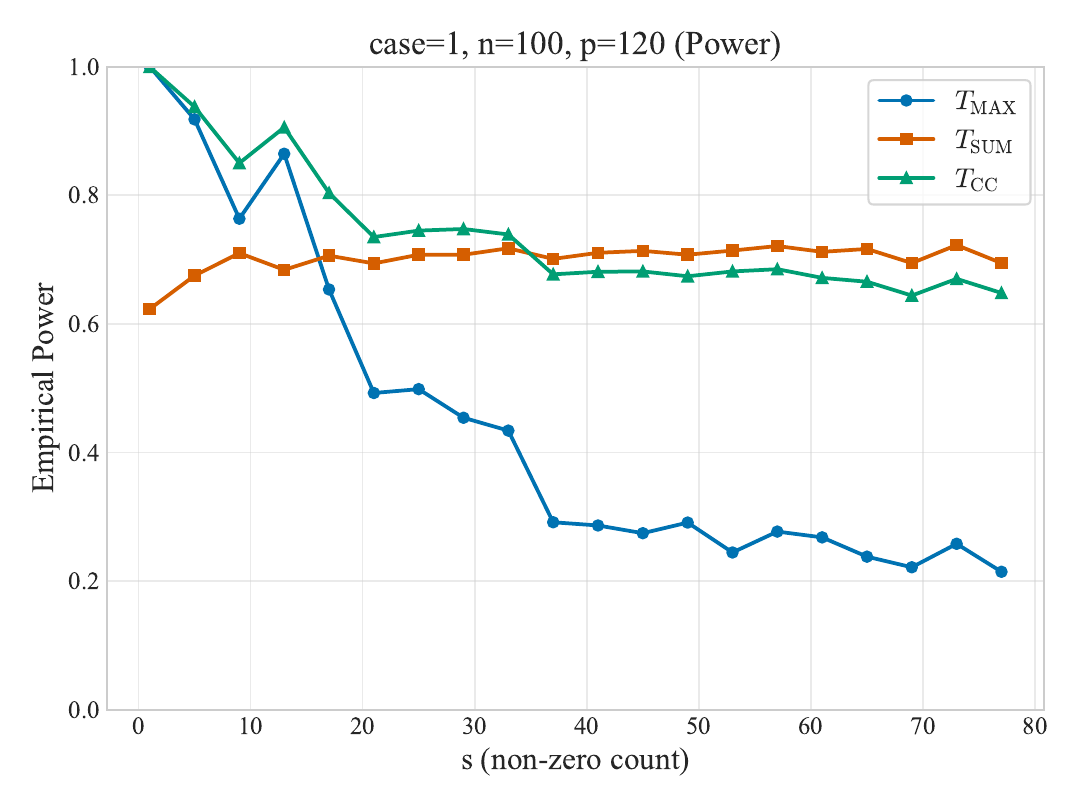}
\end{subfigure}
\hfill
\begin{subfigure}{0.23\textwidth}
    \centering
    \includegraphics[width=\linewidth]{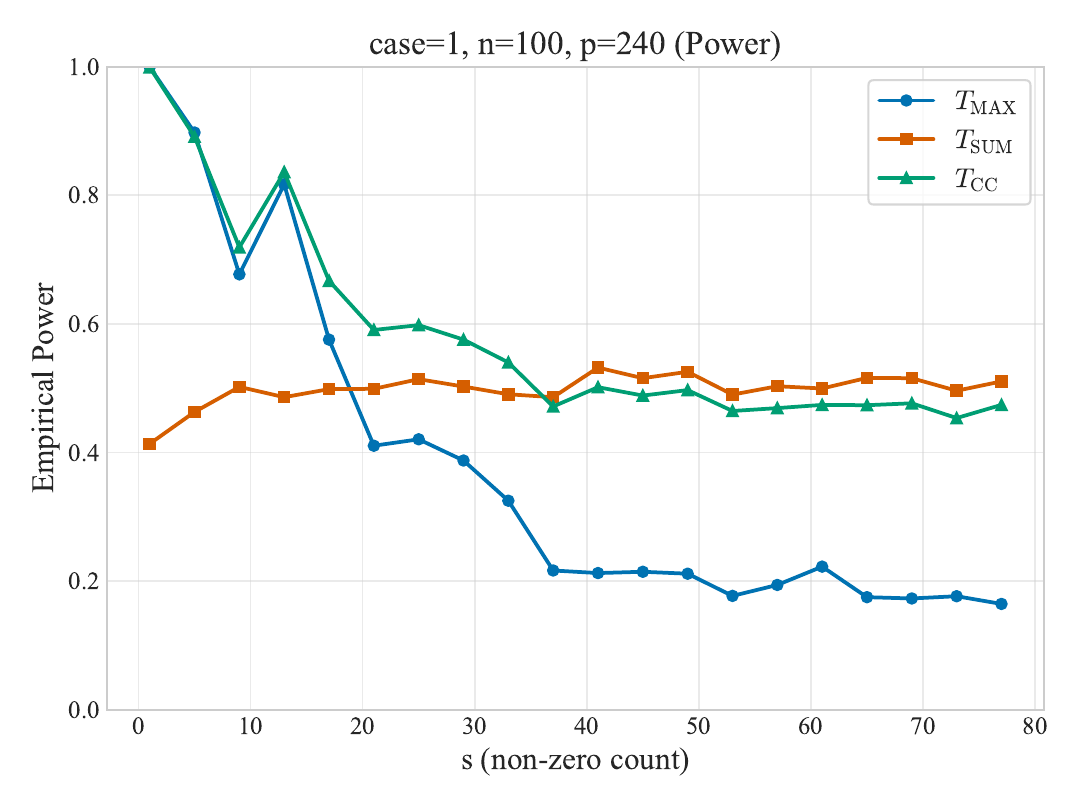}
\end{subfigure}
\hfill
\begin{subfigure}{0.23\textwidth}
    \centering
    \includegraphics[width=\linewidth]{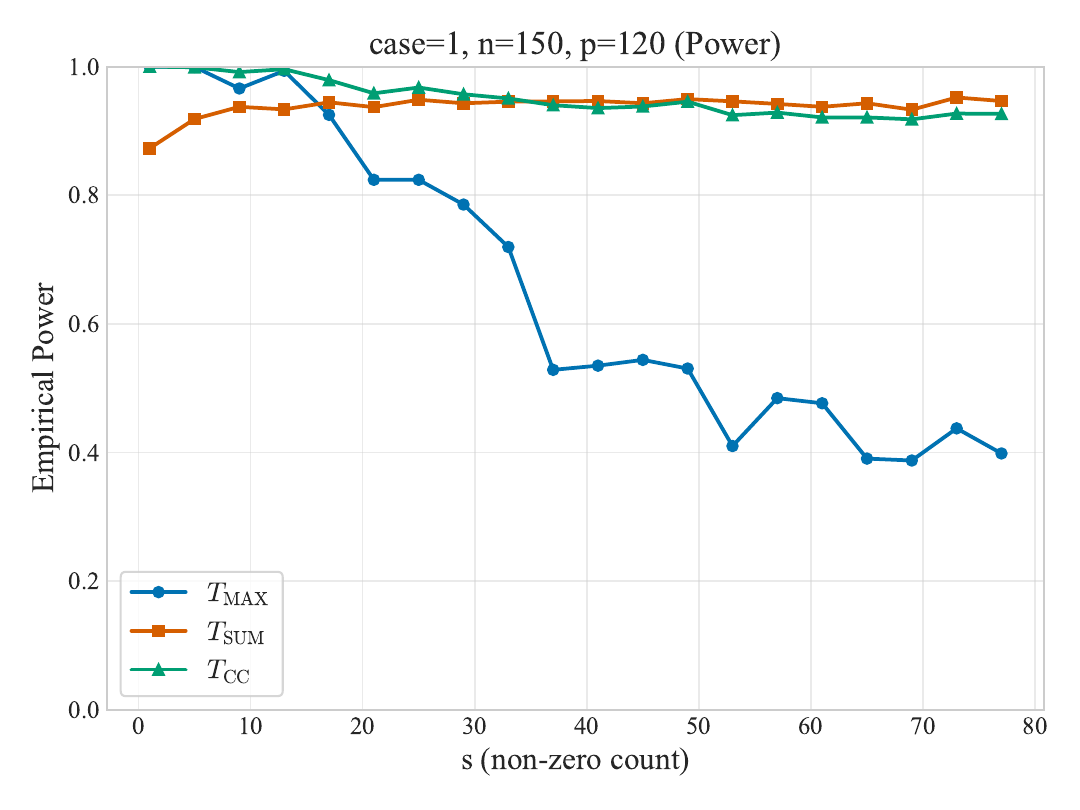}
\end{subfigure}
\hfill
\begin{subfigure}{0.23\textwidth}
    \centering
    \includegraphics[width=\linewidth]{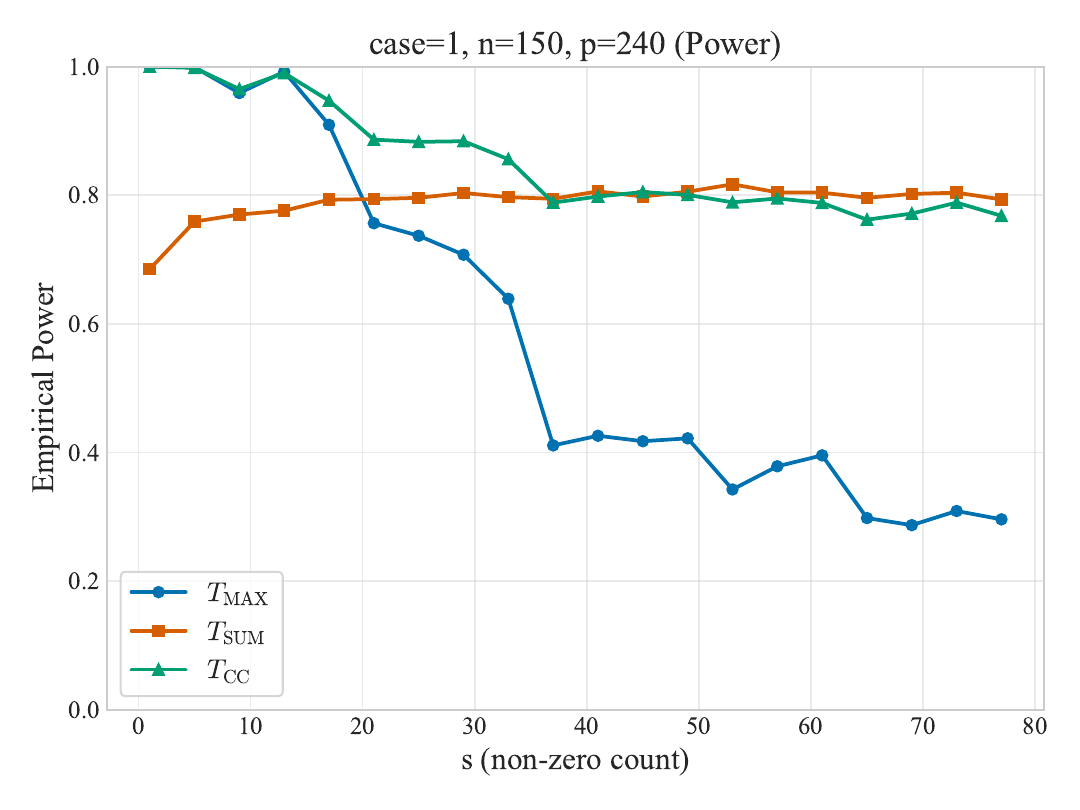}
\end{subfigure}

\vspace{0.3cm}
\begin{subfigure}{0.23\textwidth}
    \centering
    \includegraphics[width=\linewidth]{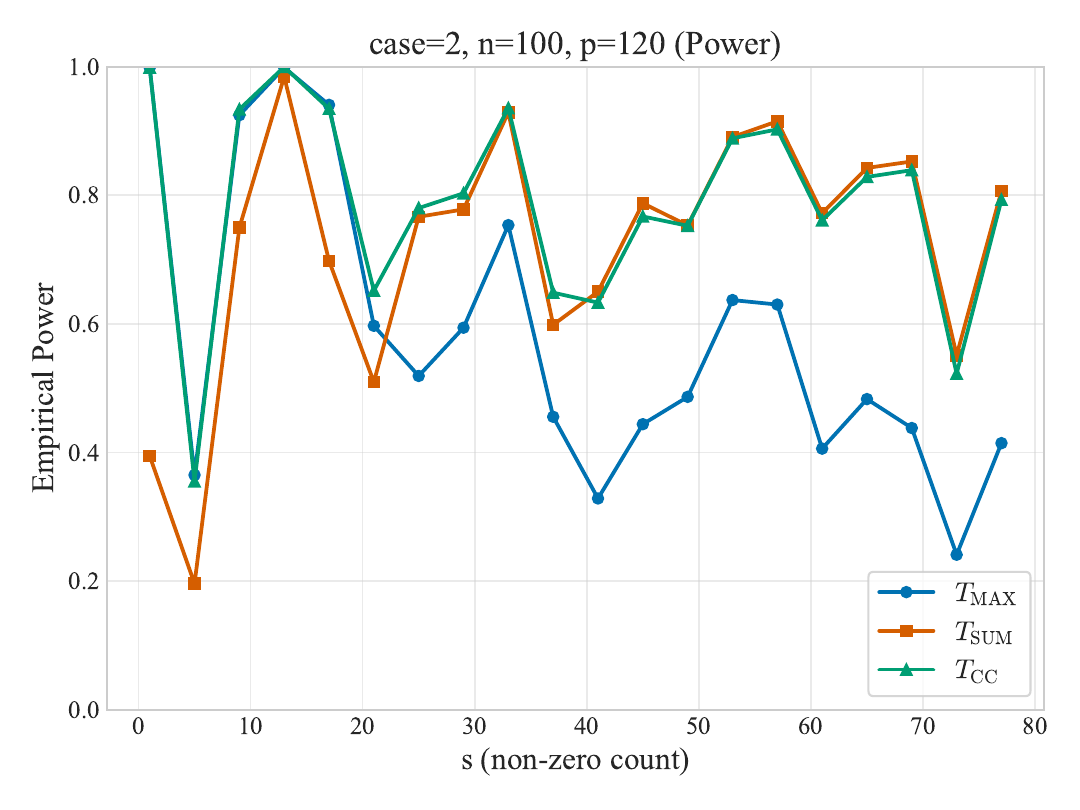}
\end{subfigure}
\hfill
\begin{subfigure}{0.23\textwidth}
    \centering
    \includegraphics[width=\linewidth]{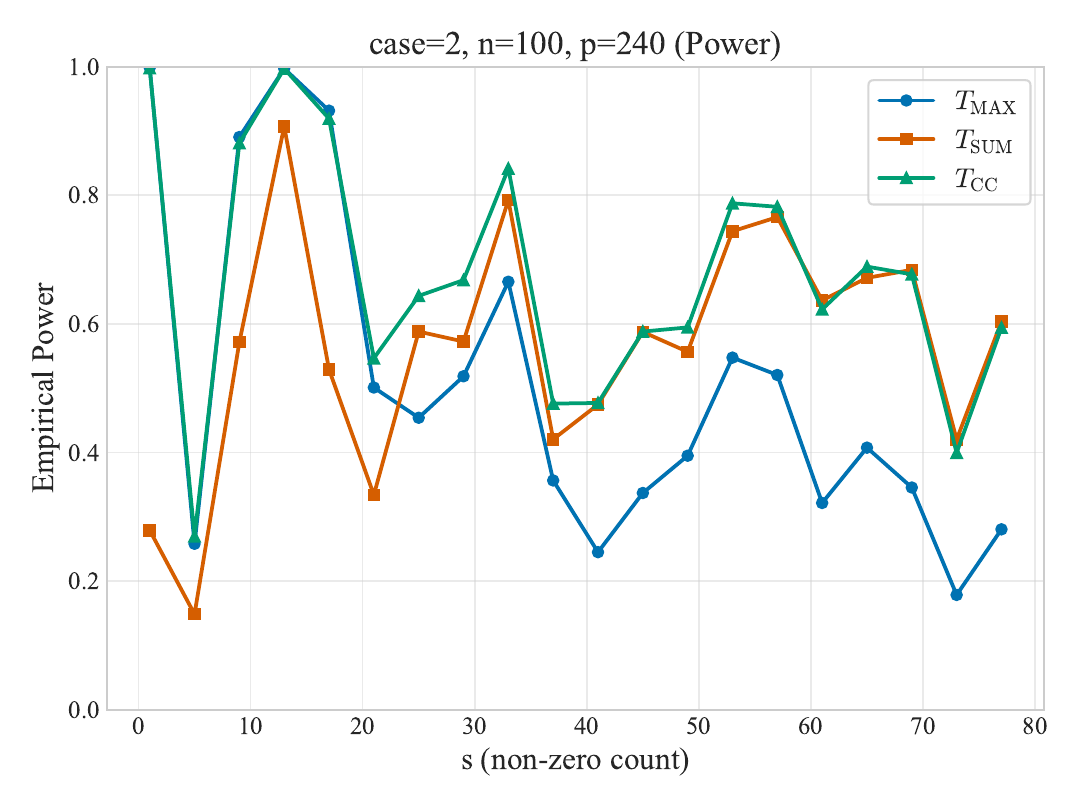}
\end{subfigure}
\hfill
\begin{subfigure}{0.23\textwidth}
    \centering
    \includegraphics[width=\linewidth]{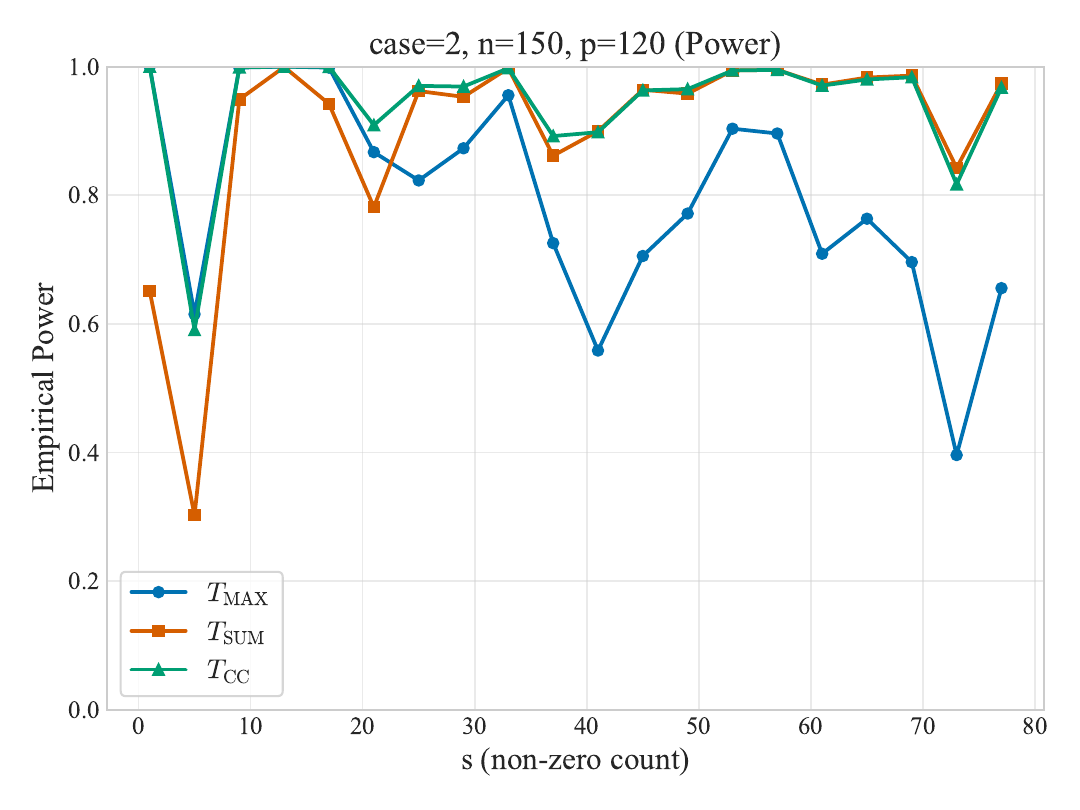}
\end{subfigure}
\hfill
\begin{subfigure}{0.23\textwidth}
    \centering
    \includegraphics[width=\linewidth]{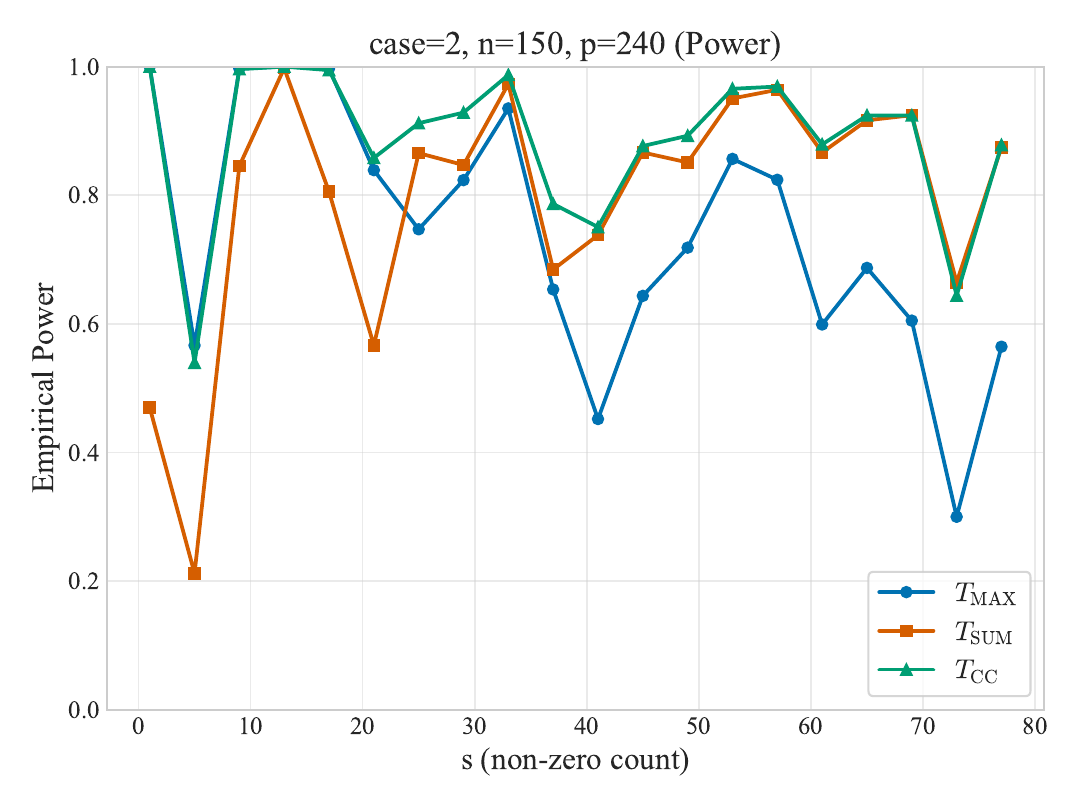}
\end{subfigure}

\vspace{0.3cm}
\begin{subfigure}{0.23\textwidth}
    \centering
    \includegraphics[width=\linewidth]{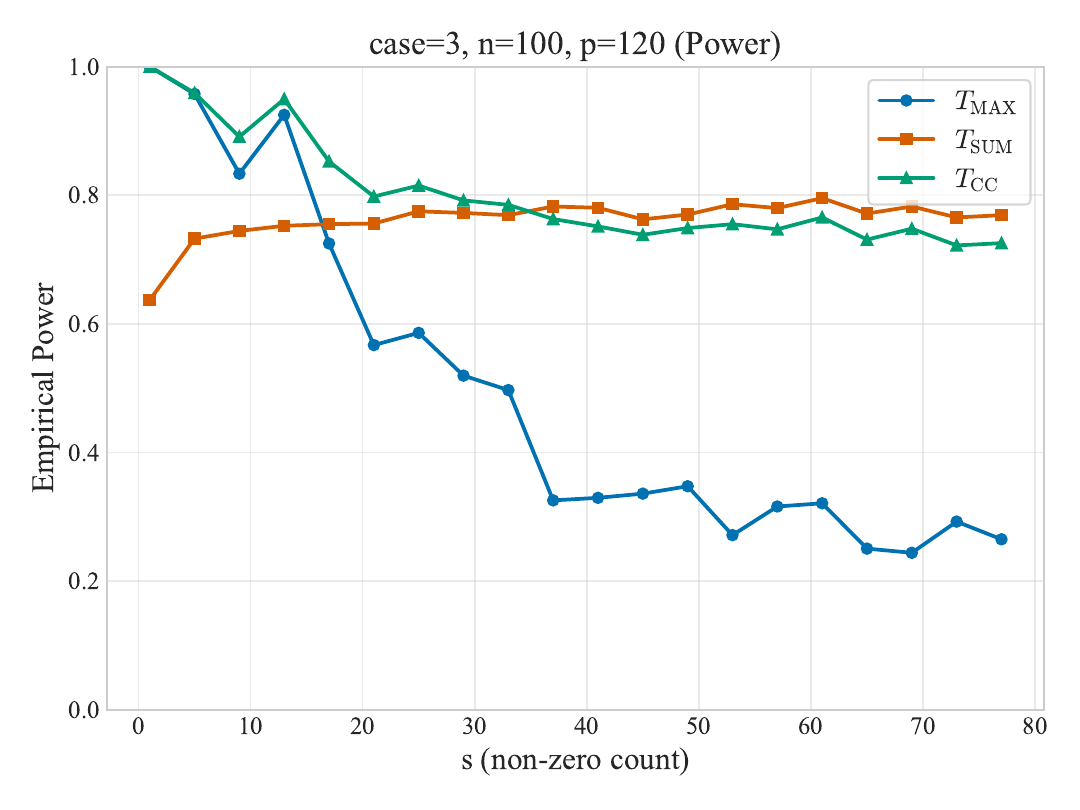}
\end{subfigure}
\hfill
\begin{subfigure}{0.23\textwidth}
    \centering
    \includegraphics[width=\linewidth]{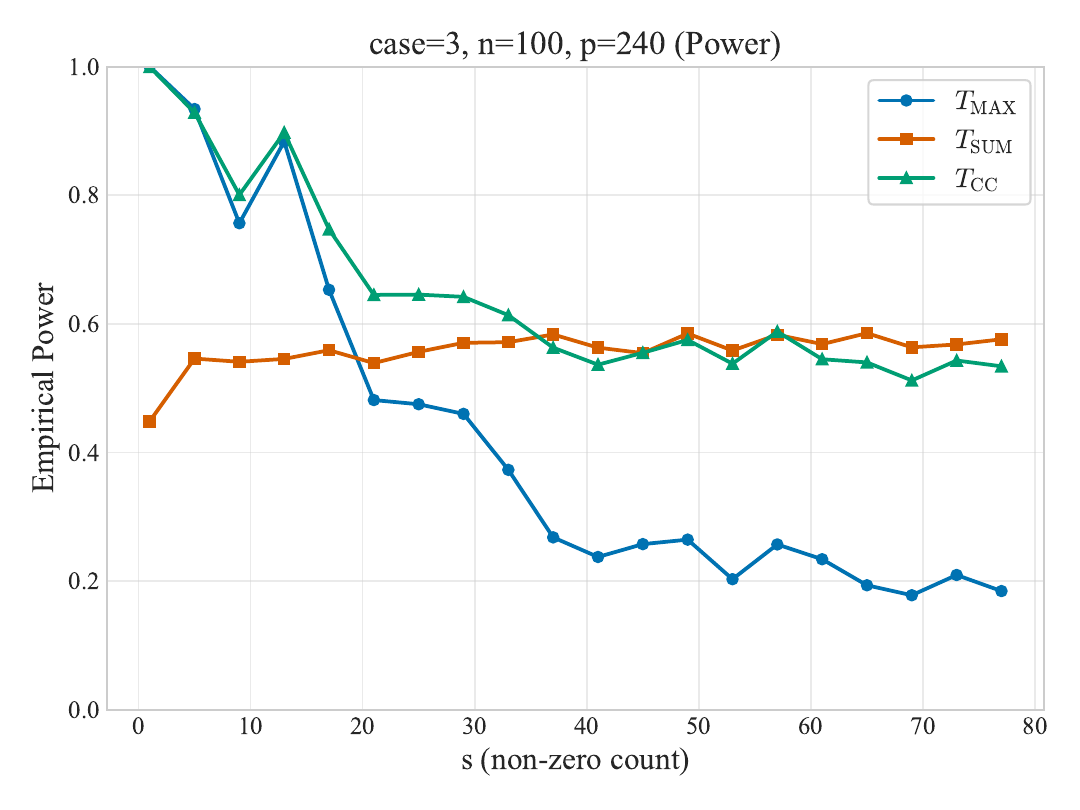}
\end{subfigure}
\hfill
\begin{subfigure}{0.23\textwidth}
    \centering
    \includegraphics[width=\linewidth]{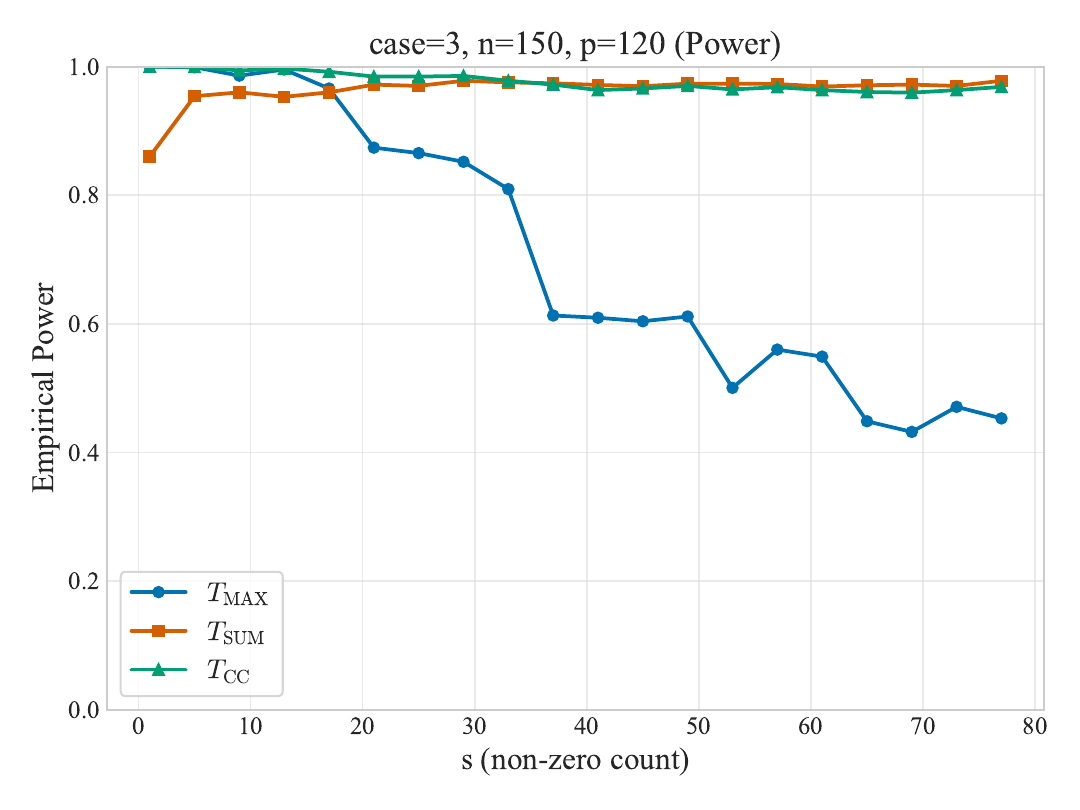}
\end{subfigure}
\hfill
\begin{subfigure}{0.23\textwidth}
    \centering
    \includegraphics[width=\linewidth]{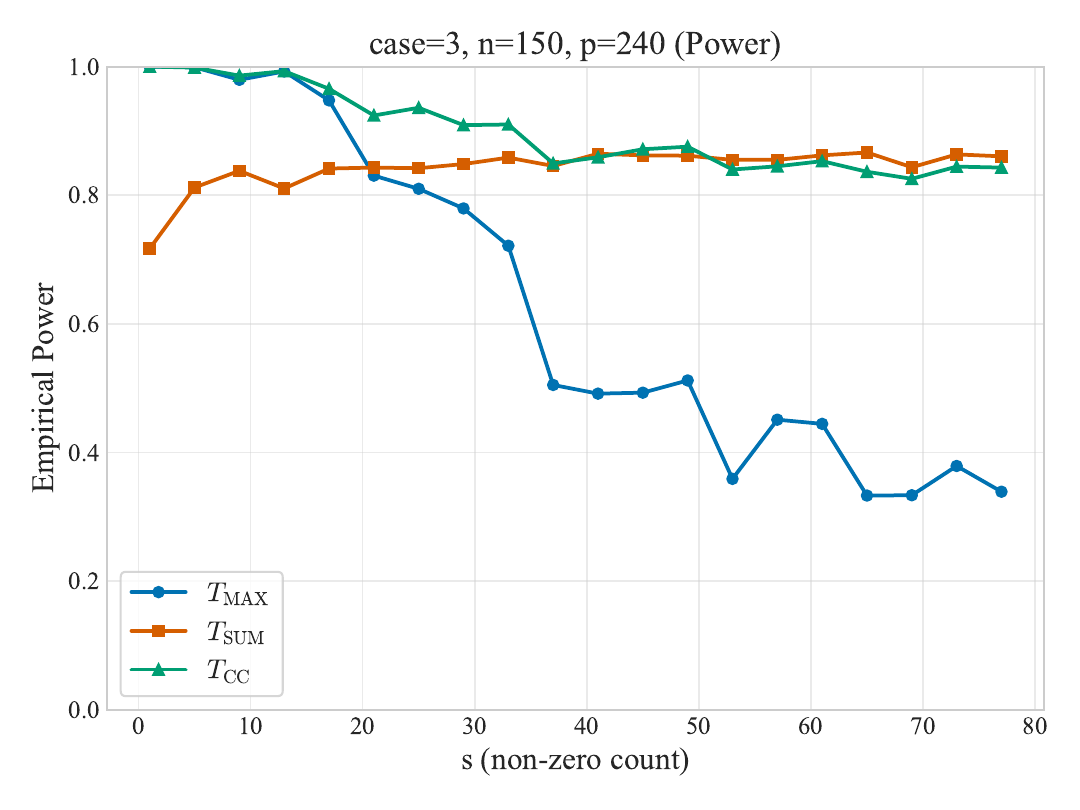}
\end{subfigure}

\caption{Empirical power as a function of $s$ for Cases~1--3 across varying $(n,p)$ 
settings under Laplace distribution ($\tau = 0.75$; 2000 replications).}
\label{fig:power75_laplace}
\end{figure}

\begin{figure}[htbp]
\centering
\begin{subfigure}{0.23\textwidth}
    \centering
    \includegraphics[width=\linewidth]{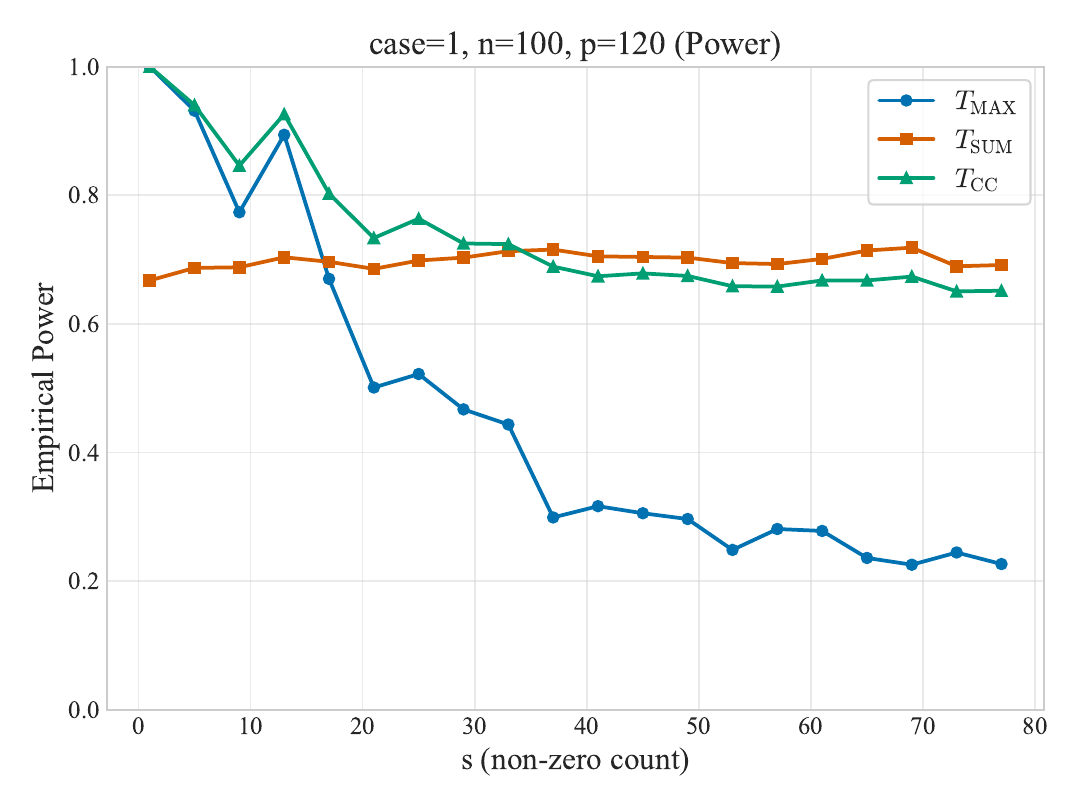}
\end{subfigure}
\hfill
\begin{subfigure}{0.23\textwidth}
    \centering
    \includegraphics[width=\linewidth]{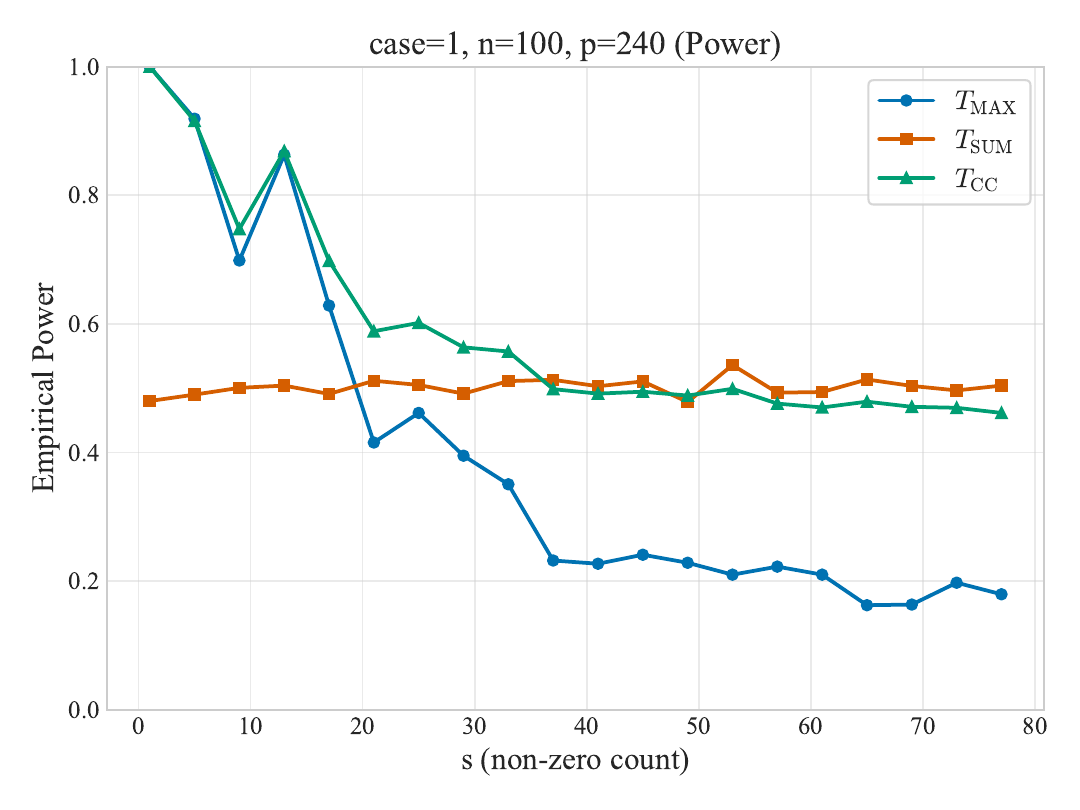}
\end{subfigure}
\hfill
\begin{subfigure}{0.23\textwidth}
    \centering
    \includegraphics[width=\linewidth]{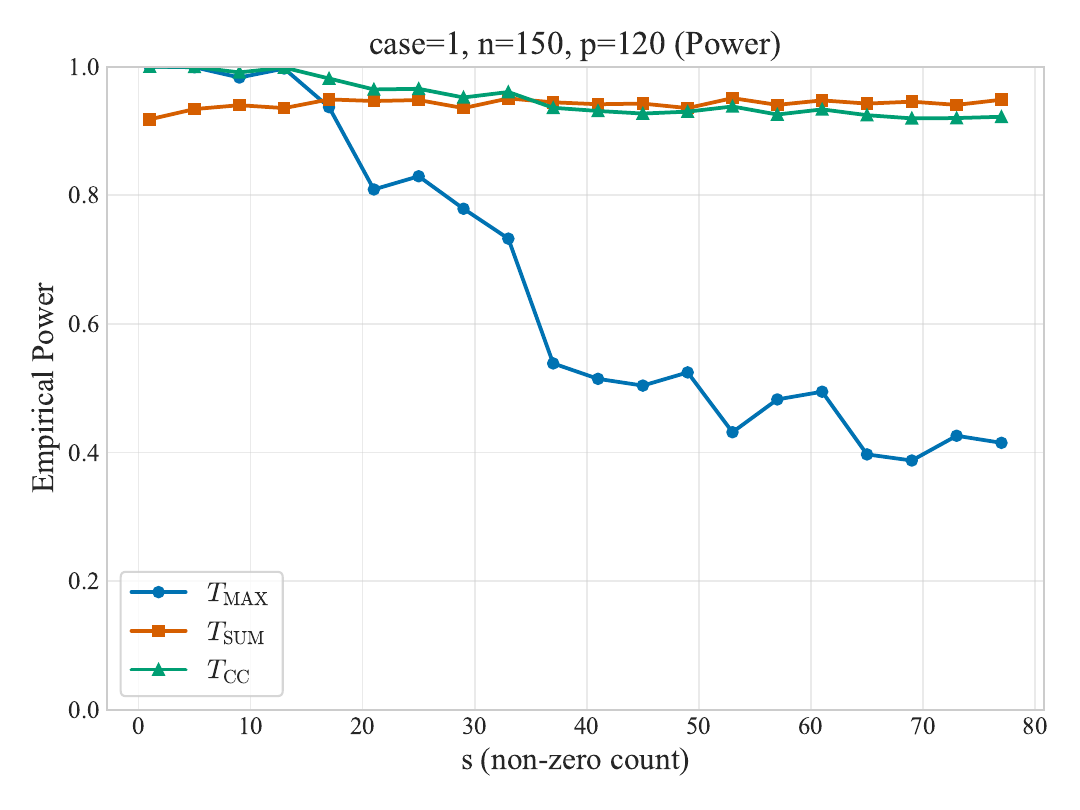}
\end{subfigure}
\hfill
\begin{subfigure}{0.23\textwidth}
    \centering
    \includegraphics[width=\linewidth]{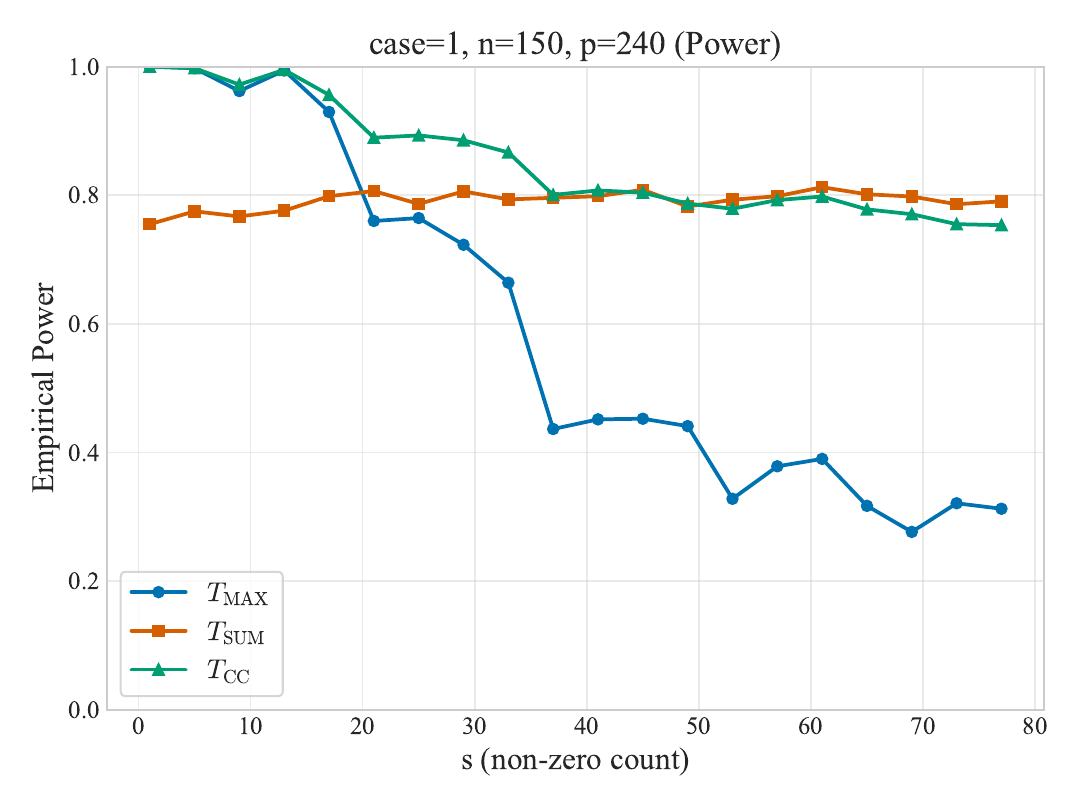}
\end{subfigure}

\vspace{0.3cm}
\begin{subfigure}{0.23\textwidth}
    \centering
    \includegraphics[width=\linewidth]{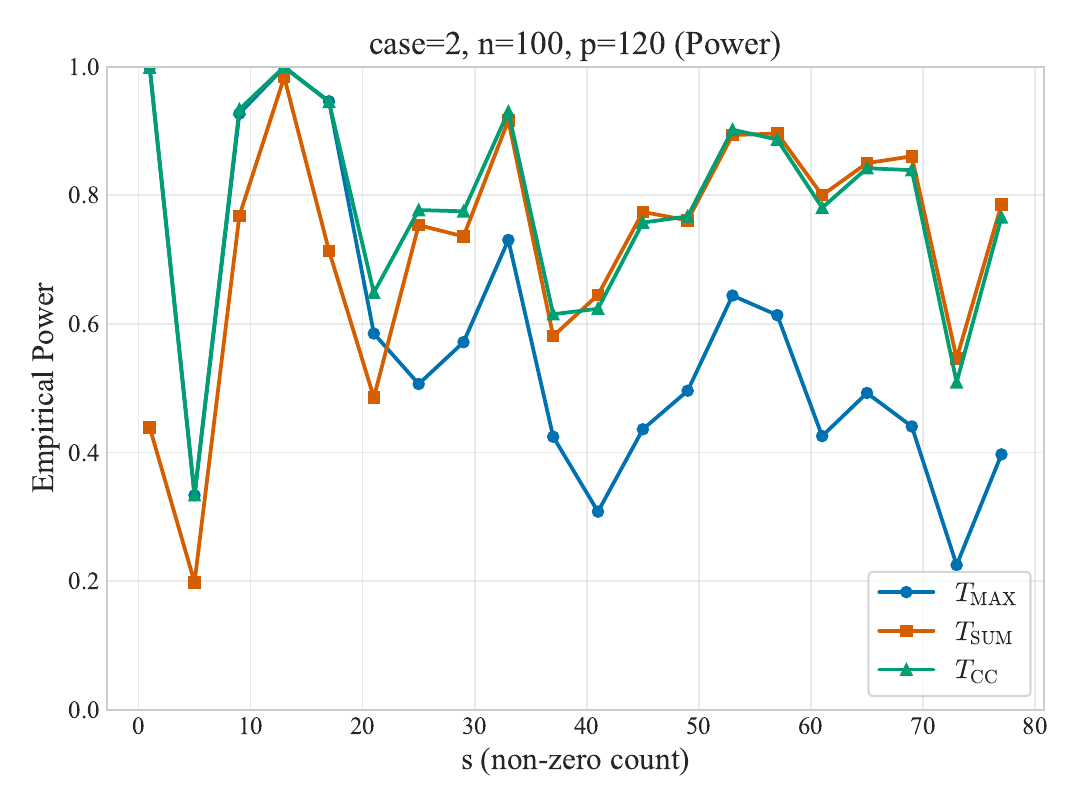}
\end{subfigure}
\hfill
\begin{subfigure}{0.23\textwidth}
    \centering
    \includegraphics[width=\linewidth]{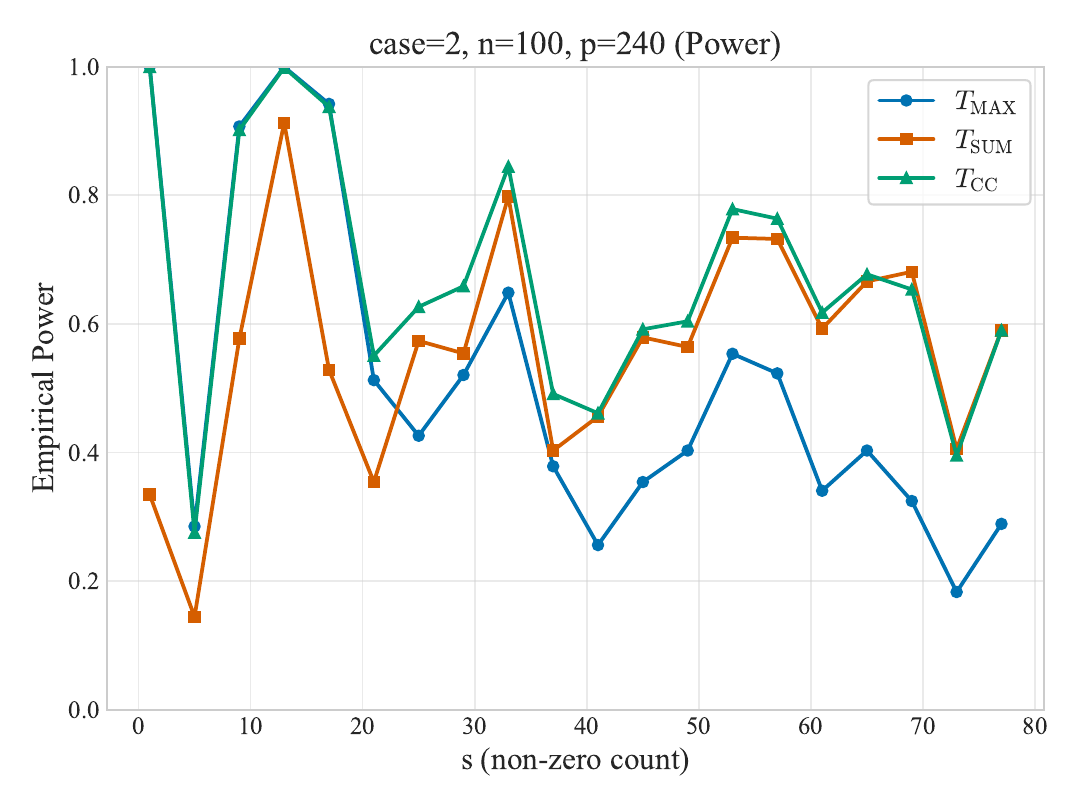}
\end{subfigure}
\hfill
\begin{subfigure}{0.23\textwidth}
    \centering
    \includegraphics[width=\linewidth]{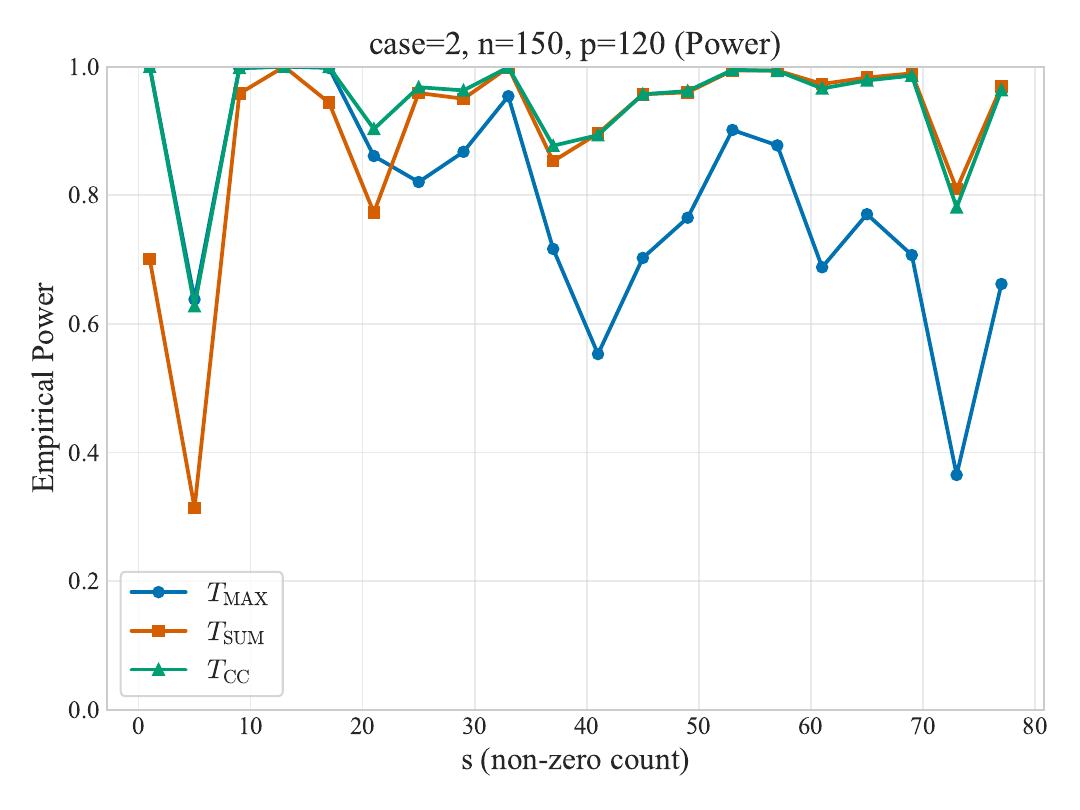}
\end{subfigure}
\hfill
\begin{subfigure}{0.23\textwidth}
    \centering
    \includegraphics[width=\linewidth]{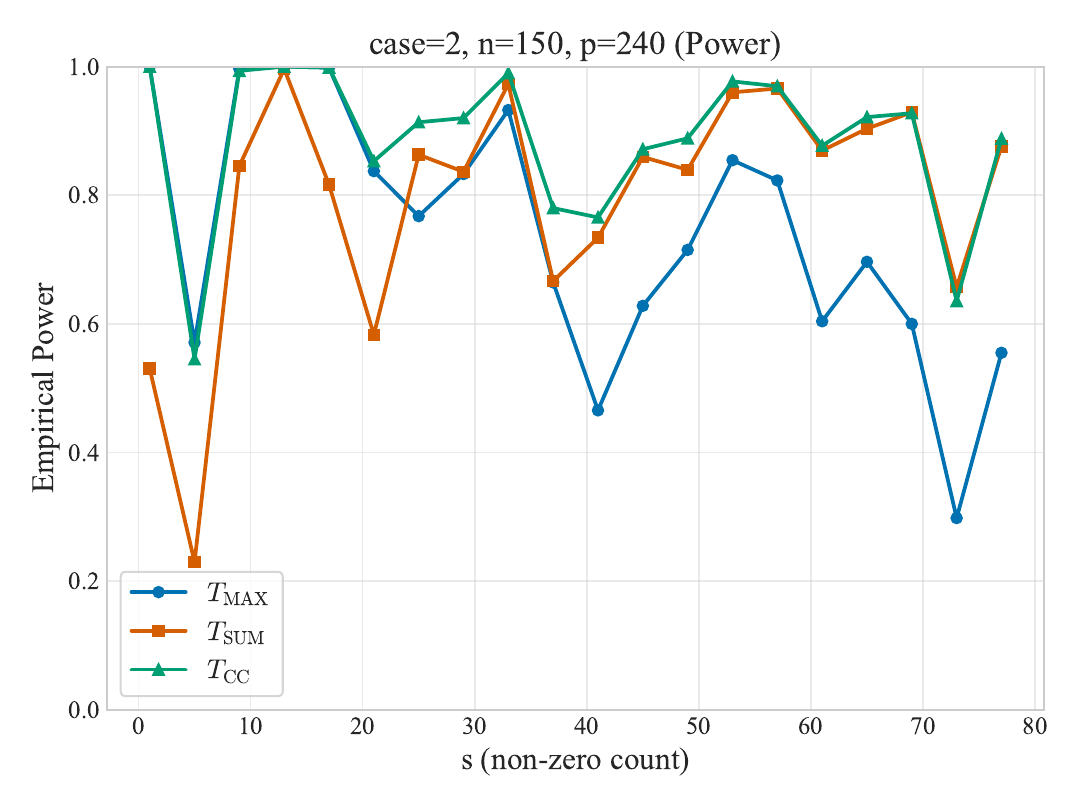}
\end{subfigure}

\vspace{0.3cm}
\begin{subfigure}{0.23\textwidth}
    \centering
    \includegraphics[width=\linewidth]{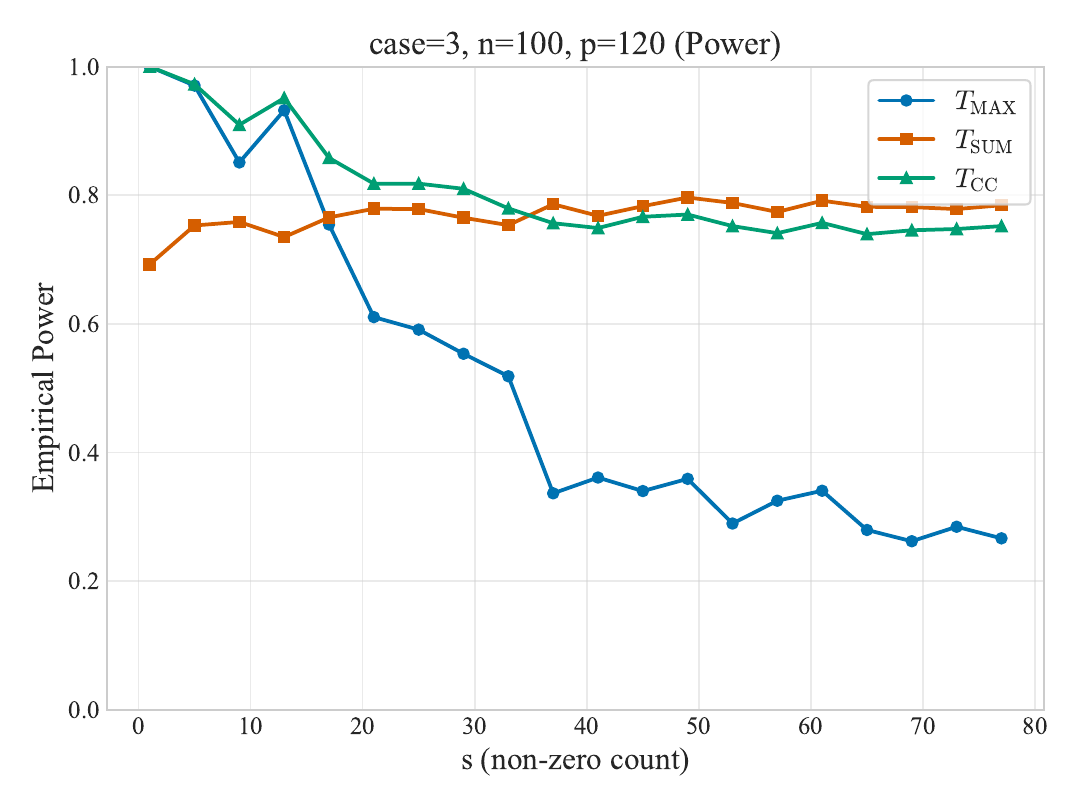}
\end{subfigure}
\hfill
\begin{subfigure}{0.23\textwidth}
    \centering
    \includegraphics[width=\linewidth]{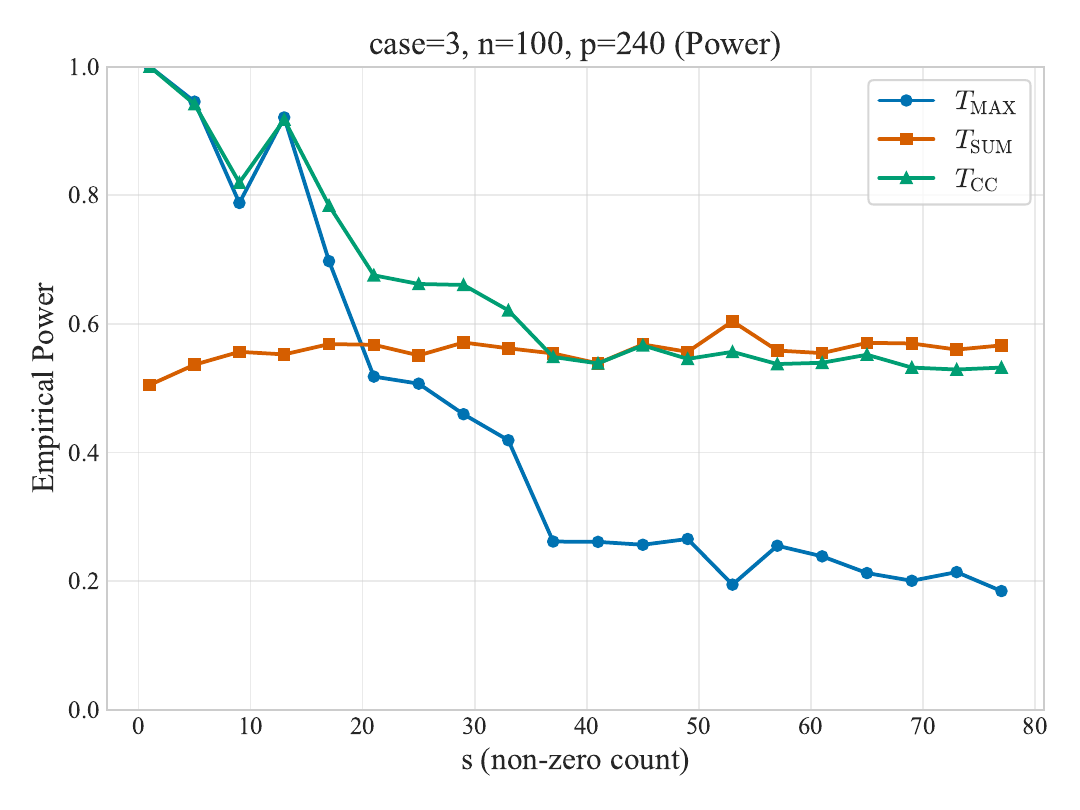}
\end{subfigure}
\hfill
\begin{subfigure}{0.23\textwidth}
    \centering
    \includegraphics[width=\linewidth]{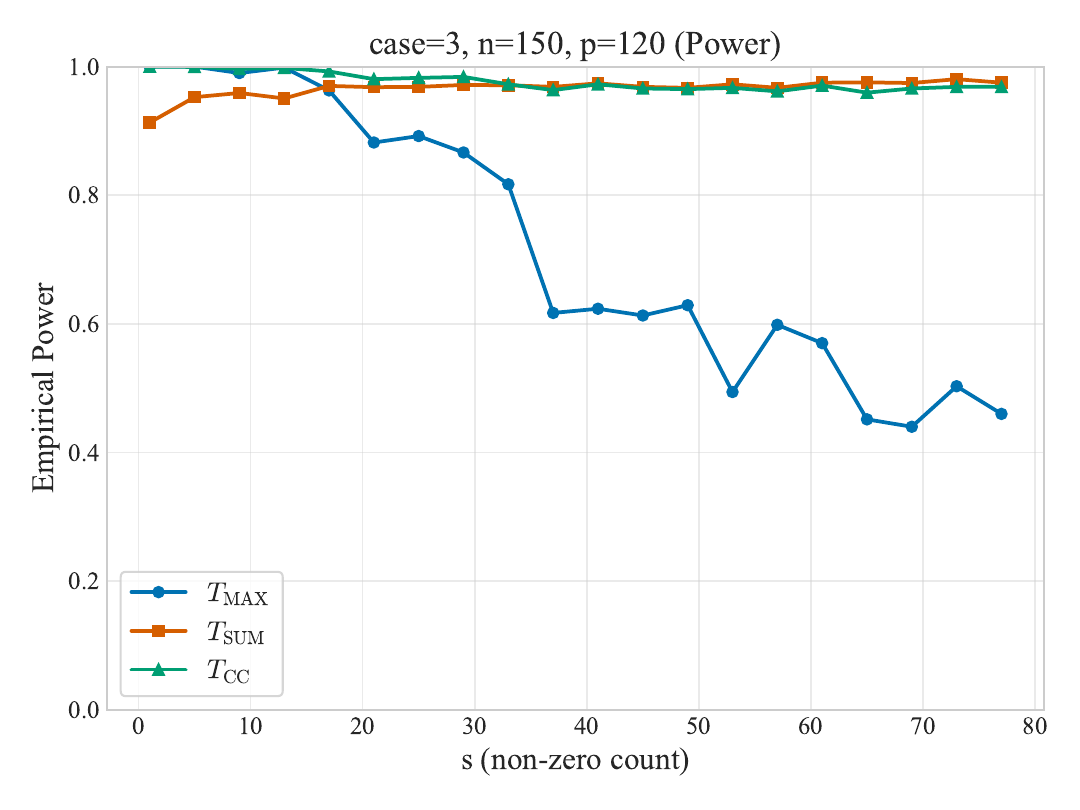}
\end{subfigure}
\hfill
\begin{subfigure}{0.23\textwidth}
    \centering
    \includegraphics[width=\linewidth]{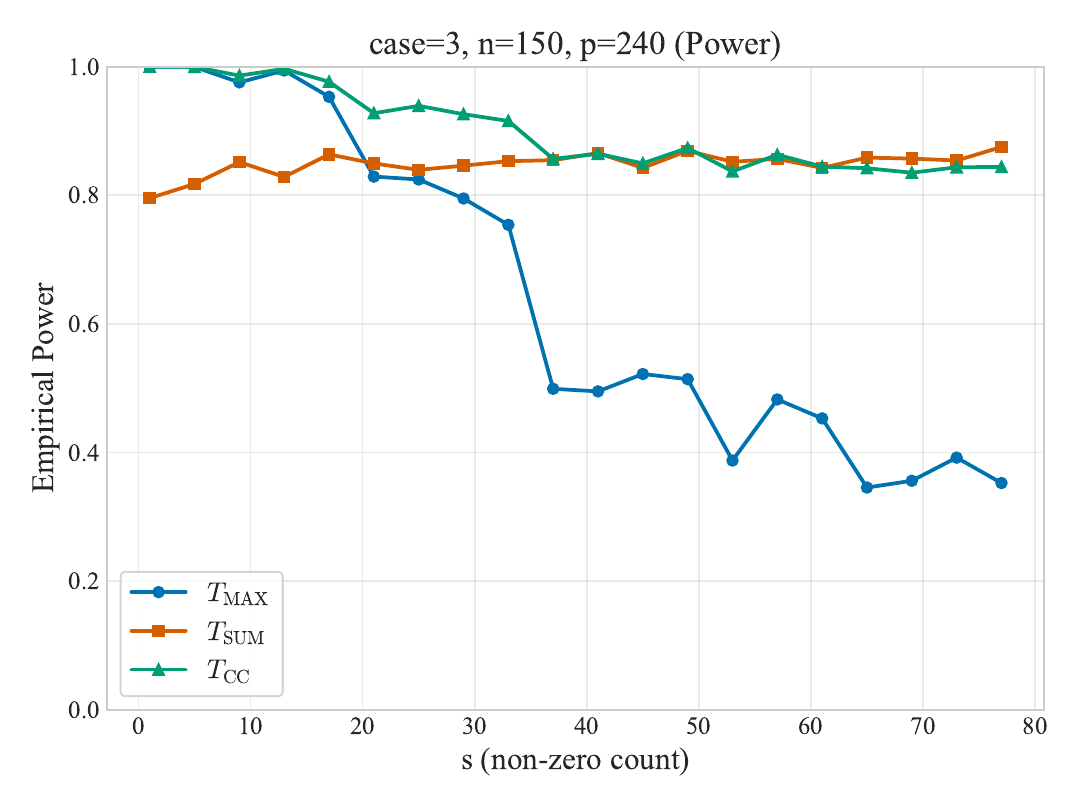}
\end{subfigure}

\caption{Empirical power as a function of $s$ for Cases~1--3 across varying $(n,p)$ 
settings under Logistic distribution ($\tau = 0.75$; 2000 replications).}
\label{fig:power75_logistic}
\end{figure}

\begin{figure}[htbp]
\centering
\begin{subfigure}{0.23\textwidth}
    \centering
    \includegraphics[width=\linewidth]{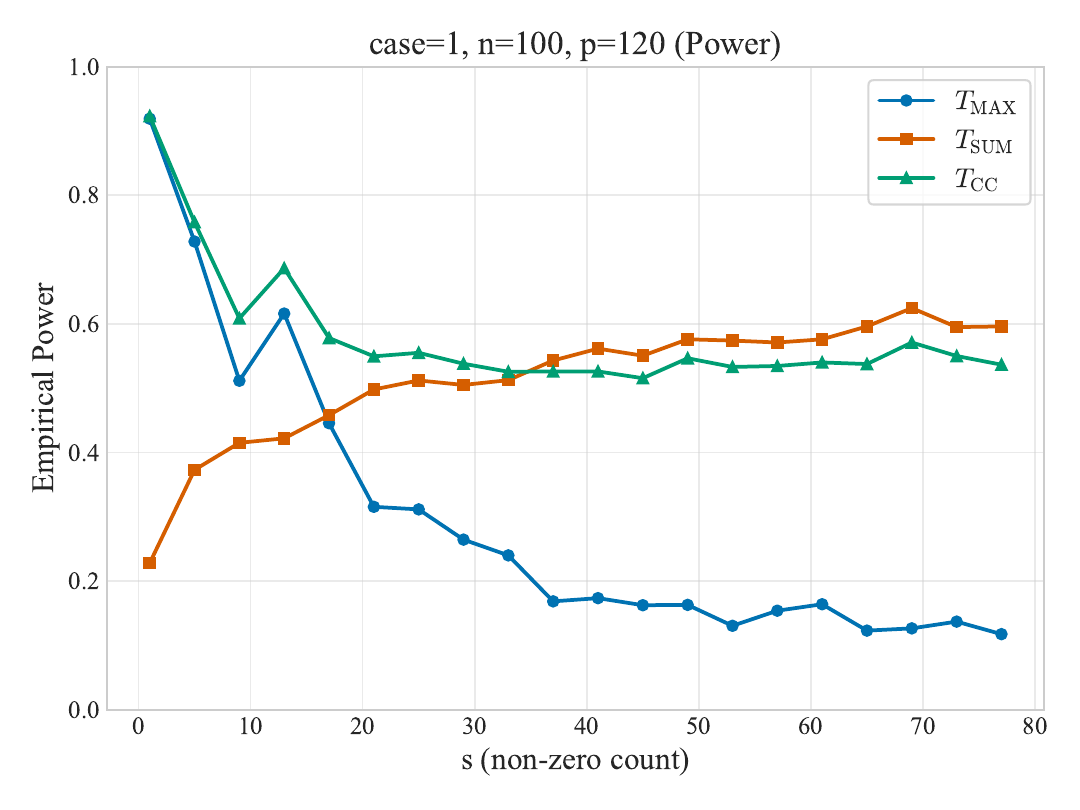}
\end{subfigure}
\hfill
\begin{subfigure}{0.23\textwidth}
    \centering
    \includegraphics[width=\linewidth]{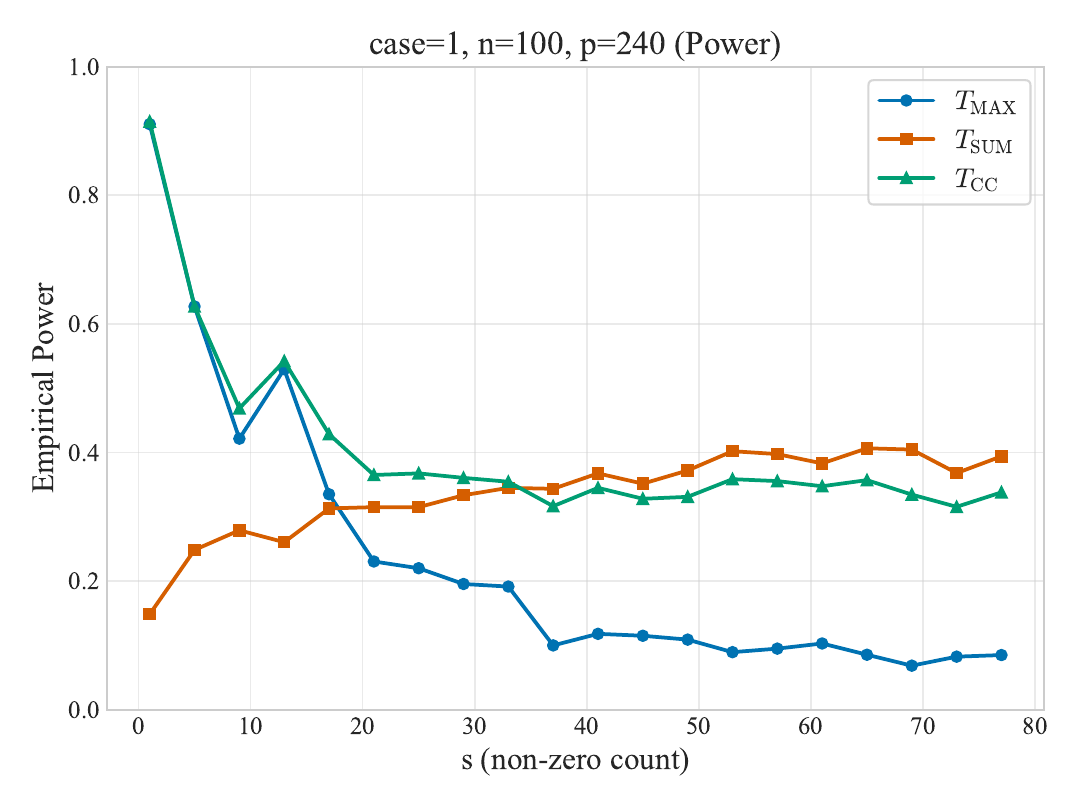}
\end{subfigure}
\hfill
\begin{subfigure}{0.23\textwidth}
    \centering
    \includegraphics[width=\linewidth]{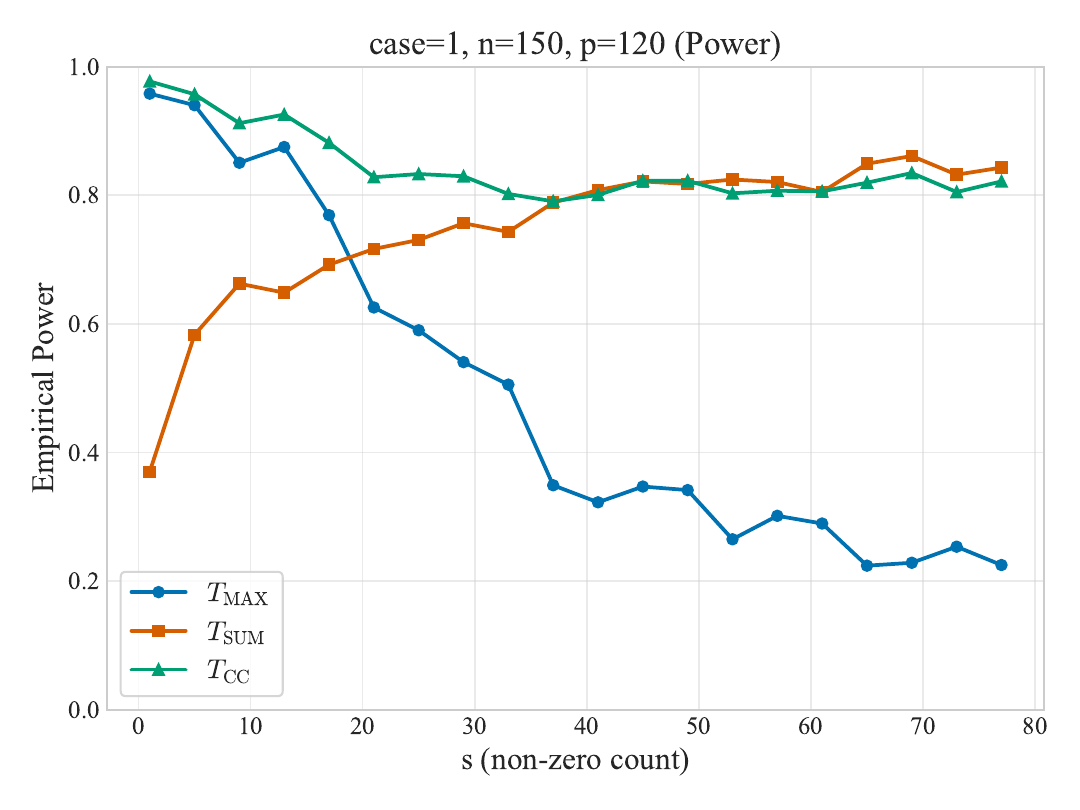}
\end{subfigure}
\hfill
\begin{subfigure}{0.23\textwidth}
    \centering
    \includegraphics[width=\linewidth]{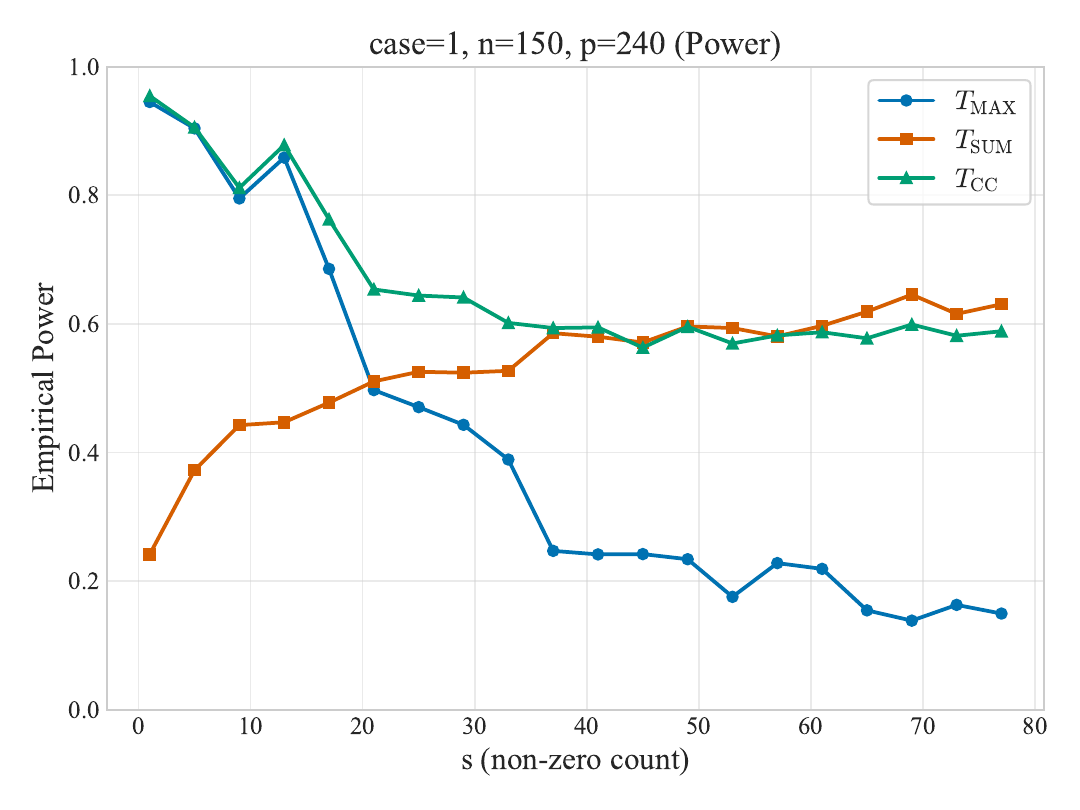}
\end{subfigure}

\vspace{0.3cm}
\begin{subfigure}{0.23\textwidth}
    \centering
    \includegraphics[width=\linewidth]{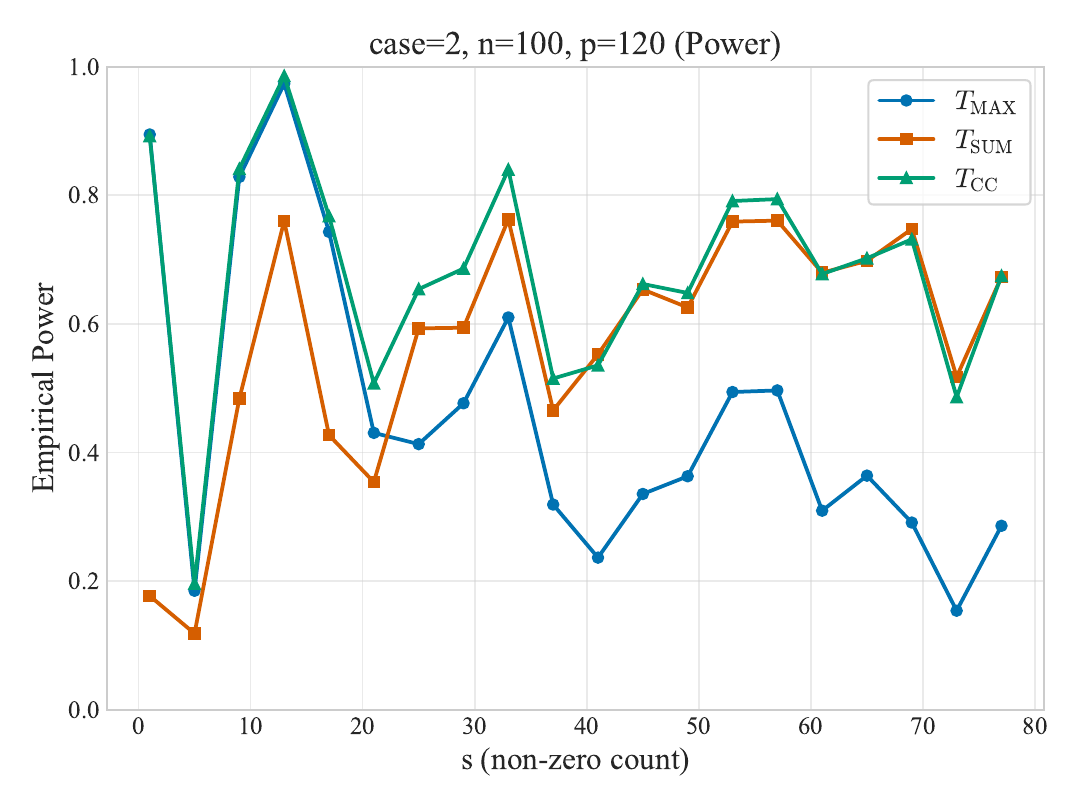}
\end{subfigure}
\hfill
\begin{subfigure}{0.23\textwidth}
    \centering
    \includegraphics[width=\linewidth]{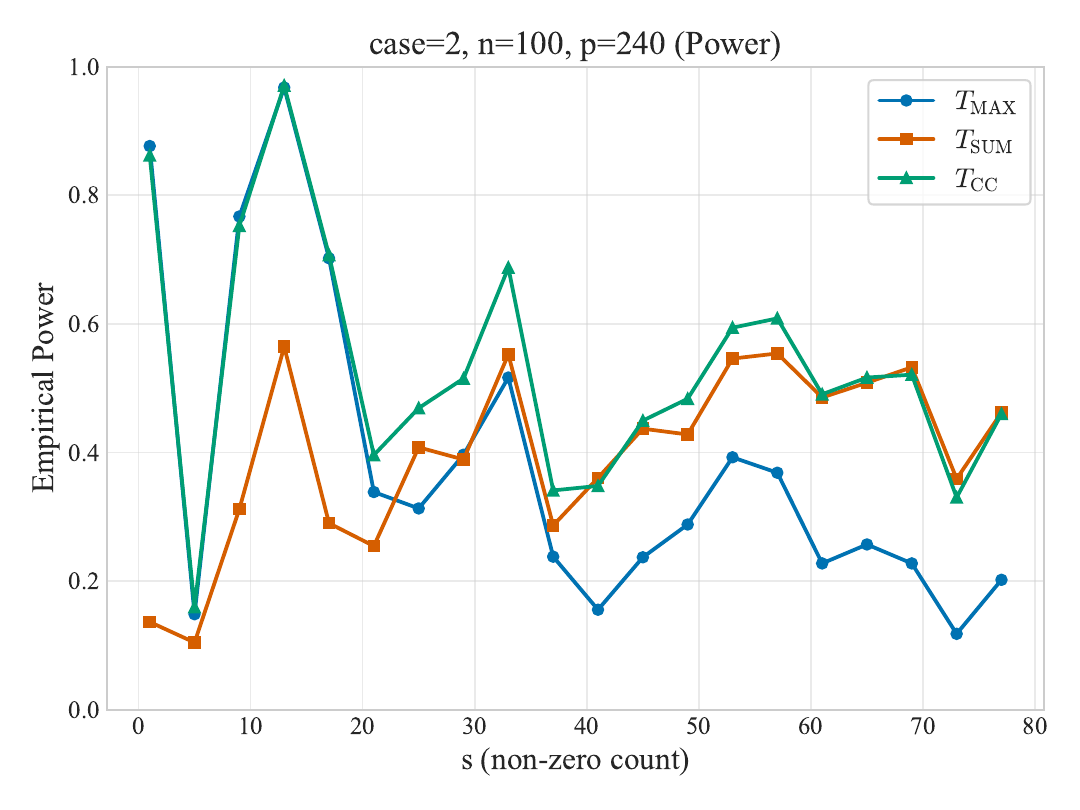}
\end{subfigure}
\hfill
\begin{subfigure}{0.23\textwidth}
    \centering
    \includegraphics[width=\linewidth]{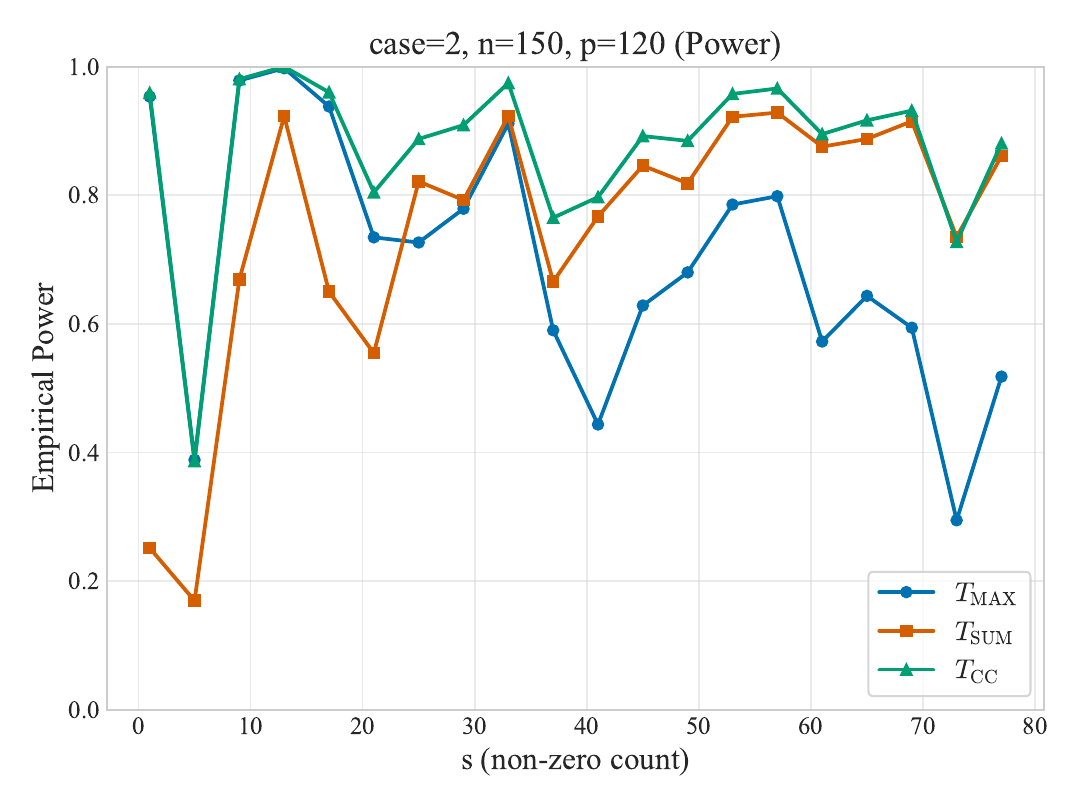}
\end{subfigure}
\hfill
\begin{subfigure}{0.23\textwidth}
    \centering
    \includegraphics[width=\linewidth]{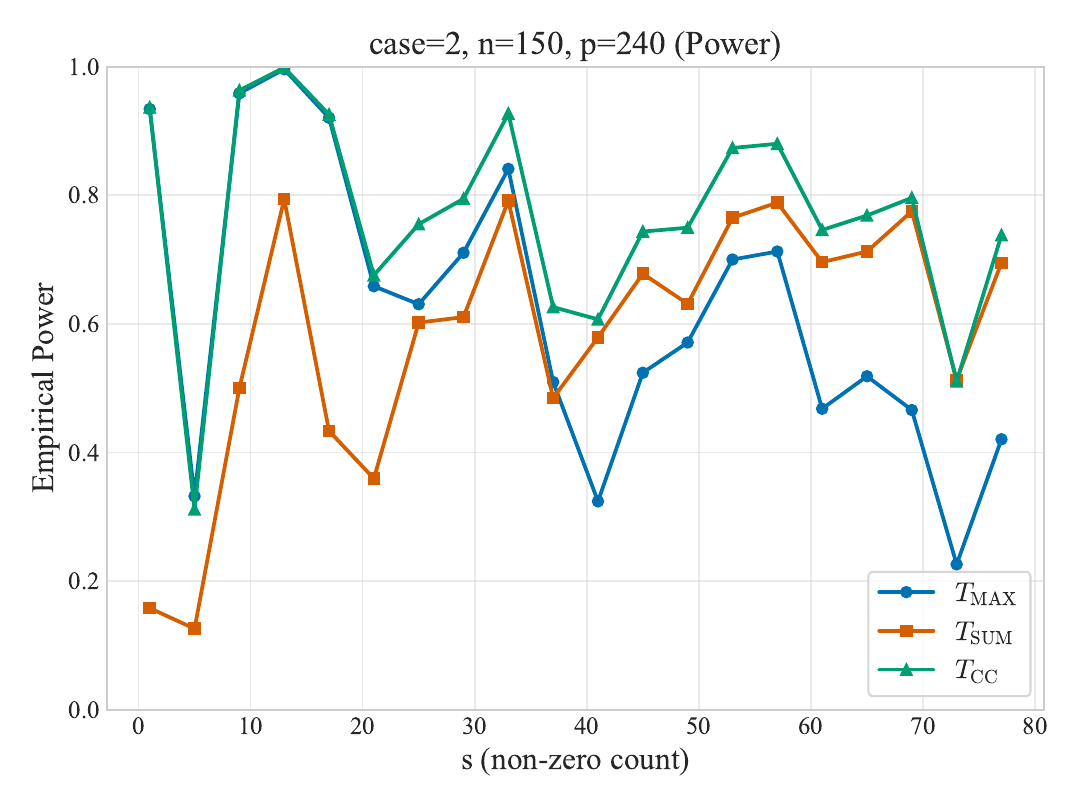}
\end{subfigure}

\vspace{0.3cm}
\begin{subfigure}{0.23\textwidth}
    \centering
    \includegraphics[width=\linewidth]{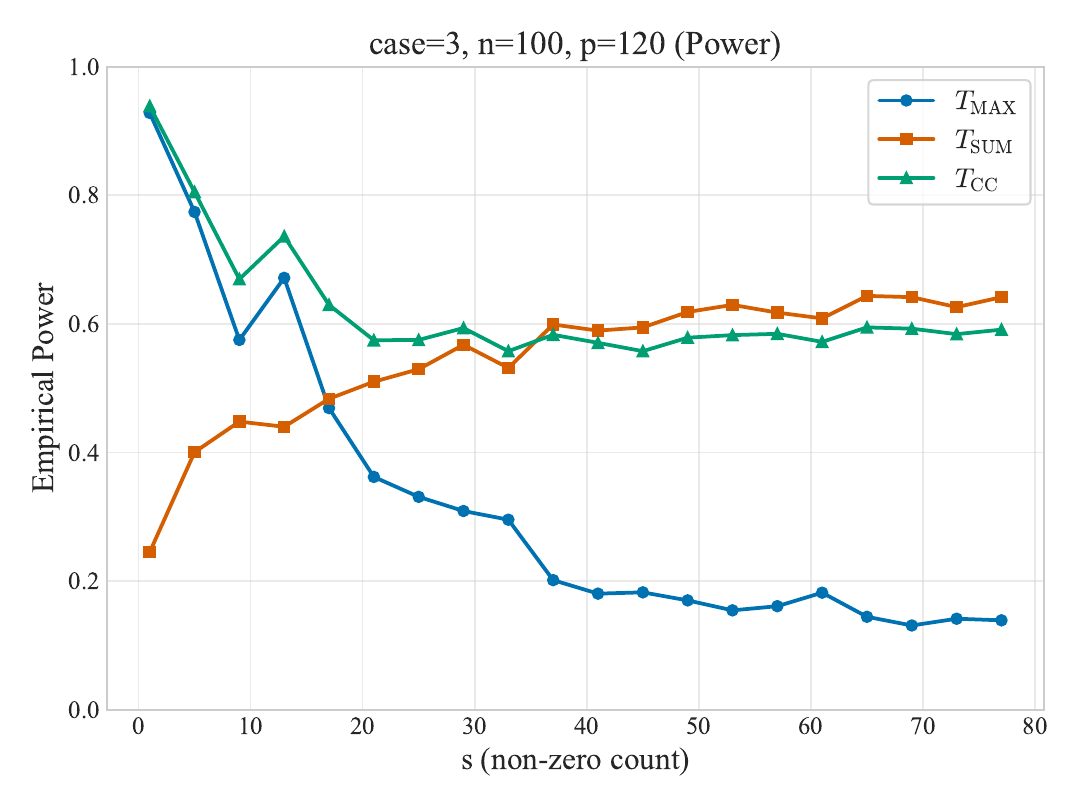}
\end{subfigure}
\hfill
\begin{subfigure}{0.23\textwidth}
    \centering
    \includegraphics[width=\linewidth]{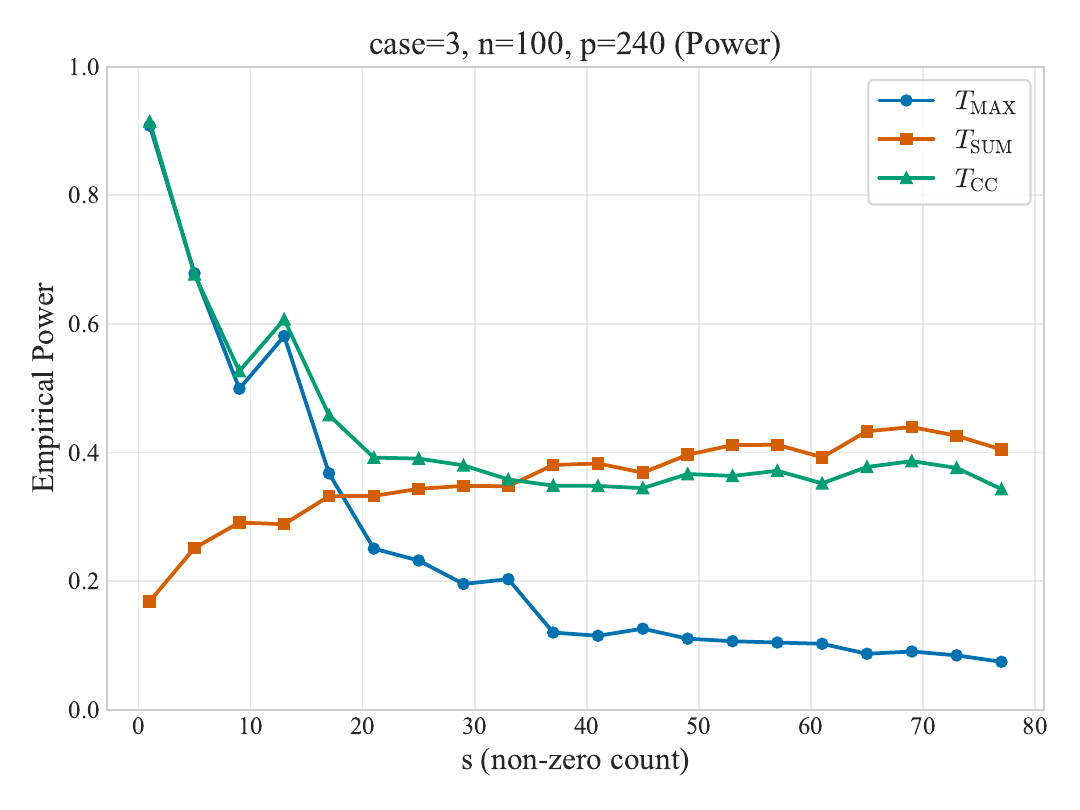}
\end{subfigure}
\hfill
\begin{subfigure}{0.23\textwidth}
    \centering
    \includegraphics[width=\linewidth]{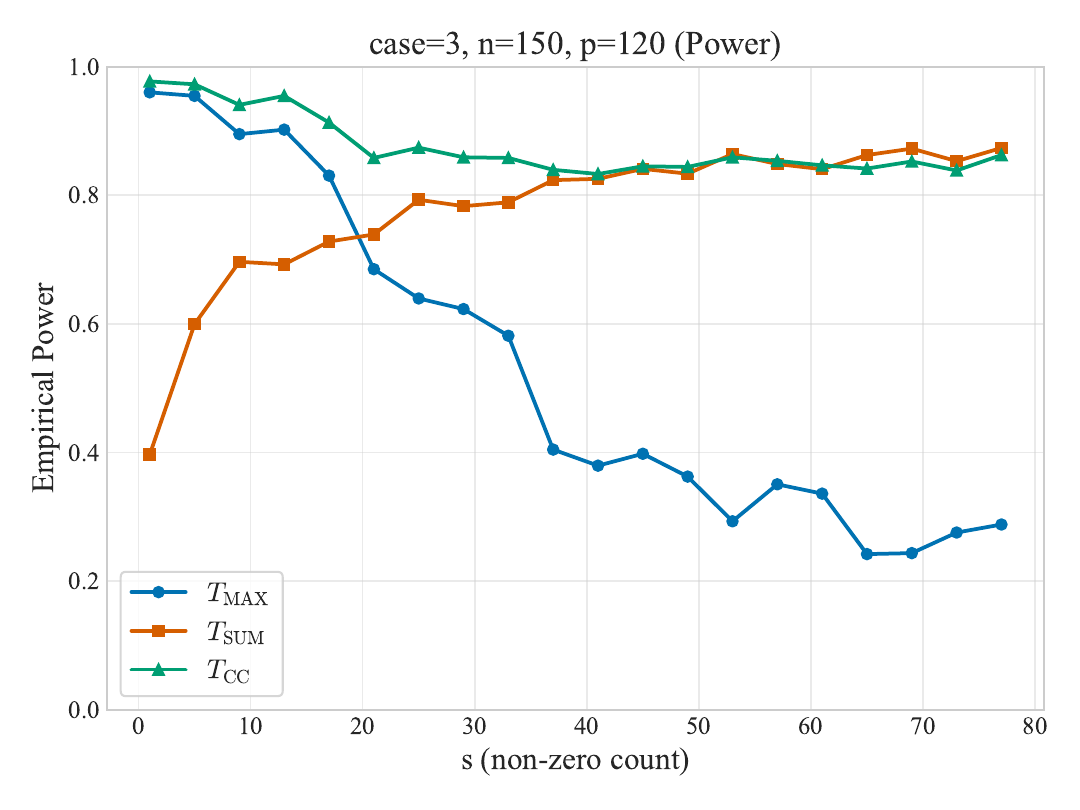}
\end{subfigure}
\hfill
\begin{subfigure}{0.23\textwidth}
    \centering
    \includegraphics[width=\linewidth]{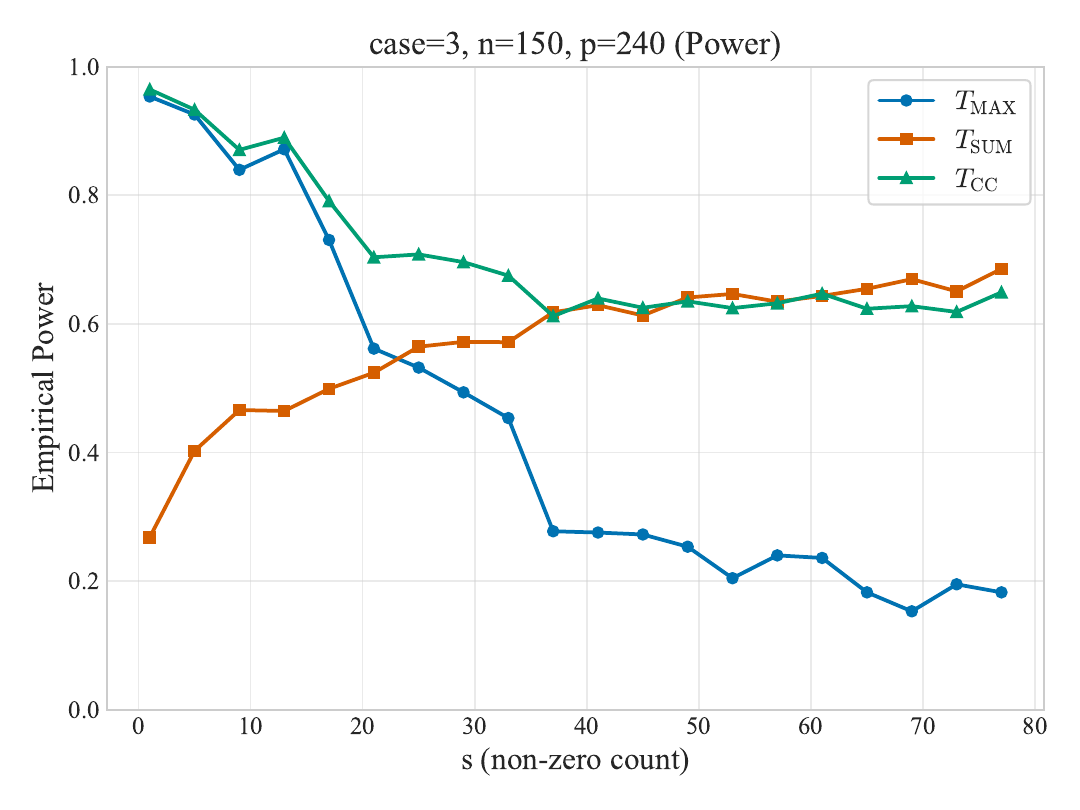}
\end{subfigure}

\caption{Empirical power as a function of $s$ for Cases~1--3 across varying $(n,p)$ 
settings under $t_2$ distribution ($\tau = 0.75$; 2000 replications).}
\label{fig:power75_t2}
\end{figure}


\end{document}